\documentclass[11pt]{article}
\usepackage[a4paper,margin=27mm]{geometry}
\usepackage{amsmath,amssymb,amsthm}
\usepackage{bm}
\usepackage{graphicx}
\usepackage{booktabs}
\usepackage{tabularx}
\usepackage{placeins}
\usepackage{needspace}

\newcolumntype{Y}{>{\raggedright\arraybackslash}X}
\usepackage{tikz}
\usetikzlibrary{arrows.meta,positioning}
\usepackage[
  colorlinks=true,
  linkcolor=black,
  citecolor=black,
  urlcolor=black,
  pdftitle={Partial Stabilizer Learning under Fixed Commuting Constraints and the Absence of a Copy-Rate Discount},
  pdfauthor={Masahito Hayashi and Yimin Lu},
  pdfsubject={Partial stabilizer learning under fixed commuting constraints},
  pdfkeywords={stabilizer learning, quantum state identification, verification, Clifford covariance, sample complexity}
]{hyperref}
\allowdisplaybreaks
\newcommand{\F}{\mathbb F}

\newcommand{\Sym}{\operatorname{Sym}}
\newcommand{\rank}{\operatorname{rank}}
\newcommand{\im}{\operatorname{Im}}

\newcommand{\Tr}{\operatorname{Tr}}

\newcommand{\1}{\mathbf 1}

\newcommand{\supp}{\operatorname{supp}}
\newcommand{\PGM}{\mathrm{PGM}}
\newcommand{\spair}[2]{\langle #1,#2\rangle_{\mathrm{sp}}}

\newcommand{\ket}[1]{\lvert #1\rangle}
\newcommand{\bra}[1]{\langle #1\rvert}
\newcommand{\braket}[2]{\langle #1\mid #2\rangle}

\theoremstyle{plain}
\newtheorem{theorem}{Theorem}[section]
\newtheorem{lemma}[theorem]{Lemma}
\newtheorem{proposition}[theorem]{Proposition}
\newtheorem{corollary}[theorem]{Corollary}
\theoremstyle{definition}
\newtheorem{definition}[theorem]{Definition}
\theoremstyle{remark}
\newtheorem{remark}[theorem]{Remark}

\title{Partial Stabilizer Learning under Fixed Commuting Constraints and the Absence of a Copy-Rate Discount}
\author{%
\parbox{0.95\textwidth}{\centering
Masahito Hayashi\textsuperscript{1,2,3}%
\thanks{Corresponding author: \href{mailto:hmasahito@cuhk.edu.cn}{hmasahito@cuhk.edu.cn}}
\quad
Yimin Lu\textsuperscript{4}\\[0.75em]
\small
\textsuperscript{1}School of Data Science, The Chinese University of Hong Kong,
Shenzhen, Longgang District, Shenzhen 518172 Guangdong, China.\\
\textsuperscript{2}International Quantum Academy, Futian District,
Shenzhen 518048 Guangdong, China.\\
\textsuperscript{3}Graduate School of Mathematics, Nagoya University,
Chikusa-ku, Nagoya 464--8602 Aichi, Japan.\\
\textsuperscript{4}Yau Mathematical Sciences Center, Tsinghua University,
Beijing 100084, China.}}
\date{}

\begin{document}
\maketitle

\begin{abstract}
We consider an unknown pure stabilizer state of \(n\) qudits with a fixed prime local dimension \(p\), together with prescribed commuting Pauli observables. Rather than learning the whole state, the learner must recover only complementary stabilizers that commute with the prescribed observables, together with their eigenvalues. Once the prescribed measurement is performed, this information and the observed outcome determine the corresponding conditional stabilizer state. We compare the number \(k\) of copies available to the learner with the number \(n\) of qudits and refer to their ratio \(k/n\) as the copy rate.

Suppose that the dimension \(m\) of the requested complementary stabilizer subspace satisfies \(m=\beta n+o(n)\) for a fixed \(0<\beta\leq1\). If the copy rate converges to a value below one, the optimal exact recovery probabilities and verification scores converge to zero, both in the worst case and under the uniform prior. Conversely, if the number \(k\) of copies minus the number \(n\) of qudits tends to positive infinity, the entire stabilizer state can be identified with probability tending to one, and the requested complementary information can then be extracted. 

Thus, restricting the learning target to the stabilizer information needed after the prescribed measurement yields no copy-rate discount.
\end{abstract}
\section{Introduction}
\label{sec:introduction}
Many quantum-information tasks ask for selected information about an unknown
state.  Examples include estimating prescribed measurement probabilities
\cite{Aaronson2018ShadowTomography,BadescuODonnell2024ImprovedQuantumData,
KingEtAl2025TriplyEfficient,Sinha2025DimensionIndependent}, predicting
state properties from a classical record
\cite{HuangKuengPreskill2020ClassicalShadows,ChenYuZengFlammia2021RobustShadows,
HadfieldEtAl2022LocallyBiasedShadows,BertoniEtAl2024ShallowShadows}, and
reconstructing specified reduced states
\cite{CotlerWilczek2020OverlappingTomography,BonetMonroigBabbushOBrien2020PartialTomography,
AraujoEtAl2022LocalOverlappingTomography,HansenneEtAl2025OptimalOverlapping}.
These tasks can use fewer samples or measurement settings than full
tomography in their respective models.  Testing and learning Pauli-structured
Hamiltonians from queries to their time evolution is another distinct
quantum-learning model \cite{ArunachalamDuttEscudero2026Hamiltonians}.
Here the resource is copies of an unknown state, not queries to a time-evolution
operator.  We ask whether a smaller exact description of a stabilizer state
can also reduce the leading number of quantum copies needed to learn it.

Let $x$ label an unknown pure stabilizer state of $n$ qudits with fixed
prime local dimension $p$.  Its stabilizer directions form an
$n$-dimensional Lagrangian subspace $L_x$ of the symplectic space
$\F_p^{2n}$, with eigenvalue assignment $\chi_x$.  A prescribed commuting
Pauli measurement is specified by an isotropic subspace $W$ of dimension
$n-m$.  Here isotropic means that the symplectic pairing vanishes on the
subspace, so the corresponding Pauli operators commute.  The learner
knows $W$ and receives $k$ copies of the initial state.

The learner reports an $m$-dimensional subspace $K$ of initial
stabilizers commuting with $W$, together with their eigenvalues.  The
required subspace satisfies
\[
 K\subseteq L_x\cap W^\perp,\qquad K\cap W=\{0\}.
\]
After the prescribed measurement, the observed eigenvalues on $W$ and
the reported eigenvalues on $K$ specify the conditional pure stabilizer
state.  Thus the report need not describe the entire initial state.
It need not determine the distribution of measurement outcomes either:
that distribution may also depend on $L_x\cap W$ and its character.
Section~\ref{sec:fixed-W-main-results} defines the report alphabet and
the chronological order of the experiment.

We evaluate two payoffs.  Exact learning accepts every report whose
subspace and eigenvalues are correct for the input.  Postmeasurement
verification tests the reported conditional-state description on an
independent copy and assigns its Born acceptance probability.  Each
payoff is optimized both in the worst case and under the uniform prior
on distinct pure stabilizer states.  Arbitrary collective measurements
on the learner's $k$ copies are allowed.

Our main theorem shows that all four criteria have first-order copy
threshold one whenever $m=\beta n+o(n)$ with fixed $0<\beta\le1$.
If $k=\alpha n+o(n)$ with $\alpha<1$, all four optimal values tend to
zero.  If $k-n\to+\infty$, complete stabilizer identification succeeds
with probability tending to one; extracting a complementary report
then makes all four values tend to one.  Theorem~\ref{thm:fixed-W-no-discount}
states the bounds and their uniformity over prescribed constraints.
The conclusion concerns this complementary-stabilizer report.
The present theorem does not determine the full success profile in the
fixed-offset window $k=n+O(1)$ or the copy requirements for sublinear
targets $m=o(n)$.

Previous complete-stabilizer learning procedures use linearly many
copies \cite{AaronsonGottesman2008,Montanaro2017BellSampling}.
For qubits, recent work gives sample-optimal complete-state identification,
including precise additive bounds near $n$ copies
\cite{ChangCaroLaroccaWest2026Stabilizer}; qudit Bell-sampling methods also
learn complete stabilizer states with $O(n)$ copies
\cite{AllcockDoriguelloIvanyosSantha2026Qudits}.
Adaptive single-copy measurements attain linear-copy complete-state learning
for qubits \cite{BittelEisertGongMeleSchatzki2026Adaptive}.
Learning stabilizer groups through the abelian state hidden subgroup problem
is a related symmetry-recovery task
\cite{HinscheEisertCarrasco2026StateHSP}.
Our converse addresses a different, smaller output: only the complementary
directions and their eigenvalues must be recovered under the prescribed
commuting measurement.  The complete-state learning bounds do not by
themselves give a converse for this reduced report.  Even this output
requires limiting copy rate at least one when $m$ is linear in $n$.
Stabilizer testing and known-target verification instead ask whether a
state has a property or is close to a specified target
\cite{HayashiMorimae2015StabilizerTesting,DangniamHanZhu2020Verification,
HinscheHelsen2025Testing}.  Limited-memory learning keeps the complete
identification target and restricts the measurement architecture
\cite{ArunachalamSchatzki2026LimitedMemory}.  Here the output is reduced
and the learner measurement is unrestricted.

The proof first evaluates an auxiliary task in which the learner also
knows an ordered basis of a Lagrangian Pauli subspace, called an auxiliary
frame.  We restrict the input stabilizer subspace to have zero intersection
with the span of that basis.  The frame then selects a unique ordered
target.  In suitable coordinates, inputs with the same target are
averaged into a mixed state, and exact learning becomes discrimination
of these conditioned states.  We prove that the pretty good measurement
(PGM) is optimal by decomposing
them into orthogonal blocks and optimizing each measurement block.
Pure geometrically uniform ensembles already have square-root
measurement optimality results \cite{EldarForney2001SRM}; the mixed
conditioned ensemble here requires the stated block argument.
As further mathematical context, Clifford tensor-power symmetries and
stabilizer-state moments have been studied through representation theory
\cite{GrossNezamiWalter2021Clifford}, and third moments of odd-prime qudit
Clifford orbits have been analyzed in
\cite{ZhuMaoYi2026Clifford}.  Those results concern orbit symmetries and
moments; our converse instead uses the probability-weighted discrimination
of mixed states conditioned on the requested complementary information.
Counting the occurring tensor-power moments, with their actual
probability weights, then bounds exact success below copy rate one.
Section~\ref{sec:converse} compares this bound with the coarser estimate
obtained from the ambient tensor-space dimension.

Two further steps connect this calculation to the original task.
We express the uniform distribution on generic inputs as a weighted mixture
of distributions conditioned on auxiliary frames.  This transfers the exact
bound to the generic original task; exceptional inputs have small total mass.
For verification, we keep the full original report alphabet.  A
report with appreciable acceptance probability yields a bounded list
of compatible exact targets; randomizing over that list compares
verification with exact learning.  Finally, complete-state
identification gives a statewise direct bound and a valid report by
classical postprocessing.

Section~\ref{sec:fixed-W-main-results} gives the operational definitions
and the main theorem.  Section~\ref{sec:model-main-results} compares the
original and frame-assisted tasks.  Sections~\ref{sec:odd-prime-foundations}--%
\ref{sec:partial-learning} construct their coordinate realizations and
prove the coordinate equivalences.
Sections~\ref{sec:common-finite-size-framework} and \ref{sec:converse}
derive the finite exact formula and the converse exponent.
Section~\ref{sec:integrated-all-prime-graded} proves the verification
converses, and Section~\ref{sec:fourier-achievability} gives achievability
and completes the main theorem.  The remaining sections treat fixed-diagonal
graph-basis submodels and the unresolved regimes.
\textbf{The main theorem, Theorem~\ref{thm:fixed-W-no-discount},
has been formalized in Lean~4}~\cite{stabilizerIdentificationLeanV1}.

\section{Operational tasks and main theorem}
\label{sec:fixed-W-main-results}

We begin with the operational problem that motivates the paper.  An
$(n-m)$-dimensional isotropic subspace $W\subseteq\F_p^{2n}$ is fixed before
the unknown state is presented.  It specifies the commuting Pauli or Weyl
observables that will be measured.  The learner is not asked to reconstruct the
entire initial stabilizer state.  Instead, it must recover complementary
stabilizer information that, together with the observed measurement outcome,
determines the resulting conditional state.  This requested complement need
not include every initial stabilizer that commutes with the prescribed
measurement.  The central question is whether this reduced output can be
learned with a copy rate below that required for complete stabilizer
identification.

We formulate two versions of this task.  The first is exact learning: the
learner must output valid complementary stabilizers and their character.  The
second is an operational verification game: a separate verifier tests the
learner's reported postmeasurement stabilizer description on an independent
conditional-state copy.  Together, the exact and verification payoffs define four optimal values:
worst-case and uniform-average performance for each payoff.

\subsection{Complementary stabilizer information after a fixed measurement}
\label{subsec:fixed-W-task}
\label{subsec:fixed-W-conditional-sufficiency}

Let $\mathsf{Stab}_{n,p}$ denote the set of projective pure stabilizer states
of $n$ qudits of local dimension $p$.  For $x\in\mathsf{Stab}_{n,p}$, let
$\rho_x$ denote the corresponding rank-one density operator, and let $L_x$
and $\chi_x$ denote its stabilizer subspace and eigenvalue character,
respectively.  
Writing phase-space vectors as $z=\binom{u}{v}$ and $z'=\binom{u'}{v'}$ with $u,v,u',v'\in\F_p^n$, we equip $\F_p^{2n}$ with the alternating bilinear form 
\begin{equation} 
\spair{z}{z'} :=u^Tv'-v^Tu'. \label{eq:global-symplectic-form} 
\end{equation} 
A Pauli or Weyl \emph{section} chooses one unitary operator $P_z$ to represent each phase-space direction $z$, with multiplication defined up to a scalar phase. Fix such a physical section and write its multiplication rule as
\begin{equation}
 P_zP_{z'}=c_p(z,z')P_{z+z'},
 \label{eq:global-Pauli-section-cocycle}
\end{equation}
where $c_p(z,z')$ is the unit-modulus scalar phase associated with the chosen section.  On an isotropic subspace $U$, a
cocycle-compatible eigenvalue assignment is a function $\eta:U\to\mathbb C$
of unit modulus satisfying
\begin{equation}
 \eta(z+z')=c_p(z,z')^{-1}\eta(z)\eta(z')
 \qquad(z,z'\in U).
 \label{eq:cocycle-compatible-assignment}
\end{equation}
It defines the joint eigenspace projector
\begin{equation}
 \Pi_{U,\eta}:=p^{-\dim U}\sum_{z\in U}\eta(z)^{-1}P_z.
 \label{eq:intrinsic-character-projector}
\end{equation}
For the odd-prime Weyl section used later, the cocycle is trivial on every
isotropic subspace and these assignments are ordinary additive characters.
In characteristic two, the coordinate sections may have a nontrivial cocycle;
the compatible assignment in Eq.~\eqref{eq:cocycle-compatible-assignment} is
the intrinsic object represented by the later enhanced coordinates.  The
subspace $L_x$ is Lagrangian, meaning that it is an $n$-dimensional maximal
isotropic subspace of $\F_p^{2n}$.  For a subspace
$U\subseteq\F_p^{2n}$, write
\begin{equation}
 U^\perp:=\{v:\spair{v}{u}=0\text{ for every }u\in U\}.
 \label{eq:symplectic-orthogonal-complement}
\end{equation}
The space $L_x\cap W^\perp$ consists precisely of the initial stabilizer
directions commuting with the whole prescribed measurement.  These subspaces
satisfy the dimension identity
\begin{equation}
 \dim(L_x\cap W^\perp)=m+\dim(L_x\cap W).
 \label{eq:fixed-W-dimension}
\end{equation}
Indeed, put $r:=\dim W=n-m$ and $t:=\dim(L_x\cap W)$.  The linear map
$L_x\to W^*$ induced by the symplectic pairing has kernel
$L_x\cap W^\perp$.  Its transpose has kernel $L_x\cap W$, because
$L_x=L_x^\perp$.  Hence its rank is $r-t$, and rank--nullity on the
$n$-dimensional space $L_x$ gives
$\dim(L_x\cap W^\perp)=n-(r-t)=m+t$, proving
Eq.~\eqref{eq:fixed-W-dimension}.  It follows that one may choose an
$m$-dimensional complement $K$ satisfying
\begin{equation}
 L_x\cap W^\perp=(L_x\cap W)\oplus K.
 \label{eq:fixed-W-direct-complement}
\end{equation}
The learner outputs a \emph{report}, a pair $(K,\eta)$ describing
complementary stabilizer directions and their eigenvalues.  The report is
\emph{syntactically admissible} when it belongs to the state-independent
alphabet
\begin{equation}
 \mathsf{Out}_{W,m}:=
 \left\{(K,\eta):
 \begin{array}{l}
 K\subseteq W^\perp\text{ is an }m\text{-dimensional isotropic subspace},\\
 K\cap W=\{0\},\text{ and }\eta\text{ is a cocycle-compatible}\\
 \text{eigenvalue assignment for the commuting family over }K
 \end{array}
 \right\}.
 \label{eq:fixed-W-output-alphabet}
\end{equation}
Compatibility is understood in the sense of
Eq.~\eqref{eq:cocycle-compatible-assignment} relative to the fixed physical
section.  The later characteristic-specific coordinate sections give its
explicit odd-prime and binary realizations; the present definition is
intrinsic and does not depend on the unknown state $x$.

Membership in $\mathsf{Out}_{W,m}$ specifies which reports the learner
may output.  Correctness also depends on the unknown input.  For a given
input $x$, a report $(K,\eta)\in\mathsf{Out}_{W,m}$ is \emph{valid}, or correct,
if and only if
\begin{equation}
 K\subseteq L_x\cap W^\perp,
 \qquad \eta=\chi_x|_K.
 \label{eq:fixed-W-valid-output}
\end{equation}
The dimension, isotropy, and transversality conditions are already part of
membership in $\mathsf{Out}_{W,m}$.  We use
\begin{equation}
 K_W(x):=L_x\cap W^\perp.
 \label{eq:fixed-W-K}
\end{equation}

The reason for this target is operational.  Let
$\mathsf P_W=\{\Pi_{W,\nu}:\nu\in\Omega_W\}$ be the joint projective
measurement associated with $W$, where $\Omega_W$ is its compatible
joint-eigenvalue outcome alphabet.  The physical update is the associated
L\"uders instrument, with operation
$\rho\mapsto\Pi_{W,\nu}\rho\Pi_{W,\nu}$ for outcome $\nu$.  The learner
knows $W$ but reports before this instrument is applied and therefore does not
know $\nu$.  Let $\nu\in\Omega_W$ be any outcome having positive
probability.  The measured family $W$ becomes part of the postmeasurement
stabilizer with assignment $\nu$, while every stabilizer in $K$ survives with
its original assignment $\eta$.  Since $W\oplus K$ is an $n$-dimensional
isotropic space, it is Lagrangian.  Since $W\cap K=\{0\}$, every element of $W\oplus K$ has a unique
form $w+k$ with $w\in W$ and $k\in K$.  Relative to the fixed physical section
in Eq.~\eqref{eq:global-Pauli-section-cocycle}, define
\begin{equation}
 \operatorname{Ext}_{W\oplus K}(\nu,\eta)(w+k)
 :=c_p(w,k)^{-1}\nu(w)\eta(k).
 \label{eq:fixed-W-explicit-cocycle-extension}
\end{equation}
The definition is unique because the direct-sum decomposition is unique.  It
restricts to $\nu$ on $W$ and to $\eta$ on $K$.  Associativity of the physical
section gives the cocycle identity for $c_p$, while $W\perp K$ ensures that the
operators over the two summands commute.  Substitution in
Eq.~\eqref{eq:fixed-W-explicit-cocycle-extension} therefore shows that the
extension is cocycle-compatible on $W\oplus K$.  Conversely, any compatible
extension with the prescribed restrictions must have the displayed value,
because $P_wP_k=c_p(w,k)P_{w+k}$.  Thus, for every outcome $\nu$ with positive probability, the conditional
state is the unique pure stabilizer state specified by the following
Lagrangian subspace and compatible eigenvalue assignment:
\begin{equation}
 \bigl(W\oplus K,\ \operatorname{Ext}_{W\oplus K}(\nu,\eta)\bigr).
 \label{eq:fixed-W-postmeasurement-space}
\end{equation}
Thus the requested output contains all initial-state information needed to
identify every positive-probability conditional state once the actual outcome
is known.  It deliberately omits stabilizers that do not survive the
measurement.  It also need not determine which outcomes can occur, because
that may depend on $L_x\cap W$ and its character.

The input is called \emph{generic} when $L_x\cap W=\{0\}$ and
\emph{exceptional} otherwise:
\begin{equation}
 \mathsf{Stab}_{n,p}^{\mathrm{gen}}(W):=\{x:L_x\cap W=\{0\}\},\quad
 \mathsf{Stab}_{n,p}^{\mathrm{exc}}(W):=\{x:L_x\cap W\ne\{0\}\}.
 \label{eq:fixed-W-generic-exceptional}
\end{equation}
For a generic input, $K_W(x)$ has dimension $m$ and is the unique valid
subspace.  For an exceptional input, the target is intentionally
set-valued: any complement in Eq.~\eqref{eq:fixed-W-direct-complement}, supplied
with the correct restricted character, is accepted.  Therefore, the exact
success event for an exceptional input includes every valid report, rather
than a single selected complement.  This distinction affects the statewise
success probability in Eq.~\eqref{eq:fixed-W-statewise-success} and the
later transfer of the generic converse to the full uniform prior.

\smallskip
For example, take two qubits and let $W$ be generated by $Z$ on the
second qubit.  Both $\ket{+0}$ and $\ket{+1}$ admit the report
$K=\operatorname{span}\{X_1\}$ with eigenvalue $+1$.  Their $Z_2$
measurement outcomes differ, but the report together with the actual
outcome specifies the corresponding conditional state.  The input
$\ket{+0}$ also admits the different correct complementary direction
$X_1Z_2$, with eigenvalue $+1$.  This illustrates both the omission of
the outcome distribution and the possibility of several correct reports
for an exceptional input.  Here $X_1$ and $Z_2$ are the usual physical
qubit Pauli operators.

\subsection{Exact learning and its two performance criteria}
\label{subsec:fixed-W-operational-criteria}

The learner receives $k$ identical copies $\rho_x^{\otimes k}$ of the unknown
initial state.  The constraint $W$ is a fixed parameter of the measurement,
whereas the unknown label $x$ is not.  An admissible strategy is a POVM on
the full $k$-copy Hilbert space, with one positive semidefinite effect for
each basis-free report:
\begin{equation}
 \mathsf M=\{M_{K,\eta}:(K,\eta)\in\mathsf{Out}_{W,m}\}.
 \label{eq:fixed-W-exact-learner-POVM}
\end{equation}
The effects satisfy $\sum_{(K,\eta)}M_{K,\eta}=I$.  The measurement
outcome $(K,\eta)$ is the proposed complementary stabilizer description.  Because an
exceptional state can admit several valid outputs,
the success probability of a fixed strategy on a fixed state is the total
probability of all valid outcomes:
\begin{equation}
 s_{\mathsf M}^{(k)}(x;W)
 :=\sum_{(K,\eta)\text{ valid for }x}
 \Tr[M_{K,\eta}\rho_x^{\otimes k}].
 \label{eq:fixed-W-statewise-success}
\end{equation}

We evaluate a strategy in two complementary ways.  The worst-case criterion
asks for a guarantee that holds uniformly over every pure stabilizer input.
The uniform-average criterion is the Bayesian success probability when the
unknown stabilizer state is sampled uniformly.  Optimizing over all collective
learner POVMs defines
\begin{align}
 P_{W,\rm wc}^{(k),*}(n,m;p)
 &:=\sup_{\mathsf M}\inf_{x\in\mathsf{Stab}_{n,p}}
 s_{\mathsf M}^{(k)}(x;W),
 \label{eq:fixed-W-wc-success}\\
 P_{W,\rm av}^{(k),*}(n,m;p)
 &:=\sup_{\mathsf M}\frac1{|\mathsf{Stab}_{n,p}|}
 \sum_{x\in\mathsf{Stab}_{n,p}}s_{\mathsf M}^{(k)}(x;W).
 \label{eq:fixed-W-av-success}
\end{align}
In the worst-case optimum, one POVM is chosen before the unknown state is
specified; the infimum over states is taken for that same POVM.  The
uniform-average optimum likewise uses one POVM and the uniform prior on
distinct pure stabilizer states.  These are the two principal exact-learning
quantities of the paper.  Our main theorem determines their first-order copy
threshold when $m$ is linear in $n$, and shows that neither criterion admits
a copy-rate discount below one.

The special case $W=\{0\}$, equivalently $m=n$, is \emph{exact complete stabilizer-state identification}: the learner must identify the entire unknown pure stabilizer state.  Indeed, $W^\perp=\F_p^{2n}$ and
$L_x\cap W^\perp=L_x$, so the only valid $n$-dimensional output is
$(K,\eta)=(L_x,\chi_x)$.  We denote the unrestricted complete-identification worst-case and
uniform-average optima by
\begin{align}
 P_{\rm stab,id,wc}^{(k),*}(n;p)
 &:=\max_{\{M_x\}}\min_{x\in\mathsf{Stab}_{n,p}}
 \Tr[M_x\rho_x^{\otimes k}],
 \label{eq:unrestricted-stabilizer-wc-success}\\
 P_{\rm stab,id,av}^{(k),*}(n;p)
 &:=\max_{\{M_x\}}\frac1{|\mathsf{Stab}_{n,p}|}
 \sum_{x\in\mathsf{Stab}_{n,p}}\Tr[M_x\rho_x^{\otimes k}].
 \label{eq:unrestricted-stabilizer-av-success}
\end{align}
The covariant construction in Section~\ref{sec:fourier-achievability}
proves equality of the two complete-identification optima.  We denote their
common value by $P_{\rm stab,id}^{(k),*}(n;p)$.

The reduction will also require the same optimization on a restricted
candidate set.  For any nonempty
$\mathcal S\subseteq\mathsf{Stab}_{n,p}$, define
\begin{align}
 P_{W,\rm wc}^{(k),*}(\mathcal S)
 &:=\sup_{\mathsf M}\inf_{x\in\mathcal S}s_{\mathsf M}^{(k)}(x;W),
 \label{eq:fixed-W-restricted-wc}\\
 P_{W,\rm av}^{(k),*}(\mathcal S)
 &:=\sup_{\mathsf M}\frac1{|\mathcal S|}
 \sum_{x\in\mathcal S}s_{\mathsf M}^{(k)}(x;W).
 \label{eq:fixed-W-restricted-av}
\end{align}
We write $P_{W,\rm gen,wc}^{(k),*}$ and
$P_{W,\rm gen,av}^{(k),*}$ for the specialization to
$\mathsf{Stab}_{n,p}^{\mathrm{gen}}(W)$.  The corresponding error probabilities
are one minus these success probabilities.

\subsection{Postmeasurement verification game}
\label{subsec:fixed-W-verification-game}

Exact identification is a stringent requirement.  To connect it directly to
the physical purpose of predicting the postmeasurement state, we also consider
a game with a separate verifier.  The chronological order is fixed: the
learner first measures the $k$ copies of the initial state and reports without
knowing $\nu$; afterward, the verifier applies the joint projective L\"uders
instrument for $W$ to an independent test copy and obtains $\nu$.  Let
$\Pi_{W,\nu}$ be the joint-eigenvalue
projector for an outcome $\nu$ of the measurement associated with $W$, and
let $\Pi_{\widehat K,\widehat\eta}$ be the joint-character projector for a
report $(\widehat K,\widehat\eta)\in\mathsf{Out}_{W,m}$.  Since
$\widehat K\subseteq W^\perp$ and $\widehat K\cap W=\{0\}$, the two projectors
commute and jointly describe the direct-sum stabilizer family
$W\oplus\widehat K$.  Thus the projector for the reported postmeasurement
description is
\begin{equation}
 \Pi_{\widehat K,\widehat\eta}^{W,\nu}
 :=\Pi_{W,\nu}\Pi_{\widehat K,\widehat\eta}.
 \label{eq:fixed-W-announced-projector-factorization}
\end{equation}
For an input $x$, the following formulas give the outcome probability
$q_x(\nu)$ and, whenever $q_x(\nu)>0$, the normalized L\"uders
postmeasurement state $\rho_{x|\nu}$:
\begin{equation}
 q_x(\nu):=\Tr[\Pi_{W,\nu}\rho_x],\quad
 \rho_{x|\nu}:=
 \frac{\Pi_{W,\nu}\rho_x\Pi_{W,\nu}}{q_x(\nu)}.
 \label{eq:fixed-W-outcome-probability}
\end{equation}

The learner measures $\rho_x^{\otimes k}$ with a POVM
$\{M_{\widehat K,\widehat\eta}\}$ and reports
$(\widehat K,\widehat\eta)$.  On one independent copy of
$\rho_{x|\nu}$, the verifier applies the projector in
Eq.~\eqref{eq:fixed-W-announced-projector-factorization}.  The verifier
assigns score one to acceptance and zero to rejection.  For fixed input,
report, and positive-probability outcome $\nu$, the expected score is the
conditional acceptance probability
\begin{equation}
 g_x^W(\widehat K,\widehat\eta\mid\nu)
 :=\Tr[\Pi_{\widehat K,\widehat\eta}^{W,\nu}\rho_{x|\nu}].
 \label{eq:fixed-W-conditional-verification-score}
\end{equation}
For each fixed report $(\widehat K,\widehat\eta)$, we first average its score
over the verifier outcomes.  The outcome-conditioned acceptance probability
$g_x^W(\widehat K,\widehat\eta\mid\nu)$ can depend on $\nu$, but the
projectors commute and $\sum_{\nu\in\Omega_W}\Pi_{W,\nu}=I$.  Averaging the
conditional acceptance probability over the measurement outcomes gives
\begin{align}
 \sum_{\substack{\nu\in\Omega_W\\q_x(\nu)>0}}
 q_x(\nu)
 \Tr\!\left[
   \Pi_{W,\nu}\Pi_{\widehat K,\widehat\eta}
   \rho_{x|\nu}
 \right]
 =
 \sum_{\nu\in\Omega_W}
 \Tr\!\left[
   \Pi_{W,\nu}\Pi_{\widehat K,\widehat\eta}
   \Pi_{W,\nu}\rho_x
 \right]
 =\Tr[\Pi_{\widehat K,\widehat\eta}\rho_x].
 \label{eq:fixed-W-outcome-averaged-score}
\end{align}
This is the expected score for a fixed input and a fixed report.  We now
also average over the learner's random report.
For a fixed initial state $x$, we define the \emph{statewise expected
verification score} of a learner POVM $\mathsf M$ as the verifier's
acceptance probability averaged over both the learner's report and the
verifier's outcome.  Write $r=(\widehat K,\widehat\eta)$, and abbreviate
the associated POVM effect and report projector by $M_r$ and $\Pi_r$.
The operational definition is
\begin{align}
 v_{\mathsf M}^{(k)}(x;W)
 :={\sum_{r\in\mathsf{Out}_{W,m}}}
 \Tr[M_r\rho_x^{\otimes k}]
\sum_{\substack{\nu\in\Omega_W\\q_x(\nu)>0}}
 q_x(\nu)\,g_x^W(r\mid\nu).
 \label{eq:fixed-W-statewise-verification-score-operational}
\end{align}
Substituting Eq.~\eqref{eq:fixed-W-outcome-averaged-score} into the
operational definition
\eqref{eq:fixed-W-statewise-verification-score-operational} gives
\begin{equation}
 v_{\mathsf M}^{(k)}(x;W)
 =\sum_{(\widehat K,\widehat\eta)\in\mathsf{Out}_{W,m}}
 \Tr[M_{\widehat K,\widehat\eta}\rho_x^{\otimes k}]
 \Tr[\Pi_{\widehat K,\widehat\eta}\rho_x].
 \label{eq:fixed-W-statewise-verification-score-reduced}
\end{equation}
The reduction to Eq.~\eqref{eq:fixed-W-statewise-verification-score-reduced}
uses the average over $\nu$; it does not assert that the conditional
acceptance probability in Eq.~\eqref{eq:fixed-W-conditional-verification-score}
is independent of $\nu$.  A valid report has Born acceptance probability
one, whereas an incorrect report can have nonzero acceptance probability.
We optimize the statewise verification score over exactly the POVMs in
Eq.~\eqref{eq:fixed-W-exact-learner-POVM}.  The worst-case and
uniform-average verification optima are defined, respectively, by
\begin{align}
 S_{W,\rm ver,wc}^{(k),*}(n,m;p)
 &:=\sup_{\mathsf M}\inf_{x\in\mathsf{Stab}_{n,p}}
 v_{\mathsf M}^{(k)}(x;W),
 \label{eq:fixed-W-verification-wc-definition}\\
 S_{W,\rm ver,av}^{(k),*}(n,m;p)
 &:=\sup_{\mathsf M}\frac1{|\mathsf{Stab}_{n,p}|}
 \sum_{x\in\mathsf{Stab}_{n,p}}v_{\mathsf M}^{(k)}(x;W).
 \label{eq:fixed-W-verification-av-definition}
\end{align}
For a specified subset of input states, the restricted versions use the
same POVM class but restrict the infimum or the uniform average to that
subset.  The verifier's independent test copy is not included among the
learner's $k$ input copies.

\subsection{Main theorem and copy rates}
\label{subsec:fixed-W-rates}
\label{subsec:rate-functions}

We now state the result for the original report alphabet.  A Clifford
change of coordinates transports the input states, the reports, and the
learner POVMs together.  It gives the following finite-size covariance;
its proof appears in Subsection~\ref{subsec:constraint-covariance-proof}.

\begin{proposition}[Dimension-only covariance of the prescribed constraint]
\label{prop:fixed-W-dimension-only-covariance}
Let $W,W'\subseteq\F_p^{2n}$ be isotropic subspaces of the same dimension
$n-m$.  Then, for every $k$, the exact worst-case and uniform-average optima
and the corresponding verification optima for $W$ and $W'$ are equal.
Consequently, these four finite-size quantities depend on the prescribed
constraint only through $n,m$, and $p$.
\end{proposition}

For the converse, define
\begin{equation}
 J(\alpha,\beta):=
 \begin{cases}
 \beta(1-\alpha-\beta/2),&0\le\alpha\le1-\beta,\\
 (1-\alpha)^2/2,&1-\beta\le\alpha<1.
 \end{cases}
 \label{eq:main-results-J}
\end{equation}

Proposition~\ref{prop:fixed-W-dimension-only-covariance} proves that,
for fixed $n,m,k,p$, all four fixed-constraint optima are independent
of the prescribed isotropic subspace $W$ of dimension $n-m$.
We therefore introduce the following dimension-only notation:
\begin{align}
 P_{\mathrm{fix,wc}}^{(k),*}(n,m;p)
 &:=P_{W,\rm wc}^{(k),*}(n,m;p),\notag\\
 P_{\mathrm{fix,av}}^{(k),*}(n,m;p)
 &:=P_{W,\rm av}^{(k),*}(n,m;p),
 \label{eq:dimension-only-exact-optima}
\end{align}
where $W$ is \emph{any} isotropic subspace with $\dim W=n-m$.
Proposition~\ref{prop:fixed-W-dimension-only-covariance} guarantees that
these definitions are independent of that choice.  The symbols
$P_{\mathrm{fix,wc}}^{(k),*}$ and $P_{\mathrm{fix,av}}^{(k),*}$ are
finite-size optima for the original basis-free task, not the refined
fixed-frame optima.

Fix $0<\beta\le1$ and $0<\epsilon<1$.  In the following two
definitions, $(m_n)$ and $(k_n)$ range over nonnegative integer sequences
with $m_n\le n$ for every $n$ and with $m_n/n\to\beta$.
The worst-case copy rate is the infimum of all $\alpha$ for which such
sequences exist with $k_n/n\to\alpha$ and the optimal worst-case success
has liminf at least $\epsilon$.  The uniform-average rate is defined in the
same way using the optimal uniform-average success.  Explicitly,
\begin{align}
 R_{\mathrm{wc}}(p,\beta,\epsilon)
 &:={\inf}\{\alpha:\exists (m_n),(k_n),\,
 m_n/n\to\beta,\ k_n/n\to\alpha,
 \liminf_nP_{\mathrm{fix,wc}}^{(k_n),*}(n,m_n;p)\ge\epsilon\},
 \label{eq:fixed-W-wc-rate}\\
 R_{\mathrm{av}}(p,\beta,\epsilon)
 &:={\inf}\{\alpha:\exists (m_n),(k_n),\,
 m_n/n\to\beta,\ k_n/n\to\alpha,
 \liminf_nP_{\mathrm{fix,av}}^{(k_n),*}(n,m_n;p)\ge\epsilon\}.
 \label{eq:fixed-W-av-rate}
\end{align}
For each $\alpha$ in the defining set, the integer sequences $(m_n),(k_n)$ are chosen
first.  Each finite optimum then has the same value for every isotropic
$W_n$ of dimension $n-m_n$, by
Proposition~\ref{prop:fixed-W-dimension-only-covariance}.  Thus the liminf
and rate infima agree whether the legal constraint sequence is prescribed
in advance or selected existentially.  The optimization order remains
that in Eqs.~\eqref{eq:fixed-W-wc-success} and
\eqref{eq:fixed-W-av-success}.

The original problem permits $m_n=0$. Auxiliary results that require $m_n\ge1$ are invoked only for sufficiently large $n$: since $m_n/n\to\beta>0$, this inequality eventually holds, and finitely many initial indices do not affect the liminf.

The converse combines the refined exact bound with a separate reduction for basis-free verification. Its assumption $k_n/n\to\alpha<1$ is distinct from the direct part's assumption $k_n-n\to+\infty$, under which exact complete stabilizer-state identification succeeds. The theorem states these two conclusions separately.

\begin{theorem}[No copy-rate discount under fixed commuting constraints]
\label{thm:fixed-W-no-discount}
Fix a prime $p$ and $0<\beta\le1$.  Let $(m_n)$ satisfy
$0\le m_n\le n$ and $m_n=\beta n+o(n)$, and let $W_n$ be any isotropic
subspace of $\F_p^{2n}$ with $\dim W_n=n-m_n$.

\medskip
\noindent\textbf{(1) Exact-learning converse.}
Suppose that $(k_n)$ is a nonnegative integer sequence with
$k_n=\alpha n+o(n)$ for some $0\le\alpha<1$.
For each fixed pair $(m_n),(k_n)$, there are a real sequence
$e_n$ with $e_n/n^2\to0$ and a constant $C_p>0$, neither depending on the
choice of $W_n$, such that, for all sufficiently large $n$ and all legal
$W_n$, the two exact optima obey
\begin{align}
 P_{W_n,\rm wc}^{(k_n),*}(n,m_n;p)
 &\le p^{-J(\alpha,\beta)n^2+e_n}\longrightarrow0,
 \label{eq:fixed-W-wc-converse}\\
 P_{W_n,\rm av}^{(k_n),*}(n,m_n;p)
 &\le C_p p^{-m_n}+p^{-J(\alpha,\beta)n^2+e_n}\longrightarrow0.
 \label{eq:fixed-W-av-converse}
\end{align}
The same $e_n$ is used in both inequalities; $e_n$ may depend on the
fixed integer sequences $(m_n),(k_n)$.

\medskip
\noindent\textbf{(2) Verification converse.}
Under the hypotheses on $(k_n)$ in (1), the worst-case and uniform-average
verification optima obey
\begin{align}
 S_{W_n,\rm ver,wc}^{(k_n),*}(n,m_n;p)&\longrightarrow0,
 \label{eq:fixed-W-ver-wc-converse}\\
 S_{W_n,\rm ver,av}^{(k_n),*}(n,m_n;p)&\longrightarrow0.
 \label{eq:fixed-W-ver-av-converse}
\end{align}
Both limits in (1) and both limits in (2) hold uniformly over legal
$W_n$: for every $\varepsilon>0$, there is an $N$ such that, for
$n\ge N$ and every isotropic $W_n$ of dimension $n-m_n$, all four
corresponding optima are below $\varepsilon$.  The choice of $N$ may
depend on $(m_n),(k_n)$; no uniformity over all pairs of admissible
integer sequences is claimed.

\medskip
\noindent\textbf{(3) Uniform direct part.}
Suppose instead that $k_n-n\to+\infty$.  Then all four original
fixed-constraint optima converge to one, uniformly over legal $W_n$:
\begin{align*}
 P_{W_n,\rm wc}^{(k_n),*}(n,m_n;p)&\to1,
 &P_{W_n,\rm av}^{(k_n),*}(n,m_n;p)&\to1,\\
 S_{W_n,\rm ver,wc}^{(k_n),*}(n,m_n;p)&\to1,
 &S_{W_n,\rm ver,av}^{(k_n),*}(n,m_n;p)&\to1.
\end{align*}
The subcritical copy-rate hypothesis of (1) is not assumed in (3).

\medskip
\noindent\textbf{(4) Copy-rate identities.}
For every $0<\epsilon<1$, the exact-learning rates defined in
Eqs.~\eqref{eq:fixed-W-wc-rate} and~\eqref{eq:fixed-W-av-rate} satisfy
\begin{equation}
 R_{\mathrm{wc}}(p,\beta,\epsilon)=R_{\mathrm{av}}(p,\beta,\epsilon)=1.
 \label{eq:fixed-W-two-rates-one}
\end{equation}
\end{theorem}

Part (4) gives the rate identities for the two exact-learning criteria
defined in Eqs.~\eqref{eq:fixed-W-wc-rate} and
\eqref{eq:fixed-W-av-rate}.  For verification, parts (2) and (3) give
the same first-order threshold: both optima tend to zero at limiting copy
rates below one and tend to one when $k_n-n\to+\infty$.

The exact converse is obtained from the finite comparison in
Theorem~\ref{thm:fixed-W-operational-comparison} and the coordinate
exact bound.  Verification uses a separate compatible-list bound for
the original report alphabet.  Complete-state identification followed by
classical postprocessing proves the direct part.  These three ingredients
are assembled in the proof of Theorem~\ref{thm:fixed-W-no-discount} in
Subsection~\ref{subsec:proof-main-fixed-W}.

In the special case $m=n$, the prescribed subspace is $W=\{0\}$, and the exact fixed-constraint task is exact complete stabilizer-state identification over the full stabilizer ensemble.
Thus the case $\beta=1$ of Eq.~\eqref{eq:fixed-W-wc-converse} gives
\begin{equation}
 P_{\rm stab,id}^{(k),*}(n;p)
 \le p^{-\frac12(1-\alpha)^2n^2+o(n^2)}\longrightarrow0
 \label{eq:unrestricted-stabilizer-converse}
\end{equation}
whenever $k=\alpha n+o(n)$ with $0\le\alpha<1$.

The quadratic exponent bounds exact success in the refined problem and
the generic original problem.  In the full
uniform-average exact bound, the exponentially rare exceptional term may
dominate the quadratically small generic term.  The present theorem does
not determine the full success profile in the bounded-offset window
$k=n+O(1)$ or the copy requirements for sublinear targets $m=o(n)$.

\section{Auxiliary frame model and operational comparisons}
\label{sec:model-main-results}
\label{sec:generic-fixed-W-frame-relation}
The original learner knows the prescribed subspace $W$ and may report any
element of $\mathsf{Out}_{W,m}$.  To bound its exact-learning probability,
we consider an auxiliary task in which the learner is given an additional
ordered basis.  We first define this basis and the target it selects, and
then compare the auxiliary task with the original one.

\subsection{The fixed-frame state family and reference-dual target}
\label{subsec:main-results-frame-model}

\begin{definition}[Auxiliary frame and transverse inputs]
\label{def:auxiliary-frame}
An \emph{auxiliary frame} is an ordered basis of a Lagrangian subspace of
the Pauli phase space $\F_p^{2n}$.  Let $\mathsf{Ref}_{n,p}$ be the set of
these ordered bases.  Fix $R\in\mathsf{Ref}_{n,p}$ and write
\begin{equation}
 R=(R_1,\ldots,R_n),\quad
 M_R:=\operatorname{span}\{R_1,\ldots,R_n\},
 \label{eq:announced-reference-frame}
\end{equation}
The input stabilizer subspace $L_x$ is \emph{transverse to the frame $R$}
when $L_x\cap M_R=\{0\}$.  The corresponding input family is
\begin{equation}
 \mathsf S(R):=\{x\in\mathsf{Stab}_{n,p}:L_x\cap M_R=\{0\}\}.
 \label{eq:frame-state-family}
\end{equation}
In the auxiliary task, $R$ is public before the learner measures the
$k$ copies, and the unknown state is promised to belong to $\mathsf S(R)$.
\end{definition}
The frame vectors are Pauli phase-space directions; eigenvalue
assignments are specified separately.

The condition $L_x\cap M_R=\{0\}$ in Eq.~\eqref{eq:frame-state-family} ensures that pairing with the $n$ frame vectors determines every vector in $L_x$. Indeed, for $x\in\mathsf S(R)$ the map
\begin{equation}
 \Gamma_{x,R}:L_x\longrightarrow\F_p^n,\quad
 \Gamma_{x,R}(s):=(\spair{s}{R_1},\ldots,\spair{s}{R_n})
 \label{eq:reference-pairing-map}
\end{equation}
is an isomorphism.  Hence there are unique reference-dual stabilizers
$S_1(x;R),\ldots,S_n(x;R)$ satisfying
\begin{equation}
 \spair{S_i(x;R)}{R_j}=\delta_{ij}.
 \label{eq:reference-dual-family}
\end{equation}
The refined target consists of the first $m$ of these directions and their
restricted character:
\begin{equation}
 y_R(x):=\bigl(S_1(x;R),\ldots,S_m(x;R),
 \chi_x|_{\operatorname{span}\{S_1(x;R),\ldots,S_m(x;R)\}}\bigr).
 \label{eq:frame-assisted-exact-target}
\end{equation}
The set of all allowed refined reports is fixed by the public frame, not
by the unknown state.  We define it as the image of the target map over the
full transverse state family:
\begin{equation}
 \Theta_{R,m}^{\mathrm{frame}}
 :=\{y_R(x'):x'\in\mathsf S(R)\}.
 \label{eq:frame-relative-report-alphabet}
\end{equation}
Thus the auxiliary task is exact identification of $y_R(x)$ from
$\rho_x^{\otimes k}$, with $R$ public.  Its learner POVMs have one effect for
each element of $\Theta_{R,m}^{\mathrm{frame}}$; the true label is the
particular element $y_R(x)$.  We denote its worst-case optimum on
$\mathsf S(R)$ by $P_{\rm frame,wc}^{(k),*}(R,m;p)$ and its uniform-average
optimum by $P_{\rm frame,av}^{(k),*}(R,m;p)$.  Unlike the original exceptional
problem, the fixed-frame target is single-valued.  This feature is essential
for the exact discrimination analysis developed later.

We call each element of $\Theta_{R,m}^{\mathrm{frame}}$ a
\emph{frame-relative report}.  Its directions are ordered by their
pairings with the public frame.
We reserve the phrase ``standardized coordinate label'' for its image
$(C,\gamma)$ under the frame-to-coordinate map.  For input $x$, only $y_R(x)$ is
the correct report.

After standardizing $R$, the frame-relative reports correspond to the
coordinate partial-label space in Eq.~\eqref{eq:partial-label-space}.  Every admissible
partial symmetric-matrix column extends to a symmetric matrix, and every
partial character extends to a full character.  Section~\ref{sec:partial-learning}
gives the coordinate realization.

Both refined criteria optimize over POVMs indexed by
$\Theta_{R,m}^{\mathrm{frame}}$.  Exact learning gives payoff one for
the report $y_R(x)$ and zero for any other report.  Verification instead
assigns to $(S_1,\ldots,S_m,\eta)$ the acceptance probability
$\Tr[\Pi_{\operatorname{span}\{S_1,\ldots,S_m\},\eta}\rho_x]$ on an independent
test copy.  The Born acceptance probability is one for the correct report
$y_R(x)$, whereas an incorrect report can have a smaller, nonzero
acceptance probability.  Let
$S_{\rm frame,ver,wc}^{(k),*}(R,m;p)$ and
$S_{\rm frame,ver,av}^{(k),*}(R,m;p)$ denote these refined fixed-frame
verification optima.

The table separates the input and report sets.  The coordinate problem
is the standardized realization of the refined fixed-frame task.
\begin{center}
\small
\renewcommand{\arraystretch}{1.15}
\begin{tabularx}{\textwidth}{@{}>{\raggedright\arraybackslash}p{0.20\textwidth}YY@{}}
\toprule
Model & State family and report & Role in the proof\\
\midrule
Original fixed-$W$ task
& All pure stabilizer states; reports in $\mathsf{Out}_{W,m}$
& The operational task; exceptional inputs may have several correct reports.\\
Generic fixed-$W$ task
& States with $L_x\cap W=\{0\}$; unique correct pair $(K_W(x),\chi_x|_{K_W(x)})$
& The exact upper bound obtained by conditioning on frames and the basis-free verification converse.\\
Refined frame-assisted task
& A public frame $R$ and input $x\in\mathsf S(R)$; the target $y_R(x)$
& Equivalent to the standardized coordinate task for both payoffs.\\
Coordinate realization
& The states in $\mathsf S(R)$ and their refined reports, expressed in matrix and character coordinates
& Supplies the conditioned ensemble used for exact discrimination.\\
Unrestricted complete identification
& All pure stabilizer states; the complete pair $(L_x,\chi_x)$
& A statewise decoder followed by classical postprocessing proves achievability.\\
\bottomrule
\end{tabularx}
\end{center}

The frame-assisted task and its coordinate form are equivalent: the
state, report, POVM, and verification projector are transported together.
For the original exact task, conditioning on a frame permits an auxiliary
learner to choose its measurement using that frame, which gives an upper
bound.  The verification argument instead keeps the original report
alphabet and uses a separate list comparison.

The comparison in Subsection~\ref{subsec:fixed-W-frame-refinement}
uses these targets to bound original exact success.  The verification
argument in Section~\ref{sec:integrated-all-prime-graded} instead uses
all reports in $\mathsf{Out}_{W,m}$.

\subsection{Standard coordinates and finite frame-to-coordinate equivalence}
\label{subsec:frame-coordinate-equivalence}

We now represent the frame-relative reports in standard coordinates
and define the two coordinate optima.  The following proposition states
the finite equivalence; its proof is deferred until the physical state
and projector formulas have been given. In the phase-space coordinates
$z=\binom{u}{v}$ fixed in Eq.~\eqref{eq:global-symplectic-form}, take the
standard ordered frame $R_j^{\mathrm{std}}=\binom{0}{e_j}$ for
$j=1,\ldots,n$, where $e_j$ is the $j$th coordinate vector. Its span is
the vertical Lagrangian subspace. For $1\leq m\leq n$, define the
canonical inclusion that selects the first $m$ horizontal coordinates:
\begin{equation}
 J_m:=\begin{pmatrix}I_m\\0\end{pmatrix}
 \in\F_p^{n\times m}.
 \label{eq:canonical-query-Jm}
\end{equation}
The admissible partial columns of symmetric $n\times n$ matrices form
\begin{equation}
 \mathcal C_{n,m}
 :=\{C\in\F_p^{n\times m}:J_m^TC=C^TJ_m\}.
 \label{eq:partial-symmetric-matrix-space}
\end{equation}
The condition $J_m^TC=C^TJ_m$ requires the upper $m\times m$ block of
$C$ to be symmetric.  These coordinate definitions apply to every prime
$p$.

The standardized input ensemble consists of the stabilizer states whose
Lagrangian subspaces are transverse to the standard vertical subspace.
Its exact report specifies the first $m$ reference-dual stabilizer directions
and their restricted character.  We call this the \emph{exact partial-label
identification task} and denote its optimal uniform-average success over
all collective $k$-copy POVMs by $P_{\mathrm{id}}^{(k),*}(n,m;p)$.

The standardized verification task has the same input ensemble and the same
report alphabet.  Its statewise payoff is instead the Born acceptance
probability of the reported commuting family on an independent verification
copy.  We denote its optimal uniform-average score by
$S_{\mathrm{ver}}^{(k),*}(n,m;p)$.  Subscripts and parameters suppressed in
later formulas are determined by the immediately stated task and theorem
hypotheses.  The verifier's independent copy is not included among the learner's $k$ copies. The coordinate formulas for the state ensemble, partial labels, POVMs, and test projectors appear in Section~\ref{sec:partial-learning}; the present discussion only fixes the tasks being compared.

To compare these coordinate tasks with the original frame-relative
problem, fix a symplectic map $F_R$ sending every $R_j$ to the corresponding
standard vertical direction.  Proposition~\ref{prop:frame-conditioned-coordinate-reduction}
will establish that the resulting Clifford transport preserves both
performance criteria. Its Clifford implementation must also be
chosen and kept fixed when defining a specific coordinate-label map.
Every transformed Lagrangian $F_RL_x$ is then the graph $L_A$ of a unique
symmetric matrix, and
\begin{equation}
 F_RS_i(x;R)=\binom{e_i}{Ae_i}.
 \label{eq:dual-family-normal-form-coordinate}
\end{equation}
Accordingly, the first $m$ reference-dual directions are encoded by
$AJ_m$.  The following proposition states the precise correspondence
between the report labels and the two statistical tasks; in particular, it
accounts for the phase conventions used in their verification projectors.

\begin{proposition}[Frame-to-coordinate equivalence]
\label{prop:frame-conditioned-coordinate-reduction}
Let $p$ be prime, $n\ge1$, $1\le m\le n$, and $k\ge0$, and let
$R\in\mathsf{Ref}_{n,p}$ be a fixed frame.  Fix a symplectic map
$F_R$ that maps $R$ to the standard vertical frame and a Clifford
unitary $U_R$ implementing $F_R$.  The following four statements hold.

\smallskip\noindent\textbf{(1) State correspondence.}
Conjugation by $U_R$ bijects the stabilizer states in $\mathsf S(R)$
with those transverse to the standard vertical frame.  If $x'$ is the
image of $x$, then
$\rho_{x'}=U_R\rho_xU_R^\dagger$ and $L_{x'}=F_RL_x$.

\smallskip\noindent\textbf{(2) Report correspondence.}
For the fixed choice $(F_R,U_R)$, the transported directions and
eigenvalue assignment define a bijection
\begin{equation}
 \Phi_{R,U_R}:\Theta_{R,m}^{\mathrm{frame}}
 \longrightarrow\mathcal C_{n,m}\times\F_p^m.
 \label{eq:frame-coordinate-report-bijection}
\end{equation}
The coordinate of the correct report $y_R(x)$ is the correct partial
label of the state $x'$ in (1).

\smallskip\noindent\textbf{(3) Probability preservation.}
Conjugating each effect of a learner POVM by $U_R^{\otimes k}$ and
relabeling its outcomes by $\Phi_{R,U_R}$ gives a bijection of the
complete POVM classes.  For every input state and every report, this
bijection preserves the probability of the reported outcome and the
Born acceptance probability of its verification projector.  Consequently,
it preserves statewise exact success and expected verification score.

\smallskip\noindent\textbf{(4) Equality of optima.}
The optimal uniform-average exact success probabilities coincide:
\begin{equation}
 P_{\rm frame,av}^{(k),*}(R,m;p)=P_{\mathrm{id}}^{(k),*}(n,m;p).
 \label{eq:finite-frame-coordinate-equivalence}
\end{equation}
The same holds for the optimal uniform-average verification scores:
\begin{equation}
 S_{\rm frame,ver,av}^{(k),*}(R,m;p)=S_{\mathrm{ver}}^{(k),*}(n,m;p).
 \label{eq:finite-frame-verification-equivalence}
\end{equation}
The corresponding worst-case optima also coincide under the bijection
of decision rules in (3).  All these optimal values are independent of
the choice of $R$.
\end{proposition}

The proposition equates the refined task with its coordinate realization.
Subsection~\ref{subsec:frame-coordinate-equivalence-proof} proves the
state, report, and POVM correspondences after the two physical realizations
have been constructed.  The next subsection compares the original task
with this refined task by a one-sided optimization bound.

\subsection{Bounding exact learning by conditioning on auxiliary frames}
\label{subsec:fixed-W-frame-refinement}

Part (1) of the following theorem compares the original exact-learning
optima with the refined exact optimum; the uniform-average bound also
includes the exceptional mass.  Part (2) records the worst-case-to-average
inequality for verification in the original task.  Define the exceptional
mass under the full uniform input prior by
\begin{equation}
 \delta_{n,m,p}:=\Pr_{X\sim{\rm Unif}(\mathsf{Stab}_{n,p})}
 \{L_X\cap W\ne\{0\}\}.
 \label{eq:fixed-W-delta}
\end{equation}

\begin{theorem}[Finite exact comparison of the fixed-constraint and refined tasks]
\label{thm:fixed-W-operational-comparison}
For every prime $p$ and integers $n\ge1$, $1\le m\le n$,
$k\ge0$, the following comparisons hold.
\medskip
\noindent\textbf{(1) Exact-learning bounds.}
The exact worst-case and uniform-average optima obey
\begin{align}
 P_{W,\rm wc}^{(k),*}(n,m;p)&\le P_{\mathrm{id}}^{(k),*}(n,m;p),
 \label{eq:fixed-W-wc-final-comparison}\\
 P_{W,\rm av}^{(k),*}(n,m;p)
 &\le\delta_{n,m,p}+(1-\delta_{n,m,p})P_{\mathrm{id}}^{(k),*}(n,m;p).
 \label{eq:fixed-W-av-comparison}
\end{align}
\medskip
\noindent\textbf{(2) Verification bound.}
The verification scores satisfy the worst-case-to-average inequality
\begin{equation}
 S_{W,\rm ver,wc}^{(k),*}(n,m;p)
 \le S_{W,\rm ver,av}^{(k),*}(n,m;p).
 \label{eq:fixed-W-verification-wc-average-only}
\end{equation}
\end{theorem}

A fixed frame restricts the input to $\mathsf S(R)$, whereas the original
generic prior includes every state with $L_x\cap W=\{0\}$.  We relate these
input distributions by completing the fixed basis of $W$ to auxiliary frames.
We first identify the correct target for a fixed completion.  We then count
the compatible state--completion pairs to express the generic uniform prior
as a weighted mixture of the fixed-frame input distributions.  Optimizing
separately for each frame gives an upper bound on the success of one
original learner POVM.
The generic comparison is proved first; the exceptional-mass estimate and
the proof of the theorem follow at the end of this subsection.  The
verification converse requires the separate list argument in
Section~\ref{sec:integrated-all-prime-graded}.

\subsubsection{Completing \texorpdfstring{$W$}{W} to a frame transverse to the input}

Fix an ordered basis $(R_{m+1},\ldots,R_n)$ of $W$. For each ordered $m$-tuple $T=(T_1,\ldots,T_m)$, define its extension by the fixed basis of $W$ as
\begin{equation}
 R(T):=(T_1,\ldots,T_m,R_{m+1},\ldots,R_n).
 \label{eq:fixed-W-ordered-extension}
\end{equation}
We call $T$ an \emph{admissible transverse completion for $x$} if
$R(T)$ is a basis of a Lagrangian subspace $M_{R(T)}$ and
$L_x\cap M_{R(T)}=\{0\}$.  Let $\mathfrak T_W(x)$ denote the set of all such
completions.  Thus $T$ completes the fixed basis of $W$, while its
admissibility also depends on $x$.  Define the set of compatible
state--completion pairs, which we call the \emph{incidence set}, by
\begin{equation}
 \mathcal I_W:=\{(x,T):x\in\mathsf{Stab}_{n,p}^{\mathrm{gen}}(W),\
 T\in\mathfrak T_W(x)\}.
 \label{eq:fixed-W-incidence-set}
\end{equation}
For every generic $x$, this set of completions is nonempty.  To see this
without choosing characteristic-dependent coordinates, pass to the symplectic
quotient $W^\perp/W$, which has dimension $2m$.  Because
$L_x\cap W=\{0\}$, the image
\begin{equation}
 \overline K_x:=((L_x\cap W^\perp)+W)/W
 \subseteq W^\perp/W
 \label{eq:generic-target-quotient-image}
\end{equation}
is an $m$-dimensional isotropic subspace and hence is Lagrangian.  Choose a
Lagrangian complement $\overline M_x$ transverse to $\overline K_x$ in the
quotient and let $M_x$ be its inverse image in $W^\perp$.  Then $M_x$ is a
Lagrangian subspace containing $W$.  Moreover, if $v\in L_x\cap M_x$, its
quotient class lies in
$\overline K_x\cap\overline M_x=\{0\}$, so $v\in W$; genericity then gives
$v\in L_x\cap W=\{0\}$.  Thus $L_x\cap M_x=\{0\}$.  Extending the fixed
ordered basis of $W$ to an ordered basis of $M_x$ produces an admissible tuple
$T$.  This quotient argument proves nonemptiness in every characteristic.

Fix an ordered tuple $T$ for which $R(T)$ in
Eq.~\eqref{eq:fixed-W-ordered-extension} is a basis of a Lagrangian
subspace containing $W$.  Whether this tuple is transverse depends on the
input $x$.  The associated \emph{incidence chart} is the set of generic
inputs for which $T$ is an admissible completion:
\begin{equation}
 \mathsf S_W(T):=\{x\in\mathsf{Stab}_{n,p}^{\mathrm{gen}}(W):
 T\in\mathfrak T_W(x)\}.
 \label{eq:fixed-W-incidence-chart}
\end{equation}
Every chartwise average and every subsequent sum over completion tuples
uses only the tuples $T$ for which $R(T)$ is an ordered Lagrangian frame
and the chart $\mathsf S_W(T)$ is nonempty.  We denote this index set by
\begin{equation}
 \mathfrak T_W^{\mathrm{inc}}
 :=\{T:\ R(T)\in\mathsf{Ref}_{n,p},\ \mathsf S_W(T)\ne\varnothing\}.
 \label{eq:fixed-W-nonempty-chart-support}
\end{equation}
The set $\mathfrak T_W^{\mathrm{inc}}$ consists precisely of the second
components of incidence pairs $(x,T)\in\mathcal I_W$ from
Eq.~\eqref{eq:fixed-W-incidence-set}.  Hence the denominator
$|\mathsf S_W(T)|$ is positive in every average indexed by this set.
Since $W\subseteq M_{R(T)}$, the transversality condition
$L_x\cap M_{R(T)}=\{0\}$ already implies $L_x\cap W=\{0\}$.  Hence the
generic restriction in the definition above is automatic, and
\begin{equation}
 \mathsf S_W(T)
 =\{x\in\mathsf{Stab}_{n,p}:L_x\cap M_{R(T)}=\{0\}\}
 =\mathsf S(R(T)).
 \label{eq:fixed-W-incidence-chart-equality}
\end{equation}
The set equality~\eqref{eq:fixed-W-incidence-chart-equality}
identifies the uniform input measure on $\mathsf S_W(T)$ with
that on the state family $\mathsf S(R(T))$ of the fixed-frame
task with parameter $R(T)$.  The conditional-law formulation of
this equality is given in
Eq.~\eqref{eq:fixed-W-incidence-conditional-law} below.

\subsubsection{Equivalence of the targets for a fixed frame}
\label{subsubsec:transverse-chart-target-equivalence}

By Eq.~\eqref{eq:fixed-W-incidence-set}, an incidence pair
$(x,T)\in\mathcal I_W$ is precisely a pair satisfying
$x\in\mathsf{Stab}_{n,p}^{\mathrm{gen}}(W)$ and
$T\in\mathfrak T_W(x)$.  Hence $R(T)$ is a Lagrangian frame
transverse to $L_x$.  On the set of states $\mathsf S_W(T)$,
we compare the single correct
basis-free report $(K_W(x),\chi_x|_{K_W(x)})$ with the ordered
reference-dual target from Eq.~\eqref{eq:frame-assisted-exact-target}.
The first lemma proves that these two targets determine one another.  The
second lemma evaluates the original verification payoff while allowing
all reports in $\mathsf{Out}_{W,m}$, rather than restricting to
$\Theta_{R(T),m}^{\mathrm{frame}}$.

\begin{lemma}[Equivalence of the generic and frame-relative targets]
\label{lem:fixed-W-frame-target-equivalence}
Fix $(x,T)\in\mathcal I_W$.  Pairing with the first $m$ frame vectors defines the following isomorphism:
\begin{equation}
 K_W(x)\longrightarrow\F_p^m,\quad
 s\longmapsto(\spair{s}{T_1},\ldots,\spair{s}{T_m}).
 \label{eq:fixed-W-pairing-isomorphism}
\end{equation}
Consequently, the basis-free pair
$(K_W(x),\chi_x|_{K_W(x)})$, together with the public frame $R(T)$, uniquely
determines the first $m$ reference-dual stabilizers and their restricted character
\[
 S_1(x;R(T)),\ldots,S_m(x;R(T)).
\]  Conversely, their span
and restricted character recover the basis-free pair.  
\end{lemma}
\begin{proof}
Every element of $K_W(x)$ is orthogonal to the last $n-m$ frame vectors,
because they span $W$.  If it is also orthogonal to $T_1,\ldots,T_m$, it lies
in $L_x\cap M_{R(T)}$, which is zero by incidence.  The displayed map is
therefore injective and, since both spaces have dimension $m$, bijective.  Its
inverse sends the standard vector $e_i$ to the unique element paired as
$\delta_{ij}$ with the first $m$ frame vectors and orthogonal to the remaining
ones.  These are precisely the reference-dual stabilizers.  Taking their span
and transporting the same character restriction gives the inverse output
map.
\end{proof}

The lemma identifies the unique correct reports for a generic input once the
frame is public.  It does not give a bijection between the full basis-free report
alphabet $\mathsf{Out}_{W,m}$ and the frame-relative alphabet
$\Theta_{R(T),m}^{\mathrm{frame}}$.

\begin{lemma}[Outcome-averaged score of a basis-free report]
\label{lem:fixed-W-frame-verification-score}
Fix $(x,T)\in\mathcal I_W$ and a syntactically admissible report
$(\widehat K,\widehat\eta)\in\mathsf{Out}_{W,m}$.  The acceptance probability
of the reported postmeasurement projector, averaged over the
positive-probability outcomes $\nu$ of the measurement associated with
$W$, is
\begin{equation}
 \sum_{\substack{\nu\in\Omega_W\\q_x(\nu)>0}}
 q_x(\nu)
 \Tr\!\left[
   \Pi_{\widehat K,\widehat\eta}^{W,\nu}
   \rho_{x|\nu}
 \right]
 =
 \Tr\!\left[
   \Pi_{\widehat K,\widehat\eta}\rho_x
 \right].
 \label{eq:fixed-W-frame-outcome-averaged-score}
\end{equation}

\end{lemma}
\begin{proof}
Equation~\eqref{eq:fixed-W-outcome-averaged-score} applies to every
report in $\mathsf{Out}_{W,m}$.  Applying that identity to the fixed
input and report in the statement gives
Eq.~\eqref{eq:fixed-W-frame-outcome-averaged-score}.
\end{proof}

Equation~\eqref{eq:fixed-W-frame-outcome-averaged-score} applies to
any admissible basis-free report; its validity does not require a
frame-relative relabeling.  It concerns the average over $\nu$ and
does not assert equal conditional acceptance probabilities for each $\nu$.

\subsubsection{Counting compatible frames and decomposing the input distribution}
\label{subsubsec:transverse-incidence-regularity}

The original uniform-average criterion uses all generic states, whereas
a fixed chart $\mathsf S_W(T)$ contains only the states transverse to
$M_{R(T)}$.  We first show that each generic state $x$ belongs to the
same number $d_W=|\mathfrak T_W(x)|$ of incidence charts.
We then count the state--tuple pairs in $\mathcal I_W$ in two ways.
The resulting probability weights express the uniform generic-state
prior as a mixture of the uniform priors on the nonempty charts.

\begin{lemma}[Transitivity on Lagrangian subspaces with zero intersection with $W$]
\label{lem:fixed-W-generic-lagrangian-transitivity}
Let
\begin{equation}
 \mathcal P_W:=\{g\in\operatorname{Sp}(2n,\F_p):gW=W\}
 \label{eq:fixed-W-parabolic-subgroup}
\end{equation}
be the setwise stabilizer of $W$, where $\operatorname{Sp}(2n,\F_p)$ is the
symplectic group of the phase space $\F_p^{2n}$.  This subgroup preserves the condition
$L\cap W=\{0\}$ and acts transitively on the Lagrangian subspaces satisfying
that condition.  In fact, for any two such Lagrangians $L,L'$, one may choose
$g\in\mathcal P_W$ that fixes $W$ pointwise and satisfies $gL=L'$.
\end{lemma}
\begin{proof}
Put $r:=\dim W=n-m$ and choose a basis $w_1,\ldots,w_r$ of $W$.  For a
Lagrangian $L$ transverse to $W$, the pairing map
\begin{equation}
 W\longrightarrow L^*,\quad
 w\longmapsto\bigl(\ell\mapsto\spair{w}{\ell}\bigr)
 \label{eq:fixed-W-pairing-injection}
\end{equation}
is injective, because its kernel is $W\cap L^\perp=W\cap L=\{0\}$.  Hence
there are vectors $\ell_1,\ldots,\ell_r\in L$ satisfying
$\spair{w_i}{\ell_j}=\delta_{ij}$.  The kernel in $L$ of the pairing with $W$
has dimension $m$, so choose a basis $a_1,\ldots,a_m$ of that kernel.
Consequently, the following ordered tuple is a basis of $L$:
\begin{equation}
 (\ell_1,\ldots,\ell_r,a_1,\ldots,a_m)
 \label{eq:fixed-W-adapted-L-basis}
\end{equation}
Each $a_j$ is orthogonal to $W$.  Construct
$\ell'_1,\ldots,\ell'_r,a'_1,\ldots,a'_m$ in the same way for $L'$, using
the same basis of $W$.

Define $f:W\oplus L\to W\oplus L'$ by
$f(w_i)=w_i$, $f(\ell_i)=\ell'_i$, and $f(a_j)=a'_j$.  The alternating form
vanishes on each of $W,L,L'$, and the displayed adapted bases have the same pairings $\spair{w_i}{\ell_j}=\spair{w_i}{\ell'_j}=\delta_{ij}$ and $\spair{w_i}{a_j}=\spair{w_i}{a'_j}=0$. All pairings within $W$, within $L$, and within $L'$ vanish. Hence the map $f:W\oplus L\to W\oplus L'$ preserves the restricted symplectic form.  The symplectic Witt extension theorem extends
$f$ to an element $g\in\operatorname{Sp}(2n,\F_p)$.  This extension fixes
$W$ pointwise and maps $L$ onto $L'$, proving transitivity.  Finally,
$gL\cap W=g(L\cap W)$ for every $g\in\mathcal P_W$, so the action
preserves the set of Lagrangian subspaces transverse to $W$.
\end{proof}

Lemma~\ref{lem:fixed-W-generic-lagrangian-transitivity} shows that,
for any generic Lagrangians $L,L'$, there is a symplectic map
$g\in\operatorname{Sp}(2n,\F_p)$ with
$g|_W=\operatorname{id}_W$ and $gL=L'$.
For instance, after extending the matched adapted bases in that lemma
to symplectic bases with matrices $B_L,B_{L'}$, one can write
$g=B_{L'}B_L^{-1}$.  Such a map sends each
Lagrangian $M$ containing $W$ and transverse to $L$ to the
Lagrangian $gM$ containing $W$ and transverse to $L'$; its inverse
is induced by $g^{-1}$.  We next count these Lagrangian subspaces
and the ordered tuples that extend the fixed basis of $W$ to their bases.

\begin{lemma}[Regularity of ordered transverse completions]
\label{lem:fixed-W-ordered-completion-regularity}
For a generic Lagrangian $L$, let $\mathcal C_W(L)$ be the set of
Lagrangian subspaces $M$ satisfying
\begin{equation}
 W\subseteq M,\quad L\cap M=\{0\}.
 \label{eq:fixed-W-frame-level-completions}
\end{equation}
Thus $M$ extends $W$ to a Lagrangian subspace while remaining
transverse to $L$; it is not an extension of $L$.  The number
$|\mathcal C_W(L)|$ is independent of the generic Lagrangian $L$.
For each $M\in\mathcal C_W(L)$, the number of ordered $m$-tuples $T$
such that $(T_1,\ldots,T_m,R_{m+1},\ldots,R_n)$ is a basis of $M$ equals
\begin{equation}
 |\operatorname{GL}(m,\F_p)|\,p^{m(n-m)}.
 \label{eq:fixed-W-ordered-completion-fiber-size}
\end{equation}
Consequently, $|\mathfrak T_W(x)|$ is constant over
$x\in\mathsf{Stab}_{n,p}^{\mathrm{gen}}(W)$.
\end{lemma}
\begin{proof}
For generic Lagrangians $L,L'$, choose $g$ as in
Lemma~\ref{lem:fixed-W-generic-lagrangian-transitivity}, so
$g|_W=\operatorname{id}_W$ and $gL=L'$.  The map $M\mapsto gM$ is a
bijection $\mathcal C_W(L)\to\mathcal C_W(L')$: containment of $W$
and the condition $L\cap M=\{0\}$ are preserved, and $g^{-1}$ gives
the inverse.  Thus $|\mathcal C_W(L)|$ is the same for all generic $L$.

Now fix $M\in\mathcal C_W(L)$.  Choose a complement $A_M$ of $W$ in $M$,
so $M=W\oplus A_M$ with $\dim A_M=m$, and choose a basis
$e_1,\ldots,e_m$ of $A_M$.  An ordered tuple $T$ extends the
prescribed ordered basis $(R_{m+1},\ldots,R_n)$ of $W$ to a basis
of $M$ precisely when
\begin{equation}
 T_i=\sum_{j=1}^m e_jG_{ji}+w_i
 \quad(1\le i\le m),\qquad
 G\in\operatorname{GL}(m,\F_p),\quad w_i\in W.
 \label{eq:fixed-W-ordered-completion-parametrization}
\end{equation}
To see the equivalence, the cosets $T_i+W$ form a basis of $M/W$
if and only if the matrix $G$ is invertible.  Adjoining the
prescribed basis of $W$ then gives a basis of $M$.
Conversely, every basis extending that of $W$ determines exactly one
$G$ and one choice of $w_1,\ldots,w_m$ in the fixed decomposition
$M=W\oplus A_M$.  There are $|\operatorname{GL}(m,\F_p)|$ choices for
$G$, and $|W|^m=p^{m(n-m)}$ choices for the $m$ vectors $w_i$.
This proves Eq.~\eqref{eq:fixed-W-ordered-completion-fiber-size}.

Finally, for every generic state $x$, the admissible tuples are
partitioned by their Lagrangian span $M_{R(T)}\in\mathcal C_W(L_x)$.
The count just proved therefore gives
\begin{equation}
 |\mathfrak T_W(x)|
 =|\mathcal C_W(L_x)|\,
 |\operatorname{GL}(m,\F_p)|\,p^{m(n-m)}.
 \label{eq:fixed-W-ordered-completion-regular-degree}
\end{equation}
The first factor is independent of the generic Lagrangian $L_x$,
and the tuple set depends on $L_x$ but not on its eigenvalue
character $\chi_x$.  Hence the displayed number is the same for
all generic stabilizer states $x$.
\end{proof}

By Lemma~\ref{lem:fixed-W-ordered-completion-regularity}, the degree
\begin{equation*}
 d_W:=|\mathfrak T_W(x)|
 =|\mathcal C_W(L_x)|\,|\operatorname{GL}(m,\F_p)|p^{m(n-m)}
\end{equation*}
is independent of the generic state $x$.
The next lemma counts the pairs in $\mathcal I_W$ by their state
component and by their tuple component.  These counts determine the
marginal distribution of $T$ and the conditional distribution of $X$
on each nonempty chart.

\begin{lemma}[Incidence decomposition of the generic prior]
\label{lem:uniform-frame-generic-decomposition}
Choose an incidence pair $(X,T)$ uniformly from $\mathcal I_W$.
Then $X$ is uniform on $\mathsf{Stab}_{n,p}^{\mathrm{gen}}(W)$,
while, conditional on $T=t$, the state $X$ is uniform on
$\mathsf S_W(t)$.  For each $T\in\mathfrak T_W^{\mathrm{inc}}$,
the probability of observing that tuple is
\begin{equation}
 \mu_W(T):=\frac{|\mathsf S_W(T)|}{|\mathcal I_W|}.
 \label{eq:fixed-W-frame-marginal}
\end{equation}
The weights $\mu_W(T)$ are the actual marginal probabilities of the
tuples; they are not assumed equal.  For every real-valued function
$f$ on the generic state family, the uniform generic average satisfies
\begin{equation}
 \frac1{|\mathsf{Stab}_{n,p}^{\mathrm{gen}}(W)|}
 \sum_{x\in\mathsf{Stab}_{n,p}^{\mathrm{gen}}(W)} f(x)
 =\sum_{T\in\mathfrak T_W^{\mathrm{inc}}}\mu_W(T)
 \frac1{|\mathsf S_W(T)|}
 \sum_{x\in\mathsf S_W(T)}f(x).
 \label{eq:fixed-W-incidence-decomposition}
\end{equation}
\end{lemma}
\begin{proof}
Write
$\mathcal G_W:=\mathsf{Stab}_{n,p}^{\mathrm{gen}}(W)$ and let
$d_W:=|\mathfrak T_W(x)|$, which is independent of $x\in\mathcal G_W$
by Lemma~\ref{lem:fixed-W-ordered-completion-regularity}.
Counting $\mathcal I_W$ first by its state component and then by
its tuple component gives
\begin{equation}
 |\mathcal I_W|=d_W|\mathcal G_W|
 =\sum_{T\in\mathfrak T_W^{\mathrm{inc}}}|\mathsf S_W(T)|.
 \label{eq:fixed-W-incidence-double-count}
\end{equation}
The uniform distribution on $\mathcal I_W$ has joint mass function
\begin{equation}
 \Pr\{X=x,T=t\}
 =\frac{\mathbf 1_{\{(x,t)\in\mathcal I_W\}}}{|\mathcal I_W|}.
 \label{eq:fixed-W-incidence-joint-law}
\end{equation}
For $x\in\mathcal G_W$, the cardinality of
$\{t:(x,t)\in\mathcal I_W\}$ is $d_W$ by
Lemma~\ref{lem:fixed-W-ordered-completion-regularity}.
Summing the joint mass over either coordinate gives
\begin{align*}
 \Pr(X=x)=\frac{d_W}{|\mathcal I_W|}
             =\frac1{|\mathcal G_W|},\quad
 \Pr(T=t)=\frac{|\mathsf S_W(t)|}{|\mathcal I_W|}
             =\mu_W(t).
\end{align*}
For every $t\in\mathfrak T_W^{\mathrm{inc}}$, division of the
joint mass by its $T$-marginal yields the conditional input law
\begin{equation}
 \Pr\{X=x\mid T=t\}
 =\frac{\mathbf 1_{\{x\in\mathsf S_W(t)\}}}{|\mathsf S_W(t)|}
 =\frac{\mathbf 1_{\{x\in\mathsf S(R(t))\}}}{|\mathsf S(R(t))|}.
 \label{eq:fixed-W-incidence-conditional-law}
\end{equation}
Applying the law of total expectation to this conditional law gives
the uniform generic-average decomposition.  The following calculation
makes the double count explicit:
\begin{align}
 \sum_{T\in\mathfrak T_W^{\mathrm{inc}}}\mu_W(T)
   \frac1{|\mathsf S_W(T)|}
   \sum_{x\in\mathsf S_W(T)}f(x)
 &=\frac1{|\mathcal I_W|}
   \sum_{T\in\mathfrak T_W^{\mathrm{inc}}}
   \sum_{x\in\mathsf S_W(T)}f(x)\notag\\
 &=\frac{d_W}{|\mathcal I_W|}
   \sum_{x\in\mathcal G_W}f(x)
 =\frac1{|\mathcal G_W|}\sum_{x\in\mathcal G_W}f(x).
 \label{eq:fixed-W-incidence-double-count-proof}
\end{align}
No uniform distribution over the completion tuples has been used.
\end{proof}

\subsubsection{Comparison of the operational optima}
\label{subsubsec:transverse-operational-comparison}
Equation~\eqref{eq:fixed-W-incidence-decomposition} expresses the
generic-state average as a weighted sum of averages on individual charts.
For fixed $n,p,k,W$, define $\mathrm{POVM}_k(\mathsf{Out}_{W,m})$
to be the set of POVMs on $(\mathbb C^{p^n})^{\otimes k}$ indexed by
the entire basis-free alphabet $\mathsf{Out}_{W,m}$.
The original optimization uses one $\mathsf M\in
\mathrm{POVM}_k(\mathsf{Out}_{W,m})$ independent of the index $t$.
The comparison optimization allows a family
$(\mathsf M_t)_{t\in\mathfrak T_W^{\mathrm{inc}}}$ of such POVMs,
with a potentially different measurement for each fixed $t$.
This enlargement of the decision class gives an upper bound on the
original optimum.
For exact identification, the unique correct basis-free report on a
fixed chart corresponds to the refined target.
For verification, the right-hand optimization still uses the
basis-free report alphabet $\mathsf{Out}_{W,m}$, and a report that is
incorrect for exact identification may have positive acceptance
probability.  The comparison therefore does not imply equality with
the refined verification optimum $S_{\mathrm{ver}}^{(k),*}(n,m;p)$.

\begin{proposition}[Comparison with chart-indexed POVMs]
\label{prop:chartwise-comparison-after-revelation}
For a learner POVM $\mathsf M$, let
$F_{\mathsf M}^{(k)}(x;W)$ denote either its exact-success probability
$s_{\mathsf M}^{(k)}(x;W)$ or its verification score
$v_{\mathsf M}^{(k)}(x;W)$.  For $t\in\mathfrak T_W^{\mathrm{inc}}$, let
$\mathsf M_t\in\mathrm{POVM}_k(\mathsf{Out}_{W,m})$ be a decision rule
indexed by $t$.  The left-hand optimization below uses a single POVM
for every $t$; the right-hand optimization permits a different POVM
for each $t$.  For either performance criterion, the following
inequality compares those two decision classes:
\begin{align}
\sup_{\mathsf M}
\frac1{|\mathsf{Stab}_{n,p}^{\mathrm{gen}}(W)|}
\sum_{x\in\mathsf{Stab}_{n,p}^{\mathrm{gen}}(W)}F_{\mathsf M}^{(k)}(x;W)
\le
\sum_{T\in\mathfrak T_W^{\mathrm{inc}}}\mu_W(T)\sup_{\mathsf M_T}
\frac1{|\mathsf S_W(T)|}
\sum_{x\in\mathsf S_W(T)}F_{\mathsf M_T}^{(k)}(x;W).
\label{eq:common-chartwise-revelation-comparison}
\end{align}

\end{proposition}

\begin{proof}
Write $\mathcal G_W:=\mathsf{Stab}_{n,p}^{\mathrm{gen}}(W)$.
For a nonempty chart $T\in\mathfrak T_W^{\mathrm{inc}}$ and any
POVM $\mathsf M$, define its chartwise average by
\begin{equation}
 f_T(\mathsf M):=\frac1{|\mathsf S_W(T)|}
 \sum_{x\in\mathsf S_W(T)}F_{\mathsf M}^{(k)}(x;W).
 \label{eq:chartwise-fixed-POVM-average}
\end{equation}
For each fixed learner POVM $\mathsf M$, applying
Eq.~\eqref{eq:fixed-W-incidence-decomposition} to the function
$x\mapsto F_{\mathsf M}^{(k)}(x;W)$ yields
\begin{align}
\frac1{|\mathsf{Stab}_{n,p}^{\mathrm{gen}}(W)|}
\sum_xF_{\mathsf M}^{(k)}(x;W)
=
\sum_{T\in\mathfrak T_W^{\mathrm{inc}}}\mu_W(T)f_T(\mathsf M).
\label{eq:common-chartwise-revelation-fixed-POVM}
\end{align}
The weights $\mu_W(t)$ are nonnegative and sum to one.
The incidence decomposition and the enlargement of the POVM class
give the annotated comparison
\begin{align}
 \sup_{\mathsf M}\frac1{|\mathcal G_W|}
 \sum_{x\in\mathcal G_W}F_{\mathsf M}^{(k)}(x;W)
 \stackrel{(a)}{=}&
 \sup_{\mathsf M}\sum_{t\in\mathfrak T_W^{\mathrm{inc}}}
 \mu_W(t)f_t(\mathsf M)\notag\\
 \stackrel{(b)}{\le}&
 \sum_{t\in\mathfrak T_W^{\mathrm{inc}}}
 \mu_W(t)\sup_{\mathsf M_t\in\mathrm{POVM}_k(\mathsf{Out}_{W,m})}
 f_t(\mathsf M_t).
 \label{eq:supremum-inside-incidence-mixture}
\end{align}
Here (a) is Eq.~\eqref{eq:fixed-W-incidence-decomposition}.  For (b),
each $f_t(\mathsf M)$ is bounded by its supremum over the fixed basis-free
POVM class; the nonnegative weights sum to one.
Equations~\eqref{eq:common-chartwise-revelation-fixed-POVM} and
\eqref{eq:supremum-inside-incidence-mixture} prove
Eq.~\eqref{eq:common-chartwise-revelation-comparison} for either
performance criterion.
\end{proof}

For exact identification on a fixed chart, the completion tuple $T$
identifies each correct basis-free report with one frame-relative target.
This identification concerns the reports that can be correct on the chart,
not all outputs in $\mathsf{Out}_{W,m}$.  The next corollary compares the
corresponding optimal exact-success probabilities.

\begin{corollary}[Identification of the chartwise exact optimum]
\label{cor:chartwise-fixed-frame-optima}
For every $T\in\mathfrak T_W^{\mathrm{inc}}$ in
Eq.~\eqref{eq:fixed-W-nonempty-chart-support}, the optimal uniform-average
exact success probability on its nonempty incidence chart equals the refined
fixed-frame optimum for the public frame $R(T)$:
\begin{equation}
\sup_{\mathsf M}\frac1{|\mathsf S_W(T)|}
\sum_{x\in\mathsf S_W(T)}s_{\mathsf M}^{(k)}(x;W)
=P_{\rm frame,av}^{(k),*}(R(T),m;p).
\label{eq:chartwise-exact-fixed-frame-optimum}
\end{equation}
\end{corollary}
\begin{proof}
Equation~\eqref{eq:fixed-W-incidence-chart-equality} identifies the chart and
its uniform prior with the fixed-frame state family.
For each input, Lemma~\ref{lem:fixed-W-frame-target-equivalence} identifies
the unique correct basis-free report with the unique reference-dual target
relative to the public frame.  Define the set of basis-free reports that can be
correct for at least one input on this chart by
\begin{equation*}
 \mathcal R_T^{\mathrm{corr}}
 :=\{(K_W(x),\chi_x|_{K_W(x)}):x\in\mathsf S_W(T)\}
 \subseteq\mathsf{Out}_{W,m}.
\end{equation*}
The lemma gives a bijection $r\mapsto y_T(r)$ from
$\mathcal R_T^{\mathrm{corr}}$ to the refined output alphabet
$\Theta_{R(T),m}^{\mathrm{frame}}$.  It does not give a bijection of
$\mathsf{Out}_{W,m}$ with that alphabet.

First consider a basis-free POVM
$\{M_r\}_{r\in\mathsf{Out}_{W,m}}$.
For each refined label, take the POVM effect of its corresponding
$r\in\mathcal R_T^{\mathrm{corr}}$.  Choose a fixed
$y_0\in\Theta_{R(T),m}^{\mathrm{frame}}$ and add to that effect the sum
of all POVM effects with $r\notin\mathcal R_T^{\mathrm{corr}}$.
Those additional outcomes have zero exact-success payoff on this chart.
Explicitly, if $r(y)$ denotes the inverse image of $y$ under the
bijection on correct reports, the effects are
\begin{equation*}
 N_y:=M_{r(y)}+
 \mathbf1_{\{y=y_0\}}\sum_{r\notin\mathcal R_T^{\mathrm{corr}}}M_r.
\end{equation*}
Each $N_y$ is positive, and $\sum_yN_y=I$ because the two sums
partition the original report alphabet.  On $\mathsf S_W(T)$,
reports outside $\mathcal R_T^{\mathrm{corr}}$ never count as an
exact success, so the displayed transformation cannot decrease
exact success.  Hence the basis-free chartwise optimum
is no greater than $P_{\rm frame,av}^{(k),*}(R(T),m;p)$.

Conversely, use the inverse bijection to relabel every output
$y\in\Theta_{R(T),m}^{\mathrm{frame}}$ of a refined POVM by the
corresponding $r\in\mathcal R_T^{\mathrm{corr}}$.
The resulting effects constitute a basis-free POVM, with zero effects on
the other elements of $\mathsf{Out}_{W,m}$, and have the same exact-success
probability on $\mathsf S_W(T)$.
Therefore the reverse inequality holds, proving
Eq.~\eqref{eq:chartwise-exact-fixed-frame-optimum}.
\end{proof}

The original exact-learning worst-case optimum is bounded by the
uniform-average exact optimum on any fixed nonempty chart
$t\in\mathfrak T_W^{\mathrm{inc}}$.  Indeed,
\begin{align}
 &P_{W,\rm wc}^{(k),*}(n,m;p)
 \stackrel{(a)}{\le}P_{W,\rm gen,wc}^{(k),*}
 \stackrel{(b)}{\le}
 \sup_{\mathsf M}\min_{x\in\mathsf S_W(t)}s_{\mathsf M}^{(k)}(x;W)
 \notag\\
 \stackrel{(c)}{\le}&
 \sup_{\mathsf M}\frac1{|\mathsf S_W(t)|}
 \sum_{x\in\mathsf S_W(t)}s_{\mathsf M}^{(k)}(x;W)
\stackrel{(d)}{=}P_{\rm frame,av}^{(k),*}(R(t),m;p)
 \stackrel{(e)}{=}P_{\mathrm{id}}^{(k),*}(n,m;p).
 \label{eq:fixed-W-wc-comparison}
\end{align}
Here (a) and (b) restrict the minimization to the generic state set
and then to the nonempty subset $\mathsf S_W(t)$; (c) is the
minimum-to-average inequality.  Equality (d) is
Corollary~\ref{cor:chartwise-fixed-frame-optima}, and (e) follows
from Proposition~\ref{prop:frame-conditioned-coordinate-reduction}.

For verification, we do not use a chartwise worst-case comparison. We begin
with the general inequality
\begin{equation}
 S_{W,\rm ver,wc}^{(k),*}(n,m;p)
 \le S_{W,\rm ver,av}^{(k),*}(n,m;p),
 \label{eq:fixed-W-ver-wc-to-av}
\end{equation}
and later split the full uniform prior into its generic and exceptional parts.

For exact success, define $f_t(\mathsf M)$ as in
Eq.~\eqref{eq:chartwise-fixed-POVM-average} with
$F_{\mathsf M}^{(k)}=s_{\mathsf M}^{(k)}$.
Then the generic uniform-average optimum satisfies
\begin{align}
 P_{W,\rm gen,av}^{(k),*}
\stackrel{(a)}{=}&
 \sup_{\mathsf M}\sum_{t\in\mathfrak T_W^{\mathrm{inc}}}
 \mu_W(t)f_t(\mathsf M)
 \stackrel{(b)}{\le}
 \sum_{t\in\mathfrak T_W^{\mathrm{inc}}}\mu_W(t)
 \sup_{\mathsf M_t}f_t(\mathsf M_t)\notag\\
\stackrel{(c)}{=}&
 \sum_{t\in\mathfrak T_W^{\mathrm{inc}}}\mu_W(t)
 P_{\rm frame,av}^{(k),*}(R(t),m;p)
 \stackrel{(d)}{=}P_{\mathrm{id}}^{(k),*}(n,m;p).
 \label{eq:fixed-W-generic-average-comparison}
\end{align}
Here (a) is the incidence decomposition in
Lemma~\ref{lem:uniform-frame-generic-decomposition}; (b) allows
an independent POVM $\mathsf M_t$ for each index $t$;
(c) is Corollary~\ref{cor:chartwise-fixed-frame-optima}; and
(d) uses Proposition~\ref{prop:frame-conditioned-coordinate-reduction}
and $\sum_t\mu_W(t)=1$.
The exact comparison is now complete.  The basis-free verification
comparison is proved in Section~\ref{sec:integrated-all-prime-graded}.
We next bound the exceptional mass $\delta_{n,m,p}$ defined in
Eq.~\eqref{eq:fixed-W-delta} to restore the full input prior.
A fixed one-dimensional subspace is contained in a uniformly random Lagrangian with
probability $1/(p^n+1)$.  The union bound over the
$(p^{n-m}-1)/(p-1)$ one-dimensional subspaces of $W$ gives
\begin{equation}
 \delta_{n,m,p}\le\frac{p^{n-m}-1}{(p-1)(p^n+1)}=O(p^{-m}).
 \label{eq:fixed-W-exceptional-bound}
\end{equation}
In particular, if $m_n=\beta n+o(n)$ for a fixed $\beta>0$, the exceptional mass tends to zero.

Equation~\eqref{eq:fixed-W-generic-average-comparison} completes the
generic exact-learning comparison.  Under the full uniform prior, exceptional
inputs contribute at most their
total probability $\delta_{n,m,p}$, because the success probability of
each input is at most one.  For the worst-case criterion, restricting the
infimum to generic inputs already gives an upper bound.  These estimates
complete the ingredients for Theorem~\ref{thm:fixed-W-operational-comparison}.

\begin{proof}[Proof of Theorem~\ref{thm:fixed-W-operational-comparison}]
\noindent\textbf{Part (1): Exact learning.}
The worst-case bound is Eq.~\eqref{eq:fixed-W-wc-comparison}.
For the average, decompose the full prior into generic mass
$1-\delta_{n,m,p}$ and exceptional mass $\delta_{n,m,p}$.
The exceptional conditional success is at most one, whereas
Eq.~\eqref{eq:fixed-W-generic-average-comparison} bounds the generic
conditional optimum by $P_{\mathrm{id}}^{(k),*}(n,m;p)$.  Taking the supremum over the original POVM class gives
Eq.~\eqref{eq:fixed-W-av-comparison}, which completes Part~(1).

\medskip
\noindent\textbf{Part (2): Verification.}
For every learner POVM, the minimum statewise verification score is at
most its uniform input average.  Taking the supremum over that same
POVM class gives
Eq.~\eqref{eq:fixed-W-verification-wc-average-only}.
\end{proof}

For the basis-free verification game, a finite comparison requires a
separate argument.  At a chosen rank-distance cutoff $D$, the learner's
reports within distance $D$ of the true target are converted into exact
guesses by randomizing over the compatible exact-target list.  Reports
outside that list contribute a verification score of at most
$p^{-(D+1)}$.  The resulting bound is the list-size factor times the
exact Bayes optimum, plus this tail term.  Section~\ref{sec:integrated-all-prime-graded}
defines and bounds the basis-free lists and proves this comparison in
Theorem~\ref{thm:generic-original-bounded-rank-graded-to-exact}.

\subsection{Covariance of the prescribed constraint}
\label{subsec:constraint-covariance-proof}
We prove Proposition~\ref{prop:fixed-W-dimension-only-covariance}.
The same Clifford transport must preserve the whole input and report
alphabets, including exceptional states, to compare the original optima.

\begin{proof}[Proof of Proposition~\ref{prop:fixed-W-dimension-only-covariance}]
A symplectic map $g$ sends $W$ onto $W'$.  Choose an implementing Clifford
unitary $U_g$.  Relative to the physical Pauli or Weyl section
$v\mapsto P_v$ used to define the state and report characters, let
$\tau_g(v)$ be the phase determined by
\begin{equation}
 U_gP_vU_g^\dagger=\tau_g(v)P_{gv}.
 \label{eq:fixed-W-clifford-section-phase}
\end{equation}
For odd $p$, $\tau_g(v)$ is a $p$th-root phase in the chosen Weyl convention;
for $p=2$ and Hermitian Pauli representatives, it lies in $\{\pm1\}$ and
includes, in particular, the possible sign produced by Clifford conjugation.

The unitary $U_g$ maps each stabilizer state $x$ bijectively to a stabilizer
state $gx$ and conjugates its $k$-copy density operator accordingly.  For a
joint character $\eta$ on $K$, define the transported joint character
$g_*\eta$ on $gK$ by
\begin{equation}
 (g_*\eta)(gv):=\tau_g(v)^{-1}\eta(v),
 \qquad v\in K.
 \label{eq:fixed-W-clifford-character-transport}
\end{equation}
Because $g$ is bijective, Eq.~\eqref{eq:fixed-W-clifford-character-transport}
assigns one phase to each element of $gK$.  We verify compatibility with
the physical-section cocycle $c_p$, defined in
Eq.~\eqref{eq:global-Pauli-section-cocycle}.
Conjugating this product and using
Eq.~\eqref{eq:fixed-W-clifford-section-phase} yields
\begin{equation}
 c_p(u,v)\tau_g(u+v)
 =\tau_g(u)\tau_g(v)c_p(gu,gv).
 \label{eq:fixed-W-clifford-cocycle-transport}
\end{equation}
For commuting $u,v\in K$, the compatibility relation
$\eta(u)\eta(v)=c_p(u,v)\eta(u+v)$ therefore gives
\begin{align}
 (g_*\eta)(gu)(g_*\eta)(gv)
 &=\tau_g(u)^{-1}\tau_g(v)^{-1}c_p(u,v)\eta(u+v)\notag\\
 &=c_p(gu,gv)\tau_g(u+v)^{-1}\eta(u+v)
 =c_p(gu,gv)(g_*\eta)(g(u+v)).
 \label{eq:fixed-W-clifford-character-compatibility}
\end{align}
Thus $g_*\eta$ is a compatible joint character on $gK$.
The true stabilizer character obeys the same transport rule,
\begin{equation}
 \chi_{gx}(gv)=\tau_g(v)^{-1}\chi_x(v),
 \qquad v\in L_x.
 \label{eq:fixed-W-true-character-transport}
\end{equation}
Hence validity and exact correctness are preserved.  The subspace conditions
are preserved because
\begin{equation}
 g(L_x\cap W^\perp)=L_{gx}\cap (W')^\perp,
 \quad gK\cap W'=g(K\cap W)=\{0\}.
 \label{eq:fixed-W-covariance-valid-output}
\end{equation}

The projector of a report $(K,\eta)$ transforms into the projector
of its Clifford-transported report $(gK,g_*\eta)$.  Applying
Eq.~\eqref{eq:fixed-W-clifford-section-phase} to the projector definition
gives the following operator identity:
\begin{align}
 U_g\Pi_{K,\eta}U_g^\dagger
 =p^{-m}\sum_{v\in K}\overline{\eta(v)}\,
   \tau_g(v)P_{gv}
 =p^{-m}\sum_{v\in K}\overline{(g_*\eta)(gv)}\,P_{gv}
 =\Pi_{gK,g_*\eta}.
 \label{eq:fixed-W-report-projector-covariance}
\end{align}
The same conjugation transports the full physical input ensemble, including
the exceptional states, and every report in $\mathsf{Out}_{W,m}$.
Thus conjugating learner POVMs and relabeling their outcomes preserves every
exact-success term and every Born verification score.  Applying the same
construction with $U_g^\dagger$ gives the inverse transport.  Therefore the
four optima for $W$ and $W'$ are equal.  The phase function depends on the
chosen implementing Clifford and section, but the resulting correspondence is
unitary and invertible, so the optimized values do not depend on that choice.
\end{proof}
\section{Quadratic stabilizer bases over odd prime fields}
\label{sec:odd-prime-foundations}

For odd prime local dimension, the standardized frame identifies each transverse Lagrangian and its character with an explicit quadratic stabilizer state.  In these coordinates,
the requested reference-dual directions become the first $m$ columns of
$(I_n;A)$ together with their character restriction.  This representation
makes partial labels, overlaps, tensor powers, and the finite-copy analysis
explicit.

Throughout this section, $p$ is an odd prime, $\F_p$ is the field with
$p$ elements, and
\begin{equation}
 \omega:=e^{2\pi i/p}.
 \label{eq:omega-definition}
\end{equation}
All vector and matrix operations appearing in exponents are evaluated over
$\F_p$.  In particular, $1/2$ denotes the multiplicative inverse of $2$ in
$\F_p$, rather than the real number interpreted before reduction modulo $p$.
The Hilbert space of $n$ qudits is
$\mathcal H_n:=(\mathbb C^p)^{\otimes n}$, with computational basis
$\{\ket{x}:x\in\F_p^n\}$.

\subsection{Weyl operators and the symmetric-matrix normal form}
\label{subsec:weyl-symmetric-matrix-normal-form}

For $u,v\in\F_p^n$, define the shift and phase operators by
\begin{equation}
  X(u)|x\rangle:=|x+u\rangle,
  \quad
  Z(v)|x\rangle:=\omega^{v^Tx}|x\rangle.
  \label{eq:odd-XZ-definitions}
\end{equation}
They obey
\begin{equation}
  X(u)Z(v)=\omega^{-u^Tv}Z(v)X(u).
  \label{eq:odd-XZ-commutation}
\end{equation}
We use the phase-adjusted Weyl operators
\begin{equation}
  W(u,v):=\omega^{-\frac12u^Tv}Z(v)X(u).
  \label{eq:odd-Weyl-definition}
\end{equation}
With the commutation convention in Eq.~\eqref{eq:odd-XZ-commutation}, this definition gives the multiplication rule
\begin{equation}
  W(u,v)W(u',v')
  =\omega^{-\frac12\spair{(u,v)}{(u',v')}}
   W(u+u',v+v').
  \label{eq:odd-Weyl-multiplication}
\end{equation}

Let $\Sym_n(\F_p)$ denote the space of symmetric $n\times n$ matrices over
$\F_p$.
For $A\in\Sym_n(\F_p)$, define the Lagrangian subspace in symmetric-matrix normal form
\begin{equation}
  L_A:=\left\{\binom{u}{Au}:u\in\F_p^n\right\}
  \subseteq\F_p^{2n}.
  \label{eq:symmetric-matrix-normal-form-LA}
\end{equation}
The symmetry of $A$ implies
\begin{equation}
  u^TAu'-(Au)^Tu'=0
  \quad
  (u,u'\in\F_p^n).
  \label{eq:symmetric-matrix-normal-form-isotropy}
\end{equation}
Consequently, the Weyl operators indexed by $L_A$ form an abelian
representation of the additive group $\F_p^n$:
\begin{equation}
  W(u,Au)W(u',Au')=W(u+u',A(u+u')).
  \label{eq:symmetric-matrix-Weyl-group}
\end{equation}
Thus, $A$ specifies a maximal commuting Weyl family and its associated joint eigenbasis.  It does not, by itself, specify one particular state in
the associated joint eigenbasis; the missing information is an eigenvalue
character.

\begin{lemma}[Symmetric-matrix normal form for transverse Lagrangians]
\label{lem:transverse-symmetric-matrix-normal-form}
Let $V=\{(0,v):v\in\F_p^n\}$. An $n$-dimensional subspace
$L\subseteq\F_p^{2n}$ is Lagrangian and satisfies $L\cap V=\{0\}$ if and only if
there is a unique $A\in\Sym_n(\F_p)$ with $L=L_A$.
\end{lemma}

\begin{proof}
For $L_A$, symmetry gives isotropy, dimension $n$ gives maximality, and
$L_A\cap V=\{0\}$. Conversely, transversality makes the projection of $L$ onto
the first component an isomorphism, so $L$ is the graph of a unique linear map
$A$. Isotropy gives $u^T(A^T-A)u'=0$ for all $u,u'$, hence $A=A^T$.
\end{proof}

The symmetric-matrix normal form describes $p^{n(n+1)/2}$ Lagrangian
subspaces out of a total of $\prod_{j=1}^n(p^j+1)$
\cite{Gross2006,SingalEtAl2023Counting}.  Their fraction is
$\prod_{j=1}^n(1+p^{-j})^{-1}$, which stays positive for fixed $p$.
Each Lagrangian supports $p^n$ compatible eigenvalue assignments, so the
corresponding stabilizer states have the same fraction among all pure
stabilizer states.  This state family is asymptotically non-negligible
but remains a proper subset of the full stabilizer-state ensemble.

\subsection{Joint eigenvectors and eigenvalue characters}
\label{subsec:odd-joint-eigenvectors}
The symmetric-matrix normal form specifies the maximal commuting Weyl family,
but a stabilizer state also requires the choice of its eigenvalue character.
The following quadratic-phase representation combines these two pieces of
information in a single explicit state vector.  It provides a
calculation-friendly label for every stabilizer state whose Lagrangian
subspace lies in the present transverse normal-form domain.  As will be
verified below, the label is unique at the level of projective states.  This
representation makes the eigenvalue character, pairwise overlaps, tensor-power
expansions, averaging over unreported coordinates, and the partial-label map
directly accessible from the parameters $(A,b)$.

For $A\in\Sym_n(\F_p)$ and $b\in\F_p^n$, define
\begin{equation}
  |\psi_{A,b}\rangle
  :=p^{-n/2}\sum_{x\in\F_p^n}
  \omega^{\frac12x^TAx+b^Tx}|x\rangle.
  \label{eq:odd-quadratic-state}
\end{equation}
Every computational-basis amplitude has modulus $p^{-n/2}$. We therefore call
this the symmetric-matrix or full-support quadratic stabilizer family.  A general
stabilizer state may instead have support on a proper affine subspace.  Such
states are excluded only from this adapted coordinate realization of the
reference-relative partial tasks; they remain part of the unrestricted complete-identification
ensemble in Eqs.~\eqref{eq:unrestricted-stabilizer-wc-success} and~\eqref{eq:unrestricted-stabilizer-av-success}.

The next lemma identifies both the commuting operator family and its
joint eigenvalues.

\begin{lemma}[Eigenvalue character]
\label{lem:odd-eigenvalue-character}
For each $u\in\F_p^n$, the quadratic state
$|\psi_{A,b}\rangle$ is an eigenvector of $W(u,Au)$ with eigenvalue
$\omega^{-b^Tu}$:
\begin{equation}
  W(u,Au)|\psi_{A,b}\rangle
  =\omega^{-b^Tu}|\psi_{A,b}\rangle.
  \label{eq:odd-eigenvalue-character}
\end{equation}
Hence $|\psi_{A,b}\rangle$ is a joint eigenvector of
$\{W(u,Au):u\in\F_p^n\}$, with character
\begin{equation}
  \chi_b(u):=\omega^{-b^Tu}.
  \label{eq:odd-character}
\end{equation}
\end{lemma}

\begin{proof}
Using Eqs.~\eqref{eq:odd-XZ-definitions} and
\eqref{eq:odd-Weyl-definition}, we obtain
\begin{align}
  W(u,Au)|\psi_{A,b}\rangle
  &={p^{-n/2}}
    \sum_{x\in\F_p^n}
    \omega^{\frac12x^TAx+b^Tx-\frac12u^TAu+(Au)^T(x+u)}
    |x+u\rangle.                                          \label{eq:eigenproof-1}
\end{align}
Set $y=x+u$.  Since $A^T=A$, the exponent in
Eq.~\eqref{eq:eigenproof-1} becomes
\begin{equation}
  \frac12(y-u)^TA(y-u)+b^T(y-u)-\frac12u^TAu+(Au)^Ty
  =\frac12y^TAy+b^Ty-b^Tu.
  \label{eq:eigenproof-exponent}
\end{equation}
Factoring out $\omega^{-b^Tu}$ gives
Eq.~\eqref{eq:odd-eigenvalue-character}.
\end{proof}

The phase-unadjusted Pauli operators have the equivalent eigenvalue relation
\begin{equation}
  Z(Au)X(u)|\psi_{A,b}\rangle
  =\omega^{\frac12u^TAu-b^Tu}|\psi_{A,b}\rangle.
  \label{eq:odd-unadjusted-eigenvalue}
\end{equation}
The quadratic term in Eq.~\eqref{eq:odd-unadjusted-eigenvalue} is a consequence
of the phase convention for $Z(Au)X(u)$.  It disappears in the Weyl convention
\eqref{eq:odd-Weyl-definition}, leaving the linear character
\eqref{eq:odd-character}.

If one prefers stabilizer operators with eigenvalue $+1$, set
\begin{equation}
  S_{A,b}(u)
  :=\omega^{b^Tu}W(u,Au)
  =\omega^{b^Tu-\frac12u^TAu}Z(Au)X(u).
  \label{eq:odd-plus-one-stabilizer}
\end{equation}
Each operator $S_{A,b}(u)$ fixes the state $|\psi_{A,b}\rangle$:
\begin{equation}
  S_{A,b}(u)|\psi_{A,b}\rangle=|\psi_{A,b}\rangle
  \quad (u\in\F_p^n),
  \label{eq:odd-plus-one-property}
\end{equation}
and $\{S_{A,b}(u):u\in\F_p^n\}$ is the stabilizer group of
$|\psi_{A,b}\rangle$ in the present convention.

\subsection{Orthonormality and pairwise overlaps of quadratic stabilizer states}
\label{subsec:odd-measurement-basis}
\label{subsec:odd-overlaps}
We first prove orthonormality for fixed $A$, and then evaluate overlaps
for different symmetric matrices.
For each fixed $A$, the joint eigenstates
$|\psi_{A,b}\rangle$, indexed by $b\in\F_p^n$, form an
orthonormal basis of $\mathcal H_n$, as shown below.
Denote this family by
\begin{equation}
  \mathcal B_A
  :=\{|\psi_{A,b}\rangle:b\in\F_p^n\}.
  \label{eq:basis-BA-definition}
\end{equation}

\begin{lemma}[Fixed-$A$ orthonormality]
\label{lem:fixed-A-orthonormality}
For every fixed $A\in\Sym_n(\F_p)$, $\mathcal B_A$ is an orthonormal basis of
$\mathcal H_n$.  More explicitly,
\begin{equation}
  \langle\psi_{A,b'}|\psi_{A,b}\rangle=\delta_{b,b'}.
  \label{eq:fixed-A-orthonormality}
\end{equation}
\end{lemma}

\begin{proof}
The quadratic phases cancel, and therefore
\begin{align}
  \langle\psi_{A,b'}|\psi_{A,b}\rangle
  =p^{-n}\sum_{x\in\F_p^n}\omega^{(b-b')^Tx}
  =\delta_{b,b'},
  \label{eq:character-orthogonality}
\end{align}
by character orthogonality on $\F_p^n$.  Since the family contains $p^n$
orthonormal vectors in the $p^n$-dimensional space $\mathcal H_n$, it is an
orthonormal basis.
\end{proof}

We now compare two states with possibly different symmetric matrices
$A$ and $A'$.  Their overlap is the normalized quadratic Gauss sum
\begin{equation}
  \langle\psi_{A',b'}|\psi_{A,b}\rangle
  =p^{-n}\sum_{x\in\F_p^n}
   \omega^{\frac12x^T(A-A')x+(b-b')^Tx}.
  \label{eq:odd-overlap-Gauss-sum}
\end{equation}
The next proposition provides the relation needed in the operational analysis.

\begin{proposition}[Rank-dependent overlap]
\label{prop:rank-dependent-overlap}
Let
\begin{equation}
  \Delta A:=A-A',
  \quad
  \Delta b:=b-b',
  \quad
  t:=\rank(\Delta A).
  \label{eq:overlap-differences}
\end{equation}
The squared overlap of the two pure states is determined by the rank $t$ and the compatibility of the linear term:
\begin{equation}
  \lvert\langle\psi_{A',b'}|\psi_{A,b}\rangle\rvert^2
  =
  \begin{cases}
    p^{-t},&\Delta b\in\im(\Delta A),\\
    0,&\Delta b\notin\im(\Delta A).
  \end{cases}.
  \label{eq:rank-dependent-overlap}
\end{equation}
\end{proposition}

\begin{proof}
The complete calculation is given in Appendix~\ref{app:quadratic-Gauss-sums}. We record the two decisive steps here.  Since $\Delta A$ is symmetric, every vector $\Delta Ay$ is orthogonal
to $\ker\Delta A$, so
$\im(\Delta A)\subseteq(\ker\Delta A)^\perp$.  Both spaces have
dimension $\rank(\Delta A)$ by rank--nullity; therefore
\begin{equation}
  \im(\Delta A)=(\ker\Delta A)^\perp.
  \label{eq:image-kernel-orthogonal}
\end{equation}
If $\Delta b\notin\im(\Delta A)$, there exists
$z\in\ker\Delta A$ with $\Delta b^Tz\ne0$.  Translating the sum in
Eq.~\eqref{eq:odd-overlap-Gauss-sum} by $z$ multiplies it by the nontrivial
phase $\omega^{\Delta b^Tz}$; hence the sum vanishes.

Suppose instead that $\Delta b\in\im(\Delta A)$.  Choose $y$ satisfying
$\Delta Ay=-\Delta b$.  Translating $x\mapsto x+y$ completes the square and
reduces the absolute value to that of the quadratic Gauss sum associated with
$\Delta A$.  After a congruence transformation, $\Delta A$ is diagonal with
exactly $t$ nonzero entries.  Each nondegenerate one-dimensional Gauss sum has
absolute value $p^{1/2}$, while each null direction contributes a factor $p$.
Thus the sum in Eq.~\eqref{eq:odd-overlap-Gauss-sum} has absolute value
$p^{n-t/2}$, and the normalized squared absolute value is $p^{-t}$.
\end{proof}

The map $(A,b)\mapsto|\psi_{A,b}\rangle$ is injective even at the level of
projective states.  If $A=A'$, fixed-$A$ orthonormality implies that
$|\psi_{A,b}\rangle$ and $|\psi_{A,b'}\rangle$ represent the same ray only when
$b=b'$.  If $A\ne A'$, then $t=\rank(A-A')\ge1$, and
Proposition~\ref{prop:rank-dependent-overlap} gives squared overlap either zero
or $p^{-t}\le p^{-1}<1$.  Hence states with distinct symmetric matrices cannot represent
the same ray.  In particular, the ensemble contains exactly $p^{n(n+3)/2}$ distinct
projective states.

The overlap formula also has a geometric interpretation.  Its
nonzero condition is equivalent to agreement of the two eigenvalue
characters on the intersection of their stabilizer subspaces.  We state
this interpretation explicitly for use in the later counting arguments.

\begin{corollary}[Common stabilizer structure]
\label{cor:common-stabilizer-structure}
When $t=\rank(A-A')$, the common Lagrangian subspace and its dimension are
\begin{equation}
 L_A\cap L_{A'}=\left\{\binom{u}{Au}:u\in\ker(A-A')\right\},
 \quad \dim(L_A\cap L_{A'})=n-t.
 \label{eq:common-lagrangian-intersection}
\end{equation}
Moreover, $b-b'\in\im(A-A')$ if and only if the two eigenvalue characters
agree on this common stabilizer subgroup.
\end{corollary}

\begin{proof}
The intersection formula follows directly from $(u,Au)=(u,A'u)$ and
rank--nullity. Character agreement is $(b-b')^Tu=0$ for every
$u\in\ker(A-A')$, which is equivalent to
$b-b'\in(\ker(A-A'))^\perp=\im(A-A')$ because $A-A'$ is symmetric.
\end{proof}

Proposition~\ref{prop:rank-dependent-overlap} separates the discrepancy
between two labels into two parts.  The matrix discrepancy $A-A'$ determines
the rank penalty $p^{-t}$, whereas the character discrepancy $b-b'$ determines
whether the overlap is nonzero at all.  In particular, if $A'=A$, then
$t=0$ and Eq.~\eqref{eq:rank-dependent-overlap} reduces to
Eq.~\eqref{eq:fixed-A-orthonormality}; the outcome $b'$ occurs with
probability one exactly when $b'=b$.

For odd $p$, the scalar $1/2$ in
Eq.~\eqref{eq:odd-quadratic-state} is defined in $\F_p$, and the
corresponding phases are $p$th roots of unity.  Division by two is
unavailable over $\F_2$.  The qubit construction in the following
section therefore uses a separate quadratic-phase convention.

\section{Characteristic-two symmetric-matrix coordinates}
\label{sec:integrated-qubit-foundations}

For qubits, division by two is unavailable in $\F_2$.  The same operational coordinate model therefore requires a $\mathbb Z_4$-valued quadratic enhancement.  This replacement preserves both the partial-learning target and the finite-copy formulation used below.

\subsection{Enhanced quadratic phases}
\label{subsec:integrated-qubit-phases}
In characteristic two, the ordinary phase $\frac12x^TAx$ used for odd primes
is unavailable.  We therefore introduce a $\mathbb Z_4$-valued quadratic
enhancement.  It plays the same role as the odd-prime quadratic phase: it
combines the symmetric-matrix normal form and the eigenvalue character into an
explicit state vector.  The resulting representation is unique at the level
of projective states and makes the later operator, overlap, tensor-power, and
finite-copy calculations tractable.

For a bit $a\in\F_2$, let $\widetilde a\in\{0,1\}\subset\mathbb Z$
denote its canonical integer lift.  All sums defining a $\mathbb Z_4$-valued
phase below are integer sums of these lifted bits, followed by reduction
modulo four.  Let $A\in\Sym_n(\F_2)$ and $b\in\F_2^n$. Define the following $\mathbb Z_4$-valued phase function on $\F_2^n$:
\begin{equation}
 q_{A,b}(x):=
 \sum_{j=1}^n \widetilde A_{jj}\widetilde x_j
 +2\sum_{1\le j<\ell\le n}
   \widetilde A_{j\ell}\widetilde x_j\widetilde x_\ell
 +2\sum_{j=1}^n\widetilde b_j\widetilde x_j\pmod4.
 \label{eq:integrated-qubit-phase}
\end{equation}
The associated normalized state vector is
\begin{equation}
 \ket{\psi_{A,b}^{(2)}}
 :=2^{-n/2}\sum_{x\in\F_2^n}i^{q_{A,b}(x)}\ket{x}.
 \label{eq:integrated-qubit-state}
\end{equation}
Equivalently, set
\begin{equation}
 \alpha_j:=A_{jj}+2b_j\pmod4,
 \quad E_{j\ell}:=A_{j\ell}\quad(j<\ell).
 \label{eq:integrated-qubit-additive-coordinates}
\end{equation}
Then $(\alpha,E)$ belongs to the additive group
\begin{equation}
 G_{n,2}:=\mathbb Z_4^n\times\F_2^{\binom n2},
 \quad |G_{n,2}|=2^{n(n+3)/2},
 \label{eq:integrated-qubit-group}
\end{equation}
and Eq.~\eqref{eq:integrated-qubit-phase} is additive in these coordinates.
For $g=(\alpha,E)\in G_{n,2}$, we write $\ket{\psi_g^{(2)}}$ for the
state $\ket{\psi_{A,b}^{(2)}}$ determined by the bijection in
Eq.~\eqref{eq:integrated-qubit-additive-coordinates}.
This is the characteristic-two replacement for the odd-prime quadratic phase
$\frac12x^TAx+b^Tx$.

\subsection{The fixed query and its quotient}
\label{subsec:integrated-qubit-query}
As in the odd-prime realization, after standardizing the public frame the
columns of $(J_m;AJ_m)$ represent the first $m$ stabilizer directions
symplectically dual to that frame.  Thus the coordinates introduced here are a
characteristic-two representation of the same intrinsic partial target, not a
separate query.  For the canonical inclusion $J_m=(I_m\ 0)^T$, the queried
symmetric-matrix action is $AJ_m$ and the induced character is $J_m^Tb$.  The sign used in the odd-prime
Weyl convention causes no ambiguity here because $-b=b$ over $\F_2$.
The quotient element records the queried label in the enhanced additive
coordinates, whereas $(AJ_m,J_m^Tb)$ is its phase-section-dependent vector
representation.  Equality of queried labels is unambiguous, but comparison
of character representatives defined by different phase sections requires the
correction $\delta_{A,C}$ introduced below.
In the additive coordinates of Eq.~\eqref{eq:integrated-qubit-additive-coordinates},
the kernel of the partial-label map is the embedded subgroup
\begin{align}
 H_{n,m,2}:=\{(\alpha,E)\in G_{n,2}:{}&
 \alpha_i=0\ \text{for }1\le i\le m,\notag\\[-1mm]
 &E_{ij}=0\ \text{whenever }i\le m\text{ or }j\le m\}.
 \label{eq:integrated-qubit-unreported-subgroup}
\end{align}
Only the coordinates $\alpha_i$ with $i>m$ and the edge coordinates
$E_{ij}$ with $m<i<j\le n$ may vary.  The unreported-coordinate subgroup therefore has the
following direct-product description:
\begin{equation}
 H_{n,m,2}\cong
 \mathbb Z_4^{n-m}\times\F_2^{\binom{n-m}{2}}.
 \label{eq:integrated-qubit-unreported-subgroup-isomorphism}
\end{equation}
The partial-label space is
\begin{equation}
 Y_{n,m,2}:=G_{n,2}/H_{n,m,2}.
 \label{eq:integrated-qubit-label-quotient}
\end{equation}
Its cardinality is
\begin{equation}
 \lvert Y_{n,m,2}\rvert
 =2^{(n+1)m-m(m-1)/2}.
 \label{eq:integrated-qubit-label-cardinality}
\end{equation}
Thus the base-$p$ visible-label exponent used later has the same expression for
$p=2$ and for odd $p$.

The quotient $Y_{n,m,2}$ and the binary coordinate-report set
$\mathcal C_{n,m}\times\F_2^m$ encode the same visible information.
The precise correspondence is a bijection of sets, not an identification
of their additive group structures.  For a quotient class represented by
$(\alpha,E)\in G_{n,2}$, set
\begin{align}
 C_{ii}&=\alpha_i\bmod2,
 &\gamma_i&=\frac{\alpha_i-\widetilde{C_{ii}}}{2}\bmod2
 &&(1\le i\le m),\notag\\
 C_{ij}&=C_{ji}=E_{ij}
 &&(1\le i<j\le m),\notag\\
 C_{ji}&=E_{ij}
 &&(1\le i\le m<j\le n).
 \label{eq:binary-quotient-coordinate-set-map}
\end{align}
In the first line, $\alpha_i$ is a residue modulo four and
$\widetilde{C_{ii}}\in\{0,1\}$ is the canonical lift of its parity.
The difference $\alpha_i-\widetilde{C_{ii}}$ is even.  Replacing an
integer representative of $\alpha_i$ by that representative plus four
changes the half-difference by two, so $\gamma_i$ is well defined
modulo two.  Since $E_{ij}$ has indices $1\leq i<j\leq n$, its visible
entries are precisely those with $i\leq m$; these give the off-diagonal
entries of the symmetric upper block and every lower row of $C$.
Equation~\eqref{eq:binary-quotient-coordinate-set-map} specifies
every entry of the partial symmetric matrix $C\in\mathcal C_{n,m}$.
Adding an element of $H_{n,m,2}$ leaves each visible $\alpha_i$ and
$E_{ij}$ unchanged.  Hence the coordinate label depends only on the
quotient class, not on its representative.

For the inverse on quotient classes, given
$(C,\gamma)\in\mathcal C_{n,m}\times\F_2^m$, choose the visible
enhanced coordinates by
\begin{equation}
 \alpha_i=\widetilde{C_{ii}}+2\widetilde{\gamma_i}\pmod4
 \quad(1\le i\le m),
 \label{eq:binary-coordinate-set-map-inverse}
\end{equation}
and recover $E_{ij}$ from the corresponding off-diagonal entry of $C$
whenever $i\le m$.  The remaining coordinates represent the subgroup
$H_{n,m,2}$ and therefore do not affect the quotient class.
The coordinate map in Eq.~\eqref{eq:binary-quotient-coordinate-set-map}
and the inverse construction in
Eq.~\eqref{eq:binary-coordinate-set-map-inverse} recover each other's
visible coordinates, so they define inverse bijections of label sets.
By
Eq.~\eqref{eq:integrated-qubit-additive-coordinates}, a complete state
with parameters $(A,b)$ maps to $(AJ_m,J_m^Tb)$, in agreement with
Eq.~\eqref{eq:frame-intrinsic-to-coordinate-label-map}.

The parameter $\gamma$ in this correspondence is a coordinate \emph{label};
the set bijection is not an additive homomorphism in general.  For
$n=m=1$, $Y_{1,1,2}\cong\mathbb Z_4$, while
$\mathcal C_{1,1}\times\F_2\cong\F_2^2$ under componentwise addition.
Thus binary Fourier analysis and differences of partial labels continue
to use the enhanced quotient group law on $Y_{n,m,2}$, not componentwise
addition on $(C,\gamma)$.

\subsection{Qubit Pauli families and verifier score}
\label{subsec:integrated-qubit-verifier}
We now pass from the enhanced quadratic coordinates to the verifier's
operator-level test.  For a binary state label $(A,b)$, write
$\rho_{A,b}^{(2)}:=\ket{\psi_{A,b}^{(2)}}\!\bra{\psi_{A,b}^{(2)}}$
for its rank-one density operator.  The true label $(A,b)$ and the
reported label $(C,\gamma)$ specify two Pauli families.  Their shared phase-space labels
determine which terms can contribute to the Born acceptance probability
$\Tr[\Pi_{C,\gamma}^{(m,2)}\rho_{A,b}^{(2)}]$.  In characteristic two,
their eigenvalue phases must also be compared using the correction defined
below.  This operator-level treatment is necessary because the
phase-adjusted Pauli representative associated with a symmetric matrix depends
on its $\mathbb Z_4$-valued quadratic enhancement.

\subsubsection{Pauli operators and quadratic phase convention}

For $u,v\in\F_2^n$, define
\begin{equation}
 X(u)\ket{x}:=\ket{x+u},
 \quad
 Z(v)\ket{x}:=(-1)^{v^Tx}\ket{x}.
 \label{eq:integrated-qubit-XZ}
\end{equation}
These operators obey the following commutation relation:
\begin{equation}
 X(u)Z(v)=(-1)^{u^Tv}Z(v)X(u).
 \label{eq:integrated-qubit-XZ-commutation}
\end{equation}
For $A\in\Sym_n(\F_2)$, write
\begin{equation}
 q_A(x):=
 \sum_{j=1}^n \widetilde A_{jj}\widetilde x_j
 +2\sum_{1\le j<\ell\le n}
   \widetilde A_{j\ell}\widetilde x_j\widetilde x_\ell
 \pmod4,
 \label{eq:integrated-qubit-phase-section-enhancement}
\end{equation}
so that
$q_{A,b}(x)=q_A(x)+2\sum_j\widetilde b_j\widetilde x_j\pmod4$.
In the polarization formula, $u+v$ denotes addition in $\F_2^n$; the
resulting bits are then canonically lifted when $q_A(u+v)$ is evaluated.
With this convention,
\begin{equation}
 q_A(u+v)-q_A(u)-q_A(v)=2u^TAv\pmod4,
 \label{eq:integrated-qubit-polarization}
\end{equation}
where the right-hand side is first evaluated in $\F_2$ and then embedded as
$0$ or $2$ modulo four.
Define the phase-adjusted symmetric-matrix Pauli operators by
\begin{equation}
 \mathsf W_A(u):=i^{-q_A(u)}Z(Au)X(u),
 \quad u\in\F_2^n.
 \label{eq:integrated-qubit-symmetric-matrix-Pauli}
\end{equation}

\subsubsection{True and reported Pauli representations}
\label{subsubsec:qubit-true-announced-representations}

We first verify that both the true family and every admissible reported family form additive Pauli representations, so that their character projectors are
well defined.

\begin{lemma}[Qubit symmetric-matrix representation and eigenvalue character]
\label{lem:integrated-qubit-symmetric-matrix-representation}
For every $A\in\Sym_n(\F_2)$, the phase-adjusted operators
$\mathsf W_A(u)$ form an additive representation of $\F_2^n$:
\begin{equation}
 \mathsf W_A(u)\mathsf W_A(v)=\mathsf W_A(u+v)
 \quad(u,v\in\F_2^n).
 \label{eq:integrated-qubit-symmetric-matrix-representation}
\end{equation}
The state vectors have the following eigenvalue character:
\begin{equation}
 \mathsf W_A(u)\ket{\psi_{A,b}^{(2)}}
 =(-1)^{b^Tu}\ket{\psi_{A,b}^{(2)}}.
 \label{eq:integrated-qubit-eigenvalue-character}
\end{equation}
Thus $u\mapsto(-1)^{b^Tu}$ is the eigenvalue character relative to the
phase-adjusted section $u\mapsto\mathsf W_A(u)$.
\end{lemma}
\begin{proof}
Using Eq.~\eqref{eq:integrated-qubit-XZ-commutation},
\begin{align*}
 \mathsf W_A(u)\mathsf W_A(v)
 &=i^{-q_A(u)-q_A(v)}(-1)^{u^TAv}
   Z(A(u+v))X(u+v)\\
 &=i^{-q_A(u+v)}Z(A(u+v))X(u+v),
\end{align*}
where the second equality follows from Eq.~\eqref{eq:integrated-qubit-polarization}, establishing the representation property in Eq.~\eqref{eq:integrated-qubit-symmetric-matrix-representation}.

For the eigenvalue relation, applying $Z(Au)X(u)$ gives the following coefficient of $\ket{y}$:
\[
 i^{q_{A,b}(y+u)}(-1)^{(Au)^Ty}.
\]
The polarization identity gives
\[
 q_{A,b}(y+u)
 =q_{A,b}(y)+q_A(u)+2y^TAu+2b^Tu\pmod4.
\]
The $y$-dependent sign cancels $(-1)^{(Au)^Ty}$.  Multiplication by
$i^{-q_A(u)}$ leaves $(-1)^{b^Tu}$, proving
Eq.~\eqref{eq:integrated-qubit-eigenvalue-character}.
\end{proof}

The map $(A,b)\mapsto|\psi_{A,b}^{(2)}\rangle$ is injective at the level of
projective states.  Indeed, the state determines its underlying symmetric-matrix-normal-form
Lagrangian $L_A$, and the symmetric-matrix normal form determines $A$
uniquely.  For fixed $A$, Lemma~\ref{lem:integrated-qubit-symmetric-matrix-representation}
shows that the joint eigenvalue character is
$u\mapsto(-1)^{b^Tu}$, which determines $b$.  Equivalently, the coordinate
change in Eq.~\eqref{eq:integrated-qubit-additive-coordinates} is bijective:
$E$ determines the off-diagonal entries of $A$, while each
$\alpha_j\in\mathbb Z_4$ determines the pair $(A_{jj},b_j)$ uniquely.
Thus the enhanced coordinates label the characteristic-two symmetric-matrix
projective stabilizer states without duplication.

Let $C\in\F_2^{n\times m}$ satisfy
$J_m^TC=C^TJ_m$, and put
\begin{equation}
 B_C:=J_m^TC\in\Sym_m(\F_2).
 \label{eq:integrated-qubit-announced-B}
\end{equation}
Let $q_{B_C}$ be the enhancement defined by
Eq.~\eqref{eq:integrated-qubit-phase-section-enhancement}, with $A$ replaced by
$B_C$.  Define the reported operators
\begin{equation}
 \mathsf W_C(x):=
 i^{-q_{B_C}(x)}Z(Cx)X(J_mx),
 \quad x\in\F_2^m.
 \label{eq:integrated-qubit-announced-Pauli}
\end{equation}

\begin{lemma}[Qubit reported-family representation and character projector]
\label{lem:integrated-qubit-announced-representation}
The operators in Eq.~\eqref{eq:integrated-qubit-announced-Pauli} satisfy
\begin{equation}
 \mathsf W_C(x)\mathsf W_C(y)=\mathsf W_C(x+y)
 \quad(x,y\in\F_2^m).
 \label{eq:integrated-qubit-announced-representation}
\end{equation}
For $\gamma\in\F_2^m$, define
\begin{equation}
 \Pi_{C,\gamma}^{(m,2)}
 :=2^{-m}\sum_{x\in\F_2^m}(-1)^{\gamma^Tx}\mathsf W_C(x).
 \label{eq:integrated-qubit-accept-projector}
\end{equation}
Then $\Pi_{C,\gamma}^{(m,2)}$ is the joint spectral projector for the
character $x\mapsto(-1)^{\gamma^Tx}$, and
\begin{equation}
 \Pi_{C,\gamma}^{(m,2)}\Pi_{C,\gamma'}^{(m,2)}
 =\delta_{\gamma,\gamma'}\Pi_{C,\gamma}^{(m,2)},
 \quad
 \sum_{\gamma\in\F_2^m}\Pi_{C,\gamma}^{(m,2)}=I.
 \label{eq:integrated-qubit-projector-resolution}
\end{equation}
\end{lemma}
\begin{proof}
The symmetry of $B_C=J_m^TC$ and the commutation relation
\eqref{eq:integrated-qubit-XZ-commutation} give
\begin{align}
 \mathsf W_C(x)\mathsf W_C(y)
 &=i^{-q_{B_C}(x)-q_{B_C}(y)}
   (-1)^{x^TB_Cy}\notag\\
 &\quad\times Z(C(x+y))X(J_m(x+y))\notag\\
 &=\mathsf W_C(x+y).
 \label{eq:qubit-announced-explicit-multiplication}
\end{align}
Here the last equality is precisely the polarization identity
\eqref{eq:integrated-qubit-polarization}, applied to the symmetric
$m\times m$ matrix $B_C$.  Thus the $n$-dimensional calculation in
Lemma~\ref{lem:integrated-qubit-symmetric-matrix-representation}
restricts to the reported $m$-dimensional family, proving
Eq.~\eqref{eq:integrated-qubit-announced-representation}.
Moreover, for every
$y\in\F_2^m$, a change of variable $x\mapsto x+y$ gives
\begin{equation}
 \mathsf W_C(y)\Pi_{C,\gamma}^{(m,2)}
 =(-1)^{\gamma^Ty}\Pi_{C,\gamma}^{(m,2)}.
 \label{eq:integrated-qubit-projector-character-action}
\end{equation}
Character orthogonality then gives pairwise orthogonality and the resolution
of the identity in Eq.~\eqref{eq:integrated-qubit-projector-resolution}.
\end{proof}

\begin{lemma}[Intrinsic realization of a qubit coordinate report]
\label{lem:qubit-intrinsic-coordinate-report-bridge}
Fix the physical Pauli section $z\mapsto P_z$ from
Eq.~\eqref{eq:global-Pauli-section-cocycle}.  For every admissible coordinate
report $(C,\gamma)$, there is a unique phase function
$s_C:\F_2^m\to\mathbb C$ of unit modulus such that
\begin{equation}
 \mathsf W_C(x)=s_C(x)P_{(J_mx,Cx)}.
 \label{eq:qubit-coordinate-to-physical-section-phase}
\end{equation}
Put
\begin{equation}
 K_C:=\{(J_mx,Cx):x\in\F_2^m\}
 \label{eq:qubit-coordinate-report-subspace}
\end{equation}
and define an intrinsic assignment on $K_C$ by
\begin{equation}
 \eta_{C,\gamma}(J_mx,Cx)
 :=s_C(x)^{-1}(-1)^{\gamma^Tx}.
 \label{eq:qubit-coordinate-to-intrinsic-assignment}
\end{equation}
Then $\eta_{C,\gamma}$ is cocycle-compatible relative to the fixed physical
section, and
\begin{equation}
 \Pi_{K_C,\eta_{C,\gamma}}=\Pi_{C,\gamma}^{(m,2)}.
 \label{eq:qubit-intrinsic-coordinate-projector-equality}
\end{equation}

\end{lemma}
\begin{proof}
The map $x\mapsto(J_mx,Cx)$ is injective because $J_m$ has full column rank.
Hence Eq.~\eqref{eq:qubit-coordinate-to-intrinsic-assignment} is well defined.
The representation property of $\mathsf W_C$ and the multiplication law for
the fixed section imply
\begin{equation}
 s_C(x+y)
 =s_C(x)s_C(y)c_2((J_mx,Cx),(J_my,Cy)).
 \label{eq:qubit-bridge-section-cocycle-relation}
\end{equation}
Taking inverses in this identity and combining the result with the additivity
of $x\mapsto(-1)^{\gamma^Tx}$ gives
\[
 \eta_{C,\gamma}(z_x+z_y)
 =c_2(z_x,z_y)^{-1}\eta_{C,\gamma}(z_x)\eta_{C,\gamma}(z_y),
 \qquad z_x=(J_mx,Cx),
\]
which is the cocycle-compatibility condition
Eq.~\eqref{eq:cocycle-compatible-assignment}.  Substituting this
assignment into the intrinsic report projector gives the operator identity
\begin{align}
 \Pi_{K_C,\eta_{C,\gamma}}
 &=2^{-m}\sum_{x\in\F_2^m}
   \eta_{C,\gamma}(J_mx,Cx)^{-1}P_{(J_mx,Cx)}
 =2^{-m}\sum_{x\in\F_2^m}(-1)^{\gamma^Tx}\mathsf W_C(x)
 =\Pi_{C,\gamma}^{(m,2)},
 \label{eq:qubit-bridge-projector-calculation}
\end{align}
where phases have unit modulus and the signs are real.  This proves the stated
projector equality.
\end{proof}

The lemma realizes each binary coordinate report as an intrinsic
fixed-constraint report with the same acceptance projector.

For the concrete physical section $P_{(u,v)}=X(u)Z(v)$,
Eqs.~\eqref{eq:integrated-qubit-announced-Pauli} and
\eqref{eq:qubit-coordinate-to-physical-section-phase} give the explicit
phase $s_C(x)=i^{q_{B_C}(x)}$, with $B_C=J_m^TC$.  Indeed, the commutation
relation gives
\begin{align}
 \mathsf W_C(x)
 =i^{-q_{B_C}(x)}(-1)^{x^TB_Cx}P_{(J_mx,Cx)}
 =i^{q_{B_C}(x)}P_{(J_mx,Cx)},
 \label{eq:qubit-section-phase-parity-calculation}
\end{align}
because $q_{B_C}(x)\equiv x^TB_Cx\pmod2$ and the two exponents of
$i$ therefore agree modulo four.
Thus the equality of the binary accept projectors in
Eq.~\eqref{eq:qubit-intrinsic-coordinate-projector-equality} includes
the fourth-root phase, rather than treating the coordinate character
as an uncorrected physical-section assignment.  When evaluating an
incorrect report against a true state, the relative correction
$\delta_{A,C}$ in Eq.~\eqref{eq:integrated-qubit-phase-correction}
must still be retained; it vanishes for a matching matrix report.

\subsubsection{Shared Pauli labels and phase correction}
\label{subsubsec:qubit-common-space-phase-correction}

The stabilizer family of the true state $A$ and the reported family $C$
may contain the same Pauli phase-space labels while assigning them different
quadratic phases.  We first identify the coefficient vectors that specify
shared labels, and then calculate the sign needed to compare the two
eigenvalue characters on those labels.

For a true symmetric matrix $A$ and a reported queried family $C$, define
\begin{equation}
 N(A,C):=\ker(C-AJ_m),
 \quad
 \kappa(A,C):=\dim N(A,C).
 \label{eq:integrated-qubit-compatibility-space}
\end{equation}
$N(A,C)$ is a subspace of the coefficient space $\F_2^m$.  Its elements
index the shared Pauli labels; it is not a space of operators on the qubits.
For any $t\in\F_2^m$, the reported phase-space label
$(J_mt,Ct)$ belongs to the true stabilizer subspace $L_A$ exactly when
$(C-AJ_m)t=0$.  Thus the intersection of the true and reported
stabilizer subspaces is
\begin{equation}
 L_A\cap K_C
 =\{(J_mt,Ct):t\in N(A,C)\},
 \label{eq:qubit-common-label-intersection}
\end{equation}
where $K_C$ was defined in
Eq.~\eqref{eq:qubit-coordinate-report-subspace}.
The Pauli operators associated with these common labels can have
different phase factors, determined by $q_A(J_mt)$ and $q_{B_C}(t)$.
To compare their eigenvalue characters, define the sign correction
on $N(A,C)$ by
\begin{equation}
 \delta_{A,C}(x)
 :=\frac{q_A(J_mx)-q_{B_C}(x)}2\pmod2,
 \quad x\in N(A,C).
 \label{eq:integrated-qubit-phase-correction}
\end{equation}
For $x\in N(A,C)$, the numerator in
Eq.~\eqref{eq:integrated-qubit-phase-correction} is even modulo four.
Its half modulo two, $\delta_{A,C}(x)$, determines the additional
sign $(-1)^{\delta_{A,C}(x)}$ required when comparing the eigenvalue
characters of the true and reported families.  It does not change the
common Pauli phase-space labels.
Put $D:=J_m^TAJ_m-B_C\in\Sym_m(\F_2)$.  If $x\in N(A,C)$, then
$Dx=0$.  Reduction modulo two of the numerator in
Eq.~\eqref{eq:integrated-qubit-phase-correction} equals
$\operatorname{diag}(D)^Tx=x^TDx=0$, where the second equality uses symmetry
of $D$ in characteristic two.  Thus the numerator is even and
$\delta_{A,C}$ is well defined.

For the shared labels indexed by $N(A,C)$, define the difference of
the two quadratic phases as the function
$f_{A,C}:N(A,C)\to\mathbb Z_4$.
Its value at $z\in N(A,C)$ is defined by
\[
f_{A,C}(z):=q_A(J_mz)-q_{B_C}(z)\pmod 4.
\]
For $x,y\in N(A,C)$, subtracting the polarization identity
for $q_{B_C}$ from that for $q_A\circ J_m$ gives
\begin{align}
f_{A,C}(x+y)-f_{A,C}(x)-f_{A,C}(y)
=2x^T\bigl(J_m^TAJ_m-B_C\bigr)y 
=2x^TDy
=0\pmod4.
\label{eq:integrated-qubit-delta-additivity-calculation}
\end{align}
The last equality follows from the defining kernel condition for $N(A,C)$.
Indeed, $B_C=J_m^TC$ and $y\in N(A,C)$ imply
\[
Dy=J_m^T(AJ_m-C)y=0.
\]
Thus the polarization difference vanishes on $N(A,C)$.
Since $f_{A,C}(z)$ is even for every $z\in N(A,C)$,
division by two modulo two yields
\[
\delta_{A,C}(x+y)
=\delta_{A,C}(x)+\delta_{A,C}(y).
\]

The calculation above proves that $\delta_{A,C}$ is a linear
functional on $N(A,C)$.

\Needspace{10\baselineskip}
\subsubsection{Born acceptance probability for qubit reports}
\label{subsubsec:qubit-verifier-score-formula}

For a true state $\rho_{A,b}^{(2)}$ and a report
$(C,\gamma)$, the verifier accepts with Born probability
$\Tr[\Pi_{C,\gamma}^{(m,2)}\rho_{A,b}^{(2)}]$.
Only the reported Pauli operators whose labels belong to the intersection
in Eq.~\eqref{eq:qubit-common-label-intersection} have nonzero expectation
in the true state.  Their indices are precisely $t\in N(A,C)$.
The next lemma evaluates the resulting character sum, including the
phase correction $\delta_{A,C}$.

\begin{lemma}[Qubit verifier score]
\label{lem:integrated-qubit-score-spectrum}
For every true label $(A,b)$ and report $(C,\gamma)$, the verifier
acceptance probability is determined by the dimension $\kappa(A,C)$ of
the shared-label coefficient space $N(A,C)$ and by compatibility of the
corrected characters.  The acceptance probability is
\begin{equation}
 \Tr\!\left[\Pi_{C,\gamma}^{(m,2)}
 \ket{\psi_{A,b}^{(2)}}\!\bra{\psi_{A,b}^{(2)}}\right]
 =\begin{cases}
 2^{\kappa(A,C)-m},&
 \bigl(\gamma+J_m^Tb+\delta_{A,C}\bigr)|_{N(A,C)}=0,\\[1mm]
 0,&
 \bigl(\gamma+J_m^Tb+\delta_{A,C}\bigr)|_{N(A,C)}\ne0.
 \end{cases}.
 \label{eq:integrated-qubit-score-spectrum}
\end{equation}
Here $\delta_{A,C}$ is regarded as a linear functional on $N(A,C)$.
In particular, the score is one if and only if
$C=AJ_m$ and $\gamma=J_m^Tb$; every false report has score either zero
or at most $1/2$.
\end{lemma}
\begin{proof}
For $x\notin N(A,C)$, put $u=J_mx$ and $v=Cx$.  Then $v-Au\ne0$.
Using the computational-basis expansion of
$\ket{\psi_{A,b}^{(2)}}$, the expectation of $Z(v)X(u)$ is a fixed phase
times
\begin{equation}
 2^{-n}\sum_{y\in\F_2^n}(-1)^{(v-Au)^Ty}=0,
 \label{eq:integrated-qubit-off-intersection-character-sum}
\end{equation}
because $y\mapsto(-1)^{(v-Au)^Ty}$ is a nontrivial character.  The scalar
phase $i^{-q_{B_C}(x)}$ in $\mathsf W_C(x)$ does not change the vanishing. Hence the expectation value of the reported Pauli operator is zero:
\begin{equation}
 \bra{\psi_{A,b}^{(2)}}\mathsf W_C(x)
 \ket{\psi_{A,b}^{(2)}}=0.
 \label{eq:integrated-qubit-off-intersection-zero}
\end{equation}
For $x\in N(A,C)$, Eqs.~\eqref{eq:integrated-qubit-symmetric-matrix-Pauli},
\eqref{eq:integrated-qubit-announced-Pauli}, and
\eqref{eq:integrated-qubit-phase-correction} give
\begin{equation}
 \mathsf W_C(x)
 =(-1)^{\delta_{A,C}(x)}\mathsf W_A(J_mx).
 \label{eq:integrated-qubit-section-comparison}
\end{equation}
Therefore Lemma~\ref{lem:integrated-qubit-symmetric-matrix-representation} yields
\begin{equation}
 \bra{\psi_{A,b}^{(2)}}\mathsf W_C(x)
 \ket{\psi_{A,b}^{(2)}}
 =(-1)^{\delta_{A,C}(x)+(J_m^Tb)^Tx}
 \quad(x\in N(A,C)).
 \label{eq:integrated-qubit-on-intersection-character}
\end{equation}
Substitute Eqs.~\eqref{eq:integrated-qubit-off-intersection-zero} and
\eqref{eq:integrated-qubit-on-intersection-character} into
Eq.~\eqref{eq:integrated-qubit-accept-projector}.  Character orthogonality on
$N(A,C)$ gives Eq.~\eqref{eq:integrated-qubit-score-spectrum}.

If the score is one, then $\kappa(A,C)=m$, so $C=AJ_m$.  In this case
$B_C=J_m^TAJ_m$ and $\delta_{A,C}=0$, and compatibility forces
$\gamma=J_m^Tb$.  Conversely, the true queried label has Born acceptance probability one.
If the label is false and $\kappa(A,C)<m$, every nonzero score is at most
$2^{-1}$; if $C=AJ_m$ but $\gamma\ne J_m^Tb$, the score is zero.
\end{proof}

The verifier score is therefore defined intrinsically.  The dimension of
$N(A,C)$ supplies the rank penalty, and the corrected character
$\gamma+J_m^Tb+\delta_{A,C}$ supplies the compatibility test.  The chosen
quadratic phase sections enter only through this correction and do not define
a different operational task.

\begin{remark}[Phase-section dependence in characteristic two]
\label{rem:integrated-qubit-phase-section-dependent-character}
The vector $b$ specifies the eigenvalue character relative to the
phase section $u\mapsto\mathsf W_A(u)$ determined by $A$.  Consequently, when two symmetric-matrix labels are compared on an intersection, equality of the vector labels alone
is not the intrinsic compatibility condition: the correction
$\delta_{A,C}$ in Eq.~\eqref{eq:integrated-qubit-phase-correction} must also be
included.  This correction vanishes when the reported queried family equals
the true queried family.  The later all-prime compatible exact-label-list formulation must
therefore either retain this correction or use the additive enhanced
coordinates of Eq.~\eqref{eq:integrated-qubit-additive-coordinates}.
\end{remark}

The graph-basis interpretation and the fixed-diagonal threshold are treated after the converse and direct arguments in Section~\ref{sec:fixed-diagonal-graph-basis-submodels}.

\section{Finite-copy analysis of the refined frame-assisted problem}
\label{sec:partial-learning}

Subsection~\ref{subsec:frame-coordinate-equivalence} states the finite
equivalence between the frame-assisted and standardized tasks.  We now
give the coordinate states and acceptance projectors used both to prove
that equivalence in
Subsection~\ref{subsec:frame-coordinate-equivalence-proof} and to
calculate the finite-copy optima.

The public frame and the queried indices are specified before the learner's
measurement.  Under the chosen frame standardization, the queried
reference-dual directions are the columns of $(J_m;AJ_m)$, not the
vertical vectors of the public standard frame.  For odd $p$ the
transverse states are $|\psi_{A,b}\rangle$; for $p=2$ they are
$|\psi_{A,b}^{(2)}\rangle$, with
$(A,b)\in\Sym_n(\F_p)\times\F_p^n$ in both cases.
The report retains the first $m$ directions and their restricted
eigenvalue assignment.  Proposition~\ref{prop:frame-conditioned-coordinate-reduction}
identifies the resulting uniform ensemble and its two optima with the
frame-assisted problem.

\subsection{Standard \texorpdfstring{$J_m$}{Jm} query and partial labels}
\label{subsec:canonical-partial-query}
\label{subsec:finite-copy-coordinate-notation}

Fix $1\leq m\leq n$ and recall $J_m$ and
$\mathcal C_{n,m}$ from Eqs.~\eqref{eq:canonical-query-Jm} and
\eqref{eq:partial-symmetric-matrix-space}.  Equation~\eqref{eq:dual-family-normal-form-coordinate}
shows why the partial matrix $AJ_m$, rather than the public vertical
frame vectors, specifies the queried stabilizer directions.  For a true
odd-prime parameter $(A,b)\in\Sym_n(\F_p)\times\F_p^n$, define the restricted
symmetric-matrix and character label by
\begin{equation}
 q_m(A,b):=(AJ_m,-J_m^Tb).
 \label{eq:true-partial-label-map}
\end{equation}
The matrix $J_m^TAJ_m$ is symmetric, so $AJ_m$ belongs to
$\mathcal C_{n,m}$. Conversely, every $C$ in that set specifies the
isotropic family whose phase-space labels are the columns of $(J_m;C)$.
We therefore define the partial label space
\begin{equation}
 \Theta_{n,m,p}^{\mathrm{part}}
 :=\mathcal C_{n,m}\times\F_p^m.
 \label{eq:partial-label-space}
\end{equation}

To record its size without reusing the subgroup symbol $H$, write an
admissible partial symmetric-matrix column in symmetric and rectangular blocks as
\begin{equation}
 C=\binom{C_{\rm sym}}{C_{\rm rect}},
 \quad C_{\rm sym}\in\Sym_m(\F_p),\quad
 C_{\rm rect}\in\F_p^{(n-m)\times m}.
 \label{eq:partial-symmetric-matrix-block-form}
\end{equation}
The block decomposition separates the symmetric and rectangular degrees of
freedom.  Their total dimension is
\begin{equation}
 \dim\mathcal C_{n,m}
 =\frac{m(m+1)}2+(n-m)m
 =nm-\frac{m(m-1)}2.
 \label{eq:partial-symmetric-matrix-dimension}
\end{equation}
Including the $m$ character coordinates gives the following cardinality for
the full partial-label space:
\begin{equation}
 M_{n,m,p}
 :=|\Theta_{n,m,p}^{\mathrm{part}}|
 =p^{(n+1)m-m(m-1)/2}.
 \label{eq:partial-label-cardinality}
\end{equation}

\begin{lemma}[Uniform partial labels and constant fibers]
\label{lem:uniform-partial-labels}
The map $q_m$ in Eq.~\eqref{eq:true-partial-label-map} is surjective.  Every
fiber has cardinality
\begin{equation}
 F_{n,m,p}
 =p^{(n-m)(n-m+1)/2+n-m}.
 \label{eq:partial-label-fiber-size}
\end{equation}
Consequently, uniform $(A,b)$ induces the uniform distribution on
$\Theta_{n,m,p}^{\mathrm{part}}$.
\end{lemma}

\begin{proof}
Decompose
\begin{equation}
 A=\begin{pmatrix}A_{11}&A_{21}^T\\A_{21}&A_{22}\end{pmatrix},
 \quad b=\binom{b_1}{b_2},
 \label{eq:partial-label-fiber-blocks}
\end{equation}
where $A_{11}\in\Sym_m(\F_p)$, $A_{21}\in\F_p^{(n-m)\times m}$,
$A_{22}\in\Sym_{n-m}(\F_p)$, $b_1\in\F_p^m$, and
$b_2\in\F_p^{n-m}$.  The partial-label map retains exactly the visible
matrix blocks and the first character block:
\begin{equation}
 q_m(A,b)=\left(\binom{A_{11}}{A_{21}},-b_1\right).
 \label{eq:partial-label-visible-blocks}
\end{equation}
Thus the abstract blocks in Eq.~\eqref{eq:partial-symmetric-matrix-block-form} are
$C_{\rm sym}=A_{11}$ and $C_{\rm rect}=A_{21}$.  Hence every admissible
partial label is attained, while $(A_{22},b_2)$ remain
arbitrary.  Their number is exactly Eq.~\eqref{eq:partial-label-fiber-size},
independently of the visible label.
\end{proof}

The conditioned exact-identification ensemble averages complete-label
tensor powers over each constant fiber of $q_m$, as established above.
For reference, Table~\ref{tab:basic-state-game-notation} summarizes the notation
for the odd-prime fixed-query problem.  The characteristic-two operational
objects are defined in Section~\ref{sec:integrated-qubit-foundations};
all-prime results use the corresponding quotient labels, not the
odd-prime phase convention recorded in this table.

\begin{table}[t]
\centering
\caption{Notation for the odd-prime realization of the fixed-query partial-learning game.}
\label{tab:basic-state-game-notation}
\small
\renewcommand{\arraystretch}{1.13}
\begin{tabularx}{\textwidth}{@{}l l Y@{}}
\toprule
Symbol & Domain or definition & Role \\
\midrule
$p$ & odd prime & Local qudit dimension and cardinality of the base field. \\
$\F_p$ & field with $p$ elements & All vector and matrix operations in the auxiliary coordinate results are over $\F_p$. \\
$n$ & integer, $n\geq1$ & Number of qudits and length of each coordinate vector in $\F_p^n$; the symplectic phase space $\F_p^{2n}$ has dimension $2n$. \\
$m$ & $1\leq m\leq n$ & Number of requested reference-dual stabilizer directions; $m=n$ requests all $n$ reference-dual directions, not the frame vectors. \\
$k$ & integer, $k\geq0$ & Number of learner copies; the verifier's independent test copy is not included. \\
$\omega$ & $e^{2\pi i/p}$ & Basic additive-character phase. \\
$A$ & $A\in\Sym_n(\F_p)$ & Symmetric matrix specifying the normal-form Lagrangian $L_A$. \\
$b$ & $b\in\F_p^n$ & Parameter specifying the eigenvalue character $u\mapsto\omega^{-b^Tu}$. \\
$\ket{\psi_{A,b}}$, $\rho_{A,b}$ & $\mathcal H_n$, $\ket{\psi_{A,b}}\!\bra{\psi_{A,b}}$ & coordinate stabilizer state and its density operator. \\
$J_m$ & $\binom{I_m}{0}\in\F_p^{n\times m}$ & Inclusion selecting the first $m$ horizontal coordinate vectors of the reference-dual family after standardization. \\
$\mathcal C_{n,m}$ & $\{C:J_m^TC=C^TJ_m\}$ & Admissible partial symmetric-matrix columns. \\
$(C,\gamma)$ & $\mathcal C_{n,m}\times\F_p^m$ & Partial label reported by the learner. \\
$q_m(A,b)$ & $(AJ_m,-J_m^Tb)$ & True partial-label map. \\
$\Theta_{n,m,p}^{\mathrm{part}}$ & $\mathcal C_{n,m}\times\F_p^m$ & Learner output alphabet: an admissible partial symmetric-matrix column and a reported character. \\
$\Pi_{C,\gamma}^{(m)}$ & Eq.~\eqref{eq:partial-accept-projector} & Verifier's accept projector for the report $(C,\gamma)$. \\
$s_m$ & $\Tr[\Pi_{C,\gamma}^{(m)}\rho_{A,b}]$ & Born acceptance probability for a fixed true state and report; for $m<n$ it is not fidelity to a unique output state. \\
$S_{\mathrm{ver}}^{(k),*}$ & optimum over collective POVMs & Optimal average partial-verification score. \\
\bottomrule
\end{tabularx}
\end{table}
\FloatBarrier

\subsection{Learner and operational verification score}
\label{subsec:partial-learning-protocol}

Nature samples $(A,b)$ uniformly and supplies the learner with $k$ copies of
\begin{equation}
 \rho_{A,b}:=\ket{\psi_{A,b}}\!\bra{\psi_{A,b}}.
 \label{eq:partial-input-state}
\end{equation}
The learner measures the $k$-copy state using an arbitrary collective POVM
with outcomes in $\Theta_{n,m,p}^{\mathrm{part}}$:
\begin{equation}
 \mathsf M
 =\{M_{C,\gamma}:(C,\gamma)\in
       \Theta_{n,m,p}^{\mathrm{part}}\}.
 \label{eq:partial-learning-POVM}
\end{equation}
Its effects are positive semidefinite and sum to the identity:
\begin{equation}
 M_{C,\gamma}\geq0,
 \quad
 \sum_{(C,\gamma)\in\Theta_{n,m,p}^{\mathrm{part}}}
 M_{C,\gamma}=I.
 \label{eq:partial-learning-POVM-normalization}
\end{equation}
Zero-payoff auxiliary outcomes may be added without changing the optimum.

For the operational test, a verifier receives one additional independent copy
of $\rho_{A,b}$.  Given the report $(C,\gamma)$, the verifier measures
the commuting Weyl family
\begin{equation}
 \{W(J_mx,Cx):x\in\F_p^m\},
 \label{eq:partial-announced-Weyl-family}
\end{equation}
and accepts only the reported character
$x\mapsto\omega^{\gamma^Tx}$.  Equivalently, the accept projector is
\begin{equation}
 \Pi_{C,\gamma}^{(m)}
 :=\frac{1}{p^m}\sum_{x\in\F_p^m}
   \omega^{-\gamma^Tx}W(J_mx,Cx).
 \label{eq:partial-accept-projector}
\end{equation}
The condition $J_m^TC=C^TJ_m$ from
Eq.~\eqref{eq:partial-symmetric-matrix-space} makes the Weyl operators
$W(J_mt,Ct)$ commute.  The operator $\Pi_{C,\gamma}^{(m)}$ in
Eq.~\eqref{eq:partial-accept-projector} is their projector onto the
simultaneous eigenspace with the reported character: for every
$t\in\F_p^m$,
\begin{equation}
 W(J_mt,Ct)\Pi_{C,\gamma}^{(m)}
 =\omega^{\gamma^Tt}\Pi_{C,\gamma}^{(m)}.
 \label{eq:partial-accept-simultaneous-eigenspace}
\end{equation}

The coordinate projector $\Pi_{C,\gamma}^{(m)}$ from
Eq.~\eqref{eq:partial-accept-projector} equals the intrinsic report
projector $\Pi_{K_C,\eta_{C,\gamma}}$ defined by
Eq.~\eqref{eq:intrinsic-character-projector} using the physical Pauli
section $z\mapsto P_z$ from
Eq.~\eqref{eq:global-Pauli-section-cocycle}.  We verify the equality
by identifying the compatible assignment.
Write $z_C(t)=(J_mt,Ct)$ and
$K_C=\{z_C(t):t\in\F_p^m\}$.  For the odd-prime Weyl convention, let
$s_C(t)$ be the unit-modulus phase defined by
$W(J_mt,Ct)=s_C(t)P_{z_C(t)}$.  The intrinsic assignment associated with
the coordinate report is
$\eta_{C,\gamma}(z_C(t))=s_C(t)^{-1}\omega^{\gamma^Tt}$.
The Weyl operators on $K_C$ form an additive representation, so this
assignment is cocycle-compatible.  Substitution into the intrinsic
projector formula gives
\begin{align}
 \Pi_{K_C,\eta_{C,\gamma}}
 =p^{-m}\sum_{t\in\F_p^m}
 \eta_{C,\gamma}(z_C(t))^{-1}P_{z_C(t)}
 =p^{-m}\sum_{t\in\F_p^m}\omega^{-\gamma^Tt}W(J_mt,Ct)
 =\Pi_{C,\gamma}^{(m)}.
 \label{eq:odd-intrinsic-coordinate-projector-equality}
\end{align}
For the fixed section $P_{(u,v)}=X(u)Z(v)$, the phase is
$s_C(t)=\omega^{t^T(J_m^TC)t/2}$ by
Eq.~\eqref{eq:odd-Weyl-definition}.  For another normalized section
of the same projective Pauli family, the corresponding $s_C$ changes
along with the intrinsic assignment, leaving the projector identity
unchanged.

For a fixed true state $\rho_{A,b}$ and a fixed reported partial label
$(C,\gamma)$, define the single-report Born acceptance probability by
\begin{equation}
 s_m((A,b),(C,\gamma))
 :=\Tr\!\left[\Pi_{C,\gamma}^{(m)}\rho_{A,b}\right].
 \label{eq:partial-score-definition}
\end{equation}
Unlike the statewise verification score
$v_{\mathsf M}^{(k)}(x;W)$ in
Eq.~\eqref{eq:fixed-W-statewise-verification-score-operational},
the quantity $s_m$ refers to a fixed report; it has not yet been averaged
over the learner's POVM outcomes.

For odd $p$, the single-report probability $s_m$ is determined by the
queried stabilizers shared by the true state and the reported family.
Put $N:=\ker(C-AJ_m)$ and $d:=\rank(C-AJ_m)$.  The probability is
$p^{-d}$ when the reported character agrees on $N$, and zero otherwise:
\begin{equation}
 s_m((A,b),(C,\gamma))
 =\begin{cases}
 p^{-d},& (\gamma+J_m^Tb)^Tx=0\text{ for every }x\in N,\\
 0,&\text{otherwise}.
 \end{cases}
 \label{eq:odd-partial-score-spectrum}
\end{equation}
To verify this formula in the chosen Weyl convention, the quadratic-state
expansion gives, for $u,v\in\F_p^n$,
\begin{align}
 \langle\psi_{A,b}|W(u,v)|\psi_{A,b}\rangle
 =\omega^{\frac12u^TAu-b^Tu-\frac12u^Tv}
 p^{-n}\sum_{z\in\F_p^n}\omega^{(v-Au)^Tz}
 =\begin{cases}
 \omega^{-b^Tu},&v=Au,\\
 0,&v\ne Au.
 \end{cases}
 \label{eq:odd-Weyl-expectation-bridge}
\end{align}
Substituting $u=J_mx$ and $v=Cx$ into the accept projector therefore yields
\begin{equation}
 s_m((A,b),(C,\gamma))
 =p^{-m}\sum_{x\in N}\omega^{-(\gamma+J_m^Tb)^Tx}.
 \label{eq:odd-partial-score-character-sum}
\end{equation}
Character orthogonality on the $(m-d)$-dimensional space $N$ proves
Eq.~\eqref{eq:odd-partial-score-spectrum}.  Thus the common-subspace
codimension determines the nonzero score, while character agreement decides
whether it is attained.  The single-report Born acceptance probability is one exactly at
$(C,\gamma)=(AJ_m,-J_m^Tb)$, consistently with
Eq.~\eqref{eq:true-partial-label-map}.  The binary counterpart is
Eq.~\eqref{eq:integrated-qubit-score-spectrum}, which retains the additional
phase-section correction.

For a POVM $\mathsf M$, the average score is
\begin{align}
 S_{\mathrm{ver}}^{(k)}(\mathsf M;n,m,p)
 := \frac{1}{p^{n(n+3)/2}}
 \sum_{A\in\Sym_n(\F_p)}\sum_{b\in\F_p^n}
 \sum_{(C,\gamma)\in\Theta_{n,m,p}^{\mathrm{part}}}
 \Tr\!\left[M_{C,\gamma}\rho_{A,b}^{\otimes k}\right]
 s_m((A,b),(C,\gamma)).
 \label{eq:partial-average-score}
\end{align}
We denote its optimum over all collective POVMs by
\begin{equation}
 S_{\mathrm{ver}}^{(k),*}(n,m;p)
 :=\max_{\mathsf M}S_{\mathrm{ver}}^{(k)}(\mathsf M;n,m,p).
 \label{eq:partial-optimal-score}
\end{equation}
The letter $S$ is reserved for verification scores, while
$P_{\mathrm{id}}^{(k),*}$ denotes exact-identification probability.  The verifier's
test copy is independent of the learner's $k$ copies.  Equations
\eqref{eq:partial-score-definition}--\eqref{eq:partial-optimal-score} give the
odd-prime coordinate form of the intrinsic Born verification payoff introduced
in the fixed-frame formulation.

The optimization is a finite semidefinite program. For a report $(C,\gamma)$ define the operator that averages the reported verification score against the $k$-copy state ensemble:
\begin{equation}
 R_{C,\gamma}^{(k)}
 :=\frac{1}{p^{n(n+3)/2}}
 \sum_{A\in\Sym_n(\F_p)}\sum_{b\in\F_p^n}
 s_m((A,b),(C,\gamma))\rho_{A,b}^{\otimes k}.
 \label{eq:partial-score-operator}
\end{equation}
The operator $R_{C,\gamma}^{(k)}$ is the uniform average of the
$k$-copy states weighted by the score of the fixed report $(C,\gamma)$.
Taking its trace against $M_{C,\gamma}$ also includes the probability that
the learner outputs that report.  Summing over reports therefore gives
the uniform-average verification score:
\begin{equation}
 S_{\mathrm{ver}}^{(k)}(\mathsf M;n,m,p)
 =\sum_{(C,\gamma)\in\Theta_{n,m,p}^{\mathrm{part}}}
 \Tr[M_{C,\gamma}R_{C,\gamma}^{(k)}].
 \label{eq:partial-score-operator-SDP}
\end{equation}
The prior and the operational score are both included in
$R_{C,\gamma}^{(k)}$.  The comparison of verification scores is developed later in Section~\ref{sec:integrated-all-prime-graded}.

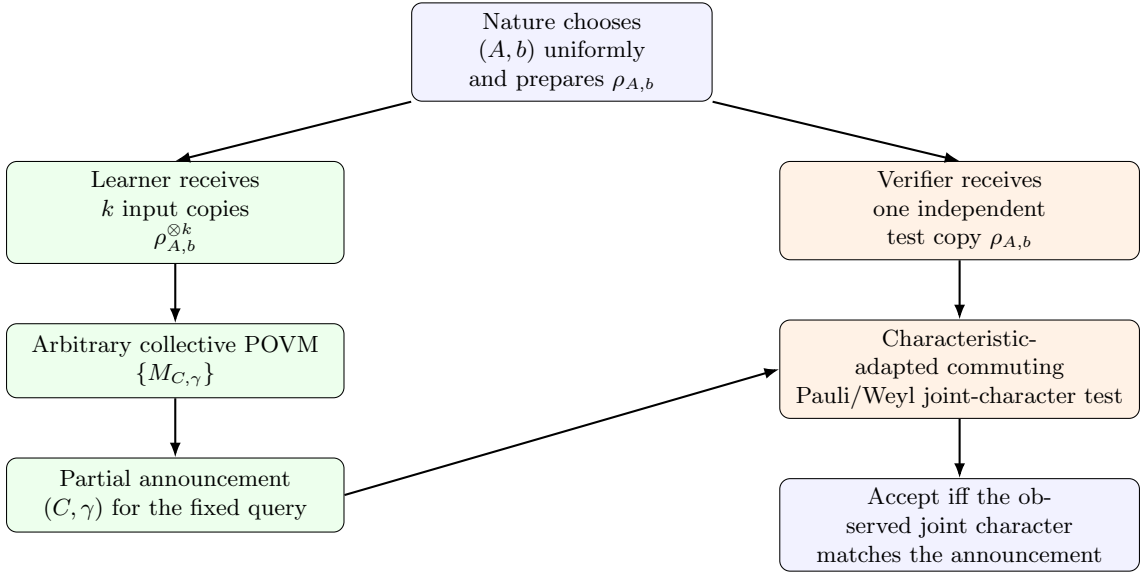
\begin{figure}[t]
\centering
\resizebox{0.96\linewidth}{!}{%
\begin{tikzpicture}[
 font=\footnotesize, node distance=8mm and 9mm,
 source/.style={draw,rounded corners,align=center,minimum height=10mm,text width=38mm,fill=blue!5},
 learner/.style={draw,rounded corners,align=center,minimum height=10mm,text width=43mm,fill=green!7},
 verifier/.style={draw,rounded corners,align=center,minimum height=10mm,text width=46mm,fill=orange!10},
 result/.style={draw,rounded corners,align=center,minimum height=10mm,text width=46mm,fill=blue!5},
 arr/.style={-{Latex[length=2mm]},thick}]
\node[source] (nature) {Nature chooses $(A,b)$ uniformly\\and prepares $\rho_{A,b}$};
\node[learner,below left=of nature] (lin) {Learner receives $k$ input copies\\$\rho_{A,b}^{\otimes k}$};
\node[verifier,below right=of nature] (vin) {Verifier receives one independent\\test copy $\rho_{A,b}$};
\node[learner,below=of lin] (povm) {Arbitrary collective POVM\\$\{M_{C,\gamma}\}$};
\node[learner,below=of povm] (announce) {Partial report\\$(C,\gamma)$ for the fixed query};
\node[verifier,below=of vin] (test) {Characteristic-adapted commuting\\Pauli/Weyl joint-character test};
\node[result,below=of test] (accept) {Accept iff the observed joint character\\matches the report};
\draw[arr] (nature.south west)--(lin.north);
\draw[arr] (nature.south east)--(vin.north);
\draw[arr] (lin)--(povm); \draw[arr] (povm)--(announce);
\draw[arr] (vin)--(test); \draw[arr] (announce.east)--(test.west);
\draw[arr] (test)--(accept);
\end{tikzpicture}%
}
\caption{Operational structure of finite-copy partial learning for every prime,
shown after the public reference frame has been sent to the standard
coordinate frame.  In these coordinates, $(C,\gamma)$ is the learner's
report for the first $m$ reference-dual stabilizer directions and their
character; the remaining $n-m$ directions are not requested.  Nature supplies
the learner's $k$ systems separately from the verifier's independent test
system, and the verifier performs the characteristic-adapted commuting
joint-character test.  At $m=n$ the acceptance probability equals full-label
fidelity within the complete coordinate ensemble, whereas for $m<n$ it is a
partial verification score.}
\label{fig:prediction-game-protocol}
\end{figure}

\subsection{Score achievable without using the input copies}
\label{subsec:partial-no-information-baseline}
The learner can obtain a nonzero verification score without measuring any
of its $k$ input copies.  To see the baseline, it fixes the matrix part
$C_0$ of the report and chooses the reported character uniformly from
$\F_p^m$.  The next proposition states the score of this input-independent
strategy, which will be used as a finite-size reference point.

\begin{proposition}[Input-independent score $p^{-m}$]
\label{prop:partial-no-information-baseline}
For every $k\geq0$, the standardized optimum verification score
is at least the input-independent baseline:
\begin{equation}
 S_{\mathrm{ver}}^{(k),*}(n,m;p)\geq p^{-m}.
 \label{eq:partial-no-information-baseline}
\end{equation}
\end{proposition}

\begin{proof}
Fix $C_0\in\mathcal C_{n,m}$.  On the learner's $k$-copy Hilbert space,
define measurement operators indexed by
$\Theta_{n,m,p}^{\mathrm{part}}=\mathcal C_{n,m}\times\F_p^m$ by
\begin{equation}
 M_{C,\gamma}^{(0)}
 :=
 \begin{cases}
 p^{-m}I,&C=C_0,\\
 0,&C\ne C_0.
 \end{cases}.
 \label{eq:partial-baseline-explicit-POVM}
\end{equation}
Every operator in Eq.~\eqref{eq:partial-baseline-explicit-POVM} is positive.
The sum over all outcomes is normalized because only the fixed matrix label
$C_0$ contributes:
\begin{equation}
 \sum_{C\in\mathcal C_{n,m}}\sum_{\gamma\in\F_p^m}
 M_{C,\gamma}^{(0)}
 =\sum_{\gamma\in\F_p^m}p^{-m}I=I.
 \label{eq:partial-baseline-POVM-normalization}
\end{equation}
Thus Eq.~\eqref{eq:partial-baseline-explicit-POVM} defines a feasible learner POVM.  
Operationally, it ignores the learner's input systems, fixes the matrix label  
$C_0$, and chooses the character label $\gamma$ uniformly.  We now verify that  
its verification score is $p^{-m}$ for every true state, rather than only after  
averaging over the prior.  For a fixed true state $\rho_{A,b}$, the conditional  
verification score is  
\begin{align}
 &\sum_{C\in\mathcal C_{n,m}}\sum_{\gamma\in\F_p^m}
 \Tr\!\left[M_{C,\gamma}^{(0)}\rho_{A,b}^{\otimes k}\right]
 \Tr\!\left[\Pi_{C,\gamma}^{(m)}\rho_{A,b}\right]
 =p^{-m}\sum_{\gamma\in\F_p^m}
 \Tr\!\left[\Pi_{C_0,\gamma}^{(m)}\rho_{A,b}\right]\notag\\
 &\quad=p^{-m}\Tr\!\left[
 \left(\sum_{\gamma\in\F_p^m}\Pi_{C_0,\gamma}^{(m)}\right)
 \rho_{A,b}\right]
 =p^{-m}.
 \label{eq:partial-baseline-average}
\end{align}
The second equality uses
$\Tr[M_{C_0,\gamma}^{(0)}\rho_{A,b}^{\otimes k}]=p^{-m}$, and the last equality
uses the joint-projector resolution of the identity and
$\Tr\rho_{A,b}=1$.  This equality holds for every $(A,b)$, so averaging over
the uniform prior gives the feasible value $p^{-m}$.  Maximizing over all
learner POVMs proves Eq.~\eqref{eq:partial-no-information-baseline}.
For $p=2$, the same input-independent construction uses the binary
joint-character projectors $\Pi_{C,\gamma}^{(m,2)}$ from
Eq.~\eqref{eq:integrated-qubit-accept-projector}; their resolution of the
identity in Eq.~\eqref{eq:integrated-qubit-projector-resolution} gives the
same value $2^{-m}$.
\end{proof}

\subsection{Proof of the frame-to-coordinate equivalence}
\label{subsec:frame-coordinate-equivalence-proof}
Both characteristic-specific state families and acceptance projectors
are now defined.  To prove
Proposition~\ref{prop:frame-conditioned-coordinate-reduction}, we first
establish bijections for the input states and report labels.  On these
bijections we check two distinct equalities: the learner's outcome
probability in Eq.~\eqref{eq:frame-coordinate-outcome-probability-preservation}
and the verifier's Born acceptance score in
Eq.~\eqref{eq:frame-coordinate-verification-score-preservation}.
The first is needed for exact identification; both are needed for the
expected verification score.  Only after these statewise equalities are
established do we compare the optimized values.

\begin{proof}[Proof of Proposition~\ref{prop:frame-conditioned-coordinate-reduction}]
Recall that the standard vertical frame is
$R_j^{\mathrm{std}}=\binom{0}{e_j}$ for $1\le j\le n$.
Choose a symplectic automorphism $F_R\in\operatorname{Sp}(2n,\F_p)$
satisfying $F_RR_j=R_j^{\mathrm{std}}$ for every $j$.  Choose a Clifford
unitary $U_R$ implementing this automorphism on Pauli labels.  Relative
to the fixed physical Pauli section $z\mapsto P_z$, define its
unit-modulus phase $\tau_R(z)$ by
\begin{equation}
 U_RP_zU_R^\dagger=\tau_R(z)P_{F_Rz},
 \qquad z\in\F_p^{2n}.
 \label{eq:frame-coordinate-Clifford-phase}
\end{equation}
We keep the symplectic standardization $F_R$ and its Clifford
implementation $U_R$ fixed throughout the proof.  The resulting report
map is denoted by $\Phi_{R,U_R}$, or $\Phi_R$ when the chosen unitary is
understood.  Neither the symplectic map nor its Clifford implementation
needs to be unique.  A different choice may change the coordinate
character labels, but gives another bijection with the same optimal values.

\noindent\textbf{Proof of (1): states.}
For $x\in\mathsf S(R)$, define $x'$ by the density operator
$\rho_{x'}:=U_R\rho_xU_R^\dagger$.  The label $x'$ exists and is unique
because conjugation by a Clifford unitary maps pure stabilizer states
to pure stabilizer states, which are labeled here by their distinct
density operators.  The transported stabilizer data satisfy
\begin{equation*}
 L_{x'}=F_RL_x,\qquad
 \chi_{x'}(F_Rz)=\tau_R(z)^{-1}\chi_x(z)\quad(z\in L_x).
\end{equation*}
The inverse map on states is conjugation by $U_R^\dagger$; thus
$x\mapsto x'$ bijects the fixed-frame state family with the
standard transverse state family.  Equation~\eqref{eq:dual-family-normal-form-coordinate}
identifies the transported reference-dual directions as
$z_C(e_i)=(e_i,Ae_i)$ for $1\leq i\leq m$, where
$C=AJ_m\in\mathcal C_{n,m}$ and
$z_C(t):=(J_mt,Ct)$ for $t\in\F_p^m$.

\noindent\textbf{Proof of (2): reports.}
Unlike the state correspondence, the report correspondence must also
track the eigenvalue phases.  The matrix $C$ records the transported
directions; the parameter $\gamma$ is determined only after transporting
the physical assignment and passing to the additive coordinate family.
We define the function $\Phi_{R,U_R}$ in
Eq.~\eqref{eq:frame-coordinate-report-bijection} directly on reports.  Let
$y=(S_1,\ldots,S_m,\eta)\in\Theta_{R,m}^{\mathrm{frame}}$.
The ordered directions $F_RS_i=z_C(e_i)$ uniquely determine
$C\in\mathcal C_{n,m}$.  Their span is the image under $F_R$ of
$K_y:=\operatorname{span}\{S_1,\ldots,S_m\}$.

First transport the intrinsic assignment by Clifford conjugation.  For every
$z\in K_y$ its value on the transformed subspace is
$\eta_y^{\mathrm{std}}(F_Rz)=\tau_R(z)^{-1}\eta(z)$, by
Eq.~\eqref{eq:frame-coordinate-Clifford-phase}.  This step is invertible
because $F_R$ is invertible and $\tau_R(z)$ has unit modulus.

Next pass from the physical Pauli section $P_{z_C(t)}$ to the additive
coordinate family $Q_C(t)$.  The latter is the family in
Eq.~\eqref{eq:partial-announced-Weyl-family} for odd $p$ and in
Eq.~\eqref{eq:integrated-qubit-announced-Pauli} for $p=2$.
Write $Q_C(t)=s_C(t)P_{z_C(t)}$, with $|s_C(t)|=1$.
The additive character of this family determines a unique
$\gamma\in\F_p^m$ by
\begin{equation}
 s_C(t)\,\eta_y^{\mathrm{std}}(z_C(t))
 =\begin{cases}
  \omega^{\gamma^Tt},&p\text{ odd},\\
  (-1)^{\gamma^Tt},&p=2,
 \end{cases}
 \qquad t\in\F_p^m.
 \label{eq:frame-coordinate-report-phase-extraction}
\end{equation}
The left side is an additive character because the operators $Q_C(t)$
form an additive commuting representation.  Every such character has a
unique parameter $\gamma$.  The phase $\tau_R$ accounts for Clifford
transport of the physical operators, whereas $s_C$ accounts for replacing
the transported physical section by the additive coordinate family.
Their combined effect is required to recover the correct eigenvalue label;
the odd-prime and binary projector identities appear in
Eqs.~\eqref{eq:odd-intrinsic-coordinate-projector-equality} and
\eqref{eq:qubit-intrinsic-coordinate-projector-equality}, respectively.
Define $\Phi_R(y):=(C,\gamma)$.

The map $\Phi_{R,U_R}(y)=(C,\gamma)$ is well defined without
choosing an input $x$ with $y_R(x)=y$:
$F_RS_i=z_C(e_i)$ determines $C$, while
Eq.~\eqref{eq:frame-coordinate-report-phase-extraction} determines
$\gamma$ from $\eta$ and the fixed phases $\tau_R,s_C$.
Its inverse at the \emph{report level} can be written explicitly.  Given
$(C,\gamma)$, put $z_C(t)=(J_mt,Ct)$ and let $\chi_\gamma(t)$ denote the
right-hand side of Eq.~\eqref{eq:frame-coordinate-report-phase-extraction}.
Recover the ordered directions and their assignment from
\begin{align}
 S_i&=F_R^{-1}z_C(e_i)\qquad(1\le i\le m),\notag\\
 \eta\bigl(F_R^{-1}z_C(t)\bigr)
 &=\tau_R\bigl(F_R^{-1}z_C(t)\bigr)
   s_C(t)^{-1}\chi_\gamma(t),\qquad t\in\F_p^m.
 \label{eq:frame-coordinate-report-inverse}
\end{align}
We verify that the reconstructed assignment satisfies the cocycle
condition~\eqref{eq:cocycle-compatible-assignment}.  Write
$z_t:=F_R^{-1}z_C(t)$ for $t\in\F_p^m$.  The Clifford relation
\eqref{eq:frame-coordinate-Clifford-phase} and the multiplication rule
\eqref{eq:global-Pauli-section-cocycle} imply
\begin{equation*}
 \tau_R(z_t)\tau_R(z_r)c_p(z_C(t),z_C(r))
 =c_p(z_t,z_r)\tau_R(z_{t+r}).
\end{equation*}
Because $Q_C(t+r)=Q_C(t)Q_C(r)$, the definition
$Q_C(t)=s_C(t)P_{z_C(t)}$ also gives
\begin{equation*}
 s_C(t+r)=s_C(t)s_C(r)c_p(z_C(t),z_C(r)).
\end{equation*}
Since $\chi_\gamma(t+r)=\chi_\gamma(t)\chi_\gamma(r)$, these two
identities and Eq.~\eqref{eq:frame-coordinate-report-inverse} yield
\begin{align*}
 \eta(z_{t+r})
 =\tau_R(z_{t+r})s_C(t+r)^{-1}\chi_\gamma(t+r)
 =c_p(z_t,z_r)^{-1}\eta(z_t)\eta(z_r).
\end{align*}
This is precisely the required cocycle compatibility on $K_y$.
The formulas reconstruct a unique report, proving injectivity; they do
not reconstruct a unique complete stabilizer state.  For a complete
state with standardized coordinates $(A,b)$, the report is
\begin{equation}
 (C,\gamma)=
 \begin{cases}
 (AJ_m,-J_m^Tb),&p\text{ odd},\\
 (AJ_m,J_m^Tb),&p=2.
 \end{cases}
 \label{eq:frame-intrinsic-to-coordinate-label-map}
\end{equation}

To prove surjectivity on report labels, take any
$(C,\gamma)\in\mathcal C_{n,m}\times\F_p^m$ and write
$C=\binom{C_0}{C_1}$, where $C_0\in\Sym_m(\F_p)$ and
$C_1\in\F_p^{(n-m)\times m}$.  The following symmetric matrix $A$
and vector $b$ give one complete-state extension:
\begin{equation}
 A=\begin{pmatrix}C_0&C_1^T\\ C_1&0\end{pmatrix},\qquad
 b=\begin{cases}
 \binom{-\gamma}{0},&p\text{ odd},\\
 \binom{\gamma}{0},&p=2.
 \end{cases}
 \label{eq:frame-coordinate-label-extension}
\end{equation}
The lower-right zero block of $A$ is an $(n-m)\times(n-m)$ matrix,
and the final zero block of $b$ has $n-m$ entries.  Since $AJ_m=C$,
this state realizes the prescribed report label.  Transporting
its report back with $U_R^\dagger$ yields an element of
$\Theta_{R,m}^{\mathrm{frame}}$.  This proves surjectivity.  The
extension $(A,b)$ is not unique when $m<n$; uniqueness of the inverse
holds for \emph{report labels}, not for complete states.  In particular,
$\Phi_R$ is a bijection of report alphabets.  In the binary case the
related set bijection through the enhanced quotient is explained in
Eq.~\eqref{eq:binary-quotient-coordinate-set-map}; its group law must
not be replaced by componentwise addition on $(C,\gamma)$.

\noindent\textbf{Proof of (3): probabilities and projectors.}
We check the measurement and verifier separately.  The measurement
identity transports the probabilities of all outcomes, while the
projector identity compares the acceptance score even for an incorrect
report.  For each learner POVM
$\{M_y\}_{y\in\Theta_{R,m}^{\mathrm{frame}}}$, define its coordinate
counterpart by
\begin{equation}
 M'_{\Phi_R(y)}=U_R^{\otimes k}M_y(U_R^\dagger)^{\otimes k}.
 \label{eq:frame-coordinate-POVM-transport}
\end{equation}
For every input $x$ and every report $y$, unitary invariance of the trace
gives the equality of the learner's outcome probabilities:
\begin{equation}
 \Tr[M_y\rho_x^{\otimes k}]
 =\Tr[M'_{\Phi_R(y)}(U_R\rho_xU_R^\dagger)^{\otimes k}].
 \label{eq:frame-coordinate-outcome-probability-preservation}
\end{equation}
The true report $y_R(x)$ maps to the true coordinate label, so the last
equality also preserves the exact-success probability.

For verification, write $\Pi_y:=\Pi_{K_y,\eta}$ for the intrinsic accept
projector of $y=(S_1,\ldots,S_m,\eta)$, with
$K_y=\operatorname{span}\{S_1,\ldots,S_m\}$.  For $\Phi_R(y)=(C,\gamma)$, write
$\Pi_{\Phi_R(y)}$ for the coordinate acceptance projector of the additive
family $Q_C(t)$ with character $\chi_\gamma(t)$.  Explicitly, this
projector is $p^{-m}\sum_{t\in\F_p^m}\chi_\gamma(t)^{-1}Q_C(t)$;
the odd-prime and binary formulas are given in
Eqs.~\eqref{eq:partial-accept-projector} and
\eqref{eq:integrated-qubit-accept-projector}.
The projector equality can be checked directly, for every report
$y$, whether or not it is correct for the input state.  Using the
intrinsic projector formula~\eqref{eq:intrinsic-character-projector},
Clifford conjugation, and
Eq.~\eqref{eq:frame-coordinate-report-phase-extraction}, we obtain
\begin{align*}
 U_R\Pi_y U_R^\dagger
 =p^{-m}\sum_{t\in\F_p^m}
   \eta(z_t)^{-1}\tau_R(z_t)P_{z_C(t)}
 =p^{-m}\sum_{t\in\F_p^m}
   \chi_\gamma(t)^{-1}Q_C(t)
 =\Pi_{\Phi_R(y)},
\end{align*}
where $z_t=F_R^{-1}z_C(t)$.  The last expression is the coordinate
accept projector for both odd primes and $p=2$; it agrees with the
convention-specific formulas in
Eqs.~\eqref{eq:odd-intrinsic-coordinate-projector-equality} and
\eqref{eq:qubit-intrinsic-coordinate-projector-equality}.
Taking the trace against $U_R\rho_xU_R^\dagger$ now gives the equality
of the Born acceptance scores:
\begin{equation}
 \Tr[\Pi_y\rho_x]
 =\Tr[\Pi_{\Phi_R(y)}\,U_R\rho_xU_R^\dagger].
 \label{eq:frame-coordinate-verification-score-preservation}
\end{equation}
Equations~\eqref{eq:frame-coordinate-outcome-probability-preservation}
and~\eqref{eq:frame-coordinate-verification-score-preservation} give
both state--report probability equalities, including for incorrect
reports.  The effects in Eq.~\eqref{eq:frame-coordinate-POVM-transport}
are positive and normalized because unitary conjugation preserves
positivity and $\sum_y M_y=I$.  The inverse maps of (1) and (2)
and conjugation by $U_R^\dagger$ provide the reverse POVM
correspondence.  Thus statewise exact and verification performances
are preserved for every learner POVM.

\smallskip\noindent\textbf{Proof of (4): optimized values.}
Both performance criteria have been preserved for each input and each
learner POVM; it remains to optimize over the corresponding POVM classes.
The state bijection in (1) preserves the uniform input distribution.
By (3), every POVM on either report alphabet has a POVM with the same
statewise performance on the other alphabet, in both directions.
Taking the uniform average and the supremum over all POVMs yields
Eqs.~\eqref{eq:finite-frame-coordinate-equivalence}
and~\eqref{eq:finite-frame-verification-equivalence}.  Taking instead
the minimum over input states before the supremum gives the two
corresponding worst-case equalities.  Their coordinate-side values
are independent of $R$, completing the proof.
\end{proof}

\subsection{Fixed query subspaces and canonical coordinates}
\label{subsec:query-coordinate-equivalence}
\label{subsubsec:fixed-query-correspondences}

We now compare two fixed query matrices within the odd-prime Weyl
realization.  Proposition~\ref{prop:frame-conditioned-coordinate-reduction}
transported intrinsic frame-relative reports to coordinates; the present
comparison changes only the queried coordinate directions.  Let $J\in\F_p^{n\times m}$ have full column rank
and define
\begin{equation}
 \mathcal C_J:=\{C\in\F_p^{n\times m}:J^TC=C^TJ\}.
 \label{eq:general-query-symmetric-matrix-space}
\end{equation}
The associated true label is $(AJ,-J^Tb)$, and the verifier tests the family
$(J;C)$.  Write
\begin{equation}
 \Theta_J:=\mathcal C_J\times\F_p^m,
 \quad
 \Theta_{\rm can}:=\mathcal C_{n,m}\times\F_p^m.
 \label{eq:general-and-canonical-query-label-spaces}
\end{equation}
For $(C,\gamma)\in\Theta_J$, let $\Pi_{C,\gamma}^{(J)}$ denote the joint-character
projector for the Weyl family $\{W(Jx,Cx):x\in\F_p^m\}$.  For $y\in\Theta_J$,
let $\sigma_y^{(J)}$ be the $k$-copy state obtained by averaging
$\rho_{A,b}^{\otimes k}$ over the complete coordinate labels $(A,b)$ satisfying
$(AJ,-J^Tb)=y$.  For $J=J_m$, the notations $\sigma_y^{(J_m)}$ and
$\Pi_{C,\gamma}^{(J_m)}$ designate, respectively, this fiber-averaged
$k$-copy state and the corresponding coordinate acceptance projector.

The general full-rank query matrix $J$ specifies which $m$ directions
are reported, whereas $J_m$ selects the first $m$ standard coordinate
directions.  To compare their learning problems, we change basis so
that $J=RJ_m$ and transport three objects together: the partial labels,
the conditioned $k$-copy states, and the verifier projectors.  This
simultaneous transport, not merely a relabeling of the reported matrix,
will show that the optimum is the same for $J$ and $J_m$.

\begin{proposition}[Explicit equivalence of fixed full-rank queries]
\label{prop:fixed-query-equivalence}
Let $J\in\F_p^{n\times m}$ have full column rank.  Choose
$R\in\operatorname{GL}(n,\F_p)$ such that $J=RJ_m$, and define
\begin{equation}
 T_R:\Theta_J\longrightarrow\Theta_{\rm can},
 \quad
 T_R(C,\gamma):=(R^TC,\gamma).
 \label{eq:fixed-query-label-map}
\end{equation}
Then $T_R$ is a bijection with inverse
\begin{equation}
 T_R^{-1}(C',\gamma)=(R^{-T}C',\gamma).
 \label{eq:fixed-query-label-map-inverse}
\end{equation}
This change of coordinates induces a corresponding transformation of every
complete coordinate label in the transverse model.  For $(A,b)$, denote the transformed label by $(A',b')$, where
\begin{equation}
 A':=R^TAR,
 \quad b':=R^Tb.
 \label{eq:fixed-query-transformed-complete-label}
\end{equation}
The true queried labels obey
\begin{equation}
 T_R(AJ,-J^Tb)=(A'J_m,-J_m^Tb').
 \label{eq:fixed-query-true-label-correspondence}
\end{equation}
Here $R$ is an invertible $n\times n$ coordinate matrix, rather than
an auxiliary frame in the $2n$-dimensional Pauli phase space.
The basis change $R$ in Eq.~\eqref{eq:fixed-query-transformed-complete-label} is implemented quantumly by the computational-basis
unitary $U_R\ket{x}:=\ket{R^{-1}x}$.  It intertwines both the verifier
projectors and the conditioned learner states.  For
$y=(C,\gamma)\in\Theta_J$, the two correspondence relations are
\begin{align}
 U_R\Pi_{C,\gamma}^{(J)}U_R^\dagger
 =\Pi_{T_R(y)}^{(J_m)},~
 U_R^{\otimes k}\sigma_y^{(J)}U_R^{\dagger\otimes k}
 =\sigma_{T_R(y)}^{(J_m)}.
 \label{eq:fixed-query-projector-correspondence}
\end{align}
Consequently, conjugation by $U_R^{\otimes k}$ together with the label map
$T_R$ gives a bijection between the feasible learner POVMs for the $J$-query
problem and those for the canonical $J_m$-query problem.  It preserves every
exact-identification term and every verification-score term.  Therefore the optimal
exact-identification probability and the optimal verification score for
query $J$ equal, respectively,
$P_{\mathrm{id}}^{(k),*}(n,m;p)$ and $S_{\mathrm{ver}}^{(k),*}(n,m;p)$.
\end{proposition}

\begin{proof}
First verify the domains in Eqs.~\eqref{eq:fixed-query-label-map} and
\eqref{eq:fixed-query-label-map-inverse}.  For $C\in\mathcal C_J$,
the transformed column satisfies the required symmetry condition:
\begin{equation}
 J_m^T(R^TC)=J^TC=C^TJ=(R^TC)^TJ_m,
 \label{eq:fixed-query-forward-admissibility}
\end{equation}
so $R^TC\in\mathcal C_{n,m}$.  The inverse transformation also preserves
admissibility.  For $C'\in\mathcal C_{n,m}$, the inverse image satisfies
the symmetry condition for the query $J$:
\begin{equation}
 J^T(R^{-T}C')=J_m^TC'=C'^TJ_m=(R^{-T}C')^TJ,
 \label{eq:fixed-query-inverse-admissibility}
\end{equation}
so $R^{-T}C'\in\mathcal C_J$.  The two displayed maps are inverse because
$R^{-T}R^T=I$.  This proves the asserted bijection, including its domain,
codomain, forward map, and inverse.

Applying $U_R\ket{x}=\ket{R^{-1}x}$ and changing the summation variable
from $x$ to $Rx$ gives
\begin{equation}
 U_R\ket{\psi_{A,b}}=\ket{\psi_{R^TAR,R^Tb}}
 =\ket{\psi_{A',b'}}.
 \label{eq:fixed-query-state-correspondence}
\end{equation}
Since $J=RJ_m$, Eq.~\eqref{eq:fixed-query-transformed-complete-label} gives
\begin{equation}
 R^TAJ=R^TARJ_m=A'J_m,
 \quad
 J^Tb=J_m^TR^Tb=J_m^Tb',
 \label{eq:fixed-query-true-coordinate-calculation}
\end{equation}
which proves Eq.~\eqref{eq:fixed-query-true-label-correspondence}.

The action of $U_R$ on computational-basis states determines its conjugation
action on the shift and phase operators.  For all $u,v\in\F_p^n$, one has
\begin{align}
 U_RX(u)U_R^\dagger=X(R^{-1}u),~
 U_RZ(v)U_R^\dagger=Z(R^Tv).
 \label{eq:fixed-query-Z-conjugation}
\end{align}
Moreover $(R^{-1}u)^T(R^Tv)=u^Tv$, so the phase
$\omega^{-u^Tv/2}$ in the Weyl convention is preserved.  Combining this phase
invariance with the two conjugation identities gives the Weyl-operator
correspondence
\begin{equation}
 U_RW(u,v)U_R^\dagger=W(R^{-1}u,R^Tv).
 \label{eq:fixed-query-general-Weyl-conjugation}
\end{equation}
For $u=Jx=RJ_mx$, $v=Cx$, and $C'=R^TC$, this becomes
\begin{equation}
 U_RW(Jx,Cx)U_R^\dagger=W(J_mx,C'x).
 \label{eq:query-coordinate-Weyl-family}
\end{equation}
Substituting Eq.~\eqref{eq:query-coordinate-Weyl-family} into the character-sum
definition of the joint projector proves
Eq.~\eqref{eq:fixed-query-projector-correspondence}.  The transformation
$(A,b)\mapsto(A',b')$ is a bijection of the uniform complete-label space, and
Eq.~\eqref{eq:fixed-query-true-label-correspondence} maps the fiber over
$y\in\Theta_J$ bijectively onto the fiber over $T_R(y)$.  Conjugating each term
in the fiber average by $U_R^{\otimes k}$ therefore proves
Eq.~\eqref{eq:fixed-query-projector-correspondence}.

Finally, a POVM $\{M_y:y\in\Theta_J\}$ for the $J$-query task
induces the following POVM indexed by
$\Theta_{\rm can}=\mathcal C_{n,m}\times\F_p^m$ for the $J_m$-query task:
\begin{equation}
 \widetilde M_{T_R(y)}
 :=U_R^{\otimes k}M_yU_R^{\dagger\otimes k}.
 \label{eq:fixed-query-POVM-correspondence}
\end{equation}
Positivity and normalization are preserved, and the inverse construction uses
$T_R^{-1}$ and $U_R^{\dagger\otimes k}$.  Hence
Eq.~\eqref{eq:fixed-query-POVM-correspondence} is a bijection of feasible POVMs.
The state correspondence in Eq.~\eqref{eq:fixed-query-state-correspondence}
preserves the probability of each relabeled learner outcome, while the
projector correspondence in Eq.~\eqref{eq:fixed-query-projector-correspondence}
preserves the verifier's Born acceptance probability for each report.
Taking the optima over the corresponding POVMs proves the final assertion.
\end{proof}

Proposition~\ref{prop:frame-conditioned-coordinate-reduction} identifies
the intrinsic frame-relative task with its standardized coordinate problem.
Within the odd-prime coordinate realization,
Proposition~\ref{prop:fixed-query-equivalence} further shows that replacing
the standard inclusion $J_m$ by any fixed full-rank query preserves each
exact-identification and verification-score term.  The first reduction transports intrinsic reports;
the second compares two coordinate query matrices.  The subsequent finite-copy analysis therefore uses the fixed query
$J_m$.
\subsection{Full-label fidelity learning in the case \texorpdfstring{$m=n$}{m=n}}
\label{subsec:complete-endpoint}
At $m=n$, the queried coordinate label becomes complete within the ensemble
conditioned on the public frame.  This is still not the unrestricted
stabilizer ensemble of Eqs.~\eqref{eq:unrestricted-stabilizer-wc-success} and~\eqref{eq:unrestricted-stabilizer-av-success}, which is
used separately as the direct benchmark.  Exact identification tests equality
of complete coordinate labels, whereas the verification payoff becomes the squared
overlap with the reported coordinate state; accordingly, the two finite
optima need not coincide.  The characteristic-specific score formulas established above ensure only
that fidelity one is attained exactly for the correct complete label.

\begin{definition}[Full-label fidelity learning]
In the case $m=n$, the learner outputs a complete coordinate label
$(A',b')$.  Against the true label $(A,b)$, the score is
\begin{equation}
 \left|\braket{\psi_{A',b'}}{\psi_{A,b}}\right|^2.
 \label{eq:full-label-fidelity-score}
\end{equation}
We call this the full-label fidelity-learning problem, and denote its optimal
average score by $S_{\mathrm{fid}}^{(k),*}(n,p)$.
\end{definition}

The symbol $S_{\mathrm{fid}}^{(k),*}(n,p)$ denotes a fidelity
score, consistently with the score notation used above.  The exact optimum remains
$P_{\mathrm{id}}^{(k),*}(n,n;p)$.

\begin{corollary}[Recovery of full-label fidelity learning]
\label{cor:complete-fidelity-endpoint}
If $m=n$, then $J_n=I_n$,
$\mathcal C_{n,n}=\Sym_n(\F_p)$, and a report can be written uniquely
as $(C,\gamma)=(A',-b')$.  Under this identification, the verifier score
coincides exactly with the fidelity between the two complete-label states:
\begin{equation}
 s_n((A,b),(A',-b'))
 =\left|\braket{\psi_{A',b'}}{\psi_{A,b}}\right|^2.
 \label{eq:partial-score-complete-fidelity}
\end{equation}
The two optimization problems have the same objective on the same
complete-label output alphabet.  Their optimal values are equal:
\begin{equation}
 S_{\mathrm{ver}}^{(k),*}(n,n;p)
 =S_{\mathrm{fid}}^{(k),*}(n,p),
 \label{eq:partial-complete-optimal-equality}
\end{equation}
where the right-hand side is the full-label fidelity-learning optimum.
\end{corollary}

\begin{proof}
For odd $p$, the complete-label identification is the symmetric-matrix
coordinate construction.  For $p=2$, it uses the enhanced binary coordinates
and the binary rank-one projector construction, rather than the odd-prime
Weyl phase formula.  In either case $m=n$ makes the unreported subgroup trivial, so the
quotient is the full coordinate-label group.  The projectors in
Eqs.~\eqref{eq:partial-accept-projector} and
\eqref{eq:integrated-qubit-accept-projector} select the common eigenspace
of $n$ independent commuting operators.  Their character labels give
the following rank-one projectors in the two characteristics:
\begin{align}
 \Pi_{A',-b'}^{(n)}
 &=\ket{\psi_{A',b'}}\!\bra{\psi_{A',b'}}
 &&(p\text{ odd}),\notag\\
 \Pi_{A',b'}^{(n,2)}
 &=\ket{\psi_{A',b'}^{(2)}}\!\bra{\psi_{A',b'}^{(2)}}
 &&(p=2).
 \label{eq:full-label-characteristic-projectors}
\end{align}
The eigenvalue relations in
Eqs.~\eqref{eq:odd-eigenvalue-character} and
\eqref{eq:integrated-qubit-eigenvalue-character} identify these
eigenspaces with the stated coordinate states.  Since $-b'=b'$ in
characteristic two, both use the label $(C,\gamma)=(A',-b')$ in the
corollary.  Born's rule now gives
Eq.~\eqref{eq:partial-score-complete-fidelity}, with the binary state
vectors interpreted as $\ket{\psi_{A,b}^{(2)}}$.  The feasible measurements
and priors are identical under this relabeling.
\end{proof}

\section{Finite-size exact identification}
\label{sec:common-finite-size-framework}
We now solve the finite exact-identification problem for the refined task.
The states obtained by averaging over complete labels compatible with a given partial label are generally mixed (they can be pure when the compatible fiber has a single state). Covariance alone therefore does not determine an optimal measurement.  After formulating the quotient
ensemble, we group computational-basis tuples that have the same moment
value.  The normalized sum over each group is a single direction in the
$k$-copy Hilbert space.  Averaging over unreported parameters preserves
coherence only between directions with the same lower moment coordinates.
We use the resulting decomposition to optimize the measurement and obtain
the exact finite-copy formula used in the converse.

\needspace{19\baselineskip}
\subsection{The quotient ensemble and the finite-copy calculation}
\label{subsec:finite-size-quotient-architecture}

The complete coordinate labels form the following finite abelian groups:
\begin{equation}
 G_{n,p}:=
 \begin{cases}
 \Sym_n(\F_p)\times\F_p^n,&p\text{ odd},\\
 \mathbb Z_4^n\times\F_2^{\binom n2},&p=2.
 \end{cases}
 \label{eq:common-complete-label-groups}
\end{equation}
For $p=2$, this is the group $G_{n,2}$ defined in
Eq.~\eqref{eq:integrated-qubit-group}, with coordinates $(\alpha,E)$.  Addition is componentwise:
\begin{equation}
 \begin{aligned}
 (A,b)+(A',b')&=(A+A',b+b') &&(p\text{ odd}),\\
 (\alpha,E)+(\alpha',E')&=(\alpha+\alpha',E+E') &&(p=2).
 \end{aligned}
 \label{eq:complete-label-group-addition}
\end{equation}
The odd-prime sums are taken over $\F_p$.  In the binary group, the
$\alpha$ coordinates add modulo four and the $E$ coordinates add modulo two.
Let $H_{n,m,p}$ be the subgroup of unreported coordinates.  For odd $p$,
this subgroup is
\begin{equation}
 H_{n,m,p}:=\left\{\left(\begin{pmatrix}0&0\\0&F\end{pmatrix},
 \binom{0}{c}\right):F\in\Sym_{n-m}(\F_p),\ c\in\F_p^{n-m}\right\}.
 \label{eq:H}
\end{equation}
For $p=2$ its enhanced-coordinate version is
Eq.~\eqref{eq:integrated-qubit-unreported-subgroup}.  Put
$Y_{n,m,p}:=G_{n,p}/H_{n,m,p}$.  For $g\in G_{n,p}$, let $\rho_g$
denote the pure coordinate stabilizer state associated with the complete
additive label $g$.  Thus $\rho_g=\rho_{A,b}$ in odd characteristic and is
the state associated with $\ket{\psi_g^{(2)}}$ in characteristic two.  The partial-label quotient has cardinality:
\begin{equation}
 |Y_{n,m,p}|=p^{(n+1)m-m(m-1)/2}.
 \label{eq:Y-size}
\end{equation}
For $y\in Y_{n,m,p}$, define the conditioned state
\begin{equation}
 \sigma_y:=\frac1{|H_{n,m,p}|}\sum_{h\in H_{n,m,p}}
 \rho_{s_Y(y)+h}^{\otimes k},
 \label{eq:sigma-y}
\end{equation}
where $s_Y(y)\in G_{n,p}$ is any representative of the quotient label $y$.
Also set
\begin{equation}
 \kappa_p^{\mathrm{GL}}:=\prod_{j=1}^{\infty}(1-p^{-j})>0.
 \label{eq:kappa-GL-product}
\end{equation}

The finite exact-identification calculation proceeds from the
complete-label group $G_{n,p}$ to its quotient
$Y_{n,m,p}=G_{n,p}/H_{n,m,p}$, which indexes the partial labels.  The
conditioned states $\sigma_y$ in Eq.~\eqref{eq:sigma-y} average the
$k$-copy states over the complete labels with the same reported label.
The next two subsections give the moment maps for odd primes and for
characteristic two.  In each case we group equal moment values, determine
which matrix elements survive the subgroup average, and write the
conditioned states explicitly.  The measurement optimization is then
proved in Section~\ref{sec:seed-sdp}.  The table summarizes the data used
in these two calculations.  
\begin{center}
\small
\setlength{\tabcolsep}{7pt}
\renewcommand{\arraystretch}{1.22}
\begin{tabularx}{0.96\textwidth}{@{}>{\raggedright\arraybackslash}p{0.20\textwidth}YY@{}}
\toprule
\textbf{Layer}
& \textbf{Odd-prime realization}
& \textbf{Characteristic-two realization}\\
\midrule
\textbf{Additive parameters}
& $\Sym_n(\F_p)\times\F_p^n$
& $\mathbb Z_4^n\times\F_2^{\binom n2}$\\
\addlinespace[1pt]
\textbf{Moment data}
& Symmetric Gram moment and linear moment
& Mod-$4$ diagonal and mod-$2$ off-diagonal moments\\
\addlinespace[1pt]
\textbf{Grouping after averaging}
& Equal lower Gram and linear coordinates
& Equal lower mod-$4$ diagonal and mod-$2$ edge coordinates\\
\addlinespace[1pt]
\textbf{Common finite-size conclusion}
& \multicolumn{2}{>{\raggedright\arraybackslash}p{0.70\textwidth}@{}}{Exact identification probability and an optimal measurement, proved after the two moment expansions}\\
\bottomrule
\end{tabularx}
\end{center}

\subsection{Moment expansion and subgroup averaging for odd primes}
\label{sec:sectors}
For odd $p$, the phase of a computational-basis tuple depends only on its
Gram matrix and its vector sum.  We collect tuples with the same two values
to obtain an orthonormal expansion of the orbit states.  We then average
over the unreported subgroup to obtain the conditioned ensemble.  The
measurement optimization follows in Section~\ref{sec:seed-sdp}.

\subsubsection{The quotient ensemble and the moment map}
For odd $p$, the dependence on the complete label $(A,b)$ factors
through two sums of the computational-basis tuple.  The next definition
collects these sums into a moment map and uses its fibers to construct
orthogonal vectors for the group action.  We write
$G=G_{n,p}=\Sym_n(\F_p)\times\F_p^n$ for the odd-prime group in
Eq.~\eqref{eq:common-complete-label-groups}, retain its subgroup $H$ in
Eq.~\eqref{eq:H}, and use the quotient $Y=G/H$ defining the conditioned
states $\{\sigma_y:y\in Y\}$ of Eq.~\eqref{eq:sigma-y}.

\begin{definition}[Moment values, fibers, and normalized sector vectors]
An ordered computational-basis tuple in the $k$-copy expansion is an element
\begin{equation}
 \boldsymbol x=(x_1,\ldots,x_k)\in(\F_p^n)^k.
 \label{eq:ordered-basis-tuple-domain}
\end{equation}
Define its moment value by the map
\begin{equation}
 \Phi_k:(\F_p^n)^k\longrightarrow
 \Sym_n(\F_p)\times\F_p^n,
 \quad
 \Phi_k(\boldsymbol x)
 :=\left(\sum_{j=1}^k x_jx_j^T,\ \sum_{j=1}^k x_j\right).
 \label{eq:moment-map}
\end{equation}
For $z=(Q,S)\in G$, its \emph{moment fiber} is the preimage
$\Phi_k^{-1}(z)$, and its multiplicity is the cardinality $N_z$:

\begin{equation}
 \Phi_k^{-1}(z)
 :=\{\boldsymbol x\in(\F_p^n)^k:\Phi_k(\boldsymbol x)=z\},
 \quad
 N_z:=|\Phi_k^{-1}(z)|.
 \label{eq:moment-fiber-and-cardinality}
\end{equation}
We say that $z$ \emph{occurs at copy number $k$} when any, and hence all, of
the following equivalent conditions hold:
\begin{equation}
 z\in\operatorname{im}\Phi_k,
 \quad
 \Phi_k^{-1}(z)\ne\varnothing,
 \quad
 N_z>0.
 \label{eq:occurring-moment-equivalences}
\end{equation}
Fix an occurring moment value $z$.  Define
\begin{equation}
 |v_z\rangle
 :=\sum_{\boldsymbol x\in\Phi_k^{-1}(z)}|\boldsymbol x\rangle,
 \quad
 |e_z\rangle:=N_z^{-1/2}|v_z\rangle.
 \label{eq:fiber-sector}
\end{equation}
We call the one-dimensional space
$\mathcal K_z:=\operatorname{span}\{|e_z\rangle\}$ the \emph{moment sector
at $z$}.  This line is generated by the normalized equal-weight sum over
the fiber.  The full span of the computational-basis vectors in that
fiber has dimension $N_z$; the moment sector selects its equal-weight
direction.  The common orbit support is
\begin{equation}
 \mathcal K_{k,n}
 :=\operatorname{span}\{|e_z\rangle:z\in\operatorname{im}\Phi_k\}.
 \label{eq:orbit-support}
\end{equation}
Equivalently, $\mathcal K_{k,n}=\bigoplus_{z\in\operatorname{im}\Phi_k}
\mathcal K_z$.  The next lemma shows why only the equal-weight direction
from each fiber occurs in the orbit expansion.
\end{definition}
Here $N_z$ is the multiplicity of one moment value in the ordered-tuple
expansion, while $|e_z\rangle$ is the normalized direction carrying its
character.  The vectors $|e_z\rangle$ are orthonormal because distinct fibers
of $\Phi_k$ are disjoint.

Use the character pairing
\begin{equation}
 \chi_{(Q,S)}(A,b):=\omega^{\frac12\Tr(AQ)+b^TS}.       \label{eq:pairing}
\end{equation}
The factor $1/2$ compensates for the doubling of off-diagonal entries in the
trace pairing on symmetric matrices.

\begin{lemma}[Moment expansion of the orbit states]
\label{lem:orbit-moment-expansion}
For $g=(A,b)\in G$, the tensor-power state has the moment expansion
\begin{equation}
 |\Psi_g\rangle:=|\psi_g\rangle^{\otimes k}
 =p^{-nk/2}\sum_{z\in\operatorname{im}\Phi_k}
 \sqrt{N_z}\,\chi_z(g)|e_z\rangle.              \label{eq:orbit-sector-expansion}
\end{equation}
\end{lemma}
\begin{proof}
Expand $|\psi_g\rangle^{\otimes k}$ in the computational basis.  The phase
of an ordered tuple $(x_1,\ldots,x_k)\in(\F_p^n)^k$ depends on it only through
$\Phi_k(x_1,\ldots,x_k)$.  Partitioning the basis tuples into the fibers in
Eq.~\eqref{eq:moment-fiber-and-cardinality} gives the stated expansion, with
coefficient $\sqrt{N_z}$ after normalization.
\end{proof}

\subsubsection{Moment coordinates associated with unreported parameters}
The moment expansion alone does not yet reflect the partial nature of the learning problem.  We must next determine which coherences survive when labels differing only in unreported coordinates are averaged together.  The answer is expressed by the restriction of a moment character to the unreported subgroup.  The terminology used below is relative to the requested partial label.  Write
the odd-prime symmetric-matrix parameters as
\begin{equation}
 A=\begin{pmatrix}A_{11}&A_{12}\\A_{12}^T&A_{22}\end{pmatrix},
 \quad b=\binom{b_1}{b_2}.
 \label{eq:finite-size-parameter-split}
\end{equation}
For the standard query $J_m$, the partial label records $(A_{11},A_{12}^T,b_1)$
and does not record $(A_{22},b_2)$.  The subgroup $H$ in Eq.~\eqref{eq:H} is
therefore the kernel of the partial-label homomorphism: its translations vary
only $(A_{22},b_2)$ and leave the requested label unchanged.

Under the character pairing in Eq.~\eqref{eq:pairing}, the parameter block
$(A_{22},b_2)$ pairs with the moment block $(Q_{22},S_2)$.  Consequently,
averaging an orbit state over $H$ preserves a matrix element between moment
values $z$ and $z'$ exactly when the characters $\chi_z$ and $\chi_{z'}$ have
the same restriction to $H$.  The next definition and lemma express this
condition as equality of the lower moment coordinates.

\begin{definition}[Character annihilator and lower moment coordinates]
Let
\begin{equation}
 H^\perp:=\{z\in G:\chi_z(h')=1\text{ for every }h'\in H\}.
 \label{eq:odd-annihilator-definition}
\end{equation}
This subgroup is the \emph{character annihilator of $H$}: it consists of
moment labels whose characters are trivial on $H$.  A coset of this
subgroup will be called an \emph{annihilator coset}.
Write each moment value $(Q,S)\in G$ in block coordinates as
\begin{equation}
 Q=\begin{pmatrix}Q_{11}&Q_{12}\\Q_{12}^T&Q_{22}\end{pmatrix},
 \quad S=\binom{S_1}{S_2}.
 \label{eq:moment-coordinate-split}
\end{equation}
The lower blocks define the coordinate map
\begin{equation}
 r:G\longrightarrow\Sym_{n-m}(\F_p)\times\F_p^{n-m},
 \quad r(Q,S):=(Q_{22},S_2).
 \label{eq:unreported-moment-early}
\end{equation}
We write $\zeta=r(z)$ for the fixed lower moment coordinates.
\end{definition}

For a tuple $\boldsymbol x=(x_1,\ldots,x_k)$, let
$X=(x_1\ \cdots\ x_k)\in\F_p^{n\times k}$ be its column matrix and let
$X_2\in\F_p^{(n-m)\times k}$ consist of its last $n-m$ rows.
Writing $\1_k\in\F_p^k$ for the all-ones column vector, the moment map gives
\[
 r(\Phi_k(\boldsymbol x))=(X_2X_2^T,X_2\1_k).
\]
Thus the lower event used in the converse is
$r(\Phi_k(\boldsymbol x))=\zeta$ for a prescribed lower value $\zeta$.

\begin{lemma}[Coordinate description of the $H^\perp$-cosets]
\label{lem:unreported-moment-cosets}
The map $r$ in Eq.~\eqref{eq:unreported-moment-early} satisfies
\begin{equation}
 H^\perp=\ker r
 =\{(Q,S):Q_{22}=0,\ S_2=0\},
 \quad |H^\perp|=|G/H|=\lvert Y\rvert.
 \label{eq:H-perp}
\end{equation}
For $z,z'\in G$, the annihilator condition is equivalent to equality
of the lower moment coordinates:
\begin{equation}
 z-z'\in H^\perp
 \quad\Longleftrightarrow\quad
 r(z)=r(z').
 \label{eq:coset-coordinate-equivalence}
\end{equation}
Hence, for every $\zeta=(\zeta_Q,\zeta_S)$ in the image of $r$, the associated
coset is the explicitly defined fiber
\begin{equation}
 \mathcal C_{\zeta}:=r^{-1}(\zeta)
 =\{(Q,S)\in G:Q_{22}=\zeta_Q,\ S_2=\zeta_S\}.
 \label{eq:coset-as-coordinate-fiber}
\end{equation}
These fibers are exactly the $H^\perp$-cosets in $G$.
\end{lemma}
\begin{proof}
The character $\chi_{(Q,S)}$ takes the following value on an element
$h'\in H$ with lower parameter blocks $(F,c)$:
\begin{equation}
 \chi_{(Q,S)}(h')
 =\omega^{\frac12\Tr(FQ_{22})+c^TS_2}.
 \label{eq:annihilator-lower-block-pairing}
\end{equation}
This character equals one for every symmetric $F$ and every $c$ if and only if
$Q_{22}=0$ and $S_2=0$.  This proves $H^\perp=\ker r$ and
Eq.~\eqref{eq:H-perp}.  Since $r$ is linear,
$z-z'\in\ker r$ if and only if $r(z)=r(z')$, proving
Eq.~\eqref{eq:coset-coordinate-equivalence}.  Equation
\eqref{eq:coset-as-coordinate-fiber} is then the coordinate description of
the corresponding coset.
\end{proof}

\begin{definition}[Coset vectors, total fiber size, and occurring moment values]
For each lower value $\zeta$, let
\begin{equation}
 \mathcal K_\zeta:=\operatorname{span}\{|e_z\rangle:
 z\in\mathcal C_\zeta\cap\operatorname{im}\Phi_k\}.
 \label{eq:odd-coset-block-space}
\end{equation}
The spaces $\mathcal K_\zeta$ are mutually orthogonal.  We call the
restriction of an operator to one such space its \emph{annihilator-coset
block} when the operator is block diagonal for this decomposition.
For $\zeta$ in the image of $r$, define
\begin{align}
 |w_{\zeta}\rangle
 &:=\sum_{z\in\mathcal C_{\zeta}\cap\operatorname{im}\Phi_k}
       \sqrt{N_z}|e_z\rangle,&
 |u_{\zeta}\rangle
 &:=\sum_{z\in\mathcal C_{\zeta}\cap\operatorname{im}\Phi_k}|e_z\rangle,
 \label{eq:w-u}\\
 T_{\zeta}
 &:=\sum_{z\in\mathcal C_{\zeta}}N_z,&
 d_{\zeta}
 &:=|\mathcal C_{\zeta}\cap\operatorname{im}\Phi_k|.
 \label{eq:early-Th-dh}
\end{align}
The integer $T_{\zeta}$ counts the ordered computational-basis tuples
$\boldsymbol x\in(\F_p^n)^k$ whose moment value lies in
$\mathcal C_{\zeta}$:
\begin{equation}
 T_{\zeta}
 =|\{\boldsymbol x\in(\F_p^n)^k:
       \Phi_k(\boldsymbol x)\in\mathcal C_{\zeta}\}|.
 \label{eq:T-zeta-counted-set}
\end{equation}
The integer $d_{\zeta}$ counts the full moment values in that coset that occur
at copy number $k$.  Thus $T_{\zeta}$ is the total tuple count in one surviving
coherence block, whereas $d_{\zeta}$ is the number of distinct occurring
moment values above the same lower coordinate.

Because $r(z)=\zeta$ fixes $(Q_{22},S_2)$, the remaining-coordinate map
\begin{equation}
 (Q,S)\longmapsto(Q_{11},Q_{12},S_1)
 \label{eq:coset-to-remaining-coordinate-map}
\end{equation}
is injective on $\mathcal C_{\zeta}$.  It identifies
$\mathcal C_{\zeta}\cap\operatorname{im}\Phi_k$ with the set of remaining
coordinate triples arising from an ordered tuple in $(\F_p^n)^k$.  In addition,
the assignment of an occurring moment to its one-dimensional sector is
\begin{equation}
 z\longmapsto\operatorname{span}\{|e_z\rangle\}.
 \label{eq:occurring-moment-to-sector-map}
\end{equation}
This assignment is a bijection from
$\mathcal C_{\zeta}\cap\operatorname{im}\Phi_k$ onto the nonzero moment
sectors in this coset.  Thus $d_{\zeta}$ has the three equivalent, explicitly
related interpretations: occurring full moment values, occurring remaining
coordinate triples, and nonzero moment sectors.
\end{definition}

\subsubsection{Conditioned and average states in moment coordinates}
Averaging over $H$ keeps a matrix element between $e_z$ and $e_{z'}$
exactly when $r(z)=r(z')$.  We now group these surviving terms by their
common lower value $\zeta$ and write the conditioned state on the
orthogonal spaces $\mathcal K_\zeta$.  This gives the ensemble needed for
minimum-error discrimination.

\begin{proposition}[Block form of the exact-identification ensemble]
\label{prop:conditioned-average-block-form}
The state conditioned on the zero partial label has the block form
\begin{equation}
 \sigma_0=p^{-nk}\sum_{\zeta}|w_{\zeta}\rangle\!\langle w_{\zeta}|.          \label{eq:sigma-zero-blocks}
\end{equation}
The blocks indexed by distinct fibers of $r$ are mutually orthogonal.  The
uniform ensemble average is
\begin{equation}
 \overline\sigma:=\frac1{\lvert Y\rvert}\sum_{y\in Y}\sigma_y
 =p^{-nk}\sum_{z:N_z>0}N_z|e_z\rangle\!\langle e_z|,      \label{eq:average-state},
\end{equation}
and hence, on $\mathcal K_{k,n}$,
\begin{equation}
 \overline\sigma^{-1/2}
 =p^{nk/2}\sum_{z:N_z>0}N_z^{-1/2}|e_z\rangle\!\langle e_z|.
                                                        \label{eq:average-inverse}
\end{equation}
\end{proposition}
\begin{proof}
Insert Eq.~\eqref{eq:orbit-sector-expansion} into the definition of $\sigma_0$.
For its matrix element $(z,z')$, the subgroup average contributes
\begin{equation}
 \frac1{|H|}\sum_{h\in H}\chi_z(h)\overline{\chi_{z'}(h)}
 \stackrel{(a)}{=}\mathbf1\{z-z'\in H^\perp\}
 \stackrel{(b)}{=}\mathbf1\{r(z)=r(z')\}.
 \label{eq:odd-subgroup-character-orthogonality}
\end{equation}
Here (a) is character orthogonality on $H$, and (b) is
Eq.~\eqref{eq:coset-coordinate-equivalence}.  Grouping the surviving
terms by $\zeta$ proves Eq.~\eqref{eq:sigma-zero-blocks}.  Averaging over
the whole quotient, equivalently over $G$, leaves only $z=z'$ by character
orthogonality and gives Eq.~\eqref{eq:average-state}.  Functional calculus on
its support gives Eq.~\eqref{eq:average-inverse}.
\end{proof}

On every nonzero space $\mathcal K_\zeta$, the corresponding block of
$\sigma_0$ is $p^{-nk}|w_\zeta\rangle\langle w_\zeta|$ and has rank one.
The space itself has dimension $d_\zeta$, which can exceed one.  This
rank-one state block, rather than a one-dimensional block space, is used
in the measurement optimization.  
The tuple count $T_{\zeta}$ gives the lower-event probability
$\pi_{\zeta}=p^{-nk}T_{\zeta}$ after normalization by the total number
$p^{nk}$ of computational-basis tuples.  The quantity $d_{\zeta}$ instead
counts the occurring full moment values above that lower event.
The converse retains both the probability and this moment count.

\subsection{Moment expansion and subgroup averaging in characteristic two}
\label{sec:integrated-qubit-moments}
In characteristic two, the phase depends on mixed mod-$4$/mod-$2$
moment coordinates.  We define their fibers and equal-weight directions,
and then determine which pairs of directions survive the unreported-label
average, just as in the odd-prime calculation.

\subsubsection{Moment group and character pairing in characteristic two}
For $p=2$, diagonal phase information is modulo four, whereas
products of distinct binary coordinates are modulo two.  We record these
two types of data in a single finite abelian moment group:
\begin{equation}
 \mathcal M_{n,2}:=\mathbb Z_4^n\times
 \F_2^{\binom n2}.
 \label{eq:binary-enhanced-moment-group}
\end{equation}
We write a parameter as $g=(\alpha,E)\in G_{n,2}$, with
$\alpha=(\alpha_a)_{a=1}^n\in\mathbb Z_4^n$ and
$E=(E_{ab})_{a<b}\in\F_2^{\binom n2}$, and a moment value as
$z=(r,T)\in\mathcal M_{n,2}$, with the same respective coordinate groups.
Define the character pairing by
\begin{equation}
 \chi_z^{(2)}(g)
 :=i^{\sum_{a=1}^n\alpha_a r_a}
   (-1)^{\sum_{1\le a<b\le n}E_{ab}T_{ab}}.
 \label{eq:binary-enhanced-character-pairing}
\end{equation}
The first exponent is evaluated modulo $4$ and the second modulo $2$.
The pairing is nondegenerate: summing $\chi_z^{(2)}(g)$ over all
$g\in G_{n,2}$ gives $|G_{n,2}|$ when $z=0$ and zero otherwise.
Thus $\mathcal M_{n,2}$ is identified explicitly with the character group
$\widehat G_{n,2}$.

For an ordered computational-basis tuple
$\boldsymbol x=(x_1,\ldots,x_k)\in(\F_2^n)^k$, define
\begin{equation}
 \Phi_k^{(2)}(\boldsymbol x):=(r(\boldsymbol x),T(\boldsymbol x))
 \in\mathcal M_{n,2}.
 \label{eq:integrated-qubit-moment-map}
\end{equation}
The components of this enhanced moment are the following linear and off-diagonal quadratic sums:
\begin{align}
 r_a(\boldsymbol x)&:=\sum_{j=1}^k(x_j)_a\pmod4,
 \label{eq:integrated-qubit-linear-moment}\\
 T_{ab}(\boldsymbol x)&:=\sum_{j=1}^k(x_j)_a(x_j)_b\pmod2,
 \quad a<b.
 \label{eq:integrated-qubit-edge-moment}
\end{align}

The following lemma explains why the enhanced moment $\Phi_k^{(2)}$ is sufficient for the tensor-power phase: every product of the single-copy phase factors can be computed from this moment alone.

\begin{lemma}[Tuple phase factors through the enhanced moment]
\label{lem:binary-tuple-phase-factorization}
For every $g=(\alpha,E)\in G_{n,2}$ and ordered tuple
$\boldsymbol x\in(\F_2^n)^k$, the product of their phase factors
depends on $\boldsymbol x$ only through its enhanced moment
$\Phi_k^{(2)}(\boldsymbol x)$:
\begin{equation}
 \prod_{j=1}^k
 i^{\sum_a\alpha_a(x_j)_a+2\sum_{a<b}E_{ab}(x_j)_a(x_j)_b}
 =\chi_{\Phi_k^{(2)}(\boldsymbol x)}^{(2)}(g).
 \label{eq:binary-tuple-phase-factorization}
\end{equation}
\end{lemma}
\begin{proof}
Multiplying the $k$ phase factors adds their exponents modulo $4$ and gives
\begin{align}
 i^{\sum_a\alpha_a\sum_j(x_j)_a}
  i^{2\sum_{a<b}E_{ab}\sum_j(x_j)_a(x_j)_b}
 =i^{\sum_a\alpha_a r_a(\boldsymbol x)}
  (-1)^{\sum_{a<b}E_{ab}T_{ab}(\boldsymbol x)},
 \label{eq:binary-tuple-phase-calculation}
\end{align}
which is the right-hand side of
Eq.~\eqref{eq:binary-tuple-phase-factorization} by
Eq.~\eqref{eq:binary-enhanced-character-pairing}.
\end{proof}

\subsubsection{Fibers of the characteristic-two moment map and orbit expansion}
Equation~\eqref{eq:binary-tuple-phase-factorization}, proved using
Eq.~\eqref{eq:binary-tuple-phase-calculation}, expresses the phase of a
tuple through its enhanced moment value.  We now group tuples that have
the same moment.  For $z\in\mathcal M_{n,2}$, define its fiber and
multiplicity by
\begin{equation}
 \mathcal F_z^{(2)}
 :=\{\boldsymbol x\in(\F_2^n)^k:
          \Phi_k^{(2)}(\boldsymbol x)=z\},
 \quad
 N_z^{(2)}:=|\mathcal F_z^{(2)}|.
 \label{eq:binary-enhanced-moment-fiber}
\end{equation}
Call $z$ occurring when $N_z^{(2)}>0$.  For every occurring $z$, put
\begin{equation}
 \ket{e_z^{(2)}}
 :=(N_z^{(2)})^{-1/2}
   \sum_{\boldsymbol x\in\mathcal F_z^{(2)}}\ket{\boldsymbol x},
 \quad
 \mathcal K_{k,n}^{(2)}
 :=\operatorname{span}\{\ket{e_z^{(2)}}:N_z^{(2)}>0\}.
 \label{eq:integrated-qubit-sector-vector}
\end{equation}
For a fixed occurring $z$, the line
$\mathcal K_z^{(2)}:=\operatorname{span}\{\ket{e_z^{(2)}}\}$ is its binary
moment sector.  As in odd characteristic, it is the equal-weight direction
in a fiber whose full computational-basis span has dimension $N_z^{(2)}$.
Distinct fibers are disjoint, so these normalized vectors are orthonormal
and $\mathcal K_{k,n}^{(2)}=\bigoplus_{z:N_z^{(2)}>0}\mathcal K_z^{(2)}$.

\begin{lemma}[Moment expansion of the binary orbit states]
\label{lem:binary-orbit-moment-expansion}
For $g=(\alpha,E)\in G_{n,2}$, the tensor-power state has the binary moment expansion
\begin{equation}
 \ket{\Psi_g^{(2)}}
 :=\ket{\psi_g^{(2)}}^{\otimes k}
 =2^{-nk/2}\sum_{z:N_z^{(2)}>0}
   \sqrt{N_z^{(2)}}\,\chi_z^{(2)}(g)\ket{e_z^{(2)}}.
 \label{eq:integrated-qubit-orbit-expansion}
\end{equation}
\end{lemma}
\begin{proof}
Expand the tensor power in the ordered computational basis.  By
Lemma~\ref{lem:binary-tuple-phase-factorization}, all tuples in the fiber
$\mathcal F_z^{(2)}$ have the common phase $\chi_z^{(2)}(g)$.
Partitioning the tuples by their moment value and replacing the uniform sum
on each nonempty fiber by
$\sqrt{N_z^{(2)}}\ket{e_z^{(2)}}$ gives
Eq.~\eqref{eq:integrated-qubit-orbit-expansion}.
\end{proof}

\subsubsection{The kernel of the partial-label map and explicit annihilator cosets}
We next impose the distinction between reported and unreported coordinates.  Averaging over the latter preserves a matrix element precisely when the two enhanced moments define the same character on the unreported subgroup.  The restriction map below records exactly the coordinates needed to test this condition.  Split the enhanced moment coordinates according to the queried indices
$1,\ldots,m$ and the lower indices $m+1,\ldots,n$.  Define
\begin{equation}
 \mathcal L_{n,m,2}
 :=\mathbb Z_4^{n-m}\times
   \F_2^{\binom{n-m}{2}},
 \label{eq:binary-lower-moment-space}
\end{equation}
and the restriction map
\begin{equation}
 r_2:\mathcal M_{n,2}\longrightarrow\mathcal L_{n,m,2},
 \quad
 r_2(r,T)
 :=\bigl((r_a)_{a>m},(T_{ab})_{m<a<b\le n}\bigr).
 \label{eq:binary-lower-moment-restriction}
\end{equation}
Under the character pairing in Eq.~\eqref{eq:binary-enhanced-character-pairing},
define the character annihilator
$H_{n,m,2}^\perp:=\{z\in\mathcal M_{n,2}:
\chi_z^{(2)}(h)=1\text{ for every }h\in H_{n,m,2}\}$.
For $\zeta\in\mathcal L_{n,m,2}$ set
\begin{equation}
 \mathcal C_\zeta^{(2)}:=r_2^{-1}(\zeta).
 \label{eq:binary-annihilator-coset-fiber}
\end{equation}

\begin{lemma}[Coordinate description of the binary annihilator cosets]
\label{lem:binary-annihilator-cosets}
Under the pairing in Eq.~\eqref{eq:binary-enhanced-character-pairing},
\begin{equation}
 H_{n,m,2}^\perp=\ker r_2.
 \label{eq:binary-annihilator-kernel}
\end{equation}
Consequently, for $z,z'\in\mathcal M_{n,2}$,
\begin{equation}
 z-z'\in H_{n,m,2}^\perp
 \quad\Longleftrightarrow\quad
 r_2(z)=r_2(z'),
 \label{eq:binary-annihilator-coset-equivalence}
\end{equation}
and the fibers $\mathcal C_\zeta^{(2)}$ are exactly the
$H_{n,m,2}^\perp$-cosets in $\mathcal M_{n,2}$.
\end{lemma}
\begin{proof}
An element $h\in H_{n,m,2}$ has arbitrary lower coordinates
$\alpha_a$ for $a>m$ and arbitrary lower--lower edge coordinates $E_{ab}$
for $m<a<b\le n$, with all other coordinates zero.
The character pairing restricted to this subgroup is
\begin{equation}
 \chi_{(r,T)}^{(2)}(h)
 =i^{\sum_{a>m}\alpha_a r_a}
  (-1)^{\sum_{m<a<b\le n}E_{ab}T_{ab}}.
 \label{eq:binary-unreported-subgroup-pairing}
\end{equation}
This equals one for every $h\in H_{n,m,2}$ if and only if every lower
$r_a$ is zero modulo $4$ and every lower--lower $T_{ab}$ is zero modulo $2$.
That condition is precisely $r_2(r,T)=0$, proving
Eq.~\eqref{eq:binary-annihilator-kernel}.  The equivalence in
Eq.~\eqref{eq:binary-annihilator-coset-equivalence} follows because $r_2$ is a
group homomorphism, and its fibers are therefore the cosets of its kernel.
\end{proof}

The corresponding block space is
\begin{equation}
 \mathcal K_\zeta^{(2)}:=\operatorname{span}\{\ket{e_z^{(2)}}:
 z\in\mathcal C_\zeta^{(2)},\ N_z^{(2)}>0\}.
 \label{eq:binary-coset-block-space}
\end{equation}
For each $\zeta\in\mathcal L_{n,m,2}$ define
\begin{align}
 \ket{w_\zeta^{(2)}}
 &:=\sum_{\substack{z\in\mathcal C_\zeta^{(2)}\\N_z^{(2)}>0}}
       \sqrt{N_z^{(2)}}\ket{e_z^{(2)}},
 \label{eq:integrated-qubit-lower-block-vector}\\
 T_\zeta^{(2)}
 &:=\sum_{z\in\mathcal C_\zeta^{(2)}}N_z^{(2)},
 &
 d_\zeta^{(2)}
 &:=|\{z\in\mathcal C_\zeta^{(2)}:N_z^{(2)}>0\}|.
 \label{eq:binary-coset-mass-and-support}
\end{align}
Thus $T_\zeta^{(2)}$ counts the ordered tuples whose lower enhanced moment is
$\zeta$, whereas $d_\zeta^{(2)}$ counts the occurring full enhanced moments
above that fixed lower value.

\subsubsection{Conditioned and average states in binary moment coordinates}
The following proposition describes the conditioned and average binary
states in block form.  The enhanced-moment orbit expansion supplies the
character sectors, and the annihilator cosets group the labels into blocks.
\begin{proposition}[Binary block form of the exact-identification ensemble]
\label{prop:binary-conditioned-average-block-form}
The binary state conditioned on the zero quotient label has the block form
\begin{equation}
 \sigma_0^{(2)}
 =2^{-nk}\sum_\zeta
   \ket{w_\zeta^{(2)}}\!\bra{w_\zeta^{(2)}}.
 \label{eq:integrated-qubit-conditioned-blocks}
\end{equation}
The summands for distinct $\zeta$ have orthogonal supports.  The uniform orbit
average is
\begin{equation}
 \overline\sigma^{(2)}
 :=\frac1{|G_{n,2}|}\sum_{g\in G_{n,2}}
   \ket{\Psi_g^{(2)}}\!\bra{\Psi_g^{(2)}}
 =2^{-nk}\sum_{z:N_z^{(2)}>0}
   N_z^{(2)}\ket{e_z^{(2)}}\!\bra{e_z^{(2)}}.
 \label{eq:integrated-qubit-average-state}
\end{equation}
\end{proposition}
\begin{proof}
Insert Eq.~\eqref{eq:integrated-qubit-orbit-expansion} into the average over
$H_{n,m,2}$.  The coefficient of
$\ket{e_z^{(2)}}\!\bra{e_{z'}^{(2)}}$ is multiplied by
\begin{equation}
 \frac1{|H_{n,m,2}|}\sum_{h\in H_{n,m,2}}
 \chi_{z-z'}^{(2)}(h)
 =\mathbf1\{z-z'\in H_{n,m,2}^\perp\}.
 \label{eq:binary-subgroup-character-orthogonality}
\end{equation}
By Lemma~\ref{lem:binary-annihilator-cosets}, the surviving pairs are exactly
those satisfying $r_2(z)=r_2(z')$.  Grouping them by their common value
$\zeta$ gives Eq.~\eqref{eq:integrated-qubit-conditioned-blocks}.  Distinct
fibers contain disjoint sector vectors and hence have orthogonal supports.
Similarly, the whole-group average multiplies the same matrix element by
\begin{equation}
 \frac1{|G_{n,2}|}\sum_{g\in G_{n,2}}
 \chi_{z-z'}^{(2)}(g)=\mathbf1\{z=z'\},
 \label{eq:binary-whole-group-character-orthogonality}
\end{equation}
which proves Eq.~\eqref{eq:integrated-qubit-average-state}.
\end{proof}

Each nonzero block of $\sigma_0^{(2)}$ on $\mathcal K_\zeta^{(2)}$ is
rank one, whereas the block space has dimension $d_\zeta^{(2)}$.

\subsection{PGM optimality for finite abelian quotient ensembles}
\label{sec:seed-sdp}
We now optimize the measurement for the conditioned ensembles obtained in
the two moment expansions.  We first formulate their common structure as
a finite abelian orbit, specify its orthogonal block spaces, and define
the PGM.  We then use covariance to optimize one POVM effect and generate
the remaining effects by translation.  The Hilbert space is finite
dimensional and the outcome set is finite, so the POVM optimum is attained.

\subsubsection{Finite abelian quotient ensembles}
Let $G$ be a finite abelian group, let $H\le G$ be a subgroup, and put
$Y:=G/H$.  Let
$\widehat G$ denote the character group of $G$.  Fix a finite set
$\mathcal Z_+\subseteq\widehat G$, an orthonormal family
$\{\ket{e_z}:z\in\mathcal Z_+\}$, and positive amplitudes
$c_z>0$ satisfying
\begin{equation}
 \sum_{z\in\mathcal Z_+}c_z^2=1.
 \label{eq:abstract-amplitude-normalization}
\end{equation}
Since $\mathcal Z_+\subseteq\widehat G$, we write $\chi_z$ for the
character represented by $z\in\mathcal Z_+$.
Write
\begin{equation}
 \mathcal K:=\operatorname{span}\{\ket{e_z}:z\in\mathcal Z_+\},
 \quad
 V_g\ket{e_z}:=\chi_z(g)\ket{e_z}.
 \label{eq:abstract-sector-representation}
\end{equation}
Define the orbit vectors
\begin{equation}
 \ket{\Psi_g}:=V_g\ket{\Psi_0},
 \quad
 \ket{\Psi_0}:=\sum_{z\in\mathcal Z_+}c_z\ket{e_z}.
 \label{eq:abstract-orbit-vectors}
\end{equation}
Choose any representative-selection map $s_Y:Y\to G$, that is, a section of
the quotient map satisfying $s_Y(y)+H=y$, and define the quotient-conditioned
states
\begin{equation}
 \sigma_y:=\frac1{\lvert H\rvert}\sum_{h\in H}
 \ket{\Psi_{s_Y(y)+h}}\!\bra{\Psi_{s_Y(y)+h}},
 \quad y\in Y.
 \label{eq:abstract-conditioned-ensemble}
\end{equation}
The state $\sigma_y$ is independent of the representative $s_Y(y)$ because
changing that representative only permutes the average over $H$.

Let
\begin{equation}
 H^\perp:=\{z\in\widehat G:\chi_z(h)=1\text{ for every }h\in H\}.
 \label{eq:abstract-annihilator}
\end{equation}
This subgroup is the annihilator of $H$: its characters are trivial on $H$.
Consequently, two characters lie in the same $H^\perp$-coset exactly when they
have the same restriction to $H$.  Partition $\mathcal Z_+$ by these cosets in
$\widehat G$.  Denote
the nonempty intersections with $\mathcal Z_+$ by $\mathcal C$.
Thus $\mathcal C$ is a set of occurring character labels, not necessarily
an entire coset.  Let
$\mathcal K_{\mathcal C}:=\operatorname{span}\{\ket{e_z}:z\in\mathcal C\}$
be its block space, and define
\begin{equation}
 \ket{w_{\mathcal C}}:=\sum_{z\in\mathcal C}c_z\ket{e_z},
 \quad
 \ket{u_{\mathcal C}}:=\sum_{z\in\mathcal C}\ket{e_z}.
 \label{eq:abstract-coset-vectors}
\end{equation}
Character orthogonality gives
\begin{equation}
 \sigma_0=\sum_{\mathcal C}
 \ket{w_{\mathcal C}}\!\bra{w_{\mathcal C}},
 \quad
 \overline\sigma:=\frac1{\lvert Y\rvert}\sum_{y\in Y}\sigma_y
 =\sum_{z\in\mathcal Z_+}c_z^2
 \ket{e_z}\!\bra{e_z}.
 \label{eq:abstract-conditioned-average-blocks}
\end{equation}
Each summand of $\sigma_0$ is a rank-one operator on the space
$\mathcal K_{\mathcal C}$, whose dimension is $|\mathcal C|$.
These spaces are mutually orthogonal and together span $\mathcal K$.

\begin{definition}[Pretty good measurement and covariant seed]
\label{def:pgm-and-seed}
For the uniform ensemble $\{|Y|^{-1},\sigma_y:y\in Y\}$, the
\emph{pretty good measurement} (PGM) has effects
\begin{equation}
 M_y^{\PGM}:=\frac1{|Y|}\overline\sigma^{-1/2}
             \sigma_y\overline\sigma^{-1/2}
 \quad (y\in Y)
 \label{eq:quotient-pgm-definition}
\end{equation}
on the common support $\mathcal K$, where the inverse is taken on that
support.  The effects sum to the identity on $\mathcal K$ and can be
completed to a POVM on the full Hilbert space; the completion does not
affect success on this ensemble.
For a covariant POVM, the effect assigned to the zero quotient label is
called its \emph{seed}.  Translating this effect by $V_{s_Y(y)}$ generates
the effect assigned to $y$.
\end{definition}

For each $z\in\mathcal Z_+$, the one-dimensional space
$\operatorname{span}\{|e_z\rangle\}$ is invariant under $V_g$, and
$V_g$ acts on it by the character $\chi_z(g)$.
Thus the structural assumptions are a finite abelian orbit, these
orthogonal one-dimensional invariant spaces, and nonnegative reference
amplitudes.  The PGM optimality below is not a generic consequence of
covariance.  It uses the one-dimensional invariant spaces
$\operatorname{span}\{|e_z\rangle\}$, the rank-one form of the
subgroup-averaged coherence blocks, and the fact that
quotient covariance fixes only the diagonal entries of each surviving seed
block.  These structural properties are sufficient to prove PGM optimality, as stated
in the next theorem.
\begin{theorem}[PGM optimality for quotient-conditioned character-sector ensembles]
\label{thm:detailed-seed-solution}
For the uniform ensemble $\{ \lvert Y\rvert^{-1},\sigma_y:y\in Y\}$ defined above, the
PGM is minimum-error optimal.  On the common support $\mathcal K$,
an optimal seed is
\begin{equation}
 \Xi^*=\frac1{\lvert Y\rvert}\sum_{\mathcal C}
 \ket{u_{\mathcal C}}\!\bra{u_{\mathcal C}},
 \label{eq:global-seed}
\end{equation}
and the optimal success probability is
\begin{equation}
 P_{\rm opt}
 =\frac1{\lvert Y\rvert}\sum_{\mathcal C}
 \left(\sum_{z\in\mathcal C}c_z\right)^2.
 \label{eq:abstract-optimal-coset-formula}
\end{equation}
Completing the POVM on $\mathcal K^\perp$ with positive effects does not
change the success probability.  The covariant POVM and its success probability are independent of
the choice of the quotient section $s_Y$.
\end{theorem}

The inverse in the PGM is the Moore--Penrose inverse on $\mathcal K$: it
inverts the nonzero eigenvalues on $\mathcal K$ and acts as zero on the
orthogonal kernel.  By the definition of $\mathcal Z_+$, the common support
$\mathcal K=\operatorname{span}\{|e_z\rangle:z\in\mathcal Z_+\}$
contains no zero-amplitude character sectors.

The proof proceeds by reducing to a covariant seed, identifying its  
annihilator-coset blocks and fixed diagonal, and optimizing each positive  
semidefinite block by a rank-one matrix.  

\begin{lemma}[Covariant seed reduction]
\label{lem:covariant-seed-reduction}
The optimum on $\mathcal K$ is attained by a quotient-covariant POVM generated
by a seed $\Xi$ satisfying
\begin{align}
 \Xi&\ge0,
 &V_h\Xi V_h^\dagger&=\Xi\quad(h\in H),
 \label{eq:seed-invariance}\\
 \sum_{y\in Y}V_{s_Y(y)}\Xi V_{s_Y(y)}^\dagger
 &=\Pi_{\mathcal K}.
 \label{eq:seed-normalization}
\end{align}
Here $\Pi_{\mathcal K}$ denotes the orthogonal projector onto the common
support $\mathcal K$.  Its average success probability is
\begin{equation}
 P(\Xi)=\Tr(\Xi\sigma_0).
 \label{eq:seed-objective}
\end{equation}
The generated POVM does not depend on the chosen representatives once the
outcomes are identified with quotient elements.
\end{lemma}
\begin{proof}
Let $\{M_y:y\in Y\}$ be any POVM on $\mathcal K$, so that
$M_y\ge0$ and $\sum_yM_y=\Pi_{\mathcal K}$.  Write
$[g]:=g+H\in Y$.  Define the group-averaged effects by
\begin{equation}
 \widetilde M_y:=\frac1{|G|}\sum_{g\in G}
 V_g M_{y-[g]}V_g^\dagger,\qquad y\in Y.
 \label{eq:covariant-seed-explicit-twirl}
\end{equation}
Each effect is positive.  Since the action of $V_g$ preserves
$\mathcal K$, summing over the labels gives
\begin{align}
 \sum_{y\in Y}\widetilde M_y
 &=\frac1{|G|}\sum_{g\in G}
 V_g\Bigl(\sum_{y\in Y}M_{y-[g]}\Bigr)V_g^\dagger\notag\\
 &=\Pi_{\mathcal K}.
 \label{eq:covariant-seed-twirl-normalization}
\end{align}
By Eq.~\eqref{eq:abstract-conditioned-ensemble},
$\sigma_{y+[g]}=V_g\sigma_yV_g^\dagger$: changing a representative
only permutes the average over $H$.  Therefore the uniform-average
success probability of the twirled POVM is
\begin{align}
 &\frac1{|Y|}\sum_{y\in Y}\Tr(\widetilde M_y\sigma_y)
 =\frac1{|G||Y|}\sum_{g\in G}\sum_{y\in Y}
 \Tr(M_{y-[g]}V_g^\dagger\sigma_yV_g)\notag\\
 &=\frac1{|G||Y|}\sum_{g\in G}\sum_{z\in Y}
 \Tr(M_z\sigma_z)
 =\frac1{|Y|}\sum_{z\in Y}\Tr(M_z\sigma_z).
 \label{eq:covariant-seed-twirl-success}
\end{align}
The change of variables is $z=y-[g]$.  Similarly,
$\widetilde M_{y+[g]}=V_g\widetilde M_yV_g^\dagger$.
Hence an optimum is attained by a quotient-covariant POVM;
only conjugation by the genuine $G$-representation $V$ is used.

Let $\Xi$ be the element assigned to the zero coset.  Elements of $H$ leave
that outcome unchanged, so covariance gives
$V_h\Xi V_h^\dagger=\Xi$, which is
Eq.~\eqref{eq:seed-invariance}.  Translating this seed over quotient labels
produces all POVM elements, and their normalization is exactly
Eq.~\eqref{eq:seed-normalization}.  Conversely, any positive seed satisfying
these two equations generates a supported quotient-covariant POVM.

For such a POVM, covariance makes every label contribute the same value after
translation back to the zero label.  The uniform average objective therefore
reduces to
$P(\Xi)=\Tr(\Xi\sigma_0)$, proving
Eq.~\eqref{eq:seed-objective}.

Finally, suppose that another quotient section satisfies
$s_Y'(y)=s_Y(y)+h_y$ with $h_y\in H$.  Seed invariance then implies
\[
 V_{s_Y'(y)}\Xi V_{s_Y'(y)}^\dagger
 =V_{s_Y(y)}\Xi V_{s_Y(y)}^\dagger,
\]
so changing quotient representatives does not change either the generated
POVM or its objective value.
\end{proof}

In Lemma~\ref{lem:covariant-seed-reduction}, the elements of the
covariant POVM are $M_y=V_{s_Y(y)}\Xi V_{s_Y(y)}^\dagger$, where
$\Xi$ is $H$-invariant.  This representation of the effects does not
require the selected representatives $s_Y(y)$ to form a group
homomorphism.

\begin{lemma}[Annihilator-coset blocks and fixed diagonal]
\label{lem:coset-fixed-diagonal}
Every feasible seed has the block form
\begin{equation}
 \Xi=\bigoplus_{\mathcal C}\Xi_{\mathcal C},
 \label{eq:Xi-blocks}
\end{equation}
where the blocks are indexed by the nonempty intersections of
$H^\perp$-cosets with $\mathcal Z_+$, and
\begin{equation}
 (\Xi_{\mathcal C})_{z,z}=\frac1{\lvert Y\rvert}
 \quad(z\in\mathcal C).
 \label{eq:fixed-diagonal}
\end{equation}
There are no further linear restrictions on the entries inside one block.
\end{lemma}
\begin{proof}
From Eq.~\eqref{eq:abstract-sector-representation}, $H$-invariance forces
$\Xi_{z,z'}=0$ unless $\chi_z|_H=\chi_{z'}|_H$, equivalently unless $z$ and
$z'$ lie in the same $H^\perp$-coset.  Take the $(z,z')$ matrix element of
Eq.~\eqref{eq:seed-normalization}.  Since $V_g$ is diagonal in the sector
basis, it equals
\begin{equation}
 \Xi_{z,z'}
 \sum_{y\in Y}
 \chi_z(s_Y(y))\overline{\chi_{z'}(s_Y(y))}
 =\delta_{z,z'}.
 \label{eq:seed-normalization-matrix-element}
\end{equation}
Within one surviving $H^\perp$-coset, the ratio
$\chi_z\overline{\chi_{z'}}$ is trivial on $H$ and hence defines a character
of $G/H$.  Quotient-character orthogonality gives
\begin{equation}
 \sum_{y\in Y}
 (\chi_z\overline{\chi_{z'}})(s_Y(y))
 =\lvert Y\rvert\,\mathbf1\{z=z'\}.
 \label{eq:quotient-character-orthogonality-seed}
\end{equation}
Thus Eq.~\eqref{eq:seed-normalization-matrix-element} is equivalent to
$\lvert Y\rvert\Xi_{z,z}=1$ and imposes no condition on off-diagonal entries inside a
surviving block.  Conversely, positive semidefinite blocks with this diagonal
satisfy every matrix element of the normalization equation, so there are no
additional linear restrictions.
\end{proof}

The seed normalization fixes the diagonal of each surviving positive semidefinite block. We next maximize the resulting blockwise quadratic form; the upper bound is attained by a rank-one block.

\begin{lemma}[Rank-one solution of one block of the seed]
\label{lem:blockwise-seed-optimum}
For each nonempty block $\mathcal C$, the optimum of the seed-block objective is
\begin{equation}
 \max_{\substack{\Xi_{\mathcal C}\ge0\\
                 \operatorname{diag}\Xi_{\mathcal C}=\lvert Y\rvert^{-1}}}
 \bra{w_{\mathcal C}}\Xi_{\mathcal C}\ket{w_{\mathcal C}}
 =\frac1{\lvert Y\rvert}\left(\sum_{z\in\mathcal C}c_z\right)^2,
 \label{eq:block-upper}
\end{equation}
and the maximum is attained by
\begin{equation}
 \Xi_{\mathcal C}^*
 :=\frac1{\lvert Y\rvert}\ket{u_{\mathcal C}}\!\bra{u_{\mathcal C}}.
 \label{eq:block-optimum}
\end{equation}
\end{lemma}
\begin{proof}
Every $2\times2$ principal minor of a positive-semidefinite block with the
prescribed diagonal gives
$|(\Xi_{\mathcal C})_{z,z'}|\le \lvert Y\rvert^{-1}$.  Since all $c_z$ are positive,
this entrywise bound yields Eq.~\eqref{eq:block-upper}.  The rank-one matrix
in Eq.~\eqref{eq:block-optimum} is positive semidefinite, has the prescribed
diagonal, and simultaneously saturates all supported entries.
\end{proof}

\begin{proof}[Proof of Theorem~\ref{thm:detailed-seed-solution}]
By Lemma~\ref{lem:covariant-seed-reduction}, it is enough to optimize one
feasible seed.  Lemma~\ref{lem:coset-fixed-diagonal} decomposes that seed into
independent annihilator-coset blocks, while the first identity in
Eq.~\eqref{eq:abstract-conditioned-average-blocks} decomposes the zero-label
state over the same orthogonal blocks.  Consequently the reduced objective is
\begin{equation}
 \Tr(\Xi\sigma_0)
 =\sum_{\mathcal C}
 \bra{w_{\mathcal C}}\Xi_{\mathcal C}\ket{w_{\mathcal C}}.
 \label{eq:abstract-seed-objective-block-decomposition}
\end{equation}
The fixed-diagonal normalization imposes no linear coupling between distinct
blocks.  Lemma~\ref{lem:blockwise-seed-optimum} may therefore be applied
independently in every block.  Summing the blockwise rank-one optimizers gives
$\Xi^*$ in Eq.~\eqref{eq:global-seed}, and summing their objective values gives
Eq.~\eqref{eq:abstract-optimal-coset-formula}.

We finally identify this optimal seed with the PGM.  On the common support,
the supported inverse square root of the average state is
\begin{equation}
 \overline\sigma^{-1/2}
 =\sum_{z\in\mathcal Z_+}c_z^{-1}
 \ket{e_z}\!\bra{e_z}.
 \label{eq:abstract-average-pseudoinverse}
\end{equation}
Applying this inverse square root to a coset vector removes its amplitudes:
$\overline\sigma^{-1/2}\ket{w_{\mathcal C}}
 =\ket{u_{\mathcal C}}$.  Hence the zero-label PGM element is
\begin{align}
 M_0^{\PGM}
 &:=\overline\sigma^{-1/2}
 \left(\frac{\sigma_0}{\lvert Y\rvert}\right)
 \overline\sigma^{-1/2}\notag\\
 &=\frac1{\lvert Y\rvert}\sum_{\mathcal C}
 \ket{u_{\mathcal C}}\!\bra{u_{\mathcal C}}
 =\Xi^*.
 \label{eq:pgm-equals-optimal}
\end{align}
Thus the PGM attains the independently optimized value in every surviving
block.

All ensemble states are supported on $\mathcal K$.  To obtain a POVM on the
entire Hilbert space, choose arbitrary operators $F_y\ge0$ on
$\mathcal K^\perp$ satisfying
$\sum_{y\in Y}F_y=\Pi_{\mathcal K^\perp}$, and replace each supported POVM
element by $M_y^{\PGM}\oplus F_y$.  This completion preserves every success
probability.  Independence from the representative-selection map $s_Y$ was
proved in Lemma~\ref{lem:covariant-seed-reduction}.
\end{proof}

\begin{remark}[Scope of the abstract PGM theorem]
The theorem applies when subgroup averaging produces rank-one coherence blocks
built from orthogonal one-dimensional character sectors and quotient covariance
fixes only their diagonals.  It does not assert PGM optimality for arbitrary
covariant mixed-state ensembles; for general background on the PGM and the
square-root measurement, see
\cite{HausladenWootters1994PGM,EldarForney2001SRM,ItenRenesSutter2017PGM}.
\end{remark}

\subsubsection{Application to the two moment expansions}
It remains to verify that the odd-prime and characteristic-two moment
expansions supply precisely the amplitudes and orthogonal block spaces
required by the abstract theorem.  Once this identification is made, no
further optimization is needed in either characteristic.

For odd $p$, Eqs.~\eqref{eq:orbit-sector-expansion},
\eqref{eq:sigma-zero-blocks}, and \eqref{eq:average-state} verify the
hypotheses with
\begin{equation}
 c_z=p^{-nk/2}\sqrt{N_z},
 \quad
 \mathcal Z_+=\{z:N_z>0\}.
 \label{eq:odd-abstract-amplitudes}
\end{equation}
For $p=2$, Eqs.~\eqref{eq:integrated-qubit-orbit-expansion},
\eqref{eq:integrated-qubit-conditioned-blocks}, and
\eqref{eq:integrated-qubit-average-state} verify the conditions of
Theorem~\ref{thm:detailed-seed-solution} with
\begin{equation}
 c_z=2^{-nk/2}\sqrt{N_z^{(2)}},
 \quad
 \mathcal Z_+=\{z:N_z^{(2)}>0\}.
 \label{eq:binary-abstract-amplitudes}
\end{equation}
In both cases, the kernel of the partial-label map supplies the same quotient-conditioning
operation, while its annihilator cosets supply the blocks $\mathcal C$.

For odd $p$, the same identification uses
$G=\Sym_n(\F_p)\times\F_p^n$, the subgroup $H$ in Eq.~\eqref{eq:H},
$Y=G/H$, and $\ket{e_z}$ from Eq.~\eqref{eq:fiber-sector}.
Disjoint moment fibers give orthonormality, and
$\sum_z N_z=p^{nk}$ gives amplitude normalization.
Proposition~\ref{prop:conditioned-average-block-form} supplies the
conditioned and average states.  The binary assumptions can be
checked with the following explicit dictionary.

For the binary realization, the objects in
Theorem~\ref{thm:detailed-seed-solution} are:
\begin{align}
 G&=G_{n,2},&
 H&=H_{n,m,2},&
 Y&=G_{n,2}/H_{n,m,2},
 \label{eq:binary-abstract-dictionary-groups}\\
 \mathcal Z_+
 &=\{z\in\mathcal M_{n,2}:N_z^{(2)}>0\},&
 c_z&=2^{-nk/2}\sqrt{N_z^{(2)}},&
 \ket{e_z}&=\ket{e_z^{(2)}}.
 \label{eq:binary-abstract-dictionary-sectors}
\end{align}
The normalization $\sum_zc_z^2=1$ follows from
$\sum_zN_z^{(2)}=2^{nk}$.  Orthonormality follows from disjointness of the moment-map fibers.
Lemma~\ref{lem:binary-orbit-moment-expansion} supplies the character-sector
orbit representation; Lemma~\ref{lem:binary-annihilator-cosets} supplies the
annihilator-coset partition; and
Proposition~\ref{prop:binary-conditioned-average-block-form} supplies the
conditioned and average state formulas.  Thus each assumption on the abstract orbit and its quotient-conditioned
states is verified by the preceding binary lemmas.

\begin{corollary}[Finite-size PGM optimality for every prime]
\label{cor:abstract-theorem-all-prime-application}
For every fixed prime $p$ and every finite $n,m,k$, the mixed-state PGM is
minimum-error optimal for exact partial-label identification, and
\begin{equation}
 P_{\mathrm{id}}^{(k),*}(n,m;p)
 =\frac{p^{-nk}}{\lvert Y_{n,m,p}\rvert}\sum_{\zeta}
 \left(\sum_{z\in\mathcal C_{\zeta}}\sqrt{N_z}\right)^2,
 \label{eq:optimal-formula-detailed}
\end{equation}
where the characteristic-two notation suppresses the superscript $(2)$ on
$N_z^{(2)}$ and $\mathcal C_{\zeta}^{(2)}$.
\end{corollary}
\begin{proof}
For odd $p$, Eq.~\eqref{eq:odd-abstract-amplitudes} gives, for each coset,
\begin{equation}
 \left(\sum_{z\in\mathcal C_{\zeta}}c_z\right)^2
 =p^{-nk}\left(\sum_{z\in\mathcal C_{\zeta}}\sqrt{N_z}\right)^2.
 \label{eq:odd-amplitude-substitution}
\end{equation}
For $p=2$, Eq.~\eqref{eq:binary-abstract-amplitudes} gives the identical
formula with $p=2$ and the suppressed binary superscripts:
\begin{equation}
 \left(\sum_{z\in\mathcal C_{\zeta}^{(2)}}c_z\right)^2
 =2^{-nk}\left(\sum_{z\in\mathcal C_{\zeta}^{(2)}}
 \sqrt{N_z^{(2)}}\right)^2.
 \label{eq:binary-amplitude-substitution}
\end{equation}
Substitute Eq.~\eqref{eq:odd-amplitude-substitution}, respectively
Eq.~\eqref{eq:binary-amplitude-substitution}, into the abstract optimum
Eq.~\eqref{eq:abstract-optimal-coset-formula}.  In both characteristics the
result is Eq.~\eqref{eq:optimal-formula-detailed}.
\end{proof}

Specializing the abstract theorem gives
\begin{equation}
 P_{\mathrm{id}}^{(k),*}(n,m;2)
 =\frac{2^{-nk}}{\lvert Y_{n,m,2}\rvert}
 \sum_\zeta\left(
 \sum_{\substack{z\in\mathcal C_\zeta^{(2)}\\N_z^{(2)}>0}}
 \sqrt{N_z^{(2)}}
 \right)^2.
 \label{eq:integrated-qubit-finite-pgm}
\end{equation}
For each $\zeta$, Cauchy--Schwarz yields
\begin{equation}
 \left(
 \sum_{\substack{z\in\mathcal C_\zeta^{(2)}\\N_z^{(2)}>0}}
 \sqrt{N_z^{(2)}}
 \right)^2
 \le d_\zeta^{(2)}T_\zeta^{(2)}.
 \label{eq:integrated-qubit-CS}
\end{equation}
The binary coset inequality in Eq.~\eqref{eq:integrated-qubit-CS}
will be used in the weighted converse.  Its starting point,
Eq.~\eqref{eq:integrated-qubit-finite-pgm}, is the exact optimum already
proved by Corollary~\ref{cor:abstract-theorem-all-prime-application}.

\begin{remark}[Finite checks]
\label{rem:checks}
If $k=0$, the common support is one-dimensional and the theorem gives
$P_{\rm opt}=\lvert Y\rvert^{-1}$.  If $m=n$, then $H=\{0\}$ and there is one
annihilator-coset block.  Finally, Cauchy--Schwarz inside each coset, together with
$\sum_zN_z=p^{nk}$, directly gives
$P_{\mathrm{id}}^{(k),*}\le1$.
\end{remark}

The all-prime coset formula now supplies the finite exact-identification input  
to the weighted converse in Section~\ref{sec:converse}.  The direct argument in
Section~\ref{sec:fourier-achievability} uses the PGM for the unrestricted
transitive ensemble of pure stabilizer states.  
\section{Converse}
\label{sec:converse}
For the refined partial-label ensemble, an elementary bound compares the
number of labels $|Y|$ with the dimension $D$ of the common support of
the conditioned states.  For fixed $p$ and linear scalings
$m=\beta n+o(n)$ and $k=\alpha n+o(n)$, with $0<\beta\le1$ and
$\alpha\ge0$, the coarse estimate $D\le p^{nk}$ gives
\[
 \log_p|Y|=\left(\beta-\frac{\beta^2}{2}\right)n^2+o(n^2),
 \qquad \log_pD\le\alpha n^2+o(n^2).
\]
Substituting these estimates into
$P_{\mathrm{id}}^{(k),*}\le\min\{1,D/|Y|\}$ yields a negative quadratic
exponent only when $\alpha<\beta-\beta^2/2$.  In contrast,
Theorem~\ref{thm:converse-common-asymptotic-reduction} below proves
vanishing exact success for every $0\le\alpha<1$.
Appendix~\ref{app:dimension-bound-insufficient} gives the elementary
calculation and comparison.  The limitation just stated concerns the
coarse substitution $D\le p^{nk}$.

To derive a sharper converse bound from the finite PGM formula, we
retain two distinct features of each lower moment $\zeta$.  Its probability
$\pi_\zeta$ weights the event, whereas the integer $d_\zeta$ counts the distinct
full moment values occurring over it.  The latter will be bounded by a
matrix-completion count, not by a list of possible reported labels.
The normalizing denominator $|Y|$ is the number of requested labels.
Cosetwise Cauchy--Schwarz gives the comparison:
\begin{align}
 P_{\mathrm{id}}^{(k),*}
 &\stackrel{(a)}{=}\frac{p^{-nk}}{|Y|}
  \sum_\zeta\left(\sum_{z\in\mathcal C_\zeta}\sqrt{N_z}\right)^2
  \notag\\
 &\stackrel{(b)}{\le}\frac{p^{-nk}}{|Y|}
  \sum_\zeta d_\zeta T_\zeta
 \stackrel{(c)}{=}\sum_\zeta\pi_\zeta\frac{d_\zeta}{|Y|}.
 \label{eq:weighted-converse-guide}
\end{align}
Here (a) is Corollary~\ref{cor:abstract-theorem-all-prime-application},
(b) is Cauchy--Schwarz within each coset, and (c) uses
$T_\zeta=p^{mk}M_\zeta$ and
$\pi_\zeta=p^{-(n-m)k}M_\zeta$.
Proposition~\ref{prop:converse-weighted-moment-fiber-reduction}
justifies this bound in both characteristics.  Theorem~\ref{thm:converse-weighted-probability-completion-reduction}
then groups the probabilities by row-space and radical dimensions
$(a,u)$ and bounds the counts using completion ranks $(\ell,t)$.
The two finite estimates feed the common exponent optimization in
Theorem~\ref{thm:converse-common-asymptotic-reduction}.  For $m,k=O(n)$,
the factors $(K_p)^k$ contribute only to the $o(n^2)$ remainder at that
asymptotic stage.

\subsection{Weighted probability--completion reduction}
\label{sec:unreported-reduction}

We now make the comparison in
Eq.~\eqref{eq:weighted-converse-guide} finite and explicit.  The event
probabilities are grouped by the row-space dimension $a$ and radical
dimension $u$ of a lower realization.  For a given lower rank $a-u$,
rank-constrained symmetric-matrix completions bound the moment-support
count $d_\zeta$.  The completion factor is capped at one before the
weighted sum is taken.

\subsubsection{Statement of the weighted reduction}
\label{subsubsec:weighted-reduction-statement}

The following theorem states the resulting bounds for odd primes and
for characteristic two.  The lower-moment matrix has dimension
$n-m+1$ in the odd-prime case because it includes the linear-moment
augmentation, whereas the binary lower matrix has dimension $n-m$.
Only the binary completion bound includes the additional factor $2^m$
from the visible enhanced-diagonal lifts.

\begin{theorem}[Weighted probability--completion reduction]
\label{thm:converse-weighted-probability-completion-reduction}
For odd $p$, put $h=n-m+1$.  For a uniform lower matrix
$X_2\in\F_p^{(n-m)\times k}$, define
\begin{equation}
 \mathcal W(X_2):=\operatorname{row}
 \begin{pmatrix}X_2\\ \1_k^T\end{pmatrix},
 \quad a(X_2):=\dim\mathcal W(X_2),
 \quad u(X_2):=\dim\bigl(\mathcal W(X_2)\cap\mathcal W(X_2)^\perp\bigr),
 \label{eq:main-odd-lower-realization-row-space-pair}
\end{equation}
and set
\begin{equation}
 \mathsf p_{n,m,k}^{\mathrm{odd}}(a,u)
 :=\Pr\{a(X_2)=a,\ u(X_2)=u\}.
 \label{eq:main-odd-lower-realization-mass-function}
\end{equation}
Here orthogonality is taken with respect to the standard symmetric bilinear
form on $\F_p^k$.  For a symmetric
$B\in\Sym_h(\F_p)$ of rank $r$, let
\begin{equation}
 \mathcal A_{m,h,k}(r)
 :=\left|\left\{(Q,C)\in\Sym_m(\F_p)\times\F_p^{m\times h}:
 \rank\begin{pmatrix}Q&C\\C^T&B\end{pmatrix}\le k\right\}\right|.
 \label{eq:main-odd-completion-count-interface}
\end{equation}
The completion count depends on $B$ only through its rank $r$.
The odd-prime exact-identification optimum satisfies the following finite bound:
\begin{equation}
 P_{\mathrm{id}}^{(k),*}(n,m;p)
 \le
 \sum_{a,u}\mathsf p_{n,m,k}^{\mathrm{odd}}(a,u)
 \min\left\{1,
 \frac{\mathcal A_{m,h,k}(a-u)}{\lvert Y_{n,m,p}\rvert}\right\}.
 \label{eq:main-odd-weighted-completion-interface}
\end{equation}

For $p=2$, let $\mathsf p_{n,m,k}^{(2)}(a,u)$ be the probability that a
uniform lower matrix $Y\in\F_2^{(n-m)\times k}$ has row-space dimension $a$
and radical dimension $u$.  For a symmetric
$B\in\Sym_{n-m}(\F_2)$ of rank $r$, define
\begin{equation}
 \mathcal N_{m,n-m,k}^{(2)}(r)
 :=\left|\left\{(Q,C)\in\Sym_m(\F_2)
 \times\F_2^{m\times(n-m)}:
 \rank\begin{pmatrix}Q&C\\C^T&B\end{pmatrix}\le k\right\}\right|.
 \label{eq:main-binary-completion-count-interface}
\end{equation}
The count $\mathcal N_{m,n-m,k}^{(2)}(r)$ is independent of the
choice of $B$ among symmetric matrices of rank $r$.
The binary exact-identification optimum satisfies:
\begin{equation}
 P_{\mathrm{id}}^{(k),*}(n,m;2)
 \le
 \sum_{a,u}\mathsf p_{n,m,k}^{(2)}(a,u)
 \min\left\{1,
 \frac{2^m\mathcal N_{m,n-m,k}^{(2)}(a-u)}{\lvert Y_{n,m,2}\rvert}
 \right\}.
 \label{eq:main-binary-weighted-completion-interface}
\end{equation}
\end{theorem}

Proposition~\ref{prop:converse-weighted-moment-fiber-reduction} gives the
weighted bound
\[
 P_{\mathrm{id}}^{(k),*}
 \le \sum_\zeta\pi_\zeta\frac{d_\zeta}{\lvert Y\rvert}.
\]
Here $\pi_\zeta$ is the probability of the lower event, and
$d_\zeta/\lvert Y\rvert$ is the normalized number of full moment values
in its annihilator coset.  It counts occurring moment values and is
normalized by $|Y|$; it does not count candidate partial labels.

The lower block held fixed in a completion is determined by $\zeta$.
Its matrix differs between the two characteristics.

\begin{definition}[Matrices associated with lower moments]
\label{def:lower-moment-matrices}
For odd $p$ and a lower moment
$\zeta=(\zeta_Q,\zeta_S)\in\Sym_{n-m}(\F_p)\times\F_p^{n-m}$, define
\begin{equation}
 B_\zeta:=\begin{pmatrix}
 \zeta_Q&\zeta_S\\
 \zeta_S^T&k
 \end{pmatrix}\in\Sym_{n-m+1}(\F_p).
 \label{eq:lower-moment-matrix-odd}
\end{equation}
The lower-right entry is the image of the copy number $k$ in $\F_p$.
For $p=2$, write a lower enhanced moment as
$\zeta=(r^{\rm low},T^{\rm low})\in\mathcal L_{n,m,2}$.
Here $r_a^{\rm low}$ is the mod-$4$ diagonal coordinate for
$1\le a\le n-m$, and $T_{ab}^{\rm low}$ is the mod-$2$ edge coordinate
for $1\le a<b\le n-m$.  Define the symmetric matrix
\begin{equation}
 B_\zeta\in\Sym_{n-m}(\F_2),
 \quad (B_\zeta)_{aa}:=r_a^{\rm low}\pmod2,
 \quad (B_\zeta)_{ab}=(B_\zeta)_{ba}:=T_{ab}^{\rm low}\quad(a<b).
 \label{eq:lower-moment-matrix-binary}
\end{equation}
\end{definition}

The characteristic-specific completion bounds estimate $d_\zeta$ in terms
of $\rank B_\zeta$.  For odd $p$, this rank is
$a(X_2)-u(X_2)$ for a realizing lower matrix $X_2$, so the average is
grouped by the realization events $(a,u)$.  The binary proof uses its
lower-matrix row-space events.  Both regroupings give the corresponding
bounds in the theorem.

\subsubsection{Weighted lower-moment bound}

We now justify the first step of the reduction independently of the
characteristic-specific completion counts.  A lower moment contributes
according to its actual probability $\pi_\zeta$; the conditional support
factor is the number $d_\zeta$ of occurring full moment values in its
coset, divided by the fixed partial-label alphabet size $|Y|$.

\begin{proposition}[Weighted counting bound]
\label{prop:converse-weighted-moment-fiber-reduction}
For every prime $p$ and all finite integers $n,m,k$ with $1\le m\le n$, let
$Y=Y_{n,m,p}$ be the partial-label quotient.  For each lower-coordinate value
$\zeta$, let $\mathcal C_{\zeta}$ be the corresponding annihilator-coset,
let
\begin{equation}
 d_{\zeta}:=
 \bigl|\mathcal C_{\zeta}\cap\operatorname{im}\Phi_k\bigr|,
 \label{eq:theorem-A-d-zeta}
\end{equation}
and let $\pi_{\zeta}$ be the probability that a uniform lower block
$X_2\in\F_p^{(n-m)\times k}$ has lower moment $\zeta$.  For odd $p$,
the probability assigned to a lower value $\zeta=(\zeta_Q,\zeta_S)$ is
\begin{equation}
 \pi_{\zeta}:=
 p^{-(n-m)k}
 \left|\left\{X_2\in\F_p^{(n-m)\times k}:
 X_2X_2^T=\zeta_Q,\ X_2\1_k=\zeta_S\right\}\right|.
 \label{eq:theorem-A-pi-zeta}
\end{equation}
For $p=2$, $\zeta$ and $\pi_{\zeta}$ denote the lower enhanced moment and its
probability under the binary moment map.  In the bound below,
$\pi_{\zeta}$ is the probability weight of the lower event, whereas
$d_{\zeta}/\lvert Y\rvert$ is the number of occurring full moment values in
its annihilator coset, normalized by the fixed partial-label alphabet size
$\lvert Y\rvert$.  Fixing a lower moment does not change this denominator.
In either characteristic, the weights form a probability distribution:
$\pi_{\zeta}\ge0$ and $\sum_{\zeta}\pi_{\zeta}=1$.
The exact-identification optimum obeys the weighted bound:
\begin{equation}
 P_{\mathrm{id}}^{(k),*}(n,m;p)
 \le
 \sum_{\zeta}\pi_{\zeta}\frac{d_{\zeta}}{\lvert Y\rvert}.
 \label{eq:weighted-support-bound}
\end{equation}

\end{proposition}
\begin{proof}[Proof of Proposition~\ref{prop:converse-weighted-moment-fiber-reduction}]
For odd $p$, Eq.~\eqref{eq:theorem-A-pi-zeta} is the lower-block
coordinate distribution induced by $X_2\mapsto(X_2X_2^T,X_2\1_k)$.
For $p=2$, the corresponding distribution comes from the enhanced-moment
coset formula; the same cosetwise Cauchy--Schwarz argument then applies.

We prove the weighted bound in three parts.  First, the lower block $X_2$
determines the annihilator coset.  Next, we count all tuples in a fixed
coset by choosing $X_2$ and then the unrestricted upper block $X_1$.
Finally, we apply Cauchy--Schwarz to the moment multiplicities within each
coset.  The corresponding binary formulas use the enhanced moment map,
but the counting and Cauchy--Schwarz steps remain the same.

\smallskip
\noindent\emph{Localization of the lower event to $X_2$.}
Suppose first that $p$ is odd.  One ordered computational-basis tuple may be
written as a matrix
\begin{equation}
 X=(x_1\ \cdots\ x_k)
 =\begin{pmatrix}X_1\\X_2\end{pmatrix},
 \quad
 X_1\in\F_p^{m\times k},\quad
 X_2\in\F_p^{(n-m)\times k}.
 \label{eq:X-blocks}
\end{equation}
For the moment value $\Phi_k(x_1,\ldots,x_k)=(Q,S)$, block multiplication gives
\begin{align}
 Q_{11}&=X_1X_1^T,&
 Q_{12}&=X_1X_2^T,&
 Q_{22}&=X_2X_2^T,\notag\\
 S_1&=X_1\1_k,&
 S_2&=X_2\1_k.
 \label{eq:moment-blocks}
\end{align}
Define the lower-moment pair associated with $X_2$ by
\begin{equation}
 \zeta(X_2):=(X_2X_2^T,X_2\1_k).
 \label{eq:zeta-of-X2}
\end{equation}
This pair depends only on $X_2$.  By the coordinate description of the
$H^\perp$-cosets in Eq.~\eqref{eq:coset-as-coordinate-fiber}, the full moment
of $X$ lies in $\mathcal C_{\zeta(X_2)}$.  Conversely, fixing $\zeta$ fixes
exactly the lower moment coordinates and leaves $X_1$ unrestricted at this
point.  The weighted reduction uses only this part of the block decomposition.  The augmented Gram matrix and its rank constraint belong to
the later completion theorem and are not used here.
Here and throughout the augmented-Gram formulas, the integer $k$ is read through its image in the prime subfield of $\F_p$.  This convention includes $p\mid k$, where that image is zero, and is consistent with the separate $k=0$ one-dimensional-support check.

For $p=2$, the same split $X=(X_1;X_2)$ is used with the enhanced moment map.
The restriction map $r_2$ in Eq.~\eqref{eq:binary-lower-moment-restriction}
keeps exactly the lower mod-$4$ linear coordinates and the lower--lower
mod-$2$ quadratic coordinates.  Hence the lower enhanced moment
$\zeta_2(X_2)$ again depends only on $X_2$, and
Eq.~\eqref{eq:binary-annihilator-coset-equivalence} identifies its fibers with
the binary annihilator cosets.

\smallskip
\noindent\emph{Normalization of the tuple count to a probability weight.}
For odd $p$, let
\begin{equation}
 M_{\zeta}:=
 \left|\left\{X_2\in\F_p^{(n-m)\times k}:
 (X_2X_2^T,X_2\1_k)=\zeta\right\}\right|.
 \label{eq:unreported-fiber-size}
\end{equation}
The total number of ordered tuples whose full moment belongs to
$\mathcal C_{\zeta}$ is
\begin{equation}
 T_{\zeta}:=\sum_{z\in\mathcal C_{\zeta}}N_z.
 \label{eq:theorem-A-T-zeta}
\end{equation}
We use this auxiliary tuple count in the following double count.  Count those tuples in two ways.  Partitioning by the full
moment gives $T_{\zeta}$.  Choosing $X_2$ first gives
$M_{\zeta}$ choices, after which $X_1\in\F_p^{m\times k}$ is arbitrary and
has $p^{mk}$ choices.  Counting all such matrices gives the
relation between tuple multiplicity and the lower moment count:
\begin{equation}
 T_{\zeta}=p^{mk}M_{\zeta}.
 \label{eq:unreported-mass-identity}
\end{equation}
The matrices $X_2$ form a uniform sample space of size $p^{(n-m)k}$, so
\begin{equation}
 \pi_{\zeta}:=\frac{M_{\zeta}}{p^{(n-m)k}}.
 \label{eq:lower-block-probability}
\end{equation}
These nonnegative weights sum to one and therefore form a probability
mass function.

In characteristic two, define
\begin{equation}
 M_{\zeta}^{(2)}
 :=\left|\left\{X_2\in\F_2^{(n-m)\times k}:
 \zeta_2(X_2)=\zeta\right\}\right|,
 \quad
 \pi_{\zeta}^{(2)}:=\frac{M_{\zeta}^{(2)}}{2^{(n-m)k}}.
 \label{eq:binary-lower-enhanced-probability}
\end{equation}
For each fixed lower matrix $X_2$ with $\zeta_2(X_2)=\zeta$,
all $2^{mk}$ upper matrices $X_1\in\F_2^{m\times k}$ are allowed.
Counting first the lower matrices and then these upper matrices gives
\begin{equation}
 T_{\zeta}^{(2)}
 =\sum_{X_2:\,\zeta_2(X_2)=\zeta}
   \sum_{X_1\in\F_2^{m\times k}}1
 =2^{mk}M_{\zeta}^{(2)}.
 \label{eq:binary-enhanced-mass-identity}
\end{equation}
The lower enhanced events partition the binary sample space, so
$\pi_{\zeta}^{(2)}$ is also a probability mass function.

\smallskip
\noindent\emph{Bounding success by the number of occurring full moments.}
For odd $p$, the exact finite-size formula
Eq.~\eqref{eq:optimal-formula-detailed} is
\begin{equation}
 P_{\mathrm{id}}^{(k),*}(n,m;p)
 =\frac{p^{-nk}}{\lvert Y\rvert}\sum_{\zeta}
 \left(\sum_{z\in\mathcal C_{\zeta}}\sqrt{N_z}\right)^2.
 \label{eq:weighted-bound-starting-formula}
\end{equation}
Only the indices $z$ for which $N_z>0$ contribute.  Hence
Cauchy--Schwarz within one coset gives
\begin{equation}
 \left(\sum_{z\in\mathcal C_{\zeta}}\sqrt{N_z}\right)^2
 \le d_{\zeta}\sum_{z\in\mathcal C_{\zeta}}N_z
 =d_{\zeta}T_{\zeta}.
 \label{eq:weighted-bound-coset-CS}
\end{equation}
Substituting Eq.~\eqref{eq:unreported-mass-identity} and using
$p^{-nk}p^{mk}=p^{-(n-m)k}$ yields
\begin{align}
 P_{\mathrm{id}}^{(k),*}(n,m;p)
 \le\frac{p^{-nk}}{\lvert Y\rvert}
 \sum_{\zeta}d_{\zeta}p^{mk}M_{\zeta}
 =\sum_{\zeta}\pi_{\zeta}\frac{d_{\zeta}}{\lvert Y\rvert}.
 \label{eq:weighted-support-bound-odd-proof}
\end{align}

For $p=2$, Eqs.~\eqref{eq:integrated-qubit-finite-pgm} and
\eqref{eq:integrated-qubit-CS} give
\begin{equation}
 P_{\mathrm{id}}^{(k),*}(n,m;2)
 \le\frac{2^{-nk}}{\lvert Y_{n,m,2}\rvert}
 \sum_{\zeta}d_{\zeta}^{(2)}T_{\zeta}^{(2)}.
 \label{eq:binary-weighted-bound-before-mass}
\end{equation}
Substituting Eq.~\eqref{eq:binary-enhanced-mass-identity} gives
\begin{equation}
 P_{\mathrm{id}}^{(k),*}(n,m;2)
 \le\sum_{\zeta}\pi_{\zeta}^{(2)}
 \frac{d_{\zeta}^{(2)}}{\lvert Y_{n,m,2}\rvert}.
 \label{eq:binary-weighted-coset-bound}
\end{equation}
Suppressing the characteristic-two superscripts in the common statement gives
Eq.~\eqref{eq:weighted-support-bound} in both characteristics.

When $k=0$, the tuple space consists of one empty tuple and the same argument
reduces to $P_{\mathrm{id}}^{(0),*}=\lvert Y\rvert^{-1}$, in agreement with
Remark~\ref{rem:checks}.  Thus no positive-copy assumption is needed for the
proposition.
\end{proof}

Proposition~\ref{prop:converse-weighted-moment-fiber-reduction} supplies
the weighted lower-moment input to
Theorem~\ref{thm:converse-weighted-probability-completion-reduction}.  It
remains to control the row-space pair $(a,u)$ and apply the
characteristic-specific completion comparisons.

\subsubsection{Common row-space and radical estimates}
\label{subsec:common-row-space-estimate}

In this subsection, orthogonal complements of subspaces of $\F_p^k$ are
taken with respect to the standard symmetric bilinear form, not the symplectic
form on phase space.  The weighted bound separates the probability of a fixed
lower-coordinate event from the number of full moments compatible with it.  The first factor is
controlled by the following row-space geometry, which is valid for every
prime; only the second factor requires characteristic-specific completion
algebra.  Recall
$\kappa_p^{\mathrm{GL}}=\prod_{j\ge1}(1-p^{-j})>0$ from
Eq.~\eqref{eq:kappa-GL-product}.

\begin{proposition}[All-prime prescribed-radical subspace bound]
\label{prop:all-prime-prescribed-radical-bound}
Fix a prime $p$.  There is a constant $K_p^{\mathrm{rad}}>1$, depending only
on $p$, such that the number $N_p(k;a,u)$ of $a$-dimensional subspaces
$H\subseteq\F_p^k$ satisfying
\begin{equation}
 \dim(H\cap H^\perp)=u
 \label{eq:all-prime-prescribed-radical-condition}
\end{equation}
is bounded by
\begin{equation}
 N_p(k;a,u)
 \le (K_p^{\mathrm{rad}})^k
 p^{a(k-a)-u(u+1)/2}.
 \label{eq:all-prime-prescribed-radical-bound}
\end{equation}
The same bound remains valid if $H$ is required to contain any prescribed
vector.  We set $N_p(k;a,u)=0$ whenever the parameters are infeasible; in
particular, this convention applies unless
$0\le u\le a$, $2u\le k$, and $a-u\le k-2u$.
\end{proposition}
\begin{proof}
There is nothing to prove outside the feasible range specified in the
statement.  We therefore assume throughout the proof that these inequalities
hold.  We first bound the possible radicals.  Let $T_p(k,u)$ be the number of
totally isotropic $u$-subspaces of the standard nondegenerate symmetric
bilinear space $\F_p^k$.  Choose an ordered basis
$v_1,\ldots,v_u$ successively.  If
$R_{j-1}=\operatorname{span}(v_1,\ldots,v_{j-1})$ is totally isotropic, then
$R_{j-1}^\perp/R_{j-1}$ is nondegenerate of dimension $k-2(j-1)$.
For an upper bound, count all vectors in this quotient, rather than only the
isotropic nonzero classes.  There are at most $p^{k-2(j-1)}$ quotient
choices, each with $p^{j-1}$ lifts, and hence at most $p^{k-j+1}$ choices at
stage $j$.  The number of ordered candidate frames is therefore at most
\begin{equation}
 p^{uk-u(u+1)/2+u}.
 \label{eq:all-prime-isotropic-ordered-frame-bound}
\end{equation}
Every $u$-space has at least
$\kappa_p^{\mathrm{GL}}p^{u^2}$ ordered bases.  Since $u\le k$, division
and $p^u\le p^k$ give
\begin{equation}
 T_p(k,u)\le
 \left(\frac{p}{\kappa_p^{\mathrm{GL}}}\right)^k
 p^{u(k-u)-u(u+1)/2}.
 \label{eq:all-prime-isotropic-subspace-bound}
\end{equation}
Here and below we use that a fixed factor
$(\kappa_p^{\mathrm{GL}})^{-1}$ is at most its $k$th power for $k\ge1$.  This argument does not distinguish
odd characteristic from characteristic two, or alternating from
nonalternating quotient type.

Now put $R=H\cap H^\perp$. Then
$R\subseteq H\subseteq R^\perp$ and $\dim R=u$.  After choosing $R$ by
Eq.~\eqref{eq:all-prime-isotropic-subspace-bound}, pass to the nondegenerate
quotient $R^\perp/R$, of dimension $k-2u$.  Dropping the requirement that the
form induced on $H/R$ be nondegenerate gives the following upper bound
on the number of choices for $H/R$:
\begin{equation}
 \genfrac{[}{]}{0pt}{}{k-2u}{a-u}_p
 \le (\kappa_p^{\mathrm{GL}})^{-1}
 p^{(a-u)(k-a-u)}.
 \label{eq:all-prime-radical-Gaussian-bound}
\end{equation}
The exponents from these two counts combine through the identity
\begin{equation}
 u(k-u)+(a-u)(k-a-u)=a(k-a).
 \label{eq:all-prime-radical-exponent-identity}
\end{equation}
Combining the two counts with this identity proves
Eq.~\eqref{eq:all-prime-prescribed-radical-bound}.  For example, one may take
\begin{equation}
 K_p^{\mathrm{rad}}
 :=p(\kappa_p^{\mathrm{GL}})^{-2}.
 \label{eq:all-prime-radical-constant-choice}
\end{equation}
The case $k=0$ is immediate.  Requiring a prescribed vector selects a
subfamily and cannot increase the count.
\end{proof}

\begin{proposition}[All-prime row-space event bound]
\label{prop:all-prime-row-space-event-bound}
Let $Y$ be uniform on $\F_p^{r\times k}$ and put
\begin{equation}
 A(Y):=\dim\operatorname{row}Y,
 \quad
 U(Y):=\dim\bigl(\operatorname{row}Y\cap
                         (\operatorname{row}Y)^\perp\bigr).
 \label{eq:all-prime-row-space-invariants}
\end{equation}
Thus $a$ records the row-space dimension and $u$ its radical dimension; the
pair $(a,u)$ is the probability-side parameter pair used below.  With
$K_p^{\mathrm{class}}:=K_p^{\mathrm{rad}}$, uniformly over all pairs
$(a,u)$,
\begin{equation}
 \Pr\{A(Y)=a,\ U(Y)=u\}
 \le (K_p^{\mathrm{class}})^k
 p^{-(k-a)(r-a)-u(u+1)/2}.
 \label{eq:all-prime-row-space-event-bound}
\end{equation}
The probability bound in
Eq.~\eqref{eq:all-prime-row-space-event-bound} also holds when the
admissible row spaces must contain a prescribed vector.
\end{proposition}
\begin{proof}
For a fixed $a$-dimensional subspace $H\subseteq\F_p^k$, the event
$\operatorname{row}Y=H$ is contained in the event that all $r$ rows lie in
$H$, whose probability is $p^{-r(k-a)}$.  For an upper bound, it is therefore sufficient to require that every
row lie in $H$; spanning $H$ need not be imposed.  Sum the resulting
bound over the spaces in
Proposition~\ref{prop:all-prime-prescribed-radical-bound}.  The resulting
exponent is simplified by
\begin{equation}
 a(k-a)-r(k-a)=-(k-a)(r-a),
 \label{eq:all-prime-row-space-exponent-identity}
\end{equation}
which proves Eq.~\eqref{eq:all-prime-row-space-event-bound}.  A prescribed
vector condition only reduces the family of spaces being summed over.
\end{proof}

The augmented event adds a prescribed vector to the rows of $Y$. Every resulting row space contains that vector, so the prescribed-vector clause of Proposition~\ref{prop:all-prime-prescribed-radical-bound} yields the same upper bound after summing over the possible subspaces.

\begin{corollary}[All-prime augmented row-space event bound]
\label{cor:all-prime-augmented-row-space-event-bound}
Fix $v\in\F_p^k$, let $Y$ be uniform on $\F_p^{r\times k}$, and define
\begin{equation}
 W_v(Y):=\operatorname{span}\bigl(\{v\}\cup\operatorname{row}Y\bigr),
 \quad
 A_v(Y):=\dim W_v(Y),
 \quad
 U_v(Y):=\dim\bigl(W_v(Y)\cap W_v(Y)^\perp\bigr).
 \label{eq:all-prime-augmented-row-space-invariants}
\end{equation}
Then, uniformly over all pairs $(a,u)$,
\begin{equation}
 \Pr\{A_v(Y)=a,\ U_v(Y)=u\}
 \le (K_p^{\mathrm{class}})^k
 p^{-(k-a)(r-a)-u(u+1)/2}.
 \label{eq:all-prime-augmented-row-space-event-bound}
\end{equation}
\end{corollary}
\begin{proof}
For every fixed $a$-space $W$ containing $v$, the event $W_v(Y)=W$ implies
that all $r$ rows of $Y$ lie in $W$, and hence has probability at most
$p^{-r(k-a)}$.  Summing over the $a$-spaces containing $v$ and having radical
dimension $u$, and then applying the prescribed-vector clause of
Proposition~\ref{prop:all-prime-prescribed-radical-bound}, gives the stated
bound.  This argument does not identify $W_v(Y)$ with
$\operatorname{row}Y$ and therefore also covers the case
$v\notin\operatorname{row}Y$.
\end{proof}

The probability side of the converse is now controlled entirely by the pair
$(a,u)$: the row-space dimension and its radical dimension.  The individual
lower matrix no longer enters this estimate.  The two characteristic-specific
branches will identify the corresponding lower rank $a-u$ and address only the
remaining completion count.

\begin{lemma}[Common Gram--radical rank identity]
\label{lem:common-Gram-radical-rank}
Let $\F$ be any field, let $Z\in\F^{r\times k}$, and put
$H:=\operatorname{row}Z\subseteq\F^k$.  For the standard bilinear form on
$\F^k$, set
\begin{equation}
 a:=\dim H,
 \quad
 u:=\dim(H\cap H^\perp).
 \label{eq:common-Gram-radical-parameters}
\end{equation}
The rank of the Gram matrix $ZZ^T$ equals the row-space dimension minus the dimension of its radical:
\begin{equation}
 \rank(ZZ^T)=a-u.
 \label{eq:common-Gram-radical-rank}
\end{equation}
\end{lemma}
\begin{proof}
Let
\begin{equation}
 \pi:\F^r\longrightarrow H,\quad \pi(c):=c^TZ,
 \label{eq:common-Gram-row-coordinate-surjection}
\end{equation}
which is surjective by the definition of the row space.  Let
$\beta_H:H\to H^*$ be the linear map induced by the restriction of the
standard bilinear form to $H$.  Under the standard coordinate
identifications, the map represented by $ZZ^T$ is the composition
\begin{equation}
 \F^r\xrightarrow{\ \pi\ }H
 \xrightarrow{\ \beta_H\ }H^*
 \xrightarrow{\ \pi^*\ }(\F^r)^*.
 \label{eq:common-Gram-pullback-factorization}
\end{equation}
Because $\pi$ is surjective, $\pi^*$ is injective.  Therefore
$\rank(\pi^*\beta_H\pi)=\rank(\beta_H)$.  Finally,
$\ker\beta_H=H\cap H^\perp$, and hence rank--nullity gives
\begin{equation}
 \rank(ZZ^T)=\rank(\beta_H)
 =\dim H-\dim(H\cap H^\perp)=a-u.
\end{equation}
\end{proof}

Lemma~\ref{lem:common-Gram-radical-rank} is characteristic-free.  The
odd-prime branch applies it to the matrix obtained by adjoining $\1_k^T$
as the last row of $X_2$, whereas the binary branch applies it to the lower
matrix $Y$.  Thus the same parameter $a-u$ is the lower rank supplied to
both completion theorems; only the type of moment data retained in addition
to the associated bilinear matrix differs.

For the binary comparison, fix a lower enhanced moment $\zeta$.
Map a full enhanced moment $z=(r,T)$ to a bilinear completion $(Q,C)$ by
reducing each visible diagonal $r_i$ modulo two and retaining all visible
edge coordinates $T_{ij}$.  Here $Q$ contains the visible--visible
bilinear entries and $C$ contains the visible--lower entries; the lower
bilinear block is fixed by $\zeta$.  The map forgets only the high bit of
each of the $m$ visible mod-$4$ diagonal coordinates, so each completion
has at most $2^m$ occurring enhanced moments above it.  Consequently,
\begin{equation}
 d_\zeta^{(2)}
 \le 2^m\mathcal N_{m,n-m,k}^{(2)}(\rank B_\zeta).
 \label{eq:main-binary-forgetting-bound}
\end{equation}
Proposition~\ref{prop:binary-enhanced-bilinear-comparison} gives the
explicit map and proves this bound.  We use \emph{diagonal lift} only
for a choice of the forgotten mod-$4$ diagonal bits above fixed bilinear
data.  Their number is bounded by $2^m$, independently of the bilinear
completion count itself.

\subsubsection{Augmented Gram representation of the odd-prime lower moment}
\label{subsec:odd-augmented-Gram-bridge}

We record here the precise bridge from the lower moment in the weighted bound
to the matrix used in the odd-prime completion theorem.  This prevents the
augmentation by the linear moment from being hidden inside the appendix.
Assume first that $k\ge1$.  For
$X_2\in\F_p^{(n-m)\times k}$, define
\begin{equation}
 Z_2:=\begin{pmatrix}X_2\\ \1_k^T\end{pmatrix}
 \in\F_p^{h\times k},
 \quad h:=n-m+1.
 \label{eq:main-odd-augmented-lower-matrix}
\end{equation}
When $X_2$ realizes the lower moment
$\zeta=(\zeta_Q,\zeta_S)=(X_2X_2^T,X_2\1_k)$, the matrix defined in
Eq.~\eqref{eq:lower-moment-matrix-odd} has the Gram representation
\begin{equation}
 B_\zeta=Z_2Z_2^T
 =\begin{pmatrix}
   \zeta_Q&\zeta_S\\
   \zeta_S^T&k
  \end{pmatrix}\in\Sym_h(\F_p).
 \label{eq:main-odd-augmented-Gram-matrix}
\end{equation}
The lower-right entry is the image of the integer $k$ in the prime subfield
of $\F_p$.  In particular, it is zero when $p\mid k$; no nonvanishing of this
entry is used.

Let
\begin{equation}
 \mathcal W(X_2):=\operatorname{row}Z_2\subseteq\F_p^k,
 \quad a(X_2):=\dim\mathcal W(X_2),
 \quad u(X_2):=\dim\bigl(\mathcal W(X_2)\cap\mathcal W(X_2)^\perp\bigr).
 \label{eq:main-odd-augmented-row-space-pair}
\end{equation}
Applying Lemma~\ref{lem:common-Gram-radical-rank} to $Z_2$ gives
\begin{equation}
 \rank B_\zeta=a(X_2)-u(X_2).
 \label{eq:main-odd-augmented-Gram-rank}
\end{equation}
The rank identity in Eq.~\eqref{eq:main-odd-augmented-Gram-rank}
holds even when $p\mid k$: it follows from the restricted bilinear
form, not from a nonzero lower-right matrix entry.  The case
$k=0$ has the one-dimensional common support treated separately in
Remark~\ref{rem:checks}; the completion estimates below are only needed for
positive copy number.

The appendix makes the remaining cardinality comparison explicit.  For fixed
$\zeta$, the map
\begin{equation}
 (Q,S)\longmapsto\bigl(Q_{11},(Q_{12}\ S_1)\bigr)
 \label{eq:main-visible-to-ambient-map-preview}
\end{equation}
is injective on the occurring full moments above $\zeta$, because the lower
coordinates $(Q_{22},S_2)=\zeta$ are fixed.  Any tuple realizing that moment
gives a full augmented Gram matrix of rank at most $k$.  Hence the occurring
visible moments inject into the ambient rank-constrained completion set.  The
full factorization and the statement that this is only an upper-bound
relaxation are given in Lemma~\ref{lem:visible-to-ambient-completion}.

\subsubsection{Proof of the weighted probability--completion reduction}

\begin{proof}[Proof of Theorem~\ref{thm:converse-weighted-probability-completion-reduction}]
Proposition~\ref{prop:converse-weighted-moment-fiber-reduction} starts from the
common weighted bound
\begin{equation}
 P_{\mathrm{id}}^{(k),*}
 \le\sum_\zeta\pi_\zeta\frac{d_\zeta}{\lvert Y\rvert}.
 \label{eq:weighted-completion-assembly-start}
\end{equation}
We now substitute the characteristic-specific completion bounds into the
normalized occurring-moment-sector factor $d_\zeta/\lvert Y\rvert$.  The
resulting functions of the lower rank are then averaged over realizing lower
matrices and grouped by their row-space parameters $(a,u)$.

For odd $p$, the completion argument in
Appendix~\ref{app:odd-prime-finite-structural-proof} gives
\begin{equation}
 \frac{d_\zeta}{\lvert Y\rvert}
 \le\min\left\{1,
 \frac{\mathcal A_{m,h,k}(\rank B_\zeta)}{\lvert Y\rvert}\right\}.
 \label{eq:weighted-completion-assembly-odd-substitution}
\end{equation}
For a realizing lower matrix $X_2$, the augmented Gram identity gives
$\rank B_\zeta=a(X_2)-u(X_2)$.  Hence, for the clipped completion factor $F$,
\begin{equation}
 \sum_\zeta\pi_\zeta F(\rank B_\zeta)
 =\mathbb E_{X_2}F\bigl(a(X_2)-u(X_2)\bigr)
 =\sum_{a,u}\mathsf p_{n,m,k}^{\mathrm{odd}}(a,u)F(a-u).
 \label{eq:main-odd-realization-expectation-regrouping}
\end{equation}
Applying this identity to Eq.~\eqref{eq:weighted-completion-assembly-start}
gives Eq.~\eqref{eq:main-odd-weighted-completion-interface}.  This is an
expectation over realizing lower matrices, not a partition of lower moments by
$(a,u)$.

For $p=2$, the enhanced-to-bilinear comparison in
Appendix~\ref{app:binary-finite-structural-proof} gives the parallel bound
\begin{equation}
 \frac{d_\zeta^{(2)}}{\lvert Y_{n,m,2}\rvert}
 \le\min\left\{1,
 \frac{2^m\mathcal N_{m,n-m,k}^{(2)}(\rank B_\zeta)}
 {\lvert Y_{n,m,2}\rvert}\right\}.
 \label{eq:weighted-completion-assembly-binary-substitution}
\end{equation}
The associated bilinear form again has rank $a-u$.  Grouping the lower
enhanced moments by the binary row-space event gives
Eq.~\eqref{eq:main-binary-weighted-completion-interface}.

In both cases, the minimum with one is the exact finite cardinality cap
$d_\zeta\le\lvert Y\rvert$; it is not introduced only for the later asymptotic analysis.
\end{proof}
\subsection{Characteristic-specific finite structural estimates}
\label{subsec:converse-characteristic-specific-structural-estimates}
Theorem~\ref{thm:converse-weighted-probability-completion-reduction}
gives a probability-weighted completion bound in each characteristic.
The lower-matrix rank is $a-u$ by
Lemma~\ref{lem:common-Gram-radical-rank}.  The two remaining ranks
$(\ell,t)$ describe, respectively, the coupling to the radical and the
residual symmetric block in the matrix completion.  These four ranks
allow us to compare the two bounds without identifying their different
lower-moment representations.  The following table records the finite
differences that must be retained before normalization.

\subsubsection{Common rank parameters in the two characteristics}
\label{subsubsec:common-characteristic-rank-parameters}

The two characteristic-specific arguments use different lower-moment
representations but lead to the same four rank parameters.  Their precise
correspondence is summarized here:
\begin{center}
\small
\begin{tabularx}{0.96\textwidth}{@{}lYY@{}}
\toprule
& Odd characteristic & Characteristic two\\
\midrule
Probability-side space
& $\operatorname{row}Z_2$, including the linear-moment augmentation
& $\operatorname{row}Y$ of the lower binary block\\
Lower matrix
& augmented Gram matrix $B_\zeta=Z_2Z_2^T$
& associated bilinear matrix $B_\zeta=YY^T$\\
Common identity
& $\rank B_\zeta=a-u$
& $\rank B_\zeta=a-u$\\
Coupling rank $\ell$
& Rank of the block coupling to $\ker B_\zeta$
& Rank of the block coupling to $\ker B_\zeta$\\
Residual rank $t$
& Rank of the remaining symmetric block after the coupling reduction
& Rank of the remaining symmetric block after the coupling reduction\\
Rank budget
& $a-u+2\ell+t\le k$
& $a-u+2\ell+t\le k$\\
Finite correction
& none beyond $h=n-m+1$
& At most $2^m$ choices of visible mod-$4$ diagonal bits; the finite exponent correction is $m-\ell$\\
\bottomrule
\end{tabularx}
\end{center}
In the remainder of the converse, $(a,u)$ denotes the common probability-side pair, and $a-u$ is the corresponding lower rank.

In the odd-prime structural estimate below, $h=n-m+1$ determines the
augmented lower block and the admissible range of $\ell$.  The binary
estimate instead uses a lower block of size $n-m$ and accounts for the
visible enhanced-diagonal lifts by $2^m$.  Both retain the completion
rank budget $a-u+2\ell+t\le k$ and the cap on the normalized completion
count.  The common exponent proof will use these rank constraints, not
a term-by-term identification of the two lower-moment alphabets.

\subsubsection{Odd-prime finite structural estimate}
\label{sec:odd-prime-finite-structural-estimate}

By Theorem~\ref{thm:converse-weighted-probability-completion-reduction}, the
odd-prime branch has already been reduced to
Eq.~\eqref{eq:main-odd-weighted-completion-interface}: a sum over the
row-space pair $(a,u)$ of its probability times the clipped completion count
$\mathcal A_{m,h,k}(a-u)/\lvert Y\rvert$.  It remains only to bound the probability
factor and the completion factor uniformly in their rank parameters.  The
completion side is indexed further by the rank pair $(\ell,t)$.

The next theorem eliminates the individual lower moments and concrete
completion matrices and expresses the odd-prime bound in the parameters used
in the common asymptotic argument.

We first define the finite rank domain and its two exponents.  This
keeps the matrix-counting parameters separate from the probability
bound in the theorem below.

Fix an odd prime $p$ and put $h:=n-m+1$.  With the augmented row-space
probability defined in Eq.~\eqref{eq:main-odd-lower-realization-mass-function},
let
\begin{equation}
 \mathfrak P_n:=\{(a,u)\in\mathbb Z_{\ge0}^2:
                \mathsf p_{n,m,k}^{\mathrm{odd}}(a,u)>0\}.
 \label{eq:odd-positive-probability-support}
\end{equation}
Every pair in this support satisfies
\begin{equation}
 0\le u\le a\le k,\quad a+u\le k.
 \label{eq:theorem-B-positive-mass-domain}
\end{equation}
For each
$(a,u)\in\mathfrak P_n$, define
\begin{equation}
 \mathcal G_{a,u}:=
 \left\{(\ell,t)\in\mathbb Z_{\ge0}^2:
 \begin{array}{l}
 0\le\ell\le\min\{m,h-a+u\},\\
 0\le t\le m-\ell,\\
 a-u+2\ell+t\le k
 \end{array}\right\}.
 \label{eq:theorem-B-G-au}
\end{equation}
Here $(a,u)$ is the row-space dimension and radical dimension.  Conditional on this pair,
$(\ell,t)$ records the completion-side ranks: $\ell$ is the rank of the
radical coupling and $t$ is the rank of the residual symmetric block.
For $(a,u,\ell,t)$ in this domain, set
\begin{align}
 E_{\rm fin}(m,h;a,u,\ell,t)
 :={}&m(a-u)+\ell(m+h-a+u-\ell)
 +\ell m-\frac{\ell(\ell-1)}2\notag\\
 &+\frac{t[2(m-\ell)-t+1]}2.
 \label{eq:theorem-B-E-fin}
\end{align}
The exponent of the visible-coordinate alphabet is defined separately by
\begin{equation}
 D_{\rm vis,fin}:=mh+\frac{m(m+1)}2
 =(n+1)m-\frac{m(m-1)}2.
 \label{eq:theorem-B-D-vis-fin}
\end{equation}
The enlarged completion count is bounded using the uniform rectangular-
and symmetric-rank estimates.  The quantity $E_{\rm fin}$ records the
$p$-power exponent in that upper bound; it is not generally the
base-$p$ logarithm of the exact completion count.
In contrast, $D_{\rm vis,fin}=\log_p|Y|$ is the exact logarithm of the
visible label-space cardinality.

\begin{theorem}[Odd-prime finite structural estimate]
\label{thm:converse-odd-prime-finite-structural-estimate}
Fix an odd prime $p$ and integers $n\ge1$, $1\le m\le n$, and $k\ge0$.  Use the
rank parameters and exponents defined above.  Every positive-mass pair
$(a,u)\in\mathfrak P_n$ satisfies Eq.~\eqref{eq:theorem-B-positive-mass-domain}.
The following two bounds hold.

\medskip\noindent\textbf{(1) Finite bound.}
Combining the completion upper bound
with the finite probability estimate for the row-space parameters
$(a,u)$ gives
\begin{align}
 P_{\mathrm{id}}^{(k),*}(n,m;p)
 &\le \sum_{(a,u)\in\mathfrak P_n}
 (K_p^{\mathrm{class}})^k
 p^{-(k-a)(n-m-a)-u(u+1)/2}
 \notag\\[-1mm]
 &\quad\times
 \min\left\{1,
 C_p^{\mathrm{comp}}(m+1)^2
 p^{\max_{(\ell,t)\in\mathcal G_{a,u}}
 E_{\rm fin}(m,h;a,u,\ell,t)-D_{\rm vis,fin}}\right\}.
 \label{eq:finite-probability-comparison-4}
\end{align}
If $\mathcal G_{a,u}$ is empty, the corresponding completion contribution is
zero.  The displayed inequality is fully finite.

\medskip\noindent\textbf{(2) Exponential form.}
For $m,h,k=O(n)$, the prefactor satisfies
$\log_p[C_p^{\mathrm{comp}}(m+1)^2]=O_p(\log n)$.
Substituting this estimate and the finite probability bound gives the
following exponential upper bound:
\begin{align}
 P_{\mathrm{id}}^{(k),*}(n,m;p)
 &\le\sum_{(a,u)\in\mathfrak P_n}
 p^{-(k-a)(n-m-a)-u(u+1)/2+O_p(n)}
 \notag\\[-1mm]
 &\quad\times
 p^{\min\left\{0,
 \max_{(\ell,t)\in\mathcal G_{a,u}}
 E_{\rm fin}(m,h;a,u,\ell,t)
 -D_{\rm vis,fin}+O_p(\log n)\right\}}.
 \label{eq:finite-probability-comparison-5}
\end{align}
For fixed $p$ and prescribed integer sequences with $m,h,k=O(n)$, the
remainder bounds are uniform over the admissible finite rank parameters.
The minimum with zero is the exact exponent form of the finite cap
$d_{\zeta}/\lvert Y\rvert\le1$ and is part of the conclusion.
\end{theorem}

The constants implicit in these remainder bounds may depend on the
fixed linear-growth bounds for the prescribed sequences $m,h,k$.
No common remainder bound over all $O(n)$ sequences with unrestricted
growth constants is asserted.

The condition $a+u\le k$ is a probability-side row-space constraint.  It
follows from $2u\le k$ and $a-u\le k-2u$ and does not alter the
completion domain $\mathcal G_{a,u}$.

\noindent\textit{Proof guide.} The complete proof is given in
Subsection~\ref{subsec:odd-finite-proof-assembly}.
The finite estimates from Appendix~\ref{app:odd-prime-finite-structural-proof}
enter in three steps.

First, Lemma~\ref{lem:visible-to-ambient-completion} includes the realizable
visible moments above a fixed lower moment in an ambient rank-constrained
completion set.  This is the only inclusion-based relaxation.

Second, Proposition~\ref{prop:exact-ambient-completion-count} counts the
enlarged set exactly.  Its cardinality depends on the lower block only
through $r_B=\rank B_\zeta=a-u$, and its completed rank is
\[
 r_B+2\ell+t.
\]
Here $\ell$ is the radical-coupling rank and $t$ is the rank of the
residual symmetric block.

Third, Proposition~\ref{prop:uniform-completion-bound} bounds that exact
count using $E_{\rm fin}$ and the stated completion ranges.  Combining
this bound with the weighted reduction and the normalized completion cap
gives the following intermediate inequalities:
\begin{align*}
 P_{\mathrm{id}}^{(k),*}(n,m;p)
 &\le\sum_{(a,u)\in\mathfrak P_n}
   \mathsf p_{n,m,k}^{\mathrm{odd}}(a,u)
   \min\left\{1,\frac{\mathcal A_{m,h,k}(a-u)}{|Y|}\right\}\\
 &\le\sum_{(a,u)\in\mathfrak P_n}
   \mathsf p_{n,m,k}^{\mathrm{odd}}(a,u)
   \min\left\{1,C_p^{\mathrm{comp}}(m+1)^2
       p^{\max_{(\ell,t)\in\mathcal G_{a,u}}E_{\rm fin}-D_{\rm vis,fin}}\right\}.
\end{align*}
Here $|Y|=p^{D_{\rm vis,fin}}$; if the completion domain is empty, its
contribution is zero.  Substituting the augmented row-space event bound
of Corollary~\ref{cor:all-prime-augmented-row-space-event-bound} gives
Eq.~\eqref{eq:finite-probability-comparison-4}; collecting uniform
finite constants on the logarithmic scale gives
Eq.~\eqref{eq:finite-probability-comparison-5}.
The cap is retained throughout.

\subsubsection{Characteristic-two finite structural estimate}
\label{sec:characteristic-two-converse-input}

By Theorem~\ref{thm:converse-weighted-probability-completion-reduction}, the
characteristic-two branch has already been reduced to
Eq.~\eqref{eq:main-binary-weighted-completion-interface}: a sum over $(a,u)$
of its probability times the clipped factor
$2^m\mathcal N_{m,n-m,k}^{(2)}(a-u)/\lvert Y_{n,m,2}\rvert$.  The finite structural
theorem expresses these two factors through the common exponent parameters.
The enhanced moment, associated bilinear form, and diagonal-lift data enter only in the
proof of this rank-grouped completion bound.

The next theorem eliminates the enhanced lower moments, their associated
bilinear forms, and the diagonal-lift data.  The resulting finite bound depends only on the two rank pairs and an
explicit characteristic-two correction.  A separate corollary below takes
the asymptotic limit.

The binary bound uses the same types of ranks, with a lower block of
dimension $n-m$ and an additional diagonal-lift term.  We define the
finite quantities before stating the bound.

Let $p=2$.  With the binary row-space probability defined in
Theorem~\ref{thm:converse-weighted-probability-completion-reduction}, let
\begin{equation}
 \mathfrak P_n^{(2)}:=\{(a,u)\in\mathbb Z_{\ge0}^2:
                       \mathsf p_{n,m,k}^{(2)}(a,u)>0\}.
 \label{eq:binary-positive-probability-support}
\end{equation}
Every pair in this support satisfies
\begin{equation}
 0\le u\le a\le k,\quad a+u\le k.
 \label{eq:theorem-C-positive-mass-domain}
\end{equation}
Define
\begin{equation}
 \mathcal G_{a,u}^{(2)}
 :=\left\{(\ell,t)\in\mathbb Z_{\ge0}^2:
 \begin{array}{l}
 0\le\ell\le\min\{m,n-m-a+u\},\\
 0\le t\le m-\ell,\\
 a-u+2\ell+t\le k
 \end{array}\right\}.
 \label{eq:main-binary-completion-ranges}
\end{equation}
Define
\begin{align}
 E_{\rm fin}^{(2)}(m,n;a,u,\ell,t)
 :={}&m+m(a-u)+\ell(n-a+u-\ell)
 +\ell m-\frac{\ell(\ell-1)}2\notag\\
 &+\frac{t[2(m-\ell)-t+1]}2.
 \label{eq:theorem-C-binary-E-fin}
\end{align}
Here $(a,u)$ is the probability-side row-space pair.  Conditional on it,
$(\ell,t)$ is the two ranks used in the completion count: $\ell$ is the radical-coupling rank
and $t$ the residual symmetric rank.  The leading term $m$ records the
separate enhanced-diagonal lift.  The probability side is kept separate by
writing
\begin{equation}
 I_{\rm hid,fin}^{(2)}(n,m,k;a,u)
 :=(k-a)(n-m-a)+\frac{u(u+1)}2.
 \label{eq:theorem-C-binary-hidden-exponent}
\end{equation}

\begin{theorem}[Characteristic-two finite structural estimate]
\label{thm:converse-binary-finite-structural-estimate}
Let $p=2$ and let $n,m,k$ be integers with $n\ge1$, $1\le m\le n$, and $k\ge0$.  Use the rank
parameters and exponents defined above.  Every positive-mass pair
$(a,u)\in\mathfrak P_n^{(2)}$ satisfies
Eq.~\eqref{eq:theorem-C-positive-mass-domain}.  The finite bound and
exact exponent comparison are as follows.

\medskip\noindent\textbf{(1) Finite bound.}
There is a constant $C_2^{\mathrm{comp}}>0$, independent of
$n,m,k,a,u,\ell,t$, such that the following fully finite inequality holds:
\begin{align}
 P_{\mathrm{id}}^{(k),*}(n,m;2)
 \le{}&\sum_{(a,u)\in\mathfrak P_n^{(2)}}
 (K_2^{\mathrm{class}})^k
 2^{-I_{\rm hid,fin}^{(2)}(n,m,k;a,u)}\notag\\
 &\quad\times
 \min\left\{1,
 C_2^{\mathrm{comp}}(m+1)^2
 2^{\max_{(\ell,t)\in\mathcal G_{a,u}^{(2)}}
 E_{\rm fin}^{(2)}(m,n;a,u,\ell,t)
 -D_{\rm vis,fin}}\right\}.
 \label{eq:main-binary-finite-exponent-bound}
\end{align}
The factor $(K_2^{\mathrm{class}})^k$ is the explicit row-space counting
constant from Eq.~\eqref{eq:binary-lower-block-probability-finite}; the factor
$C_2^{\mathrm{comp}}(m+1)^2$ collects the two dimension-independent rank-count
constants and the number of admissible completion-rank pairs.  The exponent in
the completion bound includes the separate enhanced-diagonal lift factor
$2^m$; this is the leading term $m$ in
Eq.~\eqref{eq:theorem-C-binary-E-fin}.  The minimum with one is the exact
finite cap on the normalized occurring-coordinate count.

\medskip\noindent\textbf{(2) Exact exponent comparison.}
With $h:=n-m+1$, every admissible integer quadruple satisfies the
exact finite identity
\begin{equation}
 E_{\rm fin}^{(2)}(m,n;a,u,\ell,t)
 =E_{\rm fin}(m,h;a,u,\ell,t)+(m-\ell).
 \label{eq:main-binary-finite-exponent-identity}
\end{equation}
The normalized difference satisfies the bound
\begin{equation}
 0\le\frac{m-\ell}{n^2}\le\frac1n,
 \label{eq:theorem-C-normalized-correction-bound}
\end{equation}
so the normalized characteristic-two correction vanishes uniformly.
\end{theorem}

The row-space condition $a+u\le k$ follows from $2u\le k$ and
$a-u\le k-2u$ and is independent of the binary completion-space dimension.
The normalized bound in Eq.~\eqref{eq:theorem-C-normalized-correction-bound}
uses only $0\le\ell\le m\le n$.

\noindent\textit{Proof guide.} The complete proof is given in
Subsection~\ref{subsec:binary-finite-proof-assembly}.
The characteristic-two finite bounds from
Appendix~\ref{app:binary-finite-structural-proof} enter in three steps.

First, Proposition~\ref{prop:binary-enhanced-bilinear-comparison} associates
a bilinear completion problem with each fixed lower enhanced moment.  At
most $2^m$ enhanced diagonal lifts correspond to a given bilinear completion.

Second, Lemmas~\ref{lem:binary-Schur-reduction} and
\ref{lem:binary-radical-rank-identity}, assembled in
Lemma~\ref{lem:integrated-binary-completion-count}, give the completed rank
\[
 a-u+2\ell+t.
\]
The ranks $\ell$ and $t$ have the same roles as in the odd-prime argument.

Third, the finite support count satisfies
$d_\zeta^{(2)}\le\min\{|Y_{n,m,2}|,
2^m\mathcal N_{m,n-m,k}^{(2)}(a-u)\}$.
The binary weighted reduction then gives:
\begin{align*}
 P_{\mathrm{id}}^{(k),*}(n,m;2)
 &\le\sum_{(a,u)\in\mathfrak P_n^{(2)}}
 \mathsf p_{n,m,k}^{(2)}(a,u)
 \min\left\{1,\frac{2^m\mathcal N_{m,n-m,k}^{(2)}(a-u)}{|Y_{n,m,2}|}\right\}\\
 &\le\sum_{(a,u)\in\mathfrak P_n^{(2)}}
 \mathsf p_{n,m,k}^{(2)}(a,u)
 \min\left\{1,C_2^{\mathrm{comp}}(m+1)^2
 2^{\max_{(\ell,t)\in\mathcal G_{a,u}^{(2)}}E_{\rm fin}^{(2)}-D_{\rm vis,fin}}\right\}.
\end{align*}
The $2^m$ multiplier lies outside the exact bilinear completion count;
it is the leading $m$ term of $E_{\rm fin}^{(2)}$ and is not counted again.
Applying the binary row-space probability bound now yields
Eq.~\eqref{eq:main-binary-finite-exponent-bound} with its finite cap intact.
The corresponding exponents satisfy
\[
 E_{\rm fin}^{(2)}(m,n;a,u,\ell,t)
 =E_{\rm fin}(m,n-m+1;a,u,\ell,t)+(m-\ell).
\]
Because $0\le m-\ell\le n$, the normalized correction is uniformly
$O(n^{-1})$, which proves the final assertion.

\begin{corollary}[Asymptotic form of the characteristic-two structural estimate]
\label{cor:binary-structural-asymptotic-form}
Let $(m_n)$ and $(k_n)$ be integer sequences with $m_n,k_n=O(n)$.
Assume that $1\le m_n\le n$ and $k_n\ge0$ for all sufficiently large $n$,
so that the finite bound in Eq.~\eqref{eq:main-binary-finite-exponent-bound}
applies.  Then that bound implies
\begin{align}
 P_{\mathrm{id}}^{(k_n),*}(n,m_n;2)
 \le{}&\sum_{(a,u)\in\mathfrak P_n^{(2)}}
 2^{-I_{\rm hid,fin}^{(2)}(n,m_n,k_n;a,u)+o(n^2)}\notag\\
 &\quad\times
 \min\left\{1,
 2^{\max_{(\ell,t)\in\mathcal G_{a,u}^{(2)}}
 E_{\rm fin}^{(2)}(m_n,n;a,u,\ell,t)
 -D_{\rm vis,fin}+O(\log n)}\right\},
 \label{eq:binary-structural-asymptotic-corollary}
\end{align}
where the $o(n^2)$ term is uniform over $(a,u)$ and the $O(\log n)$ term is
uniform over the completion indices.
\end{corollary}

\begin{proof}
Because $k_n=O(n)$ and $K_2^{\mathrm{class}}$ is fixed,
$k_n\log_2K_2^{\mathrm{class}}=O(n)=o(n^2)$.  Likewise,
$\log_2(C_2^{\mathrm{comp}}(m_n+1)^2)=O(\log n)$.  Substitution into the
fully finite inequality gives the stated form, uniformly over the finite
parameter ranges.
\end{proof}

The odd-prime and binary bounds now both have the probability pair
$(a,u)$, the completion pair $(\ell,t)$, and the rank budget
$a-u+2\ell+t\le k$.  The characteristic-specific ceiling on $\ell$
still differs by the single augmentation row, and the normalized
completion count is capped at one before summation.  The exact exponent
identity~\eqref{eq:main-binary-finite-exponent-identity} compares the two
branches after the binary diagonal-lift factor has been included: their
finite completion exponents differ by $m-\ell$.  Since
$0\le m-\ell\le n$, this is uniformly $O(n)=o(n^2)$.
Thus the next proof needs the normalized ranks and their constraints,
not the individual completion matrices.

\subsection{Asymptotic exponent}
\label{subsec:common-exponent-functions}
\label{subsubsec:all-prime-exponent-reduction}
We now pass from the finite bounds to a common quadratic exponent.
After normalizing $(a,u,\ell,t)$ by $n$, the probability cost will be
written as $I$ and the completion exponent minus the label exponent as
$E-D$.  The factor $\min\{1,p^{(E-D)n^2+O(n)}\}$ requires us to retain
the clipping until the per-summand bound has been established.
The proof first bounds $I+D-E$ in terms of two unused-dimension
variables, then minimizes a one-variable expression.  The transition
between the two resulting branches occurs at $\alpha=1-\beta$.

\begin{theorem}[All-prime converse exponent bound]
\label{thm:converse-common-asymptotic-reduction}
Fix a prime $p$ and suppose
\begin{equation}
 m=\beta n+o(n),\quad k=\alpha n+o(n),
 \quad 0<\beta\le1,\quad 0\le\alpha<1.
 \label{eq:direct-finite-scaling}
\end{equation}
Under this scaling, the finite converse bound yields the following asymptotic upper bound on the normalized logarithm of the optimal exact-identification probability:
\begin{equation}
 \limsup_{n\to\infty}\frac1{n^2}
 \log_p P_{\mathrm{id}}^{(k),*}(n,m;p)
 \le -J(\alpha,\beta).
 \label{eq:direct-finite-limsup}
\end{equation}
The exponent $J(\alpha,\beta)$ on the right-hand side is defined by
\begin{equation}
 J(\alpha,\beta)=
 \begin{cases}
 \beta(1-\alpha-\beta/2),&0\le\alpha\le1-\beta,\\[1mm]
 (1-\alpha)^2/2,&1-\beta\le\alpha<1.
 \end{cases}
 \label{eq:direct-finite-J}
\end{equation}
The logarithmic bound is equivalently expressed as the following probability estimate:
\begin{equation}
 P_{\mathrm{id}}^{(k),*}(n,m;p)
 \le p^{-J(\alpha,\beta)n^2+o(n^2)}\longrightarrow0.
 \label{eq:quadratic-converse}
\end{equation}
\end{theorem}

\begin{proof}
\noindent\textbf{Step 1: Finite estimates and normalized variables.}
Our first goal is to identify a common normalized probability cost $I$,
completion exponent $E$, and label exponent $D$, with errors uniform over
the finite rank domains.  We apply Theorem~\ref{thm:converse-odd-prime-finite-structural-estimate}
when $p$ is odd and
Theorem~\ref{thm:converse-binary-finite-structural-estimate} when $p=2$.
The two bounds have the same normalized probability and completion variables.
For the characteristic-two branch, Eq.~\eqref{eq:main-binary-finite-exponent-identity}
gives the exact finite relation
\begin{equation}
 E_{\rm fin}^{(2)}(m,n;a,u,\ell,t)
 =E_{\rm fin}(m,n-m+1;a,u,\ell,t)+(m-\ell).
 \label{eq:direct-finite-binary-odd-exponent-relation}
\end{equation}
The additional correction is therefore $O(n)=o(n^2)$.  Thus both finite structural theorems reduce, uniformly over their complete
finite feasible domains, to the common normalized exponent treated below.

To avoid reusing the matrix symbols $A$ and $B$, denote the two
normalized scalar ratios by
\begin{equation}
 \alpha_n:=\frac{k}{n},\quad \beta_n:=\frac{m}{n},\quad c:=1-\alpha_n.
 \label{eq:direct-finite-ABc}
\end{equation}
For an admissible quadruple $(a,u,\ell,t)$ in either characteristic branch,
write
\begin{equation}
 a_0:=\frac an,\quad u_0:=\frac un,
 \quad \lambda:=\frac\ell n,\quad \tau:=\frac tn.
 \label{eq:direct-finite-normalized-variables}
\end{equation}
The finite positive-mass and completion constraints imply, uniformly in both
characteristic branches, the following normalized feasibility conditions:
\begin{equation}
\begin{aligned}
 0&\le u_0\le\min\{a_0,\alpha_n-a_0\},\qquad
 0\le a_0\le\alpha_n,\\
 0&\le\lambda\le\min\{\beta_n,1-\beta_n-a_0+u_0\}+O(n^{-1}),\\
 0&\le\tau\le\beta_n-\lambda,\qquad
 a_0-u_0+2\lambda+\tau\le\alpha_n.
\end{aligned}
 \label{eq:direct-finite-normalized-feasibility}
\end{equation}
The condition $u_0\le\alpha_n-a_0$ comes from $a+u\le k$.
The upper bound on $\lambda$ comes from $\ell\le h-a+u$ in odd
characteristic and $\ell\le n-m-a+u$ in characteristic two; its
$O(n^{-1})$ correction accounts only for $h=n-m+1$.
The final line is the normalization of the two completion conditions
$0\le t\le m-\ell$ and $a-u+2\ell+t\le k$.  Define
\begin{align}
 I&:=(\alpha_n-a_0)(1-\beta_n-a_0)+\frac{u_0^2}{2},
 \label{eq:direct-finite-I}\\
 E&:=\beta_n(a_0-u_0)
 +\lambda(1-a_0+u_0-\lambda)
 +\beta_n\lambda-\frac{\lambda^2}{2}
 +\tau(\beta_n-\lambda)-\frac{\tau^2}{2},
 \label{eq:direct-finite-E}\\
 D&:=\beta_n-\frac{\beta_n^2}{2}.
 \label{eq:direct-finite-D}
\end{align}
Uniformly over all admissible quadruples, division of the finite exponents by
$n^2$ gives
\begin{align}
 \frac{(k-a)(n-m-a)+u(u+1)/2}{n^2}
 &=I+O(n^{-1}),
 \label{eq:direct-finite-hidden-normalization}\\
 \frac{E_{\rm fin}-D_{\rm vis,fin}}{n^2}
 &=E-D+O(n^{-1}).
 \label{eq:direct-finite-visible-normalization}
\end{align}
For odd $p$, $E_{\rm fin}$ is the exponent in
Eq.~\eqref{eq:theorem-B-E-fin}.  For $p=2$, the binary exponent differs from
this common expression by $m-\ell$, hence again by only $O(n)$.
The exact finite-to-normalized completion errors are stated in
Eq.~\eqref{eq:endpoint-clipping-exact-finite-errors}:
their absolute values are bounded by $3/(2n)$ for odd $p$ and
$1/(2n)$ for $p=2$, uniformly over the finite feasible domains.
The terms $u/2$, $\ell/2$, and $t/2$, the odd-prime relation $h=n-m+1$,
and the remaining prefactors in the two finite structural theorems
contribute only uniform $O_p(n)$ terms to the unnormalized exponents
for the fixed linearly growing integer sequences.

\medskip
\noindent\textbf{Step 2: Clipping and the visible deficit.}
For each admissible rank quadruple, we first retain the clipping and
then seek a uniform lower bound on $I+D-E$.  Combining the probability
exponent with the clipped normalized completion factor, every
summand in either finite structural theorem has normalized exponent
\begin{equation}
 -I+\min\{0,E-D\}+O_p(n^{-1}).
 \label{eq:direct-finite-normalized-clipped-summand}
\end{equation}
For the characteristic-two branch, the term $(m-\ell)/n^2$ has already
been absorbed into the displayed $O_p(n^{-1})$ remainder.  Thus the clipping
and all subsequent sign calculations apply to the same common exponent in
both characteristics.  We shall use the exact identity
\begin{equation}
 -I+\min\{0,E-D\}=-\max\{I,I+D-E\}.
 \label{eq:direct-finite-clipping-max-identity}
\end{equation}
Consequently, a uniform lower bound on $I+D-E$ already gives the required
upper bound on every clipped summand; the equality points below additionally
show that this bound is sharp for the limiting variational problem.
The remaining proof has two algebraic steps.  The variable $x$ is the
normalized nonradical hidden rank.  The slack $r$ is the unused ambient
dimension after allocating $x$ and the coupling rank $\lambda$, while $s$ is
the unused queried dimension after allocating the coupling and residual
ranks.  The first identity below controls the visible deficit $D-E$.  Put
\begin{equation}
 x:=a_0-u_0,
 \quad r:=1-\beta_n-x-\lambda,
 \quad s:=\beta_n-\lambda-\tau.
 \label{eq:direct-finite-slacks}
\end{equation}
The finite completion ranges imply
\begin{equation}
 r\ge-O(n^{-1}),
 \quad 0\le s\le \beta_n,
 \quad r+s\ge c-O(n^{-1}).
 \label{eq:direct-finite-slack-ranges}
\end{equation}
Indeed, the first inequality is the normalized coupling-rank bound; the
odd-prime endpoint has only the additional $1/n$ coming from $h=n-m+1$.
The last inequality follows from
\begin{equation}
 x+2\lambda+\tau=1-r-s,
 \label{eq:direct-finite-rank-identity}
\end{equation}
and the finite rank budget $a-u+2\ell+t\le k$.

We will first prove the identity $D-E=r(\beta_n-\lambda)+s^2/2$.
It expresses the visible deficit using the unused ambient dimension $r$
and the unused queried dimension $s$.  Expanding the existing exponent
gives this identity.  Substitute
$x=1-\beta_n-\lambda-r$ and $\tau=\beta_n-\lambda-s$ into
Eq.~\eqref{eq:direct-finite-E}.  Since $a_0-u_0=x$, the terms containing
$\tau$ satisfy
\begin{align}
 \tau(\beta_n-\lambda)-\frac{\tau^2}{2}
 &=\frac{(\beta_n-\lambda)^2}{2}-\frac{s^2}{2},
 \label{eq:direct-finite-visible-expansion-one}
\end{align}
and the remaining terms give
\begin{align}
 E
 &=\beta_n(1-\beta_n-\lambda-r)
   +\lambda(\beta_n+r)+\beta_n\lambda-\frac{\lambda^2}{2}
   +\frac{(\beta_n-\lambda)^2}{2}-\frac{s^2}{2}\notag\\
 &=\beta_n-\frac{\beta_n^2}{2}-r(\beta_n-\lambda)-\frac{s^2}{2}.
 \label{eq:direct-finite-visible-expansion-two}
\end{align}
In the first line, the factor $\lambda(\beta_n+r)$ is
$\lambda(1-x-\lambda)$ after substituting
$x=1-\beta_n-\lambda-r$.  Expanding the remaining square then cancels all
$\lambda^2$ terms, yielding the second line.  Using $D=\beta_n-\beta_n^2/2$ proves the
exact identity
\begin{equation}
 D-E=r(\beta_n-\lambda)+\frac{s^2}{2}.
 \label{eq:direct-finite-visible-deficit}
\end{equation}
The sign estimate is immediate from the finite ranges:
$r\ge-O(n^{-1})$ by Eq.~\eqref{eq:direct-finite-slack-ranges},
$\beta_n-\lambda\ge0$ by Eq.~\eqref{eq:direct-finite-normalized-feasibility}, and
$s^2/2\ge0$.  Hence $D-E\ge-O(n^{-1})$.  More explicitly, if $E-D\le0$, the
clipped term equals $E-D$; if $E-D>0$, then $E-D=-(D-E)\le O(n^{-1})$, so the clipping
replaces a positive term of only order $O(n^{-1})$ by zero.  Hence the finite clipping in
Eq.~\eqref{eq:direct-finite-normalized-clipped-summand} changes the normalized
exponent by at most $O(n^{-1})$ relative to using $E-D$.  It remains to
lower-bound $I+D-E$.

\medskip
\noindent\textbf{Step 3: The hidden-rank lower bound.}
Our goal is $I\ge r\lambda-O(n^{-1})$, which together with
Eq.~\eqref{eq:direct-finite-visible-deficit} will bound $I+D-E$.
To obtain this estimate, we factor $I-r\lambda$ using two more
normalized variables.  Define
\begin{equation}
 v:=u_0-\lambda,
 \quad \delta:=\beta_n-c-\lambda+r.
 \label{eq:direct-finite-hidden-slacks}
\end{equation}
The row-space and completion constraints give lower bounds on
$r-v$ and $\delta-v$, up to finite endpoint errors.  We establish
these bounds before deriving and applying the factorization of
$I-r\lambda$.  First,
\begin{equation}
 r-v=1-\beta_n-a_0.
 \label{eq:direct-finite-r-minus-v}
\end{equation}
The dimension bound depends on the characteristic as follows.  In the
odd-prime branch, the realized space $\mathcal W(X_2)=\operatorname{row}Z_2$
is generated by the $h=n-m+1$ rows of $Z_2$, so $a\le n-m+1$ and therefore
$r-v\ge-1/n$.  In the characteristic-two branch, the lower row space is generated by
$n-m$ rows, so $a\le n-m$ and $r-v\ge0$.  Thus uniformly in both branches,
$r-v\ge-O(n^{-1})$.  Second,
\begin{equation}
 \delta-2v=\alpha_n-a_0-u_0,
 \label{eq:direct-finite-delta-minus-two-v}
\end{equation}
which is nonnegative by $a+u\le k$.  Finally,
$r+s\ge c-O(n^{-1})$ and $s\le \beta_n-\lambda$ imply
$\delta\ge-O(n^{-1})$.

If $v\ge0$, Eq.~\eqref{eq:direct-finite-delta-minus-two-v} gives
$\delta\ge2v-O(n^{-1})$, and hence $\delta-v\ge-O(n^{-1})$.  If
$v<0$, the bound $\delta\ge-O(n^{-1})$ directly gives
$\delta-v\ge-O(n^{-1})$.  Thus $r-v$ and $\delta-v$ are both nonnegative up to uniform
$O(n^{-1})$ endpoint errors.  We now derive the factorization.
From
$a_0=x+u_0$, $x=1-\beta_n-\lambda-r$, and $u_0=\lambda+v$, one obtains
\begin{equation}
 a_0=1-\beta_n-r+v,
 \quad 1-\beta_n-a_0=r-v.
 \label{eq:direct-finite-hidden-expansion-one}
\end{equation}
Moreover, $c=1-\alpha_n$ and
$\delta=\beta_n-c-\lambda+r$ imply
\begin{equation}
 \alpha_n-a_0=\delta+\lambda-v.
 \label{eq:direct-finite-hidden-expansion-two}
\end{equation}
Substitution in Eq.~\eqref{eq:direct-finite-I} therefore gives
\begin{align}
& I-r\lambda
 =(\delta+\lambda-v)(r-v)
   +\frac{(\lambda+v)^2}{2}-r\lambda\notag\\
 =&(r-v)(\delta-v)+\frac{\lambda^2}{2}
   +\lambda\{(r-v)+v-r\}+\frac{v^2}{2}
 =(r-v)(\delta-v)+\frac{\lambda^2}{2}+\frac{v^2}{2}.
 \label{eq:direct-finite-hidden-expansion-three}
\end{align}
Thus the linear terms in $\lambda$ cancel exactly, and we obtain the
factorization
\begin{equation}
 I-r\lambda
 =\frac{\lambda^2}{2}
 +(r-v)(\delta-v)+\frac{v^2}{2}.
 \label{eq:direct-finite-hidden-gap}
\end{equation}
All normalized variables are uniformly bounded, so the endpoint errors in the
product contribute only $O(n^{-1})$.  The normalized hidden-rank
exponent consequently satisfies
\begin{equation}
 I\ge r\lambda-O(n^{-1}).
 \label{eq:direct-finite-hidden-bound}
\end{equation}
No condition $\ell+t\ge u$ is assumed, and no inequality
$k-a-\ell\ge0$ is inferred from the completion rank budget.
Combining Eqs.~\eqref{eq:direct-finite-visible-deficit} and
\eqref{eq:direct-finite-hidden-bound} gives
$I+D-E\ge r\lambda+r(\beta_n-\lambda)+s^2/2-O(n^{-1})$.
The two terms containing $\lambda$ cancel, so
\begin{equation}
 I+D-E
 \ge r\beta_n+\frac{s^2}{2}-O(n^{-1}).
 \label{eq:direct-finite-rs-bound}
\end{equation}
The last inequality in Eq.~\eqref{eq:direct-finite-slack-ranges} gives
$r\ge(c-s)_+-O(n^{-1})$, and therefore
\begin{equation}
 I+D-E
 \ge \beta_n(c-s)_++\frac{s^2}{2}-O(n^{-1}),
 \quad 0\le s\le \beta_n.
 \label{eq:direct-finite-one-variable}
\end{equation}

\medskip
\noindent\textbf{Step 4: The one-variable minimization.}
Equation~\eqref{eq:direct-finite-one-variable} reduces the bound to
minimizing $\beta_n(c-s)_++s^2/2$ over $0\le s\le\beta_n$.
The position of $c=1-\alpha_n$ relative to $\beta_n$ determines the
minimizer and hence the two branches of $J$.  If $c\ge \beta_n$, equivalently $\alpha_n\le1-\beta_n$, the function on the right is
$\beta_n(c-s)+s^2/2$ and is minimized on $[0,\beta_n]$ at $s=\beta_n$.
This minimum yields the following exponent lower bound:
\begin{equation}
 I+D-E\ge \beta_n\left(1-\alpha_n-\frac{\beta_n}{2}\right)-O(n^{-1}).
 \label{eq:direct-finite-case-one}
\end{equation}
If $c\le \beta_n$, equivalently $\alpha_n\ge1-\beta_n$, the same function is minimized at
$s=c$, giving
\begin{equation}
 I+D-E\ge\frac{(1-\alpha_n)^2}{2}-O(n^{-1}).
 \label{eq:direct-finite-case-two}
\end{equation}
At $\alpha_n=1-\beta_n$ the two branch values both equal $\beta_n^2/2$; at $\beta_n=1$ the
second branch is $(1-\alpha_n)^2/2$.  Hence the exponent is continuous across the
transition and positive for every $\alpha_n<1$ and $\beta_n>0$.
The corresponding limiting equality points are
\begin{equation}
 (a_0,u_0,\lambda,\tau)=
 \begin{cases}
  (\alpha_n,0,0,0),&0\le \alpha_n\le1-\beta_n,\\[1mm]
  (1-\beta_n,0,0,\alpha_n+\beta_n-1),&1-\beta_n\le \alpha_n<1.
 \end{cases}.
 \label{eq:direct-finite-limiting-equality-points}
\end{equation}
The normalized feasibility conditions in
Eq.~\eqref{eq:direct-finite-normalized-feasibility} can be checked at both
points.  In the first branch, $a_0=\alpha_n\le1-\beta_n$ and
$u_0=\lambda=\tau=0$; thus both $a_0+u_0$ and
$a_0-u_0+2\lambda+\tau$ equal $\alpha_n$, while the completion ceiling is
$0\le1-\beta_n-a_0$.  In the second branch,
$a_0=1-\beta_n\le\alpha_n$, $u_0=\lambda=0$, and
$\tau=\alpha_n+\beta_n-1$ satisfies
$0\le\tau\le\beta_n$ for sufficiently large $n$, since $\alpha_n<1$;
moreover, $a_0-u_0+2\lambda+\tau=\alpha_n$ and the upper completion
ceiling for $\lambda$ is zero.  Hence both points satisfy the limiting
rank restrictions, including the positive-mass condition
$a_0+u_0\le\alpha_n$.
At either point, $I=0$ and $E-D=-J(\alpha_n,\beta_n)$, so
Eq.~\eqref{eq:direct-finite-clipping-max-identity} gives the limiting clipped
exponent $-J(\alpha_n,\beta_n)$.  They also make the hidden-gap,
visible-deficit, and one-variable bounds equalities.
These are algebraic equality points of the limiting variational problem;
no exact realization by finite integer rank parameters is claimed.  Thus the limiting common objective has
maximum exactly $-J(\alpha_n,\beta_n)$, while the finite theorem retains its
$O(n^{-1})$ endpoint corrections and clipping until compactness passage.
In either case, every admissible summand in the finite structural estimate is
bounded by
\begin{equation}
 p^{-J(\alpha_n,\beta_n)n^2+O_p(n)}.
 \label{eq:direct-finite-summand-bound}
\end{equation}

\medskip
\noindent\textbf{Step 5: Summation and the converse probability bound.}
We pass from the bound for each rank quadruple to the probability bound
by summing over positive-mass $(a,u)$ pairs.  The sum contributes at most
a polynomial factor, which does not alter the quadratic exponent.
Write $\mathfrak P_{n,p}$ for the positive-mass row-space pairs
$\mathfrak P_n$ when $p$ is odd and $\mathfrak P_n^{(2)}$ when $p=2$.
For each such pair, the completion variables $(\ell,t)$ have already
been maximized in the corresponding finite structural estimate.
The uniform summand bound
Eq.~\eqref{eq:direct-finite-summand-bound} therefore gives the following
finite-sum comparison:
\begin{align}
 &P_{\mathrm{id}}^{(k),*}(n,m;p)
 \le\sum_{(a,u)\in\mathfrak P_{n,p}}
     p^{-J(\alpha_n,\beta_n)n^2+O_p(n)}\notag\\
 \le& |\mathfrak P_{n,p}|\,
     p^{-J(\alpha_n,\beta_n)n^2+O_p(n)}
 \le (k+1)^2 p^{-J(k/n,m/n)n^2+O_p(n)}.
 \label{eq:direct-finite-final-bound}
\end{align}
The same $O_p(n)$ bound applies to every summand for the fixed prescribed
sequences, and $|\mathfrak P_{n,p}|\le(k+1)^2$ because $0\le u\le a\le k$.
Taking base-$p$ logarithms, dividing by $n^2$, and using
$k/n\to\alpha$, $m/n\to\beta$, and continuity of $J$ proves
Eq.~\eqref{eq:direct-finite-limsup}.
The argument therefore applies to both the odd-prime and binary finite
structural bounds, proving the stated exponent estimate for every fixed prime
$p$.
\end{proof}

Equation~\eqref{eq:quadratic-converse} proves the exact converse stated
in Theorem~\ref{thm:refined-exact-converse}.  Section~\ref{sec:integrated-all-prime-graded}
combines this converse with the bounded-rank exact-label-list reduction to
prove the stronger verification conclusion
\begin{equation}
 S_{\mathrm{ver}}^{(k),*}(n,m;p)\longrightarrow0.
 \label{eq:graded-conclusion-forward-reference}
\end{equation}

\subsection{Exact converse for the refined frame task}
\label{subsec:main-results-operational-tasks}
The coordinate exact bound applies to each fixed frame because
Proposition~\ref{prop:frame-conditioned-coordinate-reduction} transports
the complete decision classes and preserves the uniform prior.

\begin{theorem}[Converse for exact learning in the refined problem]
\label{thm:refined-exact-converse}
Fix a prime $p$, $0<\beta\le1$, and $0\le\alpha<1$.  Let $(m_n)$ and
$(k_n)$ be integer sequences satisfying
$m_n=\beta n+o(n)$ and $k_n=\alpha n+o(n)$, and let
$R_n\in\mathsf{Ref}_{n,p}$ be any sequence of public frames.
Under these assumptions, the refined exact-learning optima satisfy:
\begin{align}
 P_{\rm frame,wc}^{(k_n),*}(R_n,m_n;p)
 \le P_{\rm frame,av}^{(k_n),*}(R_n,m_n;p)
=P_{\mathrm{id}}^{(k_n),*}(n,m_n;p)
 \le p^{-J(\alpha,\beta)n^2+o(n^2)}\longrightarrow0.
 \label{eq:main-results-converse}
\end{align}
For each fixed pair of integer sequences $(m_n)$ and $(k_n)$ as above,
the $o(n^2)$ term in the exponent can be chosen independently of the
public-frame sequence $(R_n)$.
\end{theorem}
\begin{proof}
For the sequences in the theorem, the exact-success optima satisfy
\begin{align*}
 P_{\rm frame,wc}^{(k_n),*}(R_n,m_n;p)
 \stackrel{(a)}{\le}P_{\rm frame,av}^{(k_n),*}(R_n,m_n;p)
 \stackrel{(b)}{=}P_{\mathrm{id}}^{(k_n),*}(n,m_n;p)
 \stackrel{(c)}{\le}p^{-J(\alpha,\beta)n^2+o(n^2)}.
\end{align*}
Here (a) is the worst-case-to-average inequality, (b) is
Eq.~\eqref{eq:finite-frame-coordinate-equivalence}, and (c) follows
from the all-prime converse bound
Theorem~\ref{thm:converse-common-asymptotic-reduction}, specifically
Eq.~\eqref{eq:quadratic-converse}.  The expression on the last line
tends to zero because $J(\alpha,\beta)>0$.  The remainder in (c)
depends on the fixed integer sequences but not on $R_n$, by (b).
\end{proof}

\section{Verification bounds from compatible exact-target lists}
\label{sec:integrated-all-prime-graded}
Vanishing exact-identification probability does not by itself make the
verification score vanish: an incorrect report may still have positive
Born acceptance probability.  The elementary exact-versus-verification
comparison leaves a $p^{-1}$ remainder.  We instead place the true exact
target in a bounded list whenever a report is compatible and has small
rank distance; outside this event, the Born score is small.

The conversion is common to two applications but their alphabets are
not interchangeable.  First we prove a list-to-exact lemma allowing an
arbitrary report alphabet.  We apply it to the original generic task,
whose legal reports belong to $\mathsf{Out}_{W,m}$ and whose exact target
is $y_W(x)$.  This provides the converse used for the full original
task in Theorem~\ref{thm:fixed-W-no-discount}.  We then apply the same
principle to standardized fixed-frame labels: $(C,\gamma)$ for odd $p$
and enhanced quotient labels for $p=2$.  Their compatible-list bounds
also enter Section~\ref{sec:fixed-diagonal-graph-basis-submodels}.

\subsection{Common list-to-exact conversions}
\label{subsubsec:generic-original-bayes-list-conversion}

A report supplies a finite list of possible exact targets.  Randomizing
over that list converts the given measurement into an exact decoder.
We first allow distinct target and report alphabets and an arbitrary
input prior, then give a form that retains an explicit failure outcome.

\subsubsection{Bayes conversion with distinct target and report alphabets}
\label{subsubsec:bayes-distinct-alphabets}
The target map assigns one exact label to each input state.  The reported
label may belong to a different alphabet, provided it determines a list
of possible exact labels with a common size bound.

\begin{lemma}[Bayes conversion under a uniform exact-target-list bound]
\label{lem:generic-bayes-bounded-list-conversion}
Let $x\mapsto\rho_x$ be the input-state ensemble, let $X$ have an
arbitrary prior $\pi$, and let $x\mapsto y(x)$ be a single-valued exact
target in an alphabet $\mathcal Y$.  The learner receives $k$ copies
of $\rho_X$ and reports in an alphabet $\widehat{\mathcal Y}$.
Suppose that every report $\widehat y$ is assigned a finite list
$\mathcal L_D(\widehat y)\subseteq\mathcal Y$ satisfying
$|\mathcal L_D(\widehat y)|\le L_D$ for a finite bound $L_D\ge0$.  For a learner POVM
$\{M_{\widehat y}\}$, let
\begin{equation}
 \mathsf A_D:=\{y(X)\in\mathcal L_D(\widehat Y)\}.
 \label{eq:generic-bayes-list-event}
\end{equation}
The probability that the exact target belongs to the reported list is bounded by the list size times the optimal exact-identification success probability:
\begin{equation}
 \Pr_{\mathsf M}(\mathsf A_D)\le L_D P_{\mathrm{id}}^{(k),*},
 \label{eq:generic-bayes-list-event-bound}
\end{equation}
where the exact optimum uses the same input ensemble and target map.
\end{lemma}
\begin{proof}
If $L_D=0$, every reported list is empty.  The event $\mathsf A_D$ is
then empty, so its probability is zero and
Eq.~\eqref{eq:generic-bayes-list-event-bound} holds.  For the rest of the
proof, assume $L_D>0$.
Define a sub-POVM on $\mathcal Y$ by
\begin{equation}
 N_y:=\frac1{L_D}\sum_{\widehat y:\,
 y\in\mathcal L_D(\widehat y)}M_{\widehat y}.
 \label{eq:generic-bayes-list-subpovm}
\end{equation}
The effects of this sub-POVM sum to an operator bounded above by the identity:
\begin{equation}
 \sum_yN_y
 =\frac1{L_D}\sum_{\widehat y}
 |\mathcal L_D(\widehat y)|M_{\widehat y}\le I.
 \label{eq:generic-bayes-list-subnormalization}
\end{equation}
Choose any fixed legal label $y_0\in\mathcal Y$ and add the residual
positive operator $I-\sum_yN_y$ to $N_{y_0}$.  This produces a POVM on the
declared exact-label alphabet and cannot decrease its exact Bayes success.
Before this residual contribution is counted, the success already equals
\begin{align}
 \sum_x\pi(x)\Tr[N_{y(x)}\rho_x^{\otimes k}]
 =\frac1{L_D}\sum_x\pi(x)
 \sum_{\widehat y:\,y(x)\in\mathcal L_D(\widehat y)}
 \Tr[M_{\widehat y}\rho_x^{\otimes k}]
 =\frac1{L_D}\Pr_{\mathsf M}(\mathsf A_D).
 \label{eq:generic-bayes-list-success}
\end{align}
Optimization gives Eq.~\eqref{eq:generic-bayes-list-event-bound}.  The argument
is performed on the input prior and does not require the induced prior on
$\mathcal Y$ to be uniform or the target map to be injective.
\end{proof}

\subsubsection{A size bound and a failure outcome}
The following specialization of
Lemma~\ref{lem:generic-bayes-bounded-list-conversion} uses the same alphabet
for true and reported labels and a uniform true-label prior.  It retains
an explicit failure outcome and will also be used for the fixed-diagonal
submodels.

\begin{lemma}[POVM conversion under a uniform list-size upper bound]
\label{lem:bounded-list-povm-conversion}
Let $\mathsf Y'$ be a finite label set, let
$\mathsf M=\{M_{\widehat y}:\widehat y\in\mathsf Y'\}$ be a POVM, and assign
to every reported label $\widehat y$ a list
$\mathcal L(\widehat y)\subseteq\mathsf Y'$.  Suppose that
\begin{equation}
 |\mathcal L(\widehat y)|\le L
 \quad(\widehat y\in\mathsf Y')
 \label{eq:bounded-list-size-assumption}
\end{equation}
for an integer $L\ge1$.  Define
\begin{equation}
 N_y:=\frac1L
 \sum_{\widehat y:\,y\in\mathcal L(\widehat y)}M_{\widehat y},
 \quad y\in\mathsf Y'.
 \label{eq:bounded-list-subpovm}
\end{equation}
Then $N_y\ge0$ and $\sum_{y\in\mathsf Y'}N_y\le I$.  Consequently, the
failure operator
\begin{equation}
 N_\bot:=I-\sum_{y\in\mathsf Y'}N_y
 \label{eq:bounded-list-failure-operator}
\end{equation}
is positive, and adjoining it gives a POVM.  For a uniform true label $Y$
and the report $\widehat Y$ produced by $\mathsf M$, the probability that
this POVM identifies the true label through one of the outcomes in
$\mathsf Y'$ is
\begin{equation}
 \frac{\Pr_{\mathsf M}\{Y\in\mathcal L(\widehat Y)\}}{L}.
 \label{eq:bounded-list-converted-success}
\end{equation}
If the exact decoder is required to have output alphabet $\mathsf Y'$ only,
one may add $N_\bot$ to any one of the operators $N_y$; this preserves
normalization and cannot decrease the exact-identification success
probability.
\end{lemma}

\begin{proof}
Positivity of each $N_y$ is immediate.  Exchanging the sums and applying
Eq.~\eqref{eq:bounded-list-size-assumption} gives
\begin{align}
 \sum_{y\in\mathsf Y'}N_y
 =\frac1L\sum_{\widehat y\in\mathsf Y'}
 |\mathcal L(\widehat y)|M_{\widehat y}
 \le\sum_{\widehat y\in\mathsf Y'}M_{\widehat y}=I.
 \label{eq:bounded-list-subnormalization}
\end{align}
This proves positivity of $N_\bot$.  For the joint distribution of the uniform true label $Y$ and the
measurement outcome $\widehat Y$, the success probability after the
POVM conversion in Eq.~\eqref{eq:bounded-list-subpovm} is
\begin{align*}
&\sum_{y,\widehat y\in\mathsf Y'}
 \Pr_{\mathsf M}\{Y=y,\widehat Y=\widehat y\}
 \frac{\mathbf1_{\{y\in\mathcal L(\widehat y)\}}}{L}\\
&\qquad=\frac{\Pr_{\mathsf M}\{Y\in\mathcal L(\widehat Y)\}}{L}.
\end{align*}
The failure outcome $\bot$ contributes zero to this success probability.  Finally, merging $N_\bot$ into an existing outcome
adds a positive operator and therefore cannot reduce the corresponding
correct-decoding term.
\end{proof}

\subsection{The generic fixed-constraint Bayes verification game}
\label{subsec:generic-fixed-W-Bayes-verification-game}
The generic original problem uses a single-valued exact target
$y_W(x)$ only for generic inputs, but the learner may report any legal
element of $\mathsf{Out}_{W,m}$.  Thus the report alphabet used in the
verification measurement is broader than the set of targets realized by
generic input states.  This is why we use the distinct-alphabet Bayes
conversion in Lemma~\ref{lem:generic-bayes-bounded-list-conversion}.
The original generic exact optimum is bounded by the coordinate exact
optimum in Eq.~\eqref{eq:fixed-W-generic-average-comparison}.

Fix an $(n-m)$-dimensional isotropic subspace $W$ and draw the unknown state
$X$ uniformly from $\mathsf{Stab}_{n,p}^{\mathrm{gen}}(W)$.  For a generic
input $x$, define its unique basis-free exact target by
\begin{equation}
 y_W(x):=\bigl(K_W(x),\chi_x|_{K_W(x)}\bigr),
 \qquad K_W(x)=L_x\cap W^\perp.
 \label{eq:section4-generic-original-exact-target}
\end{equation}
Genericity and Eq.~\eqref{eq:fixed-W-dimension} give
$\dim K_W(x)=m$, so this target is single-valued.  The learner receives
$k$ copies of $\rho_x$ and may report any
$\widehat y=(\widehat K,\widehat\eta)\in\mathsf{Out}_{W,m}$; in particular,
the report alphabet is the full original basis-free alphabet and not merely
the frame-relative output set $\Theta_{R(T),m}^{\mathrm{frame}}$ of a
revealed chart $T$.

For a learner POVM
$\mathsf M=\{M_{\widehat K,\widehat\eta}\}$, the generic Bayes verification
score is
\begin{align}
 S_{W,\rm ver,gen,av}^{(k)}(\mathsf M)
 :=\frac1{|\mathsf{Stab}_{n,p}^{\mathrm{gen}}(W)|}
 \sum_{x\in\mathsf{Stab}_{n,p}^{\mathrm{gen}}(W)}
 \sum_{(\widehat K,\widehat\eta)\in\mathsf{Out}_{W,m}}
 \Tr[M_{\widehat K,\widehat\eta}\rho_x^{\otimes k}]
 \Tr[\Pi_{\widehat K,\widehat\eta}\rho_x].
 \label{eq:section4-generic-original-Bayes-score}
\end{align}
The second factor is exactly the prescribed-measurement-outcome-averaged
acceptance probability by Eq.~\eqref{eq:fixed-W-outcome-averaged-score}.
Optimizing over all learner POVMs defines
$S_{W,\rm ver,gen,av}^{(k),*}(n,m;p)$.

For later reference, if $y=(K,\eta)$ is a generic exact target and
$\widehat y=(\widehat K,\widehat\eta)$ is a legal report, put
\begin{equation}
 d_W(y,\widehat y):=m-\dim(K\cap\widehat K).
 \label{eq:section4-generic-original-distance}
\end{equation}
Call the pair compatible when the two assignments agree on the common
physical Pauli or Weyl family represented by $K\cap\widehat K$.  A basis-free character-sum calculation gives
\begin{equation}
 \Tr[\Pi_{\widehat K,\widehat\eta}\rho_x]
 =\begin{cases}
 p^{-d_W(y_W(x),\widehat y)},&
 \text{if $y_W(x)$ and $\widehat y$ are compatible},\\
 0,&\text{otherwise}.
 \end{cases}
 \label{eq:section4-generic-original-score-spectrum}
\end{equation}
Equation~\eqref{eq:section4-generic-original-score-spectrum} is proved
intrinsically in Lemma~\ref{lem:generic-original-score-spectrum}.
We state the score formula here because it is needed in the following
verification converse.
Thus Born acceptance probability one is equivalent to the unique correct
report, whereas an incorrect basis-free report may still have positive
verification acceptance probability.
We control incorrect reports by their intersection distance and
compatibility, using the Bayes list conversion above.

\begin{theorem}[Verification converse for the generic fixed-constraint Bayes problem]
\label{thm:generic-fixed-W-Bayes-verification-converse}
Fix a prime $p$, $0<\beta\le1$, and $0\le\alpha<1$.  Let
$m_n=\beta n+o(n)$ and $k_n=\alpha n+o(n)$, and let
$W_n\subseteq\F_p^{2n}$ be any isotropic sequence satisfying
$\dim W_n=n-m_n$.  The generic fixed-constraint average verification optimum
converges to zero:
\begin{equation}
 S_{W_n,\rm ver,gen,av}^{(k_n),*}(n,m_n;p)\longrightarrow0,
 \label{eq:section4-generic-original-verification-converse}
\end{equation}
uniformly over the admissible sequence $(W_n)$.  More precisely,
\begin{equation}
 S_{W_n,\rm ver,gen,av}^{(k_n),*}(n,m_n;p)
 \le p^{-J(\alpha,\beta)n^2+o(n^2)}+p^{-\sqrt n+O(1)}.
 \label{eq:section4-generic-original-verification-quantitative}
\end{equation}
\end{theorem}

Theorem~\ref{thm:generic-fixed-W-Bayes-verification-converse} concerns the
full basis-free report alphabet, unlike the refined verification theorem
(Theorem~\ref{thm:refined-verification-threshold}).
Subsection~\ref{subsec:generic-original-graded-application} proves the
generic fixed-constraint verification converse using a compatible-list
bound and a high-rank score bound.  By
Eq.~\eqref{eq:compatible-rank-list-coarse-bound}, the logarithm of the
list size is $O(D(n+m+1)+\log(D+1))$ for fixed $p$.  Hence the cutoff
$D=\lfloor\sqrt n\rfloor$ makes this term $o(n^2)$, while the high-rank
contribution tends to zero.

\begin{corollary}[Verification converse for the full fixed-constraint task]
\label{cor:full-fixed-W-verification-converse}
Under the assumptions of
Theorem~\ref{thm:generic-fixed-W-Bayes-verification-converse}, both original verification optima converge to zero:
\begin{equation}
 S_{W_n,\rm ver,av}^{(k_n),*}(n,m_n;p)\longrightarrow0,
 \qquad
 S_{W_n,\rm ver,wc}^{(k_n),*}(n,m_n;p)\longrightarrow0.
 \label{eq:section4-full-original-verification-converse}
\end{equation}
\end{corollary}
\begin{proof}
Split the full uniform prior into its generic and exceptional parts.  The
optimal verification value conditioned on the generic input set vanishes by
Theorem~\ref{thm:generic-fixed-W-Bayes-verification-converse}, whereas the
verification acceptance probability for exceptional inputs is at most one,
and the exceptional prior mass satisfies
$\delta_{n,m_n,p}=O(p^{-m_n})=p^{-\beta n+o(n)}$ by
Eq.~\eqref{eq:fixed-W-exceptional-bound}.  Thus the exceptional contribution
is negligible relative to the asserted convergence, and this proves the
uniform-average claim.  The worst-case claim follows from
Eq.~\eqref{eq:fixed-W-verification-wc-average-only}.
\end{proof}

\subsection{Application to the generic original Bayes problem}
\label{subsec:generic-original-graded-application}

We now prove Theorem~\ref{thm:generic-fixed-W-Bayes-verification-converse}.
Fix an $(n-m)$-dimensional isotropic subspace $W$ and the generic prior.
The unique exact target is $y_W(x)$ from
Eq.~\eqref{eq:section4-generic-original-exact-target}; a learner report
may be any element of $\mathsf{Out}_{W,m}$.  We use the intrinsic
Born-score formula to define compatible neighborhoods within this
report alphabet, bound the number of exact targets in each neighborhood,
and then apply Lemma~\ref{lem:generic-bayes-bounded-list-conversion}.
The list-counting step therefore concerns basis-free reports, not the
coordinate labels $(C,\gamma)$.

\subsubsection{Basis-free score and exactness}
\label{subsubsec:generic-original-basis-free-score}

For a generic exact target $y=(K,\eta)$ and a legal report
$\widehat y=(\widehat K,\widehat\eta)\in\mathsf{Out}_{W,m}$, use the
basis-free distance $d_W(y,\widehat y)$ from
Eq.~\eqref{eq:section4-generic-original-distance}.
Call the two reports compatible when their joint eigenvalue assignments agree
on the common physical Pauli or Weyl family represented by
$K\cap\widehat K$.  This definition is intrinsic.  In characteristic two it
compares eigenvalues of the same physical Pauli operators, and therefore
already includes any phase-section correction needed by a coordinate
realization.

\begin{lemma}[Basis-free generic verifier score]
\label{lem:generic-original-score-spectrum}
For every generic input $x$ and every legal basis-free report
$\widehat y=(\widehat K,\widehat\eta)$,
\begin{equation}
 \Tr[\Pi_{\widehat K,\widehat\eta}\rho_x]
 =\begin{cases}
 p^{-d_W(y_W(x),\widehat y)},
 &\text{if $y_W(x)$ and $\widehat y$ are compatible},\\
 0,&\text{otherwise}.
 \end{cases}
 \label{eq:generic-original-score-spectrum}
\end{equation}
In particular, the score is one if and only if
$\widehat y=y_W(x)$, and every false report has score either zero or at most
$p^{-1}$.
\end{lemma}
\begin{proof}
Because $\widehat K\subseteq W^\perp$,
\begin{equation}
 \widehat K\cap L_x
 =\widehat K\cap(L_x\cap W^\perp)
 =\widehat K\cap K_W(x).
 \label{eq:generic-original-common-stabilizer}
\end{equation}
Use throughout the fixed physical Pauli or Weyl section $v\mapsto P_v$ from
Eq.~\eqref{eq:global-Pauli-section-cocycle}.  The intrinsic projector formula
Eq.~\eqref{eq:intrinsic-character-projector} gives
\begin{equation}
 \Pi_{\widehat K,\widehat\eta}
 =p^{-m}\sum_{v\in\widehat K}
 \widehat\eta(v)^{-1}P_v.
 \label{eq:generic-original-projector-expansion}
\end{equation}
For $v\in\widehat K$, the expectation of $P_v$ in $\rho_x$ is zero unless
$v\in L_x$.  On $\widehat K\cap L_x$, it is $\chi_x(v)$ relative to the same
fixed section.  The Born acceptance probability is consequently the
following character sum on the intersection:
\begin{equation}
 \Tr[\Pi_{\widehat K,\widehat\eta}\rho_x]
 =p^{-m}\sum_{v\in\widehat K\cap L_x}
 \widehat\eta(v)^{-1}\chi_x(v).
 \label{eq:generic-original-intersection-character-sum}
\end{equation}
Because both assignments refer to the same fixed section, their ratio is an
ordinary character of the common additive group $\widehat K\cap L_x$.  Character
orthogonality makes the sum zero unless the two eigenvalue assignments agree
on the intersection; if they agree, the sum equals
$p^{\dim(\widehat K\cap L_x)}$.  Together with
Eq.~\eqref{eq:generic-original-common-stabilizer}, this proves
Eq.~\eqref{eq:generic-original-score-spectrum}.  In characteristic two, both
characters in Eq.~\eqref{eq:generic-original-intersection-character-sum} refer
to the same physical representatives, so no comparison of incompatible phase
sections is being made.

If the score is one, then
$\dim(\widehat K\cap K_W(x))=m$.  Both subspaces have dimension $m$, hence
$\widehat K=K_W(x)$, and compatibility on the whole subspace gives
$\widehat\eta=\chi_x|_{K_W(x)}$.  The converse is immediate.
\end{proof}

Let $P_{W,\rm gen,av}^{(k),*}$ denote the exact Bayes optimum under the uniform
generic prior.  Equation~\eqref{eq:fixed-W-generic-average-comparison} and the
refined exact converse give
\begin{equation}
 P_{W,\rm gen,av}^{(k),*}(n,m;p)
 \le P_{\mathrm{id}}^{(k),*}(n,m;p)
 \le p^{-J(\alpha,\beta)n^2+o(n^2)}
 \label{eq:generic-original-exact-converse}
\end{equation}
under the scaling of Theorem~\ref{thm:refined-exact-converse}.  Only the unique
correct-report event is used in this comparison; no equality of the two full
report alphabets is required.

\subsubsection{Basis-free compatible neighborhoods and their size}
\label{subsubsec:generic-original-compatible-neighborhoods}

Let
\begin{equation}
 \mathcal Y_{W,\rm gen}:=
 \{y_W(x):x\in\mathsf{Stab}_{n,p}^{\mathrm{gen}}(W)\}
 \label{eq:generic-original-target-set}
\end{equation}
be the realizable generic exact-target set.  For
$\widehat y\in\mathsf{Out}_{W,m}$ define
\begin{equation}
 \mathcal N_{D,W}^{\mathrm{gen}}(\widehat y)
 :=\{y\in\mathcal Y_{W,\rm gen}:
 d_W(y,\widehat y)\le D,
 \ y\text{ and }\widehat y\text{ are compatible}\},
 \label{eq:generic-original-compatible-list}
\end{equation}
and put
\begin{equation}
 L_{D,W}^{\mathrm{gen}}(n,m;p)
 :=\max_{\widehat y\in\mathsf{Out}_{W,m}}
 |\mathcal N_{D,W}^{\mathrm{gen}}(\widehat y)|.
 \label{eq:generic-original-list-size}
\end{equation}
For every generic input $x$, its exact target $y_W(x)$ belongs to
$\mathcal N_{D,W}^{\mathrm{gen}}(y_W(x))$: its distance from itself is
zero and its assignment agrees with itself.  Thus the maximum
$L_{D,W}^{\mathrm{gen}}(n,m;p)$ is at least one.
Unlike the refined neighborhoods, these lists need not be translation
invariant and their sizes need not be independent of the center.

\begin{lemma}[Uniform basis-free list bound]
\label{lem:generic-original-list-count}
For every integer $0\le D\le m$, the uniform bound on the size of a
basis-free compatible list is
\begin{align}
 L_{D,W}^{\mathrm{gen}}(n,m;p)
 &\le\sum_{d=0}^{D}
 \genfrac{[}{]}{0pt}{}{m}{d}_p
 \genfrac{[}{]}{0pt}{}{n}{d}_p p^{d^2+d}
 \label{eq:generic-original-list-exact-upper}\\
 &\le (\kappa_p^{\mathrm{GL}})^{-2}(D+1)
 p^{D(n+m+1)}.
 \label{eq:generic-original-list-coarse-upper}
\end{align}
The bound is uniform in $W$ and in the center report.
\end{lemma}
\begin{proof}
The ambient space $W^\perp$ has dimension $n+m$.  Fix the $m$-subspace
$\widehat K$.  The number of $m$-dimensional subspaces $K\subseteq W^\perp$
with $\dim(K\cap\widehat K)=m-d$ is
\begin{equation}
 \genfrac{[}{]}{0pt}{}{m}{d}_p
 \genfrac{[}{]}{0pt}{}{n}{d}_p p^{d^2}.
 \label{eq:generic-original-grassmann-shell}
\end{equation}
Indeed, first choose $I=K\cap\widehat K$, an $(m-d)$-subspace of
$\widehat K$, in $\genfrac{[}{]}{0pt}{}{m}{d}_p$ ways.  Next choose the
$d$-dimensional image $U$ of $K$ in $W^\perp/\widehat K$, whose dimension is
$n$, in $\genfrac{[}{]}{0pt}{}{n}{d}_p$ ways.  For fixed $(I,U)$, the possible
subspaces $K$ are the graphs of linear maps
$U\to\widehat K/I$, giving $p^{d^2}$ choices.  This proves
Eq.~\eqref{eq:generic-original-grassmann-shell}.  Once the compatible assignment is fixed on the $(m-d)$-dimensional
intersection $I$, choose one extension to $K$.  Every other extension
is obtained uniquely by multiplying it by a character of $K/I$.
This quotient has dimension $d$, so there are exactly $p^d$ extensions.  Counting all such subspaces and extensions only enlarges
the desired list because it drops isotropy, transversality to $W$, and
realizability as a generic target.
Summing over $d\le D$ proves Eq.~\eqref{eq:generic-original-list-exact-upper}.
The standard Gaussian-binomial estimate
\begin{equation}
 \genfrac{[}{]}{0pt}{}{r}{d}_p
 \le(\kappa_p^{\mathrm{GL}})^{-1}p^{d(r-d)}
 \label{eq:generic-original-gaussian-bound}
\end{equation}
gives a summand at most
$(\kappa_p^{\mathrm{GL}})^{-2}p^{d(n+m+1)-d^2}$, and the coarser bound in
Eq.~\eqref{eq:generic-original-list-coarse-upper} follows.
\end{proof}

\subsubsection{Verification bound for the generic original problem}
\label{subsubsec:generic-original-graded-to-exact}

We apply the list-to-exact conversion to the generic original game,
whose reports are basis-free complementary stabilizer subspaces
and characters.  Its list cardinality is bounded uniformly in the
constraint $W$ and in the report, although it need not be
translation invariant as in the refined coordinate game.  The finite
inequality below uses the generic original exact optimum and its
basis-free compatible-list bound, rather than substituting the refined
verification optimum.

\begin{theorem}[Bounding original verification by exact learning and compatible lists]
\label{thm:generic-original-bounded-rank-graded-to-exact}
Let $p$ be prime, let $n\ge1$, $0\le m\le n$, and $k\ge0$ be integers,
and let $W\subseteq\F_p^{2n}$ be a fixed isotropic subspace of dimension
$n-m$.  For every integer $D$ with $0\le D\le m$, the generic original
verification optimum satisfies
\begin{equation}
 S_{W,\rm ver,gen,av}^{(k),*}(n,m;p)
 \le L_{D,W}^{\mathrm{gen}}(n,m;p)
 P_{W,\rm gen,av}^{(k),*}(n,m;p)+p^{-(D+1)}.
 \label{eq:generic-original-graded-to-exact}
\end{equation}
Splitting the original uniform prior into its generic and exceptional parts gives the following bound on the full uniform-average verification score:
\begin{align}
 S_{W,\rm ver,av}^{(k),*}(n,m;p)
 &\le\delta_{n,m,p}+(1-\delta_{n,m,p})
 \bigl[L_{D,W}^{\mathrm{gen}}(n,m;p)
 P_{W,\rm gen,av}^{(k),*}(n,m;p)+p^{-(D+1)}\bigr].
 \label{eq:full-original-finite-graded-bound}
\end{align}
\end{theorem}
\begin{proof}
Apply Lemma~\ref{lem:generic-bayes-bounded-list-conversion} to the lists in
Eq.~\eqref{eq:generic-original-compatible-list}.  The probability of the
event that the true target lies in the reported compatible rank-$D$ list is at
most
\begin{equation}
 L_{D,W}^{\mathrm{gen}}(n,m;p)
 P_{W,\rm gen,av}^{(k),*}(n,m;p).
 \label{eq:generic-original-low-rank-event-bound}
\end{equation}
On this event the score is at most one.  Outside it, an incompatible report
has score zero, while a compatible report has distance at least $D+1$ and
therefore score at most $p^{-(D+1)}$ by
Lemma~\ref{lem:generic-original-score-spectrum}.  This proves
Eq.~\eqref{eq:generic-original-graded-to-exact}.  Splitting the full uniform
prior into generic and exceptional parts and bounding the exceptional score by
one proves Eq.~\eqref{eq:full-original-finite-graded-bound}.
\end{proof}

\subsubsection{Asymptotic generic original converse}
\label{subsubsec:generic-original-asymptotic-converse}
\label{subsubsec:proof-generic-fixed-W-Bayes-verification-converse}
We now apply the finite basis-free comparison with  
$D_n=\lfloor\sqrt n\rfloor$.  Although the cutoff is the same as in the refined  
problem, the list bound and exact input are those of the generic original task,  
and the estimates must remain uniform over the admissible sequence $W_n$.  
Since $m_n=\beta n+o(n)$ with $\beta>0$, we  
have $D_n\le m_n$ for all sufficiently large $n$.  By
Lemma~\ref{lem:generic-original-list-count},
\begin{equation}
 \log_p L_{D_n,W_n}^{\mathrm{gen}}(n,m_n;p)
 \le D_n(n+m_n+1)-D_n^2+O(\log n)
 =O(n^{3/2})=o(n^2),
 \label{eq:generic-original-list-subquadratic}
\end{equation}
uniformly in $W_n$.  The exact comparison in
Eq.~\eqref{eq:generic-original-exact-converse} is also uniform in $W_n$:
its right-hand side is the standardized refined optimum and therefore does
not depend on the particular admissible subspace.  Hence
Eqs.~\eqref{eq:generic-original-exact-converse} and
\eqref{eq:generic-original-graded-to-exact} give
\begin{align}
 S_{W_n,\rm ver,gen,av}^{(k_n),*}(n,m_n;p)
 \le& L_{D_n,W_n}^{\mathrm{gen}}(n,m_n;p)
       p^{-J(\alpha,\beta)n^2+o(n^2)}
       +p^{-(D_n+1)}\notag\\
 \le & p^{-J(\alpha,\beta)n^2+o(n^2)}
       +p^{-\sqrt n+O(1)}\longrightarrow0.
 \label{eq:generic-original-stretched-exponential}
\end{align}
Both remainder estimates are uniform in the admissible sequence $(W_n)$.  This is precisely
Eqs.~\eqref{eq:section4-generic-original-verification-converse} and
\eqref{eq:section4-generic-original-verification-quantitative}, and completes
the proof of
Theorem~\ref{thm:generic-fixed-W-Bayes-verification-converse}.

The basis-free proof establishes
Theorem~\ref{thm:generic-fixed-W-Bayes-verification-converse} for generic
inputs and Corollary~\ref{cor:full-fixed-W-verification-converse} for the
full original task.  We next return to the refined fixed-frame problem.
Here every input has a unique coordinate target, and compatibility can be
counted by differences in the partial-label space.  The resulting list
bound, unlike the preceding basis-free list count, will also be used
for fixed-diagonal submodels in
Section~\ref{sec:fixed-diagonal-graph-basis-submodels}.

\subsection{Application to the refined coordinate problem}
\label{subsec:graded-refined-application}

\subsubsection{Comparing exact success with verification scores}
\label{subsec:graded-finite-sandwich}

For odd $p$, the standardized outputs are the pairs
$(C,\gamma)\in\Theta_{n,m,p}^{\mathrm{part}}$; for $p=2$, they are
labels in the enhanced quotient $Y_{n,m,2}$.  Both describe the previously
defined fixed-frame target after standardization, not a new operational game.
If $y$ is the true label and $\widehat y$ the report, let
$d(y,\widehat y)$ be the codimension of their
common queried stabilizer family and call them compatible when their characters
agree on that family.  Equations~\eqref{eq:odd-partial-score-spectrum}
and~\eqref{eq:integrated-qubit-score-spectrum} show that the Born
verification payoff depends only on this codimension and compatibility
relation.  We therefore write it intrinsically as
\begin{equation}
 s_{\rm game}(y,\widehat y):=
 \begin{cases}
 p^{-d(y,\widehat y)},&\text{if they are compatible},\\
 0,&\text{otherwise},
 \end{cases}.
 \label{eq:main-results-graded-payoff}
\end{equation}
The verification payoff $s_{\rm game}$ equals one exactly for the true
partial label:
\begin{equation}
 s_{\rm game}(y,\widehat y)=1\quad\Longleftrightarrow\quad y=\widehat y.
 \label{eq:main-results-score-one-exactness}
\end{equation}
Indeed, an incorrect but compatible report has positive codimension and
therefore scores at most $p^{-1}$; an incompatible report scores zero.
The following proposition compares exact success with the average
verification payoff using these two bounds.
\begin{proposition}[Finite exact-to-verification comparison]
\label{prop:prob-one-equivalence}
For every prime $p$, integers $n\ge1$, $1\le m\le n$, and $k\ge0$,
the standardized exact and verification optima satisfy
\begin{equation}
 P_{\mathrm{id}}^{(k),*}(n,m;p)
 \le S_{\mathrm{ver}}^{(k),*}(n,m;p)
 \le p^{-1}+(1-p^{-1})P_{\mathrm{id}}^{(k),*}(n,m;p).
 \label{eq:score-sandwich}
\end{equation}
\end{proposition}
\begin{proof}
The one-report verification payoff equals one for the true label, while any
incorrect label has payoff at most $p^{-1}$.  Splitting
the score into these two events gives the upper bound, while an exact decoder
is a feasible verification strategy and gives the lower bound.
\end{proof}

\subsubsection{Rank truncation and odd-prime difference coordinates}
\label{subsubsec:graded-rank-truncation-odd}
\label{subsec:rank-truncated-graded-to-exact}

The proposition shows that the exact and verification optima approach one together.  
It is insufficient for the converse considered here, however: exact success  
tending to zero yields only the nonvanishing bound  
$\limsup S_{\mathrm{ver}}^{(k),*}\le p^{-1}$.  To prove that the verification score  
itself vanishes, we retain the rank-dependent payoff and convert compatible  
reports of bounded rank into exact-target lists.  

We first verify, separately for odd primes and characteristic two, that rank  
and compatibility are determined by the label difference.  The binary branch  
is formulated in the enhanced quotient label group so that the phase-section  
correction is retained.  After these branchwise checks, the neighborhood count  
and the list-to-exact conversion are common to every prime.  

For odd $p$, recall
\begin{equation}
 \Theta_{n,m,p}^{\mathrm{part}}
 =\mathcal C_{n,m}\times\F_p^m,
 \quad
 \mathcal C_{n,m}
 =\{C\in\F_p^{n\times m}:J_m^TC=C^TJ_m\}.
 \label{eq:rank-truncated-label-space}
\end{equation}
For $y=(C,\gamma)$ and
$\widehat y=(\widehat C,\widehat\gamma)$, put
\begin{equation}
 B(y-\widehat y):=C-\widehat C,
 \quad
 \eta(y-\widehat y):=\gamma-\widehat\gamma.
 \label{eq:odd-rank-truncated-differences}
\end{equation}
The pair is compatible exactly when
\begin{equation}
 \eta(y-\widehat y)^Tx=0
 \quad(x\in\ker B(y-\widehat y)),
 \label{eq:odd-rank-truncated-compatibility}
\end{equation}
and its nonzero score is
$p^{-\rank B(y-\widehat y)}$.

\subsubsection{Characteristic-two enhanced quotient differences}
For $p=2$, the partial-label alphabet is the quotient group
\begin{equation}
 \mathsf Y_{n,m}^{(2)}:=Y_{n,m,2}=G_{n,2}/H_{n,m,2}.
 \label{eq:binary-graded-enhanced-label-group}
\end{equation}
We use this quotient element, rather than a phase-section-dependent pair
$(C,\gamma)$, as the primary label in the list reduction.  The visible
bilinear projection is the homomorphism
\begin{equation}
 \mathsf B:\mathsf Y_{n,m}^{(2)}\longrightarrow\mathcal C_{n,m},
 \label{eq:binary-graded-rank-projection}
\end{equation}
obtained from an enhanced representative $(\alpha,E)$ by taking the queried
symmetric-matrix column: its diagonal entries are $\alpha_i\bmod2$ for $i\le m$, its visible off-diagonal entries are the
corresponding $E_{ij}$, and its lower--visible entries are $E_{ji}$ with
$i\le m<j$.  In the quotient, a representative is determined by the visible coordinates
\begin{equation}
 ((\alpha_i)_{i\le m},(E_{ij})_{i\le m<j},
   (E_{ij})_{1\le i<j\le m}),
 \label{eq:binary-graded-visible-enhanced-coordinates}
\end{equation}
because adding an element of $H_{n,m,2}$ changes only $\alpha_j$ with
$j>m$ and $E_{ij}$ with $m<i<j$.  The displayed coordinates are therefore
unchanged coordinate by coordinate.  They determine $\mathsf B$ by reducing
$\alpha_i$ modulo two on the diagonal and using the retained edge
coordinates off the diagonal.  Hence $\mathsf B$ is a well-defined group
homomorphism, independent of the representative.  It is surjective: for a
given $C\in\mathcal C_{n,m}$ choose the visible parities and edge
coordinates prescribed by $C$, and choose either lift in $\mathbb Z_4$ for
each visible diagonal parity.

For a quotient difference $\xi\in\mathsf Y_{n,m}^{(2)}$, put
\begin{equation}
 B_\xi:=\mathsf B(\xi),
 \quad N_\xi:=\ker B_\xi,
 \quad d(\xi):=\rank B_\xi.
 \label{eq:binary-graded-difference-rank}
\end{equation}
On $N_\xi$, the true and reported phase-space labels coincide.  Comparing
their physical Pauli eigenvalues defines the intrinsic relative-character
functional
\begin{equation}
 \lambda_\xi:N_\xi\longrightarrow\F_2.
 \label{eq:binary-enhanced-relative-character}
\end{equation}
Represent the true and reported labels in phase-section-dependent vector
coordinates by $(A,b)$ and $(C,\gamma)$, respectively.  On $N(A,C)$,
the relative character has the following expression:
\begin{equation}
 \lambda_\xi(x)
 =\gamma^Tx+(J_m^Tb)^Tx+\delta_{A,C}(x),
 \quad x\in N(A,C).
 \label{eq:binary-relative-character-vector-realization}
\end{equation}
The right-hand side is linear on $N(A,C)$ by the calculation following
Eq.~\eqref{eq:integrated-qubit-phase-correction}.  For $x\in N(A,C)$, Eq.~\eqref{eq:integrated-qubit-section-comparison}
identifies the true and reported sections of the same physical Pauli label:
$\mathsf W_C(x)=(-1)^{\delta_{A,C}(x)}\mathsf W_A(J_mx)$.
Multiplying this section ratio by the reported and true eigenvalue characters
gives exactly $(-1)^{\lambda_\xi(x)}$.  Adding an unreported element of
$H_{n,m,2}$ changes only unqueried enhanced coordinates, so neither physical
operator nor either queried eigenvalue changes.  Changing a phase-section-dependent section multiplies the section and its
coordinate character by reciprocal signs.  Their product is therefore invariant.  Consequently, $\lambda_\xi$ depends
only on the quotient difference $\xi$.  The phase correction has been
absorbed, not deleted.

\begin{lemma}[Binary enhanced compatibility depends only on the quotient difference]
\label{lem:binary-enhanced-difference-compatibility}
For $y,\widehat y\in\mathsf Y_{n,m}^{(2)}$, the verifier score is nonzero if
and only if
\begin{equation}
 \lambda_{y-\widehat y}=0\quad\text{on }N_{y-\widehat y}.
 \label{eq:binary-enhanced-compatibility}
\end{equation}
On this event the score is
$2^{-d(y-\widehat y)}$.  In particular, both compatibility and score depend
only on the quotient difference.
\end{lemma}
\begin{proof}
Choose phase-section-dependent representatives.  Then
$N_{y-\widehat y}=N(A,C)$ and
$d(y-\widehat y)=m-\dim N(A,C)=\rank(C-AJ_m)$.  The compatibility condition
and score are exactly those in
Eq.~\eqref{eq:integrated-qubit-score-spectrum}, rewritten using
Eq.~\eqref{eq:binary-relative-character-vector-realization}.  Since the left
side of that realization is the relative character of the same physical
operators, changing either enhanced representative leaves it unchanged.

We verify simultaneous translation invariance directly in the additive
enhanced coordinates.  For $y=[(\alpha,E)]\in\mathsf Y_{n,m}^{(2)}$ and
$x\in\F_2^m$, let $q_y(x)\in\mathbb Z_4$ be the visible part of the
quadratic phase in Eq.~\eqref{eq:integrated-qubit-phase}:
\begin{equation}
 q_y(x):=\sum_{i=1}^m\alpha_i\widetilde x_i
       +2\sum_{1\le i<j\le m}E_{ij}\widetilde x_i\widetilde x_j
       \pmod 4.
 \label{eq:binary-enhanced-visible-quadratic-phase}
\end{equation}
The omitted coordinates belong to $H_{n,m,2}$, so this function is
independent of the enhanced representative.  Because addition of enhanced
labels is componentwise in $\mathbb Z_4$ and $\F_2$, it satisfies
$q_{y-\widehat y}=q_y-q_{\widehat y}\pmod4$.  Put
$\xi=y-\widehat y$.  For $x\in N_\xi=\ker B_\xi$, the parity of the enhanced phase is
\begin{equation}
 q_\xi(x)\bmod2
 =\operatorname{diag}(J_m^TB_\xi)^Tx
 =x^TJ_m^TB_\xi x=0\quad\text{in }\F_2.
 \label{eq:binary-enhanced-kernel-even-phase}
\end{equation}
Thus $q_\xi(x)$ is even.  Comparing the two complete quadratic phases in
Eq.~\eqref{eq:integrated-qubit-additive-coordinates} with the relative-character
formula~\eqref{eq:binary-relative-character-vector-realization} yields the
explicit equality
\begin{equation}
 \lambda_\xi(x)=\frac{q_\xi(x)}2\pmod2,
 \qquad x\in N_\xi.
 \label{eq:binary-enhanced-character-half-difference}
\end{equation}
The polarization identity, together with $B_\xi x=0$, also shows that this
half-difference is additive on $N_\xi$.  For any enhanced label
$t\in\mathsf Y_{n,m}^{(2)}$, componentwise group addition now gives
\begin{equation}
 \begin{aligned}
 (y+t)-(\widehat y+t)&=\xi,\\
 B_{(y+t)-(\widehat y+t)}&=B_\xi,\\
 N_{(y+t)-(\widehat y+t)}&=N_\xi,\\
 \lambda_{(y+t)-(\widehat y+t)}(x)
 &=\tfrac12\bigl(q_{y+t}(x)-q_{\widehat y+t}(x)\bigr)
 =\lambda_\xi(x)\quad(x\in N_\xi).
 \end{aligned}
 \label{eq:binary-enhanced-relative-character-translation-invariance}
\end{equation}
Since the score depends only on $\rank B_\xi$ and the vanishing of
$\lambda_\xi$ on $N_\xi$, the preceding calculation proves
\begin{equation}
 s_{\rm game}(y+t,\widehat y+t)=s_{\rm game}(y,\widehat y).
 \label{eq:binary-enhanced-simultaneous-translation-covariance}
\end{equation}
This argument uses the additive enhanced quotient, not componentwise addition
of its phase-section-dependent coordinates $(C,\gamma)$.
\end{proof}

\begin{lemma}[Binary compatible character fiber]
\label{lem:binary-compatible-character-fiber}
Fix $B\in\mathcal C_{n,m}$ of rank $d$.  Among the quotient differences
$\xi\in\mathsf Y_{n,m}^{(2)}$ satisfying $B_\xi=B$, exactly $2^d$ satisfy
\begin{equation}
 \lambda_\xi|_{\ker B}=0.
 \label{eq:binary-compatible-affine-fiber-condition}
\end{equation}
Equivalently, in any phase-section-dependent vector realization, the admissible
character-coordinate differences form a nonempty affine coset of
$(\ker B)^\perp$ and hence have cardinality $2^d$.
\end{lemma}
\begin{proof}
Fix one enhanced quotient difference over $B$, with visible diagonal
lifts $\alpha_i$.  Every other quotient difference over $B$ is obtained
uniquely by replacing $\alpha_i$ with $\alpha_i+2\eta_i$ modulo four,
where $\eta_i\in\F_2$ for $1\le i\le m$.  The off-diagonal coordinates
are fixed by $B$, and the unreported coordinates have already been
quotiented out.  Thus the choices are parameterized bijectively by
$\eta\in(\F_2^m)^*$, acting as the linear functional
$x\mapsto\eta^Tx$.  By Eq.~\eqref{eq:binary-enhanced-character-half-difference},
this replacement changes the functional on $\ker B$ by
$\eta|_{\ker B}$.  Thus the compatibility equation has the affine form
\begin{equation}
 \eta|_{\ker B}=\delta_B,
 \label{eq:binary-compatible-affine-restriction}
\end{equation}
where $\delta_B$ is the correction functional determined by the fixed
phase-section difference.  Its linearity follows from
Eq.~\eqref{eq:integrated-qubit-delta-additivity-calculation}.  The restriction
map
\begin{equation}
 (\F_2^m)^*\longrightarrow(\ker B)^*,
 \quad \eta\longmapsto\eta|_{\ker B},
 \label{eq:binary-character-restriction-map}
\end{equation}
is surjective, so Eq.~\eqref{eq:binary-compatible-affine-restriction} has a
solution.  Its solution set is a coset of the kernel, namely the annihilator
$(\ker B)^\perp$.  Since
$\dim(\ker B)^\perp=d$, the solution set has $2^d$ elements.
\end{proof}

\subsubsection{Branchwise compatible neighborhoods and their common size}
Having established difference-invariant rank and compatibility in both  
characteristics, we now define their common bounded-rank neighborhoods and  
count their sizes.  For either branch, write the label group as  
\begin{equation}
 \mathsf Y_{n,m,p}:=
 \begin{cases}
 \Theta_{n,m,p}^{\mathrm{part}},&p\text{ odd},\\
 \mathsf Y_{n,m}^{(2)},&p=2.
 \end{cases}.
 \label{eq:graded-common-label-group}
\end{equation}
This is the same partial-label group as $Y_{n,m,p}$ in
Section~\ref{sec:common-finite-size-framework}, expressed in the output
coordinates used here.  In odd characteristic the identification sends
a complete-label coset to $(AJ_m,-J_m^Tb)$; in characteristic two it is
the identity on the enhanced quotient.
For a difference $\xi$ in this group, let $B_\xi$ denote the odd-prime matrix
difference from Eq.~\eqref{eq:odd-rank-truncated-differences}, respectively
the binary projection in Eq.~\eqref{eq:binary-graded-difference-rank}.  Call
$\xi$ compatible when Eq.~\eqref{eq:odd-rank-truncated-compatibility},
respectively Eq.~\eqref{eq:binary-enhanced-compatibility}, holds.

\begin{definition}[Compatible rank-$D$ neighborhood]
\label{def:compatible-rank-neighborhood}
For $0\le D\le m$ and
$\widehat y\in\mathsf Y_{n,m,p}$, define
\begin{equation}
 \mathcal N_D^{(p)}(\widehat y)
 :=\{y\in\mathsf Y_{n,m,p}:
       \rank B_{y-\widehat y}\le D,
       \ y-\widehat y\text{ is compatible}\}.
 \label{eq:compatible-rank-neighborhood}
\end{equation}
Define
\begin{equation}
 L_D(n,m;p):=|\mathcal N_D^{(p)}(\widehat y)|.
 \label{eq:compatible-rank-list-size}
\end{equation}
\end{definition}

Every neighborhood contains its center, because the zero difference
has rank zero and is compatible.  Hence $L_D(n,m;p)\ge1$ for
$0\le D\le m$.  This same positive bound is used for the restricted
fixed-diagonal lists in Section~\ref{sec:fixed-diagonal-graph-basis-submodels}.
The next proposition proves that $\mathcal N_D^{(p)}(\widehat y)$
is a translate of the neighborhood at zero.  Consequently, the
cardinality $L_D(n,m;p)$ just defined does not depend on
$\widehat y$.

\begin{proposition}[Translation invariance of compatible neighborhoods]
\label{prop:compatible-neighborhood-translation}
For every prime $p$ and $0\le D\le m$, each compatible
rank-$D$ neighborhood is a translate of the neighborhood centered
at the zero label:
\begin{equation}
 \mathcal N_D^{(p)}(\widehat y)
 =\widehat y+\mathcal N_D^{(p)}(0).
 \label{eq:compatible-neighborhood-translation}
\end{equation}
Consequently, $L_D(n,m;p)$ is independent of the center.
\end{proposition}
\begin{proof}
For odd $p$, rank and compatibility depend on the additive difference through
Eqs.~\eqref{eq:odd-rank-truncated-differences} and
\eqref{eq:odd-rank-truncated-compatibility}.  For $p=2$, the difference is
taken in the enhanced quotient group, and
Lemma~\ref{lem:binary-enhanced-difference-compatibility} proves that both rank
and compatibility depend only on that difference.  Translation by
$-\widehat y$ therefore gives a bijection from the neighborhood centered at
$\widehat y$ to the neighborhood centered at zero.
\end{proof}

\begin{lemma}[Cardinality of a compatible rank neighborhood]
\label{lem:compatible-rank-neighborhood-cardinality}
Let
\begin{equation}
 N_{\mathcal C}(n,m,d;p)
 :=|\{B\in\mathcal C_{n,m}:\rank B=d\}|.
 \label{eq:admissible-rank-count}
\end{equation}
For a fixed matrix difference of rank $d$, compatibility leaves exactly
$p^d$ character differences.  Thus the rank-$D$ neighborhood size is obtained
by multiplying the number of admissible rank-$d$ matrix differences by this
character-fiber size and summing over $d\le D$:
\begin{equation}
 L_D(n,m;p)
 =\sum_{d=0}^{\min\{D,m\}}
 N_{\mathcal C}(n,m,d;p)p^d.
 \label{eq:compatible-rank-list-exact-count}
\end{equation}
To obtain a uniform upper bound, we may next drop the linear admissibility
constraint and count all rectangular matrix differences of each rank.
Recall that $\kappa_p^{\mathrm{GL}}$ was defined in
Eq.~\eqref{eq:kappa-GL-product}.  The same positive constant is used in
the following rank estimates:
\begin{equation}
 \kappa_p^{\mathrm{GL}}=\prod_{j=1}^{\infty}(1-p^{-j})>0.
 \label{eq:rank-truncated-kappa}
\end{equation}
The standard finite-field rank bounds then give the following two
estimates for $L_D(n,m;p)$:
\begin{align}
 L_D(n,m;p)
 &\le
 (\kappa_p^{\mathrm{GL}})^{-1}
 \sum_{d=0}^{\min\{D,m\}}p^{d(n+m-d+1)},
 \label{eq:compatible-rank-list-rank-sum-bound}\\
 L_D(n,m;p)
 &\le
 (\kappa_p^{\mathrm{GL}})^{-1}(D+1)
 p^{D(n+m+1)}.
 \label{eq:compatible-rank-list-coarse-bound}
\end{align}
\end{lemma}
\begin{proof}
Fix $B\in\mathcal C_{n,m}$ of rank $d$.  For odd $p$, compatible character
differences form the annihilator $(\ker B)^\perp$, which has $p^d$ elements.
For $p=2$, Lemma~\ref{lem:binary-compatible-character-fiber} gives a nonempty
affine coset of the same dimension and hence exactly $2^d$ elements.  Summing
over matrices of ranks $d\le D$ proves
Eq.~\eqref{eq:compatible-rank-list-exact-count} in both branches.

For the upper bound only, we now discard the linear admissibility condition
$J_m^TB=B^TJ_m$ and count all $n\times m$ matrices of rank $d$.  This enlarges
the set and gives
\begin{equation}
 N_{\mathcal C}(n,m,d;p)
 \le N_{\mathrm{rect}}(n,m,d;p)
 \le(\kappa_p^{\mathrm{GL}})^{-1}p^{d(n+m-d)}.
 \label{eq:rank-truncated-rectangular-bound}
\end{equation}
Multiplication by the compatible character count $p^d$ and summation over
$d\le D$ proves the first bound.  The sum has at most $D+1$ terms, and each
exponent is at most $D(n+m+1)$, proving the second.
\end{proof}

In both branches, the fiber over a partial label $y$ is a coset
$s_Y(y)+H_{n,m,p}$ of the kernel $H_{n,m,p}$ of the partial-label map.
All these fibers have the same cardinality $|H_{n,m,p}|$.
The uniform distribution on complete labels therefore induces a uniform
prior on $\mathsf Y_{n,m,p}$.  Explicitly, a complete label
has point mass $|G_{n,p}|^{-1}$, while every quotient label $y$ has exactly
$|H_{n,m,p}|$ representatives.  The induced probability of
each quotient label $y$ is therefore
\begin{equation}
 \Pr\{Y=y\}=\frac{|H_{n,m,p}|}{|G_{n,p}|}
 =|G_{n,p}/H_{n,m,p}|^{-1}.
 \label{eq:uniform-quotient-label-prior}
\end{equation}
The separate factor $|H_{n,m,p}|^{-1}$ appears only in the conditional
average over the complete labels in the coset, as in the definition of
$\sigma_y$ in Eq.~\eqref{eq:sigma-y}; it is not an additional factor in the
quotient point mass.  This argument applies equally to the odd-prime and
characteristic-two complete-label groups.  We now give the reduction from a
low-rank compatible report to exact identification.  The
conversion accepts an arbitrary normalized POVM used for verification as input; its
success guarantee follows from classical postprocessing.

\subsubsection{From compatible lists to exact identification}
\label{subsec:graded-list-to-exact}
\label{subsubsec:list-to-exact-povm-conversion}
We apply Lemma~\ref{lem:bounded-list-povm-conversion} to the refined
neighborhoods.  Here every list has exactly $L_D$ elements, by
Proposition~\ref{prop:compatible-neighborhood-translation}.  The residual
failure operator is therefore zero.  Uniform randomization over the
reported list gives a normalized POVM and an equality for its exact
success, as shown next.

\begin{lemma}[POVM conversion from a compatible list to exact identification]
\label{lem:compatible-list-povm-conversion}
Let
$\mathsf M=\{M_{\widehat y}:\widehat y\in\mathsf Y_{n,m,p}\}$
be any POVM used for verification.  For $0\le D\le m$, define
\begin{equation}
 N_y^{(D)}
 :=\frac1{L_D(n,m;p)}
 \sum_{\widehat y:\,y\in\mathcal N_D^{(p)}(\widehat y)}
 M_{\widehat y},
 \quad y\in\mathsf Y_{n,m,p}.
 \label{eq:rank-truncated-exact-povm}
\end{equation}
The operators $N_y^{(D)}$ defined in Eq.~\eqref{eq:rank-truncated-exact-povm} form a valid exact-identification POVM
$\mathsf N^{(D)}=\{N_y^{(D)}:y\in\mathsf Y_{n,m,p}\}$.  If $\mathsf A_D$ denotes the compatible low-rank event
\begin{equation}
 \mathsf A_D:=\{Y\in\mathcal N_D^{(p)}(\widehat Y)\},
 \label{eq:rank-truncated-event}
\end{equation}
under the uniform partial-label ensemble and the measurement $\mathsf M$, the exact-identification success probability of the converted POVM is
\begin{equation}
 P_{\mathrm{id}}(\mathsf N^{(D)};n,m,p)
 =\frac{\Pr_{\mathsf M}(\mathsf A_D)}{L_D(n,m;p)}.
 \label{eq:rank-truncated-converted-success}
\end{equation}
Since the converted POVM cannot exceed the optimal exact-identification success probability, the probability of the compatible-label event obeys
\begin{equation}
 \Pr_{\mathsf M}(\mathsf A_D)
 \le L_D(n,m;p)
 P_{\mathrm{id}}^{(k),*}(n,m;p).
 \label{eq:rank-truncated-list-event-bound}
\end{equation}
\end{lemma}

\begin{proof}
Every $N_y^{(D)}$ is positive semidefinite.  For normalization, exchange the
order of summation and use the center-independence of the neighborhood size:
\begin{align}
 \sum_{y\in\mathsf Y_{n,m,p}}N_y^{(D)}
 =\frac1{L_D(n,m;p)}
 \sum_{\widehat y\in\mathsf Y_{n,m,p}}
 \left(
 \sum_{y\in\mathcal N_D^{(p)}(\widehat y)}1
 \right)M_{\widehat y}
 =\sum_{\widehat y\in\mathsf Y_{n,m,p}}M_{\widehat y}
 =I.
 \label{eq:rank-truncated-povm-normalization}
\end{align}
Thus $\mathsf N^{(D)}$ is a POVM.

Let $\sigma_y$ be the conditioned $k$-copy input state from
Eq.~\eqref{eq:sigma-y} for the partial label $y$.  Since the induced prior on $\mathsf Y_{n,m,p}$ is uniform, the exact
success probability of $\mathsf N^{(D)}$ is
\begin{align}
 &P_{\mathrm{id}}(\mathsf N^{(D)};n,m,p)
 =\frac1{|\mathsf Y_{n,m,p}|}
 \sum_y\Tr(N_y^{(D)}\sigma_y)\notag\\
 =&\frac1{L_D(n,m;p)|\mathsf Y_{n,m,p}|}
 \sum_y\sum_{\widehat y:\,y\in\mathcal N_D^{(p)}(\widehat y)}
 \Tr(M_{\widehat y}\sigma_y)
 =\frac{\Pr_{\mathsf M}(\mathsf A_D)}{L_D(n,m;p)}.
 \label{eq:rank-truncated-success-calculation}
\end{align}
This establishes Eq.~\eqref{eq:rank-truncated-converted-success}.  The constructed
POVM $\mathsf N^{(D)}$ depends on the truncation radius $D$, whereas
$P_{\mathrm{id}}^{(k),*}(n,m;p)$ is the optimum over all exact-identification
POVMs for the fixed input ensemble and therefore does not depend on $D$.
Because $\mathsf N^{(D)}$ is one feasible exact-identification POVM for
the same $k$-copy ensemble, its success probability cannot exceed the
unrestricted exact optimum.  Therefore, for every $D$,
\[
 P_{\mathrm{id}}(\mathsf N^{(D)};n,m,p)
 \le P_{\mathrm{id}}^{(k),*}(n,m;p).
\]
This gives Eq.~\eqref{eq:rank-truncated-list-event-bound}.
\end{proof}

\subsubsection{Finite bounded-rank comparison}
\label{subsubsec:finite-bounded-rank-comparison}

The following theorem bounds the optimal refined verification score
using the exact-identification optimum, the rank-$D$ list size
$L_D(n,m;p)$, and the high-rank remainder $p^{-(D+1)}$.

\begin{theorem}[Finite-size verification bound from compatible exact-target lists]
\label{thm:rank-truncated-graded-to-exact}
For every fixed prime $p$, all finite $n$, $1\le m\le n$, $k\ge0$, and every
integer $0\le D\le m$, low-rank compatible reports are controlled
by exact identification after paying the list-size factor $L_D(n,m;p)$,
whereas the complementary high-rank event contributes at most
$p^{-(D+1)}$.  Quantitatively, the optimal verification score satisfies
\begin{equation}
 S_{\mathrm{ver}}^{(k),*}(n,m;p)
 \le
 L_D(n,m;p)P_{\mathrm{id}}^{(k),*}(n,m;p)
 +p^{-(D+1)}.
 \label{eq:rank-truncated-graded-to-exact}
\end{equation}
Here $L_D(n,m;p)$ is the list size defined in
Eq.~\eqref{eq:compatible-rank-list-size}.
\end{theorem}

\begin{proof}
Fix an arbitrary POVM for verification $\mathsf M$.  On the event $\mathsf A_D$ from
Eq.~\eqref{eq:rank-truncated-event}, the score is at most one.  On its
complement, an incompatible report has score zero, while a compatible
report has $\rank B_{Y-\widehat Y}\ge D+1$ and therefore score at most
$p^{-(D+1)}$.  Splitting the average score over these two events gives
\begin{align}
 S_{\mathrm{ver}}^{(k)}(\mathsf M;n,m,p)
 &\le \Pr_{\mathsf M}(\mathsf A_D)
 +p^{-(D+1)}\Pr_{\mathsf M}(\mathsf A_D^c)\notag\\
 &\le \Pr_{\mathsf M}(\mathsf A_D)+p^{-(D+1)}.
 \label{eq:rank-truncated-score-split}
\end{align}
Now substitute the list-event estimate
Eq.~\eqref{eq:rank-truncated-list-event-bound} into the first term on the
right-hand side of Eq.~\eqref{eq:rank-truncated-score-split}.  This yields
\begin{equation}
 S_{\mathrm{ver}}^{(k)}(\mathsf M;n,m,p)
 \le
 L_D(n,m;p)P_{\mathrm{id}}^{(k),*}(n,m;p)
 +p^{-(D+1)}.
 \label{eq:rank-truncated-score-bound-arbitrary}
\end{equation}
The right-hand side of
Eq.~\eqref{eq:rank-truncated-score-bound-arbitrary} is independent of the
chosen verification POVM.  Taking the supremum over $\mathsf M$ on the left-hand
side therefore gives Eq.~\eqref{eq:rank-truncated-graded-to-exact}.
\end{proof}

In the proof, the cutoff $D$ enters through the converted POVM
$\mathsf N^{(D)}$, the event $\mathsf A_D$, the list size $L_D$, and the
tail term.  The exact optimum $P_{\mathrm{id}}^{(k),*}$ does not depend on
$D$, since it is optimized over all exact-identification POVMs.
The exact count and the upper bounds for $L_D$ are given in
Eqs.~\eqref{eq:compatible-rank-list-exact-count} and
\eqref{eq:compatible-rank-list-rank-sum-bound}--
\eqref{eq:compatible-rank-list-coarse-bound}.

The converted decoder first performs the original verification measurement and then  
applies uniform classical randomization over the reported compatible  
rank-$D$ list.

For $D=0$, one has $L_0=1$, and
Eq.~\eqref{eq:rank-truncated-graded-to-exact} gives
\[
 S_{\mathrm{ver}}^{(k),*}
 \le P_{\mathrm{id}}^{(k),*}+p^{-1}.
\]
The upper side of
Proposition~\ref{prop:prob-one-equivalence} is sharper at fixed
finite size.  The advantage of
Theorem~\ref{thm:rank-truncated-graded-to-exact} is instead that $D$ may grow
with $n$, allowing the tail $p^{-(D+1)}$ to vanish while the list size remains
subexponential on the $n^2$ scale.

\subsubsection{Asymptotic verification converse}
\label{subsec:graded-asymptotic-converse}

A fixed cutoff leaves a nonzero score tail, whereas a cutoff of linear order  
can make the compatible list exponential on the $n^2$ scale.  We therefore  
choose the diverging sublinear cutoff $D_n=\lfloor\sqrt n\rfloor$, for which the  
list-size exponent is $o(n^2)$ and the tail vanishes.  

\begin{corollary}[Vanishing verification score below one learner copy per qudit]
\label{cor:graded-score-to-zero}
Fix a prime $p$ and suppose
\[
 m=\beta n+o(n),\quad k=\alpha n+o(n),
 \quad 0<\beta\le1,\quad 0\le\alpha<1.
\]
Then the optimal verification score vanishes below copy rate one, that is, below one learner copy per qudit:
\begin{equation}
 S_{\mathrm{ver}}^{(k),*}(n,m;p)\longrightarrow0.
 \label{eq:graded-score-to-zero}
\end{equation}
The quantitative bound is obtained by choosing the sublinear cutoff
$D_n=\lfloor\sqrt n\rfloor$.  With this choice,
\begin{equation}
 S_{\mathrm{ver}}^{(k),*}(n,m;p)
 \le
 p^{-J(\alpha,\beta)n^2+o(n^2)}
 +p^{-\sqrt n+O(1)}.
 \label{eq:graded-score-stretched-exponential}
\end{equation}
\end{corollary}

\begin{proof}
Choose
\begin{equation}
 D_n:=\lfloor\sqrt n\rfloor.
 \label{eq:graded-cutoff-choice}
\end{equation}
Since $m/n\to\beta>0$, Eq.~\eqref{eq:graded-cutoff-choice} satisfies
$D_n\le m$ for all sufficiently large $n$.  Hence the finite comparison
Eq.~\eqref{eq:rank-truncated-graded-to-exact} applies with $D=D_n$.

Substitute $D=D_n$ into the list-size estimate
Eq.~\eqref{eq:compatible-rank-list-coarse-bound}.  Taking base-$p$
logarithms gives
\begin{align}
 \log_p L_{D_n}(n,m;p)
 \le D_n(n+m+1)+O_p(\log n)
 =O(n^{3/2})=o(n^2).
 \label{eq:rank-truncated-list-subquadratic}
\end{align}
The exact-identification converse Eq.~\eqref{eq:quadratic-converse} gives
\begin{equation}
 P_{\mathrm{id}}^{(k),*}(n,m;p)
 \le p^{-J(\alpha,\beta)n^2+o(n^2)},
 \label{eq:graded-use-exact-converse}
\end{equation}
where $J(\alpha,\beta)>0$ for $0\le\alpha<1$.  Multiplying the bounds in
Eqs.~\eqref{eq:rank-truncated-list-subquadratic} and
\eqref{eq:graded-use-exact-converse} gives
\begin{equation}
 L_{D_n}(n,m;p)P_{\mathrm{id}}^{(k),*}(n,m;p)
 \le p^{-J(\alpha,\beta)n^2+O(n^{3/2})+o(n^2)}
 =p^{-J(\alpha,\beta)n^2+o(n^2)}.
 \label{eq:rank-truncated-list-times-exact}
\end{equation}
The tail term in Eq.~\eqref{eq:rank-truncated-graded-to-exact} satisfies
\begin{equation}
 p^{-(D_n+1)}=p^{-\sqrt n+O(1)}\longrightarrow0.
 \label{eq:rank-truncated-tail-decay}
\end{equation}
Finally, substitute Eqs.~\eqref{eq:rank-truncated-list-times-exact} and
\eqref{eq:rank-truncated-tail-decay} into
Eq.~\eqref{eq:rank-truncated-graded-to-exact}.  This gives
Eq.~\eqref{eq:graded-score-stretched-exponential}, and hence
Eq.~\eqref{eq:graded-score-to-zero}.
\end{proof}

The coordinate bound also gives the following frame-relative conclusion.

\begin{theorem}[Verification converse for the refined problem]
\label{thm:refined-verification-threshold}
Fix the prime, integer sequences, and public frames as in Theorem~\ref{thm:refined-exact-converse}. Then the optimal refined worst-case verification score is at most the refined uniform-average score, which equals the standardized verification optimum and converges to zero:
\begin{equation}
 S_{\rm frame,ver,wc}^{(k_n),*}(R_n,m_n;p)
 \le S_{\rm frame,ver,av}^{(k_n),*}(R_n,m_n;p)
 =S_{\mathrm{ver}}^{(k_n),*}(n,m_n;p)\longrightarrow0.
 \label{eq:refined-verification-converse}
\end{equation}
The convergence is uniform over the public-frame sequence $(R_n)$.
\end{theorem}
\begin{proof}
The verification optima satisfy
\begin{align*}
 S_{\rm frame,ver,wc}^{(k_n),*}(R_n,m_n;p)
 \stackrel{(a)}{\le}S_{\rm frame,ver,av}^{(k_n),*}(R_n,m_n;p)
 \stackrel{(b)}{=}S_{\mathrm{ver}}^{(k_n),*}(n,m_n;p)
 \stackrel{(c)}{\longrightarrow}0.
\end{align*}
Here (a) compares worst-case and uniform-average performance,
(b) is Eq.~\eqref{eq:finite-frame-verification-equivalence}, and
(c) is Corollary~\ref{cor:graded-score-to-zero}, whose convergence
is given in Eq.~\eqref{eq:graded-score-to-zero}.  The coordinate
score in (b) is independent of $R_n$, so the same upper bound
holds for every frame sequence.
\end{proof}

\section{Direct bound and proof of the main theorem}
\label{sec:fourier-achievability}
For achievability, we solve the stronger problem of complete stabilizer
identification and then extract the requested report by deterministic
classical postprocessing.  This route also handles exceptional inputs without
choosing a preferred complement in advance.  First, we prove PGM optimality
for transitive pure-state ensembles and bound its error by pairwise overlaps.
For stabilizer states, grouping those overlaps by Lagrangian-intersection
codimension gives a finite error bound that vanishes when $k-n\to+\infty$.
Clifford covariance makes the success probability state independent, so the
postprocessing gives both worst-case and uniform-average achievability for
the fixed-$W$ and frame-assisted tasks.  Correct reports have verification
score one, which supplies the direct bound for that criterion as well.

\subsection{Complete identification and equal statewise success}
The exact complete-identification optima were defined in
Eqs.~\eqref{eq:unrestricted-stabilizer-wc-success} and
\eqref{eq:unrestricted-stabilizer-av-success}.  We first prove their equality
by averaging a decoder over the transitive Clifford action.

The finite Clifford--Pauli action on pure stabilizer states is transitive.
Let $G$ be a finite group implementing this transitive action, with unitary
representatives $U_g$ (global phases are immaterial), so that
$U_g\rho_xU_g^\dagger=\rho_{gx}$.  For any complete-identification
POVM $\{M_x\}_{x\in\mathsf{Stab}_{n,p}}$, its covariant average is
\begin{equation}
 \widetilde M_x:=\frac1{|G|}\sum_{g\in G}
 U_g^{\otimes k}M_{g^{-1}x}(U_g^\dagger)^{\otimes k}.
 \label{eq:complete-id-twirled-decoder}
\end{equation}
The averaged effects remain positive and sum to the identity.  Using
$U_g^\dagger\rho_xU_g=\rho_{g^{-1}x}$, the success probabilities of the
averaged decoder satisfy
\begin{equation}
 \Tr[\widetilde M_x\rho_x^{\otimes k}]
 =\frac1{|G|}\sum_{g\in G}
 \Tr[M_{g^{-1}x}\rho_{g^{-1}x}^{\otimes k}].
 \label{eq:complete-id-twirled-statewise-success}
\end{equation}
The right-hand side is independent of $x$ by transitivity and equals the
uniform-average success of the original decoder.  A worst-case score is
always bounded by its uniform average, while the twirled decoder attains
that average on every input.  Optimizing gives the equality:
\begin{equation}
 P_{\rm stab,id,wc}^{(k),*}(n;p)
 =P_{\rm stab,id,av}^{(k),*}(n;p).
 \label{eq:unrestricted-stabilizer-wc-av-equality}
\end{equation}
We write $P_{\rm stab,id}^{(k),*}(n;p)$ for this common value.

The complete-identification problem will provide the direct part for every
prescribed constraint $W$: after identifying $(L_x,\chi_x)$, the learner can
choose a valid complementary report by classical postprocessing.  We first
record the full-state identification guarantee used for this construction.

\begin{theorem}[Complete stabilizer identification with diverging overhead]
\label{thm:complete-identification-achievability}
Fix a prime $p$.
\medskip
\noindent\textbf{(1) Finite-size equality.}
For every $n\ge1$ and $k\ge0$, the optimal worst-case and
uniform-average exact identification probabilities coincide:
\begin{equation}
 P_{\rm stab,id,wc}^{(k),*}(n;p)
 =P_{\rm stab,id,av}^{(k),*}(n;p)
 =P_{\rm stab,id}^{(k),*}(n;p).
 \label{eq:complete-identification-wc-av-common-value}
\end{equation}
\medskip
\noindent\textbf{(2) Diverging-overhead achievability.}
If $k_n-n\to+\infty$, the optimal worst-case and average
complete-state identification probabilities converge to one:
\begin{equation}
 P_{\rm stab,id,wc}^{(k_n),*}(n;p)
 =P_{\rm stab,id,av}^{(k_n),*}(n;p)
 \longrightarrow1.
 \label{eq:complete-identification-wc-av-achievability}
\end{equation}
Equivalently, in the fixed-constraint special case $W=\{0\}$ and $m=n$, both
exact-learning success probabilities tend to one.
\end{theorem}
The finite PGM bound below proves the diverging-overhead assertion.

The verification objective at $W=\{0\}$ should not be confused with the exact
indicator above: it is the Born acceptance probability of the reported
full-state projector, so an incorrect nonorthogonal report may have a
positive acceptance probability.  A correct report of the entire state $(L_x,\chi_x)$ has verification acceptance probability one, because its reported projector contains the true state. Hence the exact complete stabilizer-state identification decoder also supplies the direct part for this verification task.

\subsection{Transitive PGM optimality and its error estimate}
\label{subsec:unrestricted-stabilizer-direct}
We first establish the minimum-error optimality of the PGM for a finite
transitive pure-state ensemble, including the case of a noninvertible
average state.  The associated pairwise-overlap estimate then bounds its
error.  We apply these two results to stabilizer states by counting
Lagrangian intersection shells.

\subsubsection{Transitive PGM formula}
\label{subsubsec:transitive-pgm-formula}

The next theorem isolates the property of the pure-state ensemble that
makes its PGM exactly optimal: the unitary group action is transitive on
the labels, so the Gram matrix has a positive square root with constant
diagonal.  The proof will also check the dual optimality conditions on
the common support, without assuming that the Gram matrix is invertible.

\begin{theorem}[PGM optimality for transitive pure-state ensembles]
\label{thm:transitive-pure-pgm}
Let $\{N^{-1},\ket{\psi_x}\!\bra{\psi_x}:x\in X\}$ be a finite
uniform pure-state ensemble.  Suppose that a group $\mathcal G$ acts
transitively on the label set $X$ and that, for every $g\in\mathcal G$, there
is a unitary $U_g$ such that
\begin{equation}
 U_g\ket{\psi_x}=e^{i\phi(g,x)}\ket{\psi_{g\cdot x}}
 \quad(x\in X)
 \label{eq:transitive-labeled-ensemble-covariance}
\end{equation}
for suitable phases $\phi(g,x)$.  Then the PGM is minimum-error optimal.  The
statement remains valid when the label-space Gram matrix is singular; all
inverses are Moore--Penrose inverses on the common support.
\end{theorem}
The conclusion is not inferred from covariance alone.  The transitive action
on the labels first forces the positive square root of the Gram matrix to have
constant diagonal;
the proof then uses that constant-diagonal property to construct an explicit
dual feasible operator whose objective equals the PGM primal value.  Thus the
argument verifies the minimum-error optimality conditions directly, including
when the label-space Gram matrix is singular.
\begin{proof}
\noindent\textbf{Step 1: The supported discrimination problem and PGM.}
Put $\ket{\varphi_x}:=N^{-1/2}\ket{\psi_x}$, and let $S$ be the
synthesis operator with columns $\ket{\varphi_x}$.  Write
\begin{equation}
 G:=S^\dagger S,\quad R:=SS^\dagger,
 \quad \mathcal K:=\supp R=\operatorname{ran}S.
 \label{eq:transitive-pgm-G-R-support}
\end{equation}
We first formulate minimum-error discrimination on the common support
$\mathcal K$.  With the prior absorbed into the subnormalized vectors
$\ket{\varphi_x}$, the primal and dual semidefinite programs are
\begin{align}
 \text{maximize}\quad
 &\sum_{x\in X}\bra{\varphi_x}M_x\ket{\varphi_x},
 &M_x&\ge0,\quad \sum_{x\in X}M_x=\Pi_{\mathcal K},
 \label{eq:transitive-pgm-primal-sdp}\\
 \text{minimize}\quad
 &\Tr_{\mathcal K}Y,
 &Y&\ge\ket{\varphi_x}\!\bra{\varphi_x}
       \quad(x\in X),
 \label{eq:transitive-pgm-dual-sdp}
\end{align}
where the dual variable is supported on $\mathcal K$.  Any supported POVM
can be completed arbitrarily on $\mathcal K^\perp$, and every ensemble state
is supported on $\mathcal K$.  Hence this completion changes neither the
outcome probabilities nor the success probability.

All inverse powers below are Moore--Penrose inverse powers.  Define the
supported PGM vectors by
\begin{equation}
 \ket{\mu_x}:=SG^{+1/2}\ket{x}.
 \label{eq:transitive-pgm-vectors-supported}
\end{equation}
They form a POVM on the common support because
\begin{equation}
 \sum_{x\in X}\ket{\mu_x}\!\bra{\mu_x}
 =SG^+S^\dagger=\Pi_{\mathcal K}.
 \label{eq:transitive-pgm-supported-normalization}
\end{equation}
Moreover, functional calculus on $\supp G$ gives
$G^{+1/2}G=\sqrt G$.  Therefore the PGM overlap matrix is obtained by the
explicit calculation
\begin{equation}
 \langle\mu_x|\varphi_y\rangle
 =\bra{x}G^{+1/2}S^\dagger S\ket{y}
 =\bra{x}\sqrt G\ket{y}
 =(\sqrt G)_{xy}.
 \label{eq:transitive-pgm-overlap-matrix}
\end{equation}

\medskip
\noindent\textbf{Step 2: Covariance and the constant-diagonal property.}
For each $g\in\mathcal G$, define a monomial unitary on the label space
by its action on the standard basis:
\begin{equation*}
 W_g\ket{x}:=e^{i\phi(g,x)}\ket{g\cdot x},\qquad
 U_gS=SW_g.
\end{equation*}
The second identity follows by applying both sides to every $\ket{x}$
and using Eq.~\eqref{eq:transitive-labeled-ensemble-covariance}.
Since $U_g$ is unitary, this identity gives
\begin{equation*}
 W_g^\dagger G W_g
 =(SW_g)^\dagger(SW_g)
 =(U_gS)^\dagger(U_gS)=G.
\end{equation*}
Functional calculus therefore gives
$W_g^\dagger\sqrt G W_g=\sqrt G$.
This conjugation invariance for each $g$ does not require the operators
$g\mapsto W_g$ to form a strict linear representation on the label space.
Since the given action is transitive
on $X$, all diagonal entries of $\sqrt G$ are equal.  Write
$(\sqrt G)_{xx}=d$.  The ensemble is nonzero, so
$\Tr\sqrt G>0$; the constant-diagonal property implies $d>0$.  Using
Eq.~\eqref{eq:transitive-pgm-overlap-matrix}, the supported PGM has primal
value
\begin{equation}
 P^{\PGM}
 =\sum_{x\in X}\lvert\langle\mu_x|\varphi_x\rangle\rvert^2
 =\sum_{x\in X}d^2=Nd^2.
 \label{eq:transitive-pgm-primal-value}
\end{equation}

\medskip
\noindent\textbf{Step 3: A dual feasible certificate.}
We next construct a dual certificate with the same value.  Define
\begin{equation}
 Y:=d\sqrt R
 \quad\text{on }\mathcal K,
 \label{eq:transitive-pgm-dual-certificate}
\end{equation}
and set it equal to zero on $\mathcal K^\perp$.  To evaluate its
pseudoinverse on the signal vectors, take a reduced singular-value
decomposition $S=U\Sigma V^\dagger$.  Then
$R=U\Sigma^2U^\dagger$, $G=V\Sigma^2V^\dagger$, and hence
\begin{equation}
 S^\dagger(\sqrt R)^+S
 =V\Sigma U^\dagger U\Sigma^+U^\dagger U\Sigma V^\dagger
 =V\Sigma V^\dagger=\sqrt G.
 \label{eq:transitive-pgm-svd-pseudoinverse-identity}
\end{equation}
The expectation of $Y^+$ in each subnormalized signal vector is therefore
\begin{equation}
 \bra{\varphi_x}Y^+\ket{\varphi_x}
 =d^{-1}\bra{x}S^\dagger(\sqrt R)^+S\ket{x}
 =d^{-1}(\sqrt G)_{xx}=1.
 \label{eq:transitive-pgm-dual-saturation}
\end{equation}

The saturation identity in Eq.~\eqref{eq:transitive-pgm-dual-saturation}
also verifies the dual constraints of
Eq.~\eqref{eq:transitive-pgm-dual-sdp}.  For a positive operator $Y$ and
$v\in\supp Y$, conjugation on $\supp Y$ by $Y^{+1/2}$ gives the rank-one
domination criterion
\begin{equation}
 Y\ge\ket{v}\!\bra{v}
 \quad\Longleftrightarrow\quad
 I\ge\ket{Y^{+1/2}v}\!\bra{Y^{+1/2}v}
 \quad\Longleftrightarrow\quad
 \bra{v}Y^+\ket{v}\le1.
 \label{eq:rank-one-pseudoinverse-domination}
\end{equation}
Each $\ket{\varphi_x}$ belongs to $\mathcal K=\supp Y$, so
Eq.~\eqref{eq:transitive-pgm-dual-saturation} and the criterion imply
$Y\ge\ket{\varphi_x}\!\bra{\varphi_x}$ for every $x$.  Thus $Y$ is dual
feasible.

\medskip
\noindent\textbf{Step 4: Equality of primal and dual values.}
Finally, $R$ and $G$ have the same nonzero eigenvalues.  The dual objective
therefore equals the PGM primal value:
\begin{equation}
 \Tr_{\mathcal K}Y
 =d\Tr\sqrt R=d\Tr\sqrt G=Nd^2=P^{\PGM}.
 \label{eq:transitive-pgm-primal-dual-match}
\end{equation}
A primal feasible measurement and a dual feasible operator have attained the
same value, so weak duality forces both to be optimal.  Completing the
supported PGM on $\mathcal K^\perp$ proves the theorem on the full Hilbert
space, including the singular-Gram-matrix case.
\end{proof}

We next bound the transitive PGM error by the sum of pairwise squared
overlaps.  This estimate will be evaluated for the unrestricted stabilizer
ensemble by its Lagrangian-intersection classes.
\subsubsection{Overlap-energy error control}
\label{subsubsec:transitive-pgm-error-control}

\begin{lemma}[Squared-overlap control of the transitive PGM error]
\label{lem:transitive-pgm-overlap-energy}
Under the assumptions of Theorem~\ref{thm:transitive-pure-pgm}, fix
$x_0\in X$.  The PGM error probability is bounded by the sum of squared
overlaps with $\ket{\psi_{x_0}}$:
\begin{equation}
 1-P^{\PGM}
 \le\sum_{x\ne x_0}|\langle\psi_{x_0}|\psi_x\rangle|^2.
 \label{eq:transitive-pgm-overlap-energy}
\end{equation}
For the $k$-copy ensemble, the summands on the right are replaced by
$|\langle\psi_{x_0}|\psi_x\rangle|^{2k}$.
\end{lemma}
\begin{proof}
Let $H$ be the unweighted Gram matrix, so $H_{xx}=1$.  The same transitive monomial action used in
Theorem~\ref{thm:transitive-pure-pgm} makes
$d_0:=(\sqrt H)_{xx}$ independent of $x$, and the PGM success is $d_0^2$.
Let $\mu$ be the spectral probability measure of $H$ at the coordinate
vector $\ket{x_0}$.  Its first moment is one, and its
square-root moment equals the common diagonal entry $d_0$:
\begin{equation}
 \int\lambda\,d\mu(\lambda)=1,
 \quad d_0=\int\sqrt\lambda\,d\mu(\lambda).
\end{equation}
Jensen's inequality gives
$d_0=\int\sqrt\lambda\,d\mu(\lambda)\le1$.  Because
$(\sqrt\lambda-1)^2\le(\lambda-1)^2$ for $\lambda\ge0$,
\begin{align}
 1-d_0^2
 &=(1-d_0)(1+d_0)\le2(1-d_0)
 =\int(\sqrt\lambda-1)^2\,d\mu(\lambda)\notag\\
 &\le\int(\lambda-1)^2\,d\mu(\lambda)
 =(H^2)_{x_0x_0}-1
 =\sum_{x\ne x_0}\lvert H_{x_0x}\rvert^2.
\end{align}
This proves the first claim.  Tensor powers replace each Gram entry by its
$k$th power.
\end{proof}

\subsection{Stabilizer overlaps and Lagrangian intersection classes}
\label{subsec:direct-overlaps-shells}

We now specialize the transitive-ensemble bound.  For two stabilizer states,
the codimension $r$ of the Lagrangian intersection determines both ingredients
of the overlap energy: compatibility gives squared overlap $p^{-r}$, and the
intersection-shell count gives the number of underlying Lagrangians at that
codimension.  This common parameter therefore converts the abstract PGM bound
into a one-dimensional finite sum.

Fix $x_0\in\mathsf{Stab}_{n,p}$.  For
$x\in\mathsf{Stab}_{n,p}$ put
\begin{equation}
 r(x,x_0):=n-\dim(L_x\cap L_{x_0}).
 \label{eq:unrestricted-intersection-codimension}
\end{equation}
Call $x$ and $x_0$ \emph{compatible on their intersection} if they have the
same eigenvalue on every common physical Pauli operator associated with
$L_x\cap L_{x_0}$.  Equality of the two eigenvalues is independent of the
chosen common phase representative, because changing that representative
multiplies both eigenvalues by the same phase.  Thus the condition is
intrinsic and does not require a common phase section or a
characteristic-independent character coordinate.
For every projective pure state label $z$, choose a normalized vector
$\ket{\phi_z}$ satisfying $\rho_z=\ket{\phi_z}\!\bra{\phi_z}$.
\begin{lemma}[Unrestricted stabilizer overlaps]
\label{lem:unrestricted-stabilizer-overlaps}
When $r=r(x,x_0)$, the squared stabilizer overlap is
\begin{equation}
 |\langle\phi_{x_0}|\phi_x\rangle|^2
 =\begin{cases}
 p^{-r},&x\text{ and }x_0\text{ are compatible on their intersection},\\
 0,&\text{otherwise}.
 \end{cases}.
 \label{eq:unrestricted-stabilizer-overlap}
\end{equation}
For every fixed Lagrangian $L$ satisfying
$\dim(L\cap L_{x_0})=n-r$, exactly $p^r$ projective stabilizer states with
underlying Lagrangian $L$ are compatible with $x_0$ on the intersection.
\end{lemma}
\begin{proof}
Put $K:=L_x\cap L_{x_0}$.  For a projective stabilizer state $z$, let
$\mathcal S_z$ denote its phase-normalized stabilizer group, containing one
operator above each vector of $L_z$.  Its rank-one projector has the intrinsic
expansion
\begin{equation}
 \ket{\phi_z}\!\bra{\phi_z}
 =p^{-n}\sum_{P\in\mathcal S_z}P.
 \label{eq:unrestricted-stabilizer-projector-expansion}
\end{equation}
Taking the trace of the product of two stabilizer projectors expresses their squared state overlap in terms of the intersection of the stabilizer groups:
\begin{align}
 \lvert\langle\phi_{x_0}\mid\phi_x\rangle\rvert^2
 &=\Tr(\rho_{x_0}\rho_x)\notag\\
 &=p^{-2n}\sum_{P\in\mathcal S_{x_0}}
                 \sum_{Q\in\mathcal S_x}\Tr(PQ).
 \label{eq:unrestricted-overlap-projector-trace}
\end{align}
We evaluate the surviving terms by their common phase-space direction.  Fix
any Pauli section $v\mapsto D_v$ on the intersection $K$.  For
$z\in\{x_0,x\}$, write the phase-normalized stabilizer operator above
$v\in K$ as
\begin{equation}
 P_z(v)=\chi_z(v)^{-1}D_v,
 \label{eq:unrestricted-common-direction-representative}
\end{equation}
where changing the section multiplies both $D_v$ and the two coordinate
characters by the same phase.  Hence the ratio
$\chi_x(v)\chi_{x_0}(v)^{-1}$ is intrinsic.  Pauli trace orthogonality reduces
the double sum in
Eq.~\eqref{eq:unrestricted-overlap-projector-trace} to the common directions
$v\in K$ and gives
\begin{equation}
 \Tr(\rho_{x_0}\rho_x)
 =p^{-n}\sum_{v\in K}
 \chi_x(v)\chi_{x_0}(v)^{-1}.
 \label{eq:unrestricted-overlap-common-character-sum}
\end{equation}
The summand is the ratio of the two restricted characters.  If the states are
compatible on $K$, this ratio is the trivial character, so the sum equals
$p^{-n}|K|=p^{-n}p^{n-r}=p^{-r}$.  If compatibility fails, the ratio is a
nontrivial character of the additive group $K$, and character orthogonality
makes the sum zero.  This proves
Eq.~\eqref{eq:unrestricted-stabilizer-overlap} for every prime without choosing
a common phase convention globally.

For the counting statement, fix one cocycle-compatible assignment on $L$.
The ratio of any two compatible assignments on $L$ is an ordinary character
of the additive space $L$.  Thus every compatible assignment on $L$
is obtained uniquely by multiplying the fixed assignment by an element
of $L^*$.  After fixing one compatible restriction on
$K=L\cap L_{x_0}$, the remaining counting problem is the restriction map
$L^*\to K^*$.  This map is surjective: extend a basis of $K$ to a basis of
$L$ and choose the remaining character values arbitrarily.  Its kernel has
dimension $n-(n-r)=r$, so the prescribed compatible restriction has exactly
$p^r$ extensions.  This parameterization by ratios of assignments remains valid in
characteristic two, where the compatible assignments themselves need not be
identified with ordinary characters relative to the fixed physical section.
\end{proof}

The direct bound also requires the number of competing Lagrangian subspaces having a prescribed intersection dimension with $L_{x_0}$. We count these subspaces before combining the count with the compatible-character extensions established above.

\begin{lemma}[Classes of Fixed Lagrangian-Intersection Codimension]
\label{lem:lagrangian-intersection-shells}
The number of Lagrangians $L\subseteq\F_p^{2n}$ satisfying
$\dim(L\cap L_{x_0})=n-r$ is
\begin{equation}
 N_{\rm Lag}(n,r;p)
 =\genfrac{[}{]}{0pt}{}{n}{r}_p p^{r(r+1)/2}.
 \label{eq:lagrangian-intersection-shell-count}
\end{equation}
\end{lemma}
\begin{proof}
First choose
$K=L\cap L_{x_0}$, an $(n-r)$-dimensional subspace of $L_{x_0}$; there are
$\genfrac{[}{]}{0pt}{}{n}{r}_p$ choices.  Since $K$ is isotropic, the quotient
$K^\perp/K$ is a $2r$-dimensional symplectic space.  Every Lagrangian $L$
containing $K$ determines the Lagrangian image $L/K$ in this quotient, and
\begin{equation}
 L\cap L_{x_0}=K
 \quad\Longleftrightarrow\quad
 (L/K)\cap(L_{x_0}/K)=\{0\}.
 \label{eq:lagrangian-shell-quotient-transversality}
\end{equation}
Conversely, the inverse image in $K^\perp$ of every quotient Lagrangian
transverse to $L_{x_0}/K$ is a unique Lagrangian containing $K$ and satisfying
this exact intersection condition.  Thus no additional multiplicity is
introduced when passing to the quotient.

Put $\bar L_0:=L_{x_0}/K$ in the symplectic quotient
$K^\perp/K$ and choose a complementary Lagrangian $\bar L_1$.
For any quotient Lagrangian $\bar L$ transverse to $\bar L_0$,
projection onto $\bar L_1$ is a linear isomorphism.  Thus there is a
unique linear map $S:\bar L_1\to\bar L_0$ such that
\begin{equation}
 \bar L=\{v+Sv:v\in\bar L_1\}.
 \label{eq:lagrangian-shell-quotient-graph}
\end{equation}
As in Lemma~\ref{lem:transverse-symmetric-matrix-normal-form},
isotropy of this graph is equivalent to symmetry of the bilinear form
$(v,w)\mapsto\langle Sv,w\rangle$.  Relative to dual bases,
$S$ is therefore an arbitrary symmetric $r\times r$ matrix.
There are $p^{r(r+1)/2}$ such matrices, and the projection argument
is a bijection, with no multiplicity from the chosen complement.
Multiplying this number by the choices of $K$ gives
Eq.~\eqref{eq:lagrangian-intersection-shell-count}.
\end{proof}

\subsection{Finite unrestricted error bound and fixed-\texorpdfstring{$W$}{W} postprocessing}
\label{subsec:direct-finite-bound-postprocessing}

For $k\ge1$, the abstract PGM estimate can now be summed shell by shell.
Orthogonal one-copy states have zero $k$-copy overlap in this range, so
only compatible eigenvalue assignments contribute.  Each shell
contributes its Lagrangian multiplicity, the $p^r$ compatible stabilizer
characters, and the $k$-copy squared overlap $p^{-kr}$.  This gives the finite
energy below and, after a Gaussian-binomial bound, the unit-rate sequence
condition.  Only after complete-state recovery has been established do we
apply deterministic postprocessing to the fixed-$W$ target.

For a fixed true stabilizer state and $k\ge1$, the overlap energy is the
sum of its squared $k$-copy overlaps with all other stabilizer states.
Grouping this sum by intersection codimension gives
\begin{equation*}
 \sum_{x\ne x_0}|\langle\phi_{x_0}|\phi_x\rangle|^{2k}
 =\sum_{r=1}^n N_{\rm Lag}(n,r;p)\,p^r p^{-kr}.
\end{equation*}
We denote this energy by
\begin{equation}
 \mathcal E_{n,k}^{\mathrm{stab}}(p)
 :=\sum_{r=1}^n
 \genfrac{[}{]}{0pt}{}{n}{r}_p
 p^{r(r+1)/2}p^{-(k-1)r}.
 \label{eq:unrestricted-stabilizer-energy}
\end{equation}
Indeed, Lemma~\ref{lem:lagrangian-intersection-shells} supplies the
Lagrangian count, Lemma~\ref{lem:unrestricted-stabilizer-overlaps} supplies
$p^r$ compatible projective stabilizer states and $k$-copy squared overlap
$p^{-kr}$. Summing these contributions gives the following bound on the failure probability of the PGM for exact complete stabilizer-state identification:
\begin{equation}
 1-P_{\rm stab,id}^{(k),\PGM}(n;p)
 \le\mathcal E_{n,k}^{\mathrm{stab}}(p).
 \label{eq:unrestricted-stabilizer-pgm-energy-bound}
\end{equation}

The restriction $k\ge1$ applies to the overlap sum used here.  At
$k=0$, every input has the same empty tensor power, so orthogonal
one-copy inputs cannot be omitted from that sum.  The zero-copy task
remains defined as before; the finite error theorem below uses
$k\ge n+2$.

The next theorem evaluates this energy uniformly in the dimension.  It gives
the finite unrestricted error bound and hence complete-identification success
tending to one for $k=n+\omega(1)$, for every fixed prime.

\begin{theorem}[Finite unrestricted PGM bound and proof of Theorem~\ref{thm:complete-identification-achievability}]
\label{thm:unrestricted-stabilizer-direct}
For every prime $p$, the PGM is minimum-error optimal for the uniform ensemble
of all pure $n$-qudit stabilizer states.  When $k\ge n+2$, its optimal
error probability obeys the finite bound
\begin{equation}
 1-P_{\rm stab,id}^{(k),*}(n;p)
 \le(\kappa_p^{\mathrm{GL}})^{-1}
 \frac{p^{-(k-n-1)}}{1-p^{-(k-n-1)}}.
 \label{eq:unrestricted-stabilizer-direct-bound}
\end{equation}
The right-hand side of the finite error bound vanishes whenever
$k_n-n\to+\infty$.  Consequently, the unrestricted success probability tends
to one:
\begin{equation}
 P_{\rm stab,id}^{(k_n),*}(n;p)\longrightarrow1.
 \label{eq:unrestricted-stabilizer-direct-conclusion}
\end{equation}
\end{theorem}
\begin{proof}
PGM optimality follows from Theorem~\ref{thm:transitive-pure-pgm}.  Here the
label set consists of the distinct projective pure stabilizer states, and the
Clifford group acts transitively on this label set while satisfying the
covariance relation
Eq.~\eqref{eq:transitive-labeled-ensemble-covariance}.  The Gaussian-binomial estimate gives
\begin{align}
 \mathcal E_{n,k}^{\mathrm{stab}}(p)
 &\le(\kappa_p^{\mathrm{GL}})^{-1}
 \sum_{r=1}^n
 p^{r(n-r)+r(r+1)/2-(k-1)r}\notag\\
 &=(\kappa_p^{\mathrm{GL}})^{-1}
 \sum_{r=1}^n
 p^{-r(k-n-1)-r(r-1)/2}.
 \label{eq:unrestricted-stabilizer-energy-estimate}
\end{align}
For $k\ge n+2$, discard the nonpositive quadratic term and extend the sum to
infinity.  Equation~\eqref{eq:unrestricted-stabilizer-pgm-energy-bound} then
gives Eq.~\eqref{eq:unrestricted-stabilizer-direct-bound}.

For the auxiliary frame task, fix the public frame $R$ and target dimension
$m$.  The report alphabet $\Theta_{R,m}^{\mathrm{frame}}$ is nonempty:
a Lagrangian complement to $M_R$, with a compatible eigenvalue assignment,
defines a state in $\mathsf S(R)$.  Choose a report
$r_0(R,m)\in\Theta_{R,m}^{\mathrm{frame}}$ before running the complete-state
decoder.  Write $\widehat X$ for its estimated pure stabilizer state label.
Define a map on all possible estimates, with values in the legal report
alphabet, by
\begin{equation}
 f_{R,m}(\widehat x):=
 \begin{cases}
  y_R(\widehat x),&\widehat x\in\mathsf S(R),\\
  r_0(R,m),&\widehat x\notin\mathsf S(R),
 \end{cases}
 \qquad \widehat x\in\mathsf{Stab}_{n,p}.
 \label{eq:complete-to-frame-assisted-partial-postprocessing}
\end{equation}
Here $y_R$ is the reference-dual target defined in
Eq.~\eqref{eq:frame-assisted-exact-target}.  For any normalized complete-state
decoder with effects $\{M_{\widehat x}\}$, grouping outcomes according to this
map gives positive effects indexed by $\Theta_{R,m}^{\mathrm{frame}}$:
\begin{align}
 N_r^{R,m}&:=\sum_{\widehat x:\,f_{R,m}(\widehat x)=r}M_{\widehat x},
 \qquad r\in\Theta_{R,m}^{\mathrm{frame}},\notag\\
 \sum_{r\in\Theta_{R,m}^{\mathrm{frame}}}N_r^{R,m}
 &=\sum_{\widehat x\in\mathsf{Stab}_{n,p}}M_{\widehat x}=I.
 \label{eq:complete-to-frame-povm-normalization}
\end{align}
Thus these effects form a POVM with the required output alphabet.
For every true input $x\in\mathsf S(R)$, the identity $f_{R,m}(x)=y_R(x)$
implies the statewise success bound
\begin{equation}
 \Pr_x\{f_{R,m}(\widehat X)=y_R(x)\}
 \ge\Pr_x\{\widehat X=x\},
 \label{eq:complete-to-frame-success-lower-bound}
\end{equation}
where $\Pr_x$ is the decoder's outcome law on $\rho_x^{\otimes k}$.
On every successful complete identification, the resulting report has
verification score one.  Since all verification scores are nonnegative,
the expected verification score has the same lower bound.
\end{proof}

The theorem applies to $k_n=n+\lceil\sqrt n\rceil$, whose normalized copy
number tends to one.  It therefore gives an achievable copy rate of one.
The finite bound also gives an error estimate at each offset $k-n\ge2$.
Its convergence-to-one consequence uses $k_n-n\to+\infty$.
The theorem does not assert that $k=n$ copies suffice with success tending
to one or determine the optimal success probabilities throughout the
fixed-offset window $k=n+O(1)$.

\begin{corollary}[Statewise complete-identification and fixed-$W$ achievability]
\label{cor:statewise-fixed-W-achievability}
For the covariant unrestricted PGM of
Theorem~\ref{thm:unrestricted-stabilizer-direct}, the success probability is
constant over all projective pure stabilizer states.  Hence the average error
bound in Eq.~\eqref{eq:unrestricted-stabilizer-direct-bound} is also its
worst-case error bound.  After exact complete stabilizer-state identification, the deterministic
postprocessing
\begin{equation}
 (L_x,\chi_x)\longmapsto
 \bigl(K,\chi_x|_K\bigr),
 \quad
 L_x\cap W^\perp=(L_x\cap W)\oplus K,
 \label{eq:fixed-W-direct-postprocessing}
\end{equation}
produces a valid fixed-$W$ output for every state.  Consequently, whenever
$k_n-n\to+\infty$, both the worst-case and uniform Bayesian fixed-$W$ success
probabilities tend to one.
\end{corollary}
\begin{proof}
The PGM is covariant under the transitive Clifford action.  Conjugating a
state and its correctly relabelled POVM outcome leaves the conditional success
probability invariant, so transitivity makes it constant.  The postprocessing first computes $R=L_x\cap W$ and
$S=L_x\cap W^\perp$, and then selects, by a fixed Gaussian-elimination rule, a
linear complement $K$ satisfying $S=R\oplus K$.
Equation~\eqref{eq:fixed-W-dimension} gives $\dim K=m$,
$K\cap W=\{0\}$, and validity of the resulting output.  The finite
complete-identification error bound therefore applies statewise and also after
deterministic postprocessing.
\end{proof}

Thus the complete-state decoder, its statewise covariance, and the
postprocessing map together establish
Theorem~\ref{thm:complete-identification-achievability}.

The statewise decoder now supplies the direct ingredient in the original copy-rate theorem.

\subsection{Proof of the fixed-constraint copy-rate theorem}
\label{subsec:proof-main-fixed-W}
We combine the exact comparison, the basis-free verification converse,
and the statewise complete-identification decoder.  The proof follows
the four conclusions of Theorem~\ref{thm:fixed-W-no-discount}.

\begin{proof}[Proof of Theorem~\ref{thm:fixed-W-no-discount}]
\noindent\textbf{Part (1): Exact-learning converse.}
For fixed $(m_n),(k_n)$ as in (1), put
$Q_n:=P_{\mathrm{id}}^{(k_n),*}(n,m_n;p)$ and
$\delta_n:=\delta_{n,m_n,p}$.  Theorem~\ref{thm:refined-exact-converse}
and the all-prime converse bound give a sequence $e_n=o(n^2)$,
independent of $W_n$, with
$Q_n\le p^{-J(\alpha,\beta)n^2+e_n}$.  By
Theorem~\ref{thm:fixed-W-operational-comparison},
\begin{align*}
 P_{W_n,\rm wc}^{(k_n),*}
 &\stackrel{(a)}{\le}Q_n
 \stackrel{(b)}{\le}p^{-J(\alpha,\beta)n^2+e_n},\\
 P_{W_n,\rm av}^{(k_n),*}
 &\stackrel{(c)}{\le}\delta_n+(1-\delta_n)Q_n
 \stackrel{(d)}{\le}C_p p^{-m_n}
       +p^{-J(\alpha,\beta)n^2+e_n}.
\end{align*}
Here the arguments $(n,m_n;p)$ and the superscript $(k_n),*$ are
suppressed on the two left-hand sides only.  Steps (a) and (c) are the
finite comparison in
Eqs.~\eqref{eq:fixed-W-wc-final-comparison} and
\eqref{eq:fixed-W-av-comparison}; (b) is the refined converse; and
(d) uses $0\le1-\delta_n\le1$ and
$\delta_n\le C_p p^{-m_n}$ from
Eq.~\eqref{eq:fixed-W-exceptional-bound}.
Since $\beta>0$ and $\alpha<1$, both upper bounds tend to zero.
The common $e_n$ and the exceptional-mass constant do not depend on
$W_n$, proving the uniformity asserted in (1).

\medskip
\noindent\textbf{Part (2): Verification converse.}
Let $S_{W_n,\rm ver,gen,av}^{(k_n),*}$ denote the
uniform-average verification optimum under the generic conditional
prior, with the full basis-free report alphabet.  Splitting the uniform
prior into generic and exceptional parts gives
\begin{align*}
 S_{W_n,\rm ver,wc}^{(k_n),*}
 &\stackrel{(a)}{\le}S_{W_n,\rm ver,av}^{(k_n),*}\\
 &\stackrel{(b)}{\le}\delta_n+(1-\delta_n)
        S_{W_n,\rm ver,gen,av}^{(k_n),*}
 \stackrel{(c)}{\longrightarrow}0.
\end{align*}
In (a) the minimum over inputs is bounded by the uniform average.
For (b), exceptional inputs contribute acceptance probability at most
one.  In (c), Theorem~\ref{thm:generic-fixed-W-Bayes-verification-converse}
controls the generic term, and
Eq.~\eqref{eq:fixed-W-exceptional-bound} gives $\delta_n\to0$.
Corollary~\ref{cor:full-fixed-W-verification-converse} records this
full-prior consequence.  Each bound is uniform over legal $W_n$.
The report alphabet in this argument is $\mathsf{Out}_{W_n,m_n}$;
no equality with the refined verification optimum is used.

\medskip
\noindent\textbf{Part (3): Uniform direct part.}
By Theorem~\ref{thm:complete-identification-achievability} and the
statewise covariance argument of
Corollary~\ref{cor:statewise-fixed-W-achievability}, there is a
complete-state decoder with success probability at least $q_n$ on every
pure stabilizer input, where $q_n\to1$ when $k_n-n\to+\infty$.
The sequence $q_n$ is independent of the input state and the prescribed
isotropic subspace $W_n$.
On a successful complete identification, the recovered pair $(L_x,\chi_x)$
determines an $m_n$-dimensional complement $K$ of
$L_x\cap W_n$ inside $L_x\cap W_n^\perp$, by
Eq.~\eqref{eq:fixed-W-dimension}.  A fixed Gaussian-elimination rule
selects such a $K$; the postprocessed report is $(K,\chi_x|_K)$.
By Eq.~\eqref{eq:fixed-W-valid-output}, this report is exactly correct,
and Eq.~\eqref{eq:fixed-W-statewise-verification-score-reduced} gives
its verification acceptance probability as one.  The constructed decoder
therefore has both exact success and expected verification score at
least $q_n$ on every input $x$ and for every legal $W_n$.
Taking input minima or averages and optimizing the learner POVM gives
\begin{align*}
 q_n&\le\min\bigl\{
 P_{W_n,\rm wc}^{(k_n),*},P_{W_n,\rm av}^{(k_n),*},
 S_{W_n,\rm ver,wc}^{(k_n),*},S_{W_n,\rm ver,av}^{(k_n),*}
 \bigr\}\le1.
\end{align*}
The displayed optima have the common arguments $(n,m_n;p)$.
Since $q_n$ is independent of $x$ and $W_n$, this proves the four
uniform limits.  The asymptotic bound requires only existence of the
complement and a fixed selection rule, not a complexity bound for its
implementation.

\medskip
\noindent\textbf{Part (4): Copy-rate identities.}
Part (1) excludes every limiting copy rate $\alpha<1$ for either exact
criterion.  Conversely, the integer sequences
\[
 m_n=\max\{1,\lfloor\beta n\rfloor\},\qquad
 k_n=n+\lfloor\sqrt n\rfloor
\]
satisfy $m_n/n\to\beta$, $k_n/n\to1$, and
$k_n-n\to+\infty$.  Part (3) supplies success tending to one, so
copy rate one is achievable for both exact criteria.  This proves (4).
\end{proof}

For the refined direct part, apply
Theorem~\ref{thm:complete-identification-achievability} and the total report
map $f_{R_n,m_n}$ in
Eq.~\eqref{eq:complete-to-frame-assisted-partial-postprocessing}.
Equation~\eqref{eq:complete-to-frame-success-lower-bound} bounds exact
success from below by the complete-identification success on every input
$x\in\mathsf S(R_n)$, and the verification score has the same lower bound.
By Corollary~\ref{cor:statewise-fixed-W-achievability}, this complete-state
success probability is independent of the input and tends to one when
$k_n-n\to+\infty$.  Taking input minima or uniform averages proves the
refined exact and verification direct parts.  Together with the two converse
theorems, this gives copy rate one for every sequence
$m_n=\beta n+o(n)$ with $\beta>0$.

\section{Fixed-diagonal qubit graph-basis submodels}
\label{sec:fixed-diagonal-graph-basis-submodels}
This section considers a restriction of the standardized qubit model,
not a new promise on the original basis-free input.  The diagonal
$d=\operatorname{diag}A$ of the symmetric-matrix coordinate is publicly
known, while its off-diagonal entries and $b$ remain unknown.  Thus the
input states and legal partial reports form proper subsets of those in
the unrestricted coordinate problem.

The product phase gate $S(d)$ identifies each prescribed-diagonal family
with the zero-diagonal graph-basis family, including its report
projectors.  We establish that the restriction does not lower the
first-order copy threshold.  For exact identification, the unrestricted
converse transfers with at most a $2^{O(n)}$ multiplicative loss.
Achievability follows by mapping every complete-state decoder output
to a legal partial report for the prescribed diagonal and grouping the
decoder's POVM effects accordingly.  The compatible-list comparison
transfers the converse to the verification score.  These arguments establish the copy threshold without
asserting a separate finite PGM formula for the restricted ensemble.

\subsection{Graph-basis realization of a fixed diagonal}
\label{subsec:fixed-diagonal-graph-basis-submodels}

Write
\begin{equation}
 A=\Gamma+\operatorname{diag}(d),
 \quad \operatorname{diag}(\Gamma)=0.
 \label{eq:fixed-diagonal-decomposition}
\end{equation}
For a bit $a\in\F_2$, let $\widetilde a\in\{0,1\}\subset\mathbb Z$
denote its integer lift.  Define the product phase gate by
\begin{equation}
 S(d)\ket{x}
 :=i^{\sum_{j=1}^n\widetilde d_j\widetilde x_j}\ket{x},
 \qquad
 S(d)=\bigotimes_{j=1}^n\operatorname{diag}(1,i^{\widetilde d_j}),
 \label{eq:fixed-diagonal-phase-gate}
\end{equation}
where the exponent is read modulo four.  In particular, the exponent in
Eq.~\eqref{eq:fixed-diagonal-phase-gate} is the integer dot product of the
lifted bits, not the binary dot product in $\F_2$. With this integer-lift convention, the binary stabilizer state has the following graph-state decomposition:
\begin{equation}
 \ket{\psi_{A,b}^{(2)}}=S(d)Z(b)\ket{G_\Gamma},
 \label{eq:fixed-diagonal-graph-basis-decomposition}
\end{equation}
where $\ket{G_\Gamma}:=\ket{\psi_{\Gamma,0}^{(2)}}$ is the ordinary graph
state associated with the zero-diagonal adjacency matrix $\Gamma$
\cite{HeinEtAl2006GraphStates}.  Thus
$d=0$ gives the graph-basis family
$\{Z(b)\ket{G_\Gamma}:\Gamma=\Gamma^T,\operatorname{diag}(\Gamma)=0\}$,
with both the graph and the basis element unknown.  If $d$ is fixed and
publicly known, conjugation by $S(d)^\dagger$ reduces the model to this
zero-diagonal family.
Let $x(A,b)\in\mathsf{Stab}_{n,2}$ denote the projective state represented by
$\ket{\psi_{A,b}^{(2)}}$.  In these standardized coordinates, the physical
input family with prescribed diagonal $d$ is
\begin{equation}
 \mathsf S_d:=\{x(A,b):A\in\Sym_n(\F_2),\ \operatorname{diag}A=d,
                       \ b\in\F_2^n\}.
 \label{eq:fixed-diagonal-physical-state-family}
\end{equation}
The report alphabet for this promised model consists of pairs $(C,\gamma)$
in the coordinate output alphabet with
$\operatorname{diag}(J_m^TC)=J_m^Td$; the character $\gamma$ is otherwise
unrestricted within that alphabet.  We call this set the fixed-diagonal
report slice.  The same known unitary transports its report projectors:
for every report in this slice,
\begin{equation}
 S(d)^\dagger\Pi_{C,\gamma}^{(m,2)}S(d)
 =\Pi_{C-\operatorname{diag}(d)J_m,\gamma}^{(m,2)}.
 \label{eq:fixed-diagonal-report-projector-conjugation}
\end{equation}
Indeed, conjugating each tensor factor gives
$S(d)^\dagger Z(Cx)X(J_mx)S(d)$ with the $Z$ coordinate shifted by
$\operatorname{diag}(d)J_mx$; the change in the quadratic section phase is
exactly the corresponding diagonal term.  Hence no additional character shift
appears in Eq.~\eqref{eq:fixed-diagonal-report-projector-conjugation}.  Thus
the unitary equivalence acts simultaneously on the state family and on the
exact and verification report tests, rather than only on the state vectors.

For fixed $n,m,k$ and a publicly known diagonal $d\in\F_2^n$, let
$P_{\mathrm{id},d}^{(k),*}(n,m)$ and $S_{\mathrm{ver},d}^{(k),*}(n,m)$ denote the optimal
exact success probability and verification score when $(\Gamma,b)$ is uniform,
$A=\Gamma+\operatorname{diag}(d)$, and both the true and reported labels are
restricted to this fixed-diagonal alphabet.  The projection $\mathsf B$
in Eq.~\eqref{eq:binary-graded-rank-projection} assigns to each enhanced
quotient label $y$ an $n\times m$ partial symmetric-matrix column
$\mathsf B(y)\in\mathcal C_{n,m}$.  Its upper block
$J_m^T\mathsf B(y)\in\Sym_m(\F_2)$ is symmetric, as in the block
decomposition of Eq.~\eqref{eq:partial-symmetric-matrix-block-form}.
For a quotient label $y=[(\alpha,E)]$, the diagonal entries are
the visible parities:
\begin{equation}
 \operatorname{diag}\bigl(J_m^T\mathsf B(y)\bigr)
 =\bigl(\alpha_1\bmod2,\ldots,\alpha_m\bmod2\bigr)^T.
 \label{eq:fixed-diagonal-visible-parities}
\end{equation}
For the quotient label of a complete parameter $(A,b)$,
Eqs.~\eqref{eq:integrated-qubit-additive-coordinates} and
\eqref{eq:binary-quotient-coordinate-set-map} give
$\alpha_i=A_{ii}+2b_i\pmod4$ and $\mathsf B(y)=AJ_m$.
Thus the promise $\operatorname{diag}A=d$ gives
\begin{equation}
 \operatorname{diag}\bigl(J_m^T\mathsf B(y)\bigr)=J_m^Td.
 \label{eq:fixed-diagonal-true-label-projection}
\end{equation}
Define the restricted alphabet by
\begin{equation}
 Y_d:=\{y\in Y_{n,m,2}:\operatorname{diag}(J_m^T\mathsf B(y))=J_m^Td\}.
 \label{eq:fixed-diagonal-partial-label-slice}
\end{equation}
For $x=x(A,b)\in\mathsf S_d$, write $g(A,b)=(\alpha,E)$ for the enhanced
parameter in Eq.~\eqref{eq:integrated-qubit-additive-coordinates}.
The partial-target map $y_d:\mathsf S_d\to Y_d$ is
\begin{equation}
 y_d(x):=g(A,b)+H_{n,m,2}\in Y_d.
 \label{eq:fixed-diagonal-state-to-partial-label}
\end{equation}
Each state in $\mathsf S_d$ has a unique pair $(A,b)$, and
Eq.~\eqref{eq:fixed-diagonal-true-label-projection} places this quotient
label in $Y_d$.  Its coordinate report is given by
Eq.~\eqref{eq:binary-quotient-coordinate-set-map}.
Every value of the remaining visible coordinates has the same
number of complete-label extensions with diagonal $d$.  Hence uniform
$(\Gamma,b)$ induces the uniform prior on $Y_d$ required by
Lemma~\ref{lem:bounded-list-povm-conversion}.  The conversion
needs only a common upper bound on the restricted list sizes;
it does not require those sizes to be equal.  The diagonal condition
is a promise on the symmetric-matrix coordinates obtained after standardizing
the public frame; it is not a frame-independent property of the original
stabilizer state.
Fixing the diagonal removes $n$ binary coordinates from the complete
label and $m$ from the partial label.  Write $\mathsf X_d$ for the complete-label set with the prescribed
diagonal $d$.  The parameter-to-state correspondence identifies
$\mathsf X_d$ with the physical state-label family $\mathsf S_d$.
The remaining cardinalities are
\begin{equation}
 |\mathsf X_d|=2^{\binom n2+n}=2^{n(n+1)/2},
 \qquad |Y_d|=2^{-m}|Y_{n,m,2}|.
 \label{eq:fixed-diagonal-label-cardinalities}
\end{equation}
The reductions change the logarithms of the corresponding alphabet sizes
by $n$ and $m$, respectively.  If $m=\beta n+o(n)$, both changes are
$O(n)$ and leave the quadratic label-rate exponent
$(\beta-\beta^2/2)n^2+o(n^2)$ unchanged.

\subsection{Copy-rate threshold}
\label{subsec:fixed-diagonal-copy-rate-threshold}

The fixed-diagonal rate statement combines the two main mechanisms already
established.  For the converse, a decoder for the restricted alphabet is
embedded into one for the full binary symmetric-matrix model by guessing the
missing diagonal.  The resulting factor is only exponential in $n$ and is
therefore negligible relative to the quadratic converse exponent.  For the
direct part, the complete-state decoder is followed by a map into $Y_d$:
estimates in the promised state family yield their partial labels, and all
other estimates are assigned a fixed legal report.

\begin{theorem}[Copy-rate threshold for fixed-diagonal graph-basis submodels]
\label{thm:fixed-diagonal-graph-basis-threshold}
Let $m_n=\beta n+o(n)$ with $0<\beta\le1$.  Uniformly over every
sequence of publicly known diagonals $d_n\in\F_2^n$, the following
statements hold.

\medskip\noindent\textbf{(1) Exact converse.}
For the converse, let
$k_n=\alpha n+o(n)$ with $0\le\alpha<1$.  Then the exact success satisfies the quadratic converse bound:
\begin{equation}
 P_{\mathrm{id},d_n}^{(k_n),*}(n,m_n)
 \le 2^{-J(\alpha,\beta)n^2+o(n^2)}\longrightarrow0,
 \label{eq:fixed-diagonal-exact-converse}
\end{equation}

\medskip\noindent\textbf{(2) Verification converse.}
Under the same subcritical assumptions, the verification score satisfies
\begin{equation}
 S_{\mathrm{ver},d_n}^{(k_n),*}(n,m_n)\longrightarrow0.
 \label{eq:fixed-diagonal-graded-converse}
\end{equation}

\medskip\noindent\textbf{(3) Direct part.}
For the achievability direction, let instead $k_n-n\to+\infty$.  Then
$P_{\mathrm{id},d_n}^{(k_n),*}(n,m_n)$ and
$S_{\mathrm{ver},d_n}^{(k_n),*}(n,m_n)$ both tend to one.  Hence every fixed-diagonal submodel, including the
zero-diagonal graph-basis submodel, has the same proved first-order threshold:
vanishing performance at every normalized rate below one and successful
recovery under a diverging supercritical overhead.
\end{theorem}
\begin{proof}
\noindent\textbf{Part (1): Exact converse.}
All fixed diagonals are unitarily equivalent by
Eq.~\eqref{eq:fixed-diagonal-graph-basis-decomposition}, so it suffices to
prove the converse for $d=0$.  Given any exact decoder for the zero-diagonal
submodel, construct a decoder for the full binary symmetric-matrix model by
choosing $\widehat d\in\F_2^n$ uniformly, applying
$S(\widehat d)^\dagger$ to every input copy, running the zero-diagonal decoder,
and restoring $\widehat d$ in the reported label.  Write $s_0(\mathsf M_0)$ for the success probability of the chosen
zero-diagonal decoder, averaged over uniform $(\Gamma,b)$ in that slice.
Under the full prior, the actual diagonal $d\in\F_2^n$ is uniform and
independent of $(\Gamma,b)$; the guess $\widehat d$ is uniform and independent
of the input.  Conditional on $d=\widehat d$, the conjugation in
Eq.~\eqref{eq:fixed-diagonal-report-projector-conjugation} maps both the
state and the report back to the zero-diagonal problem.  The
conditional success averaged over $(\Gamma,b)$ is therefore
$s_0(\mathsf M_0)$ for every $d$.  All other outcomes contribute
nonnegative success probability, so the constructed full-model decoder
$\mathsf M_{\rm emb}$ satisfies
\begin{equation}
 \begin{aligned}
 s_{\rm full}(\mathsf M_{\rm emb})
 &\ge \sum_{d\in\F_2^n}
       \Pr(D=d)\Pr(\widehat D=d)\,s_0(\mathsf M_0)\\
 &=\sum_{d\in\F_2^n}2^{-2n}s_0(\mathsf M_0)
 =2^{-n}s_0(\mathsf M_0).
 \end{aligned}
 \label{eq:fixed-diagonal-decoder-average-embedding}
\end{equation}
Taking the supremum over zero-diagonal decoders and using the full-model
optimum yields
\begin{equation}
 P_{\mathrm{id},0}^{(k),*}(n,m)
 \le 2^n P_{\mathrm{id}}^{(k),*}(n,m;2).
 \label{eq:fixed-diagonal-decoder-comparison}
\end{equation}
The equality of conditional averages is not a statewise equality for an
arbitrary decoder.
Theorem~\ref{thm:refined-exact-converse} absorbs the factor $2^n$ into its
$o(n^2)$ term and proves Eq.~\eqref{eq:fixed-diagonal-exact-converse}.  This
argument uses finite PGM optimality only for the full symmetric-matrix model;
no finite-size PGM theorem for the restricted submodel is required.

\medskip\noindent\textbf{Part (2): Verification converse.}
For the verification task, let $Y_d$ be the fixed-diagonal partial-label alphabet and
put
\begin{equation}
 \mathcal N_{D,d}(\widehat y)
 :=\mathcal N_D^{(2)}(\widehat y)\cap Y_d.
 \label{eq:fixed-diagonal-restricted-neighborhood}
\end{equation}
Its cardinality is at most the full-model list bound $L_D(n,m;2)$, although
the restricted cardinality need not be independent of the center
$\widehat y$.  Apply Lemma~\ref{lem:bounded-list-povm-conversion} with
$\mathsf Y'=Y_d$, $\mathcal L(\widehat y)=\mathcal N_{D,d}(\widehat y)$, and
$L=L_D(n,m;2)$.  The resulting operators on $Y_d$ form a sub-POVM; adjoining
the positive failure operator, or merging it into an existing exact-label
outcome, gives a normalized exact decoder without decreasing its success.
Thus the conversion requires only the displayed uniform upper bound and not
equality of the restricted list sizes.  Repeating the low-rank/tail split of
Theorem~\ref{thm:rank-truncated-graded-to-exact} consequently gives
\begin{equation}
 S_{\mathrm{ver},d}^{(k),*}(n,m)
 \le L_D(n,m;2)P_{\mathrm{id},d}^{(k),*}(n,m)+2^{-(D+1)}.
 \label{eq:fixed-diagonal-graded-list-bound}
\end{equation}
With $D=\lfloor\sqrt n\rfloor$, the logarithm of the list factor is $o(n^2)$,
while the second term vanishes.  Equation
\eqref{eq:fixed-diagonal-exact-converse} therefore proves
Eq.~\eqref{eq:fixed-diagonal-graded-converse}.

\medskip\noindent\textbf{Part (3): Direct part.}
Apply the covariant unrestricted PGM of
Theorem~\ref{thm:unrestricted-stabilizer-direct}, with complete-state
effects $\{M_{\widehat x}\}$ and estimate $\widehat X$.
Its success probability is the same for every pure stabilizer state by
Corollary~\ref{cor:statewise-fixed-W-achievability}.
Before running this decoder, fix $y_d^0\in Y_d$; for example, take the
partial label of $x(\operatorname{diag}(d),0)\in\mathsf S_d$.
Define the following map from all complete-state estimates into $Y_d$:
\begin{equation}
 f_{d,m}(\widehat x):=
 \begin{cases}
  y_d(\widehat x),&\widehat x\in\mathsf S_d,\\
  y_d^0,&\widehat x\notin\mathsf S_d,
 \end{cases}
 \qquad \widehat x\in\mathsf{Stab}_{n,2}.
 \label{eq:fixed-diagonal-direct-total-map}
\end{equation}
This membership test includes both the standardized state family and the
prescribed diagonal.  Group the PGM effects into the legal reports by
\begin{equation}
 N_y^{d,m}:=\sum_{\widehat x:\,f_{d,m}(\widehat x)=y}M_{\widehat x},
 \qquad y\in Y_d.
 \label{eq:fixed-diagonal-direct-povm}
\end{equation}
These effects are positive and sum to $I$ by the same partition of complete
outcomes as in Eq.~\eqref{eq:complete-to-frame-povm-normalization}.
For every true input $x\in\mathsf S_d$, successful complete identification
therefore implies the correct partial report:
\begin{equation}
 \Pr_x\{f_{d,m}(\widehat X)=y_d(x)\}
 \ge\Pr_x\{\widehat X=x\}
 =P_{\rm stab,id}^{(k),*}(n;2).
 \label{eq:fixed-diagonal-direct-success-lower-bound}
\end{equation}
The same lower bound applies to the expected verification score, since the
correct report has score one and every score is nonnegative.
The right-hand side is independent of $x$ and $d$, and tends to one when
$k_n-n\to+\infty$.  Averaging over the promised inputs proves both direct
limits uniformly over $d_n$; the construction also supplies the same
statewise lower bound on every fixed-diagonal family.
\end{proof}

The zero-diagonal parameter set is a natural subgroup of the enhanced binary
parameter group and can in principle be given its own finite character-sector
and PGM analysis.  Theorem~\ref{thm:fixed-diagonal-graph-basis-threshold} does
not claim such a finite formula.  Its point is that the quadratic converse for
the full symmetric-matrix model already implies the copy-rate converse for
every fixed-diagonal graph-basis submodel, and the restricted compatible-list
argument transfers the conclusion to the verification task.

We conclude by interpreting the absence of a first-order copy-rate
discount and recording the remaining regimes.

\section{Conclusions and outlook}  
\label{sec:discussion}  
For fixed prime local dimension and $m_n/n\to\beta\in(0,1]$,
the complementary-stabilizer report has first-order learner-copy
threshold one.  Exact recovery and outcome-averaged verification,
under both worst-case and uniform-average criteria, vanish at every
limiting copy rate below one.  All four values tend to one when
$k_n-n\to+\infty$.

The requested report is smaller than a complete initial-state
description when $m<n$.  Its complementary subspace and eigenvalues,
together with the observed outcome on $W$, determine the conditional
pure stabilizer state.  They may omit the initial stabilizers in
$L_x\cap W$ and need not determine which measurement outcomes occur.
The threshold therefore concerns this specific report relation.

Two distinctions are essential to the proof.  A public transverse
frame gives an auxiliary single-valued target and permits an exact
discrimination calculation in coordinates.  The original learner
instead uses the full basis-free report alphabet.  Its verification
converse follows from a direct compatible-list bound on that alphabet.
The public auxiliary frame is not a resource in the original task.

This result does not transfer the sample complexities of shadow or
prediction protocols to exact stabilizer learning.  Those protocols
estimate specified properties under their own accuracy and confidence
requirements \cite{HuangKuengPreskill2020ClassicalShadows,
ChenYuZengFlammia2021RobustShadows,HadfieldEtAl2022LocallyBiasedShadows,
BertoniEtAl2024ShallowShadows}.  Output size alone does not fix the
number of copies needed for those different tasks.

The present results leave several quantitative questions unresolved.
First, the optimal success probabilities throughout the bounded-offset
window $k=n+O(1)$ are not determined here.  The finite PGM bound gives
an error estimate for $k\ge n+2$, while its convergence-to-one consequence
uses a diverging positive overhead.

Second, the regime $m=o(n)$ requires separate analysis.
The fixed-positive-$\beta$ theorem cannot be specialized to
$\beta=0$: its converse exponent and the list-size balance use a
positive linear target dimension.

Third, the present bounds do not determine the optimal success exponent
below threshold.
The full-prior upper bound contains an exceptional contribution
$O(p^{-m})$, which may dominate the generic quadratic bound.
The equality points in the limiting relaxed optimization establish
sharpness of that algebraic calculation, but do not give an
operational lower bound matching $J(\alpha,\beta)$.

Finally, arbitrary collective POVMs are permitted in the model.
The converse also bounds restricted measurement classes, whereas
efficient or local achievability needs separate analysis.
Agreement of first-order copy rates does not settle finite-size
constants, classical storage, postprocessing costs, or measurement
implementation.
Related studies of critical offsets, smaller targets, exponent
tightness, and implementation costs are described as work in Ref.~\cite{Pre}; none of their prospective conclusions
is used in the present proofs.

\section*{AI-assisted research and author responsibility}
Microsoft 365 Copilot and Codex, using GPT-5.4/5.6 for thinking and 
GPT-5.6 and 6.1 Sol/6 Astra for reviewing and editing, made substantive contributions to the refined formulation of the partial stabilizer identification problem. The authors formulated both the original basis-free problem and the refined frame-assisted problem and determined the conceptual relationship between them.

AI was used to carry out and organize calculations for the refined problem, to develop derivations and proof arguments, and to assist in converting these calculations into a complete mathematical presentation. The arguments transferring the refined results to the original problem were based principally on ideas supplied by the authors.

All AI-assisted material was critically reviewed, verified, and revised by the authors, who take full responsibility for the accuracy and integrity of the manuscript.

\appendix

\section*{Appendices}

\addcontentsline{toc}{section}{Appendices}
The main text contains the operational reductions, the common limiting
exponent, the unrestricted direct theorem, and the list-to-exact argument for verification.
The appendices are reserved for the longer characteristic-specific algebra
needed to justify the finite estimates used there.  The first two appendices
prove the odd-prime and characteristic-two structural estimates invoked by the
converse, including the finite completion exponents and exact clipping step.
Appendix~\ref{app:common-clipping} then gives the common endpoint-uniform
passage to normalized exponents.  The remaining appendices
give the diagnostic support-dimension comparison, odd-prime overlap and
rank-enumeration identities, and an auxiliary binary Fourier-energy
consistency check.  Each appendix opens
by identifying the main-text result it supports and the reason the details are
deferred.

\section{Proof of the odd-prime finite structural estimate}
\label{app:odd-prime-finite-structural-proof}
This appendix proves Theorem~\ref{thm:converse-odd-prime-finite-structural-estimate}, which supplies the finite odd-prime moment-count and lower-block probability bounds used in the converse.  The theorem statement and its role in the common exponent argument remain in the main text.  The details are deferred because they require a separate rank-constrained symmetric-completion count, exact rank reductions, and a row-space probability estimate.  We first count the enlarged completion set, then estimate the probability of the relevant lower moment, and finally assemble the finite exponent.

\subsection{Rank-constrained completion}
\label{sec:singular-completion}
Equation~\eqref{eq:weighted-support-bound} bounds exact success by
$\sum_{\zeta}\pi_{\zeta}d_{\zeta}/\lvert Y\rvert$.  Fix one lower moment
$\zeta=(Q_{22},S_2)$.  This subsection bounds its occurring-coordinate
factor $d_{\zeta}/\lvert Y\rvert$; the probability weight $\pi_{\zeta}$ is
estimated separately.

Every realizable upper coordinate above $\zeta$ determines a symmetric
completion of the augmented lower Gram matrix $B_{\zeta}$ with rank at
most $k$.  We enlarge this realizable set to all completions satisfying
that rank constraint, then parameterize and count the enlarged set by
its radical coupling and residual symmetric rank.  This is the sole
enlargement of the underlying completion set; later bounds on its rank
counts are inequalities for the resulting factors.  The next theorem
makes the inclusion and exact count precise.

\subsubsection{Completion set and moment-derived inclusion}
\label{subsubsec:odd-completion-set-inclusion}

\begin{theorem}[Rank-constrained completions above a fixed lower moment]
\label{thm:completion-bound-converse}
Let $h:=n-m+1$ and fix a lower moment
$\zeta=(\zeta_Q,\zeta_S)\in
\Sym_{n-m}(\F_p)\times\F_p^{n-m}$.  Its augmented Gram matrix is
\begin{equation}
 B_{\zeta}:=
 \begin{pmatrix}\zeta_Q&\zeta_S\\\zeta_S^T&k\end{pmatrix}
 \in\Sym_h(\F_p).
 \label{eq:completion-theorem-B-zeta}
\end{equation}
The realizable visible coordinates above this fixed lower moment form the set
\begin{equation}
 \mathcal D_{\zeta}:=
 \left\{(Q_{11},Q_{12},S_1):
 \begin{array}{l}
 \text{there exists }\boldsymbol x\in(\F_p^n)^k\text{ with }\\
 \Phi_k(\boldsymbol x)=(Q,S),\quad
 (Q_{22},S_2)=\zeta
 \end{array}\right\}.
 \label{eq:completion-theorem-D-zeta}
\end{equation}
By construction, the occurring-coordinate count is the cardinality of this set:
\begin{equation}
 d_{\zeta}=|\mathcal D_{\zeta}|.
 \label{eq:completion-theorem-d-zeta}
\end{equation}
Write the rank and nullity of the fixed lower block as
\begin{equation}
 r_B:=\rank B_{\zeta},\quad s:=h-r_B,
 \label{eq:lower-block-rank-nullity}
\end{equation}
and let $N_{\rm rect}(M,S,L;p)$ and $N_{\rm sym}(D,T;p)$ denote the numbers of
rank-$L$ rectangular matrices and rank-$T$ symmetric matrices, respectively.
If $k<r_B$, then $d_{\zeta}=0$.  When $k\ge r_B$, the number of
occurring moments is bounded by the following completion count:
\begin{align}
 d_{\zeta}\le \mathcal A_{m,h,k}(r_B)
 :={}&p^{m r_B}
 \sum_{\ell=0}^{\min\{m,s,\lfloor(k-r_B)/2\rfloor\}}
 N_{\rm rect}(m,s,\ell;p)
 p^{\ell m-\ell(\ell-1)/2}\notag\\
 &\quad\times
 \sum_{t=0}^{\min\{m-\ell,k-r_B-2\ell\}}
 N_{\rm sym}(m-\ell,t;p).
 \label{eq:exact-completion-count}
\end{align}
If $B,B'\in\Sym_h(\F_p)$ have the same rank, their enlarged completion
counts are equal:
\begin{equation}
 \mathcal A_{m,h,k}[B]=\mathcal A_{m,h,k}[B']
 =\mathcal A_{m,h,k}(\rank B).
 \label{eq:completion-count-rank-only-theorem}
\end{equation}

To average the fixed-lower-block estimate, let $X_2$ be uniform on
$\F_p^{(n-m)\times k}$ and define its augmented Gram rank by
\begin{equation}
 R(X_2):=\rank
 \begin{pmatrix}X_2X_2^T&X_2\1_k\\
 \1_k^TX_2^T&k\end{pmatrix}.
 \label{eq:random-augmented-rank}
\end{equation}
Grouping the weighted lower-moment bound by this random rank gives
\begin{equation}
 P_{\mathrm{id}}^{(k),*}(n,m;p)
 \le\sum_{r=0}^{h}\Pr\{R(X_2)=r\}
 \min\left\{1,\frac{\mathcal A_{m,h,k}(r)}{\lvert Y\rvert}\right\}.
 \label{eq:rank-grouped-bound}
\end{equation}
\end{theorem}
The proof is deferred until after
Lemma~\ref{lem:visible-to-ambient-completion},
Lemma~\ref{lem:first-congruence},
Lemma~\ref{lem:radical-rank-identity}, and
Proposition~\ref{prop:exact-ambient-completion-count}.  These results establish,
in that order, the actual-to-enlarged relaxation, the Schur-complement rank
reduction, the radical-coupling rank identity, and the exact count of the
enlarged set.
Accordingly, these four displayed results form the complete forward proof map for Theorem~\ref{thm:completion-bound-converse}; the only inequality in that path is the initial enlargement of the realizable moment set.

The initial inclusion enlarges realizable moments to rank-constrained
completions.  Within that enlarged set, a congruence separates the
invertible part of $B_\zeta$ from its radical.
Lemma~\ref{lem:radical-rank-identity} then gives the completed rank
$r_B+2\ell+\rank A_{22}$.  Finally,
Proposition~\ref{prop:exact-ambient-completion-count} counts the
completion parameters with this rank budget.

\medskip\noindent\emph{Relaxing realizable remaining-coordinate moments to low-rank completions.}\par

Fix $\zeta=(Q_{22},S_2)$.  The moment-derived visible data above this lower
coordinate are collected directly, without choosing a tuple factorization, in
\begin{equation}
 \mathcal V_{\zeta}
 :=\left\{
 \bigl(Q_{11},(Q_{12}\ S_1)\bigr):
 \begin{array}{l}
 (Q_{11},Q_{12},Q_{22},S_1,S_2)\in\operatorname{im}\Phi_k,\\
 (Q_{22},S_2)=\zeta
 \end{array}
 \right\}.
 \label{eq:actual-visible-completion-set}
\end{equation}
The set $\mathcal V_\zeta$ represents the same realizable data as
$\mathcal D_\zeta$ in Eq.~\eqref{eq:completion-theorem-D-zeta}, with the
last two coordinates joined into one block.  The joining map is
\begin{equation}
 \begin{aligned}
 \mathcal D_\zeta&\longrightarrow\mathcal V_\zeta,\\
 (Q_{11},Q_{12},S_1)&\longmapsto\bigl(Q_{11},(Q_{12}\ S_1)\bigr).
 \end{aligned}
 \label{eq:visible-triple-block-pair-bijection}
\end{equation}
Its inverse splits the $m\times(n-m+1)$ block into its first $n-m$
columns $Q_{12}$ and its last column $S_1$.  Both descriptions require
realization by the same full moment, so joining and splitting give
inverse maps.  Thus the occurring-coordinate count is the cardinality of
either set:
\begin{equation}
 d_\zeta=|\mathcal D_\zeta|=|\mathcal V_\zeta|.
 \label{eq:visible-realizable-cardinalities}
\end{equation}
Joining the blocks preserves realizability; the
enlargement to all rank-constrained completions is the next step.
The set defined in Eq.~\eqref{eq:actual-visible-completion-set}
depends only on the moment value; its definition does not require
choosing a particular factorization through $X_2$.  To prove the rank
constraint below, we choose a realizing ordered tuple in $(\F_p^n)^k$
and arrange its vectors as the columns of a matrix.

The fixed augmented matrix associated with $\zeta$ is
$B_{\zeta}=\bigl(\begin{smallmatrix}Q_{22}&S_2\\S_2^T&k\end{smallmatrix}\bigr)$
from Eq.~\eqref{eq:completion-theorem-B-zeta}.  It depends only on $\zeta$,
even though a factor $Z_2$ satisfying $B_{\zeta}=Z_2Z_2^T$ need not be unique.
Relax the realizable set to all symmetric completions satisfying the same rank
constraint.  The enlarged set is
\begin{equation}
 \mathcal R_{\zeta}
 :=\left\{(Q,C)\in\Sym_m(\F_p)\times\F_p^{m\times h}:
 \rank\begin{pmatrix}Q&C\\C^T&B_{\zeta}\end{pmatrix}\le k\right\}.
 \label{eq:enlarged-low-rank-completion-set}
\end{equation}
Its cardinality is denoted by
\begin{equation}
 \mathcal A_{m,h,k}[B_{\zeta}]
 :=|\mathcal R_{\zeta}|.
 \label{eq:A-completion-definition}
\end{equation}
The square brackets indicate that the matrix $B_{\zeta}$ itself is fixed.

\begin{lemma}[Inclusion of moment-derived data in the low-rank completion set]
\label{lem:visible-to-ambient-completion}
For fixed $\zeta$, let $\mathcal V_{\zeta}$ and $\mathcal R_{\zeta}$ be the sets
in Eqs.~\eqref{eq:actual-visible-completion-set} and
\eqref{eq:enlarged-low-rank-completion-set}.  Every completion
realized by a moment belongs to the enlarged set:
\begin{equation}
 \mathcal V_{\zeta}\subseteq\mathcal R_{\zeta}.
 \label{eq:moment-derived-completion-inclusion}
\end{equation}
The inclusion of moment-realizable completions in the enlarged completion set gives the following upper bound on the original support size:
\begin{equation}
 d_{\zeta}=|\mathcal V_{\zeta}|
 \le |\mathcal R_{\zeta}|
 =\mathcal A_{m,h,k}[B_{\zeta}].
 \label{eq:realizable-versus-ambient-completions}
\end{equation}
\end{lemma}
\begin{proof}
For $z=(Q,S)\in\mathcal C_{\zeta}\cap\operatorname{im}\Phi_k$, use the
block form of Eq.~\eqref{eq:moment-blocks} and retain precisely its visible
coordinates through the map
\begin{equation}
 \iota_{\zeta}(z)
 :=\bigl(Q_{11},(Q_{12}\ S_1)\bigr).
 \label{eq:visible-moment-to-completion-map}
\end{equation}
Because $(Q_{22},S_2)=\zeta$ is fixed throughout
$\mathcal C_{\zeta}$, the tuple $(Q_{11},Q_{12},S_1)$ uniquely specifies the
full moment value
$z=(Q,S)$.  Thus the map in
Eq.~\eqref{eq:visible-moment-to-completion-map} is injective.  Its image is exactly the set in
Eq.~\eqref{eq:actual-visible-completion-set}.  Hence this map identifies
the two finite sets and gives
\begin{equation}
 d_{\zeta}
 =|\mathcal C_{\zeta}\cap\operatorname{im}\Phi_k|
 =|\mathcal V_{\zeta}|.
 \label{eq:visible-moment-completion-bijection-count}
\end{equation}

To complete the comparison, we prove the inclusion $\mathcal V_{\zeta}\subseteq\mathcal R_{\zeta}$.  Take
$(Q_{11},C)\in\mathcal V_{\zeta}$.  By
Eq.~\eqref{eq:actual-visible-completion-set}, there is a realizing ordered
tuple $(x_1,\ldots,x_k)\in(\F_p^n)^k$.  Arrange it as the column matrix
$X=(x_1\ \cdots\ x_k)=\bigl(\begin{smallmatrix}X_1\\X_2\end{smallmatrix}\bigr)$
for the corresponding full moment.  For this realization, block multiplication gives
$C=X_1Z_2^T=(Q_{12}\ S_1)$.  Appending the linear moment coordinate then
produces the full augmented Gram factorization
\begin{equation}
 \begin{pmatrix}Q_{11}&C\\C^T&B_{\zeta}\end{pmatrix}
 =\begin{pmatrix}X\\\1_k^T\end{pmatrix}
  \begin{pmatrix}X^T&\1_k\end{pmatrix}.
 \label{eq:actual-completion-Gram-factorization}
\end{equation}
The matrix on the right of
Eq.~\eqref{eq:actual-completion-Gram-factorization} is a product of an
$(n+1)\times k$ matrix and a $k\times(n+1)$ matrix, so its rank is at most
$k$.  Hence $(Q_{11},C)\in\mathcal R_{\zeta}$ by
Eq.~\eqref{eq:enlarged-low-rank-completion-set}.  Combining this inclusion
with Eq.~\eqref{eq:visible-moment-completion-bijection-count} proves
Eq.~\eqref{eq:realizable-versus-ambient-completions}.
\end{proof}

The set $\mathcal R_{\zeta}$ contains every pair
$(Q_{11},C)$ obtained from a moment-derived triple
$(Q_{11},Q_{12},S_1)$ by setting $C=(Q_{12}\ S_1)$.
A general member of $\mathcal R_{\zeta}$ need not have the special Gram
factorization in Eq.~\eqref{eq:actual-completion-Gram-factorization}.  Thus the
count below is exact for the enlarged set $\mathcal R_{\zeta}$ and supplies
only an upper bound for the actual number $d_{\zeta}$.

\medskip\noindent\emph{Separating the invertible and radical parts of a fixed symmetric matrix.}\par
For a symmetric matrix $B$, write
$\operatorname{rad}(B):=\ker B$.  The next lemma separates the nondegenerate
part of $B$ from its radical and records every congruence used in that
reduction.

\subsubsection{Schur reduction and radical coupling}
\label{subsubsec:odd-schur-radical-reduction}

\begin{lemma}[Schur-complement rank reduction]
\label{lem:first-congruence}
Let $B\in\Sym_h(\F_p)$ have rank $r_B$ and nullity
$s=h-r_B$.  Choose a congruence that separates its nondegenerate and radical
parts:
\begin{equation}
 P^TBP=\begin{pmatrix}D_{r_B}&0\\0&0_s\end{pmatrix},
 \quad D_{r_B}\in\Sym_{r_B}(\F_p)
 \text{ invertible}.
 \label{eq:first-congruence-assumptions}
\end{equation}
Use the same change of basis to split the coupling matrix according to those
two parts:
\begin{equation}
 CP=(C_1\ C_0),
 \quad C_1\in\F_p^{m\times r_B},\quad
 C_0\in\F_p^{m\times s}.
 \label{eq:first-congruence-C-split}
\end{equation}
The Schur complement of the invertible block is
\begin{equation}
 A:=Q-C_1D_{r_B}^{-1}C_1^T.
 \label{eq:first-congruence-A-definition}
\end{equation}
The completed rank separates into the fixed nondegenerate contribution and a
residual radical-coupling rank:
\begin{equation}
 \rank\begin{pmatrix}Q&C\\C^T&B\end{pmatrix}
 =r_B+
 \rank\begin{pmatrix}A&C_0\\C_0^T&0_s\end{pmatrix}.
 \label{eq:first-rank-reduction}
\end{equation}
For fixed $C_1$, the map $Q\mapsto A$ is a bijection of
$\Sym_m(\F_p)$.
\end{lemma}
\begin{proof}
Apply the congruence $\operatorname{diag}(I_m,P)$ and use
Eqs.~\eqref{eq:first-congruence-assumptions} and
\eqref{eq:first-congruence-C-split}.  Set
\begin{equation}
 T:=\begin{pmatrix}
 I_m&0&0\\
 -D_{r_B}^{-1}C_1^T&I_{r_B}&0\\
 0&0&I_s
 \end{pmatrix}.
 \label{eq:first-congruence-explicit-T}
\end{equation}
The matrix $T$ is invertible and direct block multiplication gives
\begin{equation}
 T^T\begin{pmatrix}
 Q&C_1&C_0\\
 C_1^T&D_{r_B}&0\\
 C_0^T&0&0_s
 \end{pmatrix}T
 =\begin{pmatrix}
 A&0&C_0\\
 0&D_{r_B}&0\\
 C_0^T&0&0_s
 \end{pmatrix}.
 \label{eq:first-congruence-explicit-product}
\end{equation}
Permuting the second and third blocks now gives the direct sum
\begin{equation}
 \begin{pmatrix}
 Q&C_1&C_0\\
 C_1^T&D_{r_B}&0\\
 C_0^T&0&0_s
 \end{pmatrix}
 \sim
 D_{r_B}\oplus
 \begin{pmatrix}A&C_0\\C_0^T&0_s\end{pmatrix},
 \quad
 A=Q-C_1D_{r_B}^{-1}C_1^T.
 \label{eq:schur-direct-sum-form}
\end{equation}
Since $D_{r_B}$ has rank $r_B$, this proves
Eq.~\eqref{eq:first-rank-reduction}.  Moreover, $D_{r_B}^{-1}$ is symmetric,
so $A\in\Sym_m(\F_p)$; for fixed $C_1$, the map $Q\mapsto A$ is a
translation of $\Sym_m(\F_p)$ and hence a bijection.
\end{proof}

\medskip\noindent\emph{The exact rank cost of radical coupling.}\par
\begin{lemma}[Radical-coupling rank identity]
\label{lem:radical-rank-identity}
Let $A\in\Sym_m(\F_p)$ and let $C_0\in\F_p^{m\times s}$ have rank
$\ell$.  Choose row and column bases that put this coupling into its standard
rank normal form:
\begin{equation}
 RC_0S=\begin{pmatrix}I_\ell&0\\0&0\end{pmatrix}.
 \label{eq:radical-lemma-normal-form}
\end{equation}
Decompose the transformed symmetric matrix relative to the image of the
coupling, using $\F_p^m=\F_p^\ell\oplus\F_p^{m-\ell}$:
\begin{equation}
 RAR^T=\begin{pmatrix}A_{11}&A_{12}\\A_{12}^T&A_{22}\end{pmatrix}.
 \label{eq:radical-lemma-A-split}
\end{equation}
The coupling contributes a hyperbolic block of rank $2\ell$, leaving only the
rank of the residual symmetric block:
\begin{equation}
 \rank\begin{pmatrix}A&C_0\\C_0^T&0_s\end{pmatrix}
 =2\ell+\rank A_{22}.
 \label{eq:radical-rank-identity}
\end{equation}
\end{lemma}
\begin{proof}
The congruence $\operatorname{diag}(R^T,S)$, followed by deletion of the zero
$(s-\ell)$-dimensional direct summand, reduces the matrix to
\begin{equation}
 \begin{pmatrix}
 A_{11}&A_{12}&I_\ell\\
 A_{12}^T&A_{22}&0\\
 I_\ell&0&0
 \end{pmatrix}.
 \label{eq:radical-reduced-normal-form}
\end{equation}
Write $M_0$ for the matrix in
Eq.~\eqref{eq:radical-reduced-normal-form}, and define the invertible
block triangular matrix
\begin{equation}
 T_1:=\begin{pmatrix}
 I_\ell&0&0\\
 0&I_{m-\ell}&0\\
 0&-A_{12}&I_\ell
 \end{pmatrix}.
 \label{eq:radical-explicit-congruence-one}
\end{equation}
Direct multiplication eliminates both off-diagonal $A_{12}$ blocks:
\begin{equation}
 M_1:=T_1^TM_0T_1
 =\begin{pmatrix}
 A_{11}&0&I_\ell\\
 0&A_{22}&0\\
 I_\ell&0&0
 \end{pmatrix}.
 \label{eq:radical-explicit-intermediate}
\end{equation}
Since $p$ is odd, $2$ is invertible.  The second invertible
congruence matrix is
\begin{equation}
 T_2:=\begin{pmatrix}
 I_\ell&0&0\\
 0&I_{m-\ell}&0\\
 -\frac12 A_{11}&0&I_\ell
 \end{pmatrix}.
 \label{eq:radical-explicit-congruence-two}
\end{equation}
Since $A_{11}=A_{11}^T$, the upper-left block of
$T_2^TM_1T_2$ is $A_{11}-\frac12 A_{11}-\frac12 A_{11}=0$.
Permuting the resulting blocks shows that
Eq.~\eqref{eq:radical-reduced-normal-form} is congruent to
\begin{equation}
 \begin{pmatrix}0&I_\ell\\I_\ell&0\end{pmatrix}
 \oplus A_{22}.
 \label{eq:radical-hyperbolic-direct-sum}
\end{equation}
The first summand is nonsingular of rank $2\ell$, which proves
Eq.~\eqref{eq:radical-rank-identity}.
\end{proof}

The congruence that eliminates $A_{11}$ uses division by two and is not
available over $\F_2$.  Lemma~\ref{lem:binary-radical-rank-identity} in
Appendix~\ref{app:binary-finite-structural-proof} obtains the same rank cost
by an explicit kernel calculation without dividing by two.  Thus the rank constraint is common, while the elimination itself is
characteristic-dependent.

Substitute Eq.~\eqref{eq:radical-rank-identity} into the second rank term on
the right-hand side of Eq.~\eqref{eq:first-rank-reduction}.  Writing $t:=\rank A_{22}$, the two exact rank reductions combine to give
\begin{equation}
 \rank\begin{pmatrix}Q&C\\C^T&B_{\zeta}\end{pmatrix}
 =r_B+2\ell+t.
 \label{eq:full-rank-identity}
\end{equation}
Comparing Eq.~\eqref{eq:full-rank-identity} with the defining rank constraint
in Eq.~\eqref{eq:enlarged-low-rank-completion-set}, a completion belongs to
$\mathcal R_{\zeta}$ exactly when
\begin{equation}
 r_B+2\ell+t\le k.
 \label{eq:completion-rank-budget}
\end{equation}

\medskip\noindent\emph{Counting the enlarged low-rank completions.}\par
The objects counted below are the completion pairs $(Q,C)$ in
$\mathcal R_{\zeta}$.  We first make the parameterization underlying the
product count explicit.

\subsubsection{Fixed-rank parameterization and exact count}
\label{subsubsec:odd-fixed-rank-parameterization}

\begin{lemma}[Bijective parameterization at fixed completion ranks]
\label{lem:completion-bijective-parameterization}
Fix $B$, a matrix $P$ satisfying
Eq.~\eqref{eq:first-congruence-assumptions}, and integers $(\ell,t)$.  For each
rank-$\ell$ matrix $C_0\in\F_p^{m\times s}$, choose once and for all matrices
$R(C_0)$ and $S(C_0)$ satisfying
Eq.~\eqref{eq:radical-lemma-normal-form}.  For fixed ranks $\rank C_0=\ell$ and $\rank A_{22}=t$, the completion pairs
$(Q,C)$ are parametrized bijectively by the tuple
\begin{equation}
 (C_1,C_0,A_{11},A_{12},A_{22}),
 \label{eq:completion-parameter-tuple}
\end{equation}
where $C_1$ is arbitrary, $C_0$ has rank $\ell$,
$A_{11}\in\Sym_\ell(\F_p)$ and
$A_{12}\in\F_p^{\ell\times(m-\ell)}$ are arbitrary, and
$A_{22}\in\Sym_{m-\ell}(\F_p)$ has rank $t$.
\end{lemma}
\begin{proof}
For each fixed $C_1$, Lemma~\ref{lem:first-congruence} identifies $Q$
affinely and bijectively with $A\in\Sym_m(\F_p)$.  For each $C_0$, the
changes of basis $R(C_0),S(C_0)$ fixed in the statement give the unique
block decomposition
\[
 R(C_0)AR(C_0)^T
 =\begin{pmatrix}A_{11}&A_{12}\\A_{12}^T&A_{22}\end{pmatrix}.
\]
Conversely, these blocks reconstruct $A$, then $Q$, while
$C=(C_1\ C_0)P^{-1}$.  Hence the correspondence is bijective.  The fixed
choice of these basis changes for each $C_0$ prevents overcounting, and
Lemma~\ref{lem:radical-rank-identity} gives the stated rank classification.
\end{proof}

\begin{proposition}[Exact count of the enlarged low-rank completion set]
\label{prop:exact-ambient-completion-count}
The matrix-dependent count $\mathcal A_{m,h,k}[B]$ depends on
$B\in\Sym_h(\F_p)$ only through $r_B=\rank B$.  It is zero when
$k<r_B$.  When $k\ge r_B$, it equals
\begin{align}
 \mathcal A_{m,h,k}(r_B)
 :={}&p^{m r_B}
 \sum_{\ell=0}^{\min\{m,s,\lfloor(k-r_B)/2\rfloor\}}
 N_{\rm rect}(m,s,\ell;p)
 p^{\ell m-\ell(\ell-1)/2}\notag\\
 &\quad\times
 \sum_{t=0}^{\min\{m-\ell,k-r_B-2\ell\}}
 N_{\rm sym}(m-\ell,t;p).
 \label{eq:exact-enlarged-completion-count-derived}
\end{align}
\end{proposition}
\begin{proof}
If $k<r_B$, Eq.~\eqref{eq:first-rank-reduction} shows that every
completed matrix has rank at least $r_B>k$, so the enlarged set is
empty.

Assume $k\ge r_B$.  Partition the enlarged set by the two residual ranks
$\ell=\rank C_0$ and $t=\rank A_{22}$.  By
Lemma~\ref{lem:completion-bijective-parameterization}, the components of the
tuple in Eq.~\eqref{eq:completion-parameter-tuple} can be chosen freely and
successively within their stated sets.  Their cardinalities are:
\begin{align}
 |\{C_1\}|&=p^{m r_B},
 \label{eq:completion-count-C1}\\
 |\{C_0:\rank C_0=\ell\}|&=N_{\rm rect}(m,s,\ell;p),
 \label{eq:completion-count-C0}\\
 |\{(A_{11},A_{12})\}|
 &=p^{\ell(\ell+1)/2}\,p^{\ell(m-\ell)}
 =p^{\ell m-\ell(\ell-1)/2},
 \label{eq:completion-count-free-A-blocks}\\
 |\{A_{22}:\rank A_{22}=t\}|
 &=N_{\rm sym}(m-\ell,t;p).
 \label{eq:completion-count-A22}
\end{align}
Here $A_{11}$ has $\ell(\ell+1)/2$ independent symmetric entries and
$A_{12}$ has $\ell(m-\ell)$ arbitrary entries.

The exact rank identity imposes $r_B+2\ell+t\le k$.  Combining this budget
with the matrix dimensions gives the admissible summation ranges
\begin{equation}
 0\le\ell\le
 \min\{m,s,\lfloor(k-r_B)/2\rfloor\},
 \quad
 0\le t\le
 \min\{m-\ell,k-r_B-2\ell\}.
 \label{eq:completion-count-ranges}
\end{equation}
Multiplying Eqs.~\eqref{eq:completion-count-C1}--
\eqref{eq:completion-count-A22} for fixed $(\ell,t)$ and summing over the
admissible range proves
Eq.~\eqref{eq:exact-enlarged-completion-count-derived}.  This is exactly the
completion formula stated earlier in Eq.~\eqref{eq:exact-completion-count}.

Finally, the transformations $C\mapsto CP$ and $Q\mapsto A$ are bijective,
and every factor in the count depends only on $r_B$ and
$s=h-r_B$.  Hence the count depends on $B$ only through its rank.  In
particular, this conclusion does not use a classification of nondegenerate
symmetric forms and does not assume that rank determines their congruence
type.  Any dependence on the chosen invertible block $D_{r_B}$ is removed by
the affine translation $Q\mapsto A=Q-C_1D_{r_B}^{-1}C_1^T$ for each fixed
$C_1$.
\end{proof}

The proposition proves the rank-only identity
\begin{equation}
 \mathcal A_{m,h,k}[B]
 =\mathcal A_{m,h,k}(\rank B).
 \label{eq:A-completion-rank-definition}
\end{equation}

\begin{proof}[Proof of Theorem~\ref{thm:completion-bound-converse}]
The proof was deferred until the actual-to-enlarged relaxation and the exact
count of the enlarged set had been established.  If
$k<r_B$, Eq.~\eqref{eq:first-rank-reduction} implies that no actual
completion can have rank at most $k$, and hence $d_{\zeta}=0$.

Assume $k\ge r_B$.  Substitute the rank-only identity
Eq.~\eqref{eq:A-completion-rank-definition} into the last term of the
relaxation Eq.~\eqref{eq:realizable-versus-ambient-completions}.  This yields
\begin{equation}
 d_{\zeta}
 \le\mathcal A_{m,h,k}(\rank B_{\zeta}).
 \label{eq:dh-completion-bound}
\end{equation}
Independently,
$\mathcal C_{\zeta}\cap\operatorname{im}\Phi_k\subseteq\mathcal C_{\zeta}$
and Eq.~\eqref{eq:H-perp} give
\begin{equation}
 d_{\zeta}\le|\mathcal C_{\zeta}|=\lvert Y\rvert.
 \label{eq:dh-trivial-visible-cap}
\end{equation}
Intersecting the completion bound with the trivial coset-size cap gives the
normalized moment-count bound capped at one
\begin{equation}
 \frac{d_{\zeta}}{\lvert Y\rvert}
 \le\min\left\{1,
 \frac{\mathcal A_{m,h,k}(\rank B_{\zeta})}{\lvert Y\rvert}\right\}.
 \label{eq:dh-clipped-completion-bound}
\end{equation}
Substitute Eq.~\eqref{eq:dh-clipped-completion-bound} into the weighted bound
Eq.~\eqref{eq:weighted-support-bound}.  This yields
\begin{equation}
 P_{\mathrm{id}}^{(k),*}(n,m;p)
 \le\sum_{\zeta}\pi_{\zeta}
 \min\left\{1,
 \frac{\mathcal A_{m,h,k}(\rank B_{\zeta})}{\lvert Y\rvert}\right\}.
 \label{eq:rank-bound-before-grouping}
\end{equation}
Let $R:=\rank B_{\zeta}$ when $\zeta$ is distributed according to
$\pi_{\zeta}$.  Regrouping any function of the lower block by this rank gives
\begin{equation}
 \sum_{\zeta}\pi_{\zeta}f(\rank B_{\zeta})
 =\sum_{r=0}^{h}\Pr\{R=r\}f(r).
 \label{eq:lower-block-rank-regrouping-identity}
\end{equation}
Apply Eq.~\eqref{eq:lower-block-rank-regrouping-identity} to
Eq.~\eqref{eq:rank-bound-before-grouping} with
$f(r)=\min\{1,\mathcal A_{m,h,k}(r)/\lvert Y\rvert\}$.  The result is
Eq.~\eqref{eq:rank-grouped-bound}.
\end{proof}
If $s=0$, then only $\ell=0$ occurs, and the formula reduces to the ordinary Schur-complement count with an invertible fixed lower-coordinate block
\[
 \mathcal A_{m,h,k}(h)=p^{mh}
 \sum_{t=0}^{\min\{m,k-h\}}N_{\rm sym}(m,t;p).
\]
If $r_B=0$, every coupled radical direction costs exactly two units of rank.
More generally, increasing $\ell$ creates additional free entries of the completion matrix but
uses two units of the rank constraint, while increasing $t$ uses the remaining
budget one unit at a time.  This competition becomes the later optimization
over the scaled variables $\lambda$ and $\tau$.

The next subsection turns from the completion count to the probability pair
$(a,u)$, which determines the lower rank through $r_B=a-u$.

\subsection{Lower row-space invariants and probability bound}
\label{sec:lower-block-classes}
\label{subsubsec:odd-lower-class-probability}
The completion count depends on the fixed lower moment $\zeta$ only through
$r_B=\rank B_{\zeta}$. To weight that count by the correct probability,
we introduce the row space $\mathcal W=\operatorname{row}(Z_2)\subseteq\F_p^k$ and its dimension and radical dimension:
$a=\dim\mathcal W$ and $u=\dim(\mathcal W\cap\mathcal W^\perp)$.
The orthogonal complement here is taken with respect to the standard bilinear form on $\F_p^k$, not the symplectic form on $\F_p^{2n}$.
The identity $B_{\zeta}=Z_2Z_2^T$ gives $r_B=a-u$.
We now bound the probability of each $(a,u)$-event by the row-space estimate established in the main text.

\begin{table}[t]
\centering
\caption{Data passed from augmented-row-space event counting to the completion-count bound.}
\label{tab:unreported-visible-relation}
\small
\begin{tabularx}{\textwidth}{@{}l Y Y@{}}
\toprule
Quantity & Definition & Role in the weighted bound\\
\midrule
$\zeta$ & Fixed lower moment coordinates $(Q_{22},S_2)$ & Indexes the fiber $r^{-1}(\zeta)$.\\
$\pi_{\zeta}$ & Probability that uniform $X_2$ satisfies $\zeta(X_2)=\zeta$ & Supplies the probabilistic weight.\\
$d_{\zeta}/\lvert Y\rvert$ & Normalized cardinality of $r^{-1}(\zeta)\cap\operatorname{im}\Phi_k$ & Supplies the normalized occurring-moment count.\\
$B_{\zeta}$ & Augmented Gram matrix determined by $\zeta$ & Defines the rank-constrained completion count.\\
$(a,u)$ & $a=\dim\mathcal W$, $u=\dim(\mathcal W\cap\mathcal W^\perp)$ & Gives $\rank B_{\zeta}=a-u$ and indexes the probability estimate.\\
\bottomrule
\end{tabularx}
\end{table}

We now define the row-space parameters used in the probability theorem and
then bound the mass of every integer pair that can occur.  The dimension $a$
alone is insufficient because the completion rank is $a-u$.

\medskip\noindent\emph{Row-space parameters and the rank identity.}\par
The radical $\mathcal W\cap\mathcal W^\perp$ is the subspace of vectors in $\mathcal W$ orthogonal to every vector of $\mathcal W$; its dimension is denoted by $u$.

Let $X_2$ be uniform on $\F_p^{(n-m)\times k}$.  Adjoining the prescribed
all-ones row defines the lower augmented matrix and its row space:
\begin{equation}
 Z_2(X_2):=\begin{pmatrix}X_2\\\1_k^T\end{pmatrix},
 \quad
 \mathcal W(X_2):=\operatorname{row}(Z_2(X_2))\subseteq\F_p^k.
 \label{eq:random-unreported-row-space}
\end{equation}
The two invariants needed by the completion estimate are the row-space
dimension and the dimension of its radical:
\begin{equation}
 A(X_2):=\dim\mathcal W(X_2),
 \quad
 U(X_2):=\dim\bigl(\mathcal W(X_2)\cap\mathcal W(X_2)^\perp\bigr).
 \label{eq:a-u-definitions}
\end{equation}
Their joint event probability is denoted by
\begin{equation}
 \mathsf p_{n,m,k}^{\mathrm{odd}}(a,u)
 :=\Pr\{A(X_2)=a,\ U(X_2)=u\}.
 \label{eq:lower-block-class-mass-function}
\end{equation}
For a fixed realization of $X_2$, abbreviate
$Z_2=Z_2(X_2)$ and $\mathcal W=\mathcal W(X_2)$.  Because the final row of $Z_2$ is
$\1_k^T$, every such $\mathcal W$ contains the prescribed vector $\1_k^T$.  The Gram
matrix of the rows of $Z_2$ is $B_{\zeta(X_2)}$.  To compute the Gram rank without assuming that the rows of $Z_2$ are
independent, use the surjective coefficient-to-row-space map
\begin{equation}
 \varphi:\F_p^h\longrightarrow\mathcal W,
 \quad \varphi(c):=c^TZ_2.
 \label{eq:unreported-row-coefficient-map}
\end{equation}
Let $\langle\cdot,\cdot\rangle$ be the standard nondegenerate bilinear form
on $\F_p^k$.  The lower augmented Gram matrix $B_{\zeta}=Z_2Z_2^T$ represents the pullback
of the ambient bilinear form along this map:
\begin{equation}
 b_{\varphi}(c,d)
 :=\langle\varphi(c),\varphi(d)\rangle
 =c^TB_{\zeta}d.
 \label{eq:unreported-pullback-bilinear-form}
\end{equation}
A coefficient vector $c$ belongs to the radical of $b_{\varphi}$ exactly when
$\varphi(c)$ is orthogonal to every vector in $\mathcal W$, equivalently when
$\varphi(c)\in \mathcal W\cap\mathcal W^\perp$.  Thus the
radical of the pulled-back form is the following inverse image:
\begin{equation}
 \ker B_{\zeta}
 =\varphi^{-1}(\mathcal W\cap\mathcal W^\perp).
 \label{eq:unreported-Gram-kernel-preimage}
\end{equation}
Since $\varphi$ is surjective, $\dim\ker\varphi=h-a$.  Its inverse image of the
$u$-dimensional space $\mathcal W\cap\mathcal W^\perp$ therefore has dimension
$(h-a)+u$.  Rank--nullity therefore gives
the completion rank directly as
\begin{equation}
 r_B=\rank B_{\zeta}
 =h-\bigl((h-a)+u\bigr)=a-u.
 \label{eq:lower-block-rank-row-space}
\end{equation}
Therefore the completion factor for an $(a,u)$-class is
$\mathcal A_{m,h,k}(a-u)$: the rank-only count in
Eq.~\eqref{eq:A-completion-rank-definition} evaluated at
$r_B=a-u$.  Equation~\eqref{eq:lower-block-rank-row-space} remains
valid whether or not $\1_k$ is isotropic, and hence whether or not $p$ divides
$k$.

\medskip\noindent\emph{Application of the common row-space bound.}\par
The augmented row space $\mathcal W=\operatorname{row}(Z_2)$ equals
$W_{\1_k}(X_2)=\operatorname{span}(\{\1_k\}\cup\operatorname{row}X_2)$.
It is therefore governed directly by
Corollary~\ref{cor:all-prime-augmented-row-space-event-bound}, with
$v=\1_k$ and $r=n-m$.  This formulation does not require
$\1_k\in\operatorname{row}X_2$ and gives the following branch-specific
bound.

\begin{proposition}[Odd-prime lower-block class probability]
\label{prop:lower-block-class-probability}
For every positive-mass integer pair $(a,u)$, the augmented row-space event
obeys the uniform bound
\begin{equation}
 \mathsf p_{n,m,k}^{\mathrm{odd}}(a,u)
 \le
 (K_p^{\mathrm{class}})^k
 p^{-(k-a)(n-m-a)-u(u+1)/2}.
 \label{eq:lower-block-class-probability}
\end{equation}
\end{proposition}
\begin{proof}
Apply Corollary~\ref{cor:all-prime-augmented-row-space-event-bound} with
$Y=X_2$, $v=\1_k$, and $r=n-m$.  The conclusion is independent of whether
$\1_k$ is isotropic, and hence also covers the case $p\mid k$.
\end{proof}

The probability side is now closed: it depends only on the pair $(a,u)$.
The next subsection returns to the completion side and places its exact count on a
uniform exponential scale before the two factors are combined.

\subsection{Uniform completion exponent}
\label{sec:uniform-completion-exponent}
Fixing $(a,u)$ fixes the lower Gram rank $a-u$. We now bound the exact enlarged completion count from above using the ranks $(\ell,t)$ of its rectangular and symmetric blocks. The probability of the lower row-space event is estimated separately, not in this subsection.

\subsubsection{Uniform rectangular- and symmetric-rank estimates}
\label{subsubsec:odd-uniform-rank-counts}
The completion bound uses two rank-counting factors.  We first bound the
number of rectangular matrices and then bound the number of symmetric
matrices, in each case using constants independent of the matrix dimensions.

The rectangular factor is controlled from its exact rank formula
\begin{equation}
 N_{\rm rect}(M,S,L;p)
 =\prod_{i=0}^{L-1}
 \frac{(p^M-p^i)(p^S-p^i)}{p^L-p^i}.
 \label{eq:rectangular-exact-formula}
\end{equation}
Factoring out the powers of $p$ and using
Eq.~\eqref{eq:kappa-GL-product} gives the dimension-uniform bound
\begin{equation}
 N_{\rm rect}(M,S,L;p)
 \le (\kappa_p^{\mathrm{GL}})^{-1}p^{L(M+S-L)}.
 \label{eq:rectangular-uniform-bound}
\end{equation}

The symmetric factor is bounded separately by the following lemma.

\begin{lemma}[Uniform symmetric-rank bound]
\label{lem:symmetric-rank-bound}
For fixed odd $p$, there is a constant $K_p^{\mathrm{sym}}>0$,
independent of $D$ and $T$, such that the symmetric-matrix count
obeys the following bound for every $0\le T\le D$:
\begin{equation}
 N_{\rm sym}(D,T;p)
 \le K_p^{\mathrm{sym}}p^{T(2D-T+1)/2}.
 \label{eq:symmetric-rank-bound}
\end{equation}
\end{lemma}

\begin{proof}
We separate the representation, its change of basis, and the count.
Every rank-$T$ matrix $M\in\Sym_D(\F_p)$ admits a full-column-rank
factorization of the form
\begin{equation}
 M=XJX^T,
 \label{eq:symmetric-rank-factorization}
\end{equation}
where $X\in\F_p^{D\times T}$ has full column rank and
$J\in\Sym_T(\F_p)$ is invertible.  Choose the columns of $X$ as a basis of
$\operatorname{im}M$, choose a left inverse $L$ with $LX=I_T$, and put
\begin{equation}
 J:=LML^T.
 \label{eq:v84-symmetric-factor-J}
\end{equation}
Then $J$ is symmetric.  Since $XL$ is the identity on
$\operatorname{im}X=\operatorname{im}M$, one has $XLM=M$; transposing and
using $M=M^T$ gives $ML^TX^T=M$.  The proposed factor
$J=LML^T$ consequently satisfies the required factorization:
\begin{equation}
 XJX^T=XLML^TX^T=M.
 \label{eq:v84-symmetric-factor-verification}
\end{equation}
Finally,
$T=\rank M=\rank(XJX^T)\le\rank J\le T$, so $J$ is invertible.

Changing the basis of the image space produces the action
\begin{equation}
 (X,J)\cdot G=(XG,G^{-1}JG^{-T}).
 \label{eq:symmetric-factor-change-of-basis}
\end{equation}
Equation~\eqref{eq:symmetric-rank-factorization} is unchanged under
Eq.~\eqref{eq:symmetric-factor-change-of-basis}.  Conversely, two full-rank
factorizations of the same $M$ have the same column space
$\operatorname{im}M$, hence their matrices $X$ differ by a unique
$G\in\operatorname{GL}(T,p)$; the equality of the two factorizations then
forces the displayed transformation of $J$.  Thus the fibers of the map
$(X,J)\mapsto XJX^T$ have cardinality exactly
$|\operatorname{GL}(T,p)|$.

The number of full-column-rank $D\times T$ matrices is
$N_{\rm rect}(D,T,T;p)$ and, by Eq.~\eqref{eq:rectangular-uniform-bound}, is at
most
\begin{equation}
 (\kappa_p^{\mathrm{GL}})^{-1}p^{DT}.
 \label{eq:symmetric-count-full-rank-X}
\end{equation}
The exact nonsingular symmetric-form count is
Eq.~\eqref{eq:MacWilliams-full-rank-count}; its uniform bound
Eq.~\eqref{eq:nonsingular-symmetric-uniform-bounds} gives
\begin{equation}
 |\{J\in\Sym_T(\F_p):J\text{ invertible}\}|
 \le p^{T(T+1)/2}.
 \label{eq:symmetric-count-invertible-J}
\end{equation}
Thus no dimension- or orthogonal-type-dependent constant is needed in this
factor.  The remaining group-order factor has the uniform
lower bound
\begin{equation}
 |\operatorname{GL}(T,p)|
 \ge\kappa_p^{\mathrm{GL}}p^{T^2}.
 \label{eq:symmetric-count-GL-lower}
\end{equation}
Multiply Eqs.~\eqref{eq:symmetric-count-full-rank-X} and
\eqref{eq:symmetric-count-invertible-J}, divide by
Eq.~\eqref{eq:symmetric-count-GL-lower}, and simplify the exponent:
\begin{equation}
 DT+\frac{T(T+1)}2-T^2
 =\frac{T(2D-T+1)}2.
 \label{eq:symmetric-rank-exponent-assembly}
\end{equation}
Absorbing the two fixed product constants into $K_p^{\mathrm{sym}}$ proves
Eq.~\eqref{eq:symmetric-rank-bound}.
\end{proof}

\begin{corollary}[Symmetric rank partial sums]
\label{cor:symmetric-rank-partial-sum}
For $0\le R\le D$, the sum of the symmetric rank counts satisfies
\begin{equation}
 \sum_{T=0}^R N_{\rm sym}(D,T;p)
 \le K_p^{\mathrm{sym}}(D+1)p^{R(2D-R+1)/2}.
 \label{eq:symmetric-partial-sum}
\end{equation}
\end{corollary}

\begin{proof}
The exponent $f_D(T)=T(2D-T+1)/2$ is nondecreasing on $0\le T\le D$, because $f_D(T+1)-f_D(T)=D-T$.  Hence every term with $T\le R$ is bounded by the $T=R$ exponent, and there are at most $D+1$ terms.
\end{proof}

\medskip\noindent\emph{Finite exponent in the completion bound.}\par
Fix a positive-mass probability class $(a,u)$.  For completion ranks
$\ell,t\ge0$, define the exponent of the $p$-power part of the upper bound on each corresponding summand of the enlarged completion formula by
\begin{align}
 E_{\rm fin}(m,h;a,u,\ell,t)
 :={}&m(a-u)+\ell(m+h-a+u-\ell)
 +\ell m-\frac{\ell(\ell-1)}2\notag\\
 &+\frac{t[2(m-\ell)-t+1]}2.
 \label{eq:E-fin}
\end{align}
Equation~\eqref{eq:exact-completion-count} counts the enlarged completion
family exactly.  The original support size $d_\zeta$ is bounded above by
that count.  Its factors $N_{\rm rect}$ and $N_{\rm sym}$ are exact
cardinalities, not necessarily powers of $p$.

The uniform rank estimates in
Eqs.~\eqref{eq:rectangular-uniform-bound} and
\eqref{eq:symmetric-rank-bound} bound each summand of the enlarged
completion count by $p^{E_{\rm fin}(m,h;a,u,\ell,t)}$ times a constant
depending only on $p$.  At fixed ranks $\rank C_0=\ell$ and
$\rank A_{22}=t$, the four contributions to this upper-bound exponent are
\begin{equation}
 \begin{aligned}
 E_{C_1}&=m(a-u),\\
 E_{C_0}&=\ell(m+h-a+u-\ell),\\
 E_{A_{11},A_{12}}&=\ell m-\frac{\ell(\ell-1)}2,\\
 E_{A_{22}}&=\frac{t[2(m-\ell)-t+1]}2.
 \end{aligned}
 \label{eq:odd-completion-factor-exponent-dictionary}
\end{equation}
The factors $p^{E_{C_1}}$ and $p^{E_{A_{11},A_{12}}}$ give exact
numbers of choices for $C_1$ and the free $(A_{11},A_{12})$ coordinates.
By contrast, $E_{C_0}$ and $E_{A_{22}}$ occur in upper bounds for
fixed-rank matrix counts.  Their four contributions sum to $E_{\rm fin}$
in Eq.~\eqref{eq:E-fin}.

The dimension-independent rank-count constants enter as
$C_p^{\mathrm{comp}}$ in Eq.~\eqref{eq:completion-bound-before-max}.
The at most $(m+1)^2$ admissible rank pairs lead to
Eq.~\eqref{eq:completion-max-term-bound}.

\subsubsection{Completion exponent bound}
\label{subsubsec:odd-completion-exponent-bound}

The rectangular- and symmetric-rank estimates bound the enlarged completion
count for lower blocks of fixed rank $a-u$.  The following proposition gives
the finite bound in terms of $E_{\rm fin}$.

\begin{proposition}[Finite completion bound for rank $a-u$]
\label{prop:uniform-completion-bound}
Let
\begin{equation}
 C_p^{\mathrm{comp}}:=(\kappa_p^{\mathrm{GL}})^{-1}K_p^{\mathrm{sym}}.
 \label{eq:odd-completion-uniform-constant}
\end{equation}
For every $(a,u)\in\mathfrak P_n$, the exact completion sum satisfies the
fully finite bound
\begin{equation}
 \mathcal A_{m,h,k}(a-u)
 \le C_p^{\mathrm{comp}}(m+1)^2
 p^{\max_{(\ell,t)\in\mathcal G_{a,u}}
 E_{\rm fin}(m,h;a,u,\ell,t)}.
 \label{eq:uniform-completion-bound}
\end{equation}
Here the maximum is taken over the integer pairs $(\ell,t)$ satisfying
the following admissibility conditions:
\begin{equation}
 0\le\ell\le\min\{m,h-a+u\},\quad
 0\le t\le m-\ell,\quad
 a-u+2\ell+t\le k.
 \label{eq:finite-completion-ranges}
\end{equation}
If no such pair exists, then $\mathcal A_{m,h,k}(a-u)=0$.
The constant $C_p^{\mathrm{comp}}$ depends only on $p$ and is independent of
$m,h,k,a,u,\ell,t$.
\end{proposition}

\begin{proof}
Start from the exact count in Eq.~\eqref{eq:exact-completion-count}, with
$r_B=a-u$ and $s=h-a+u$.  For a fixed admissible pair $(\ell,t)$,
substitute the rectangular estimate Eq.~\eqref{eq:rectangular-uniform-bound}
and the symmetric estimate Eq.~\eqref{eq:symmetric-rank-bound} into the
corresponding summand of Eq.~\eqref{eq:exact-completion-count}.  Up to a
constant depending only on $p$, that summand is bounded by
\begin{align}
 p^{m(a-u)}
 p^{\ell[m+(h-a+u)-\ell]}
 p^{\ell m-\ell(\ell-1)/2}
 p^{t[2(m-\ell)-t+1]/2}
 =p^{E_{\rm fin}(m,h;a,u,\ell,t)},
 \label{eq:completion-summand-exponent-assembly}
\end{align}
where the equality is exactly the definition in Eq.~\eqref{eq:E-fin}.  The
rectangular and symmetric rank bounds contribute the dimension-independent constant
$C_p^{\mathrm{comp}}$ from
Eq.~\eqref{eq:odd-completion-uniform-constant}.  Bounding each
admissible rank pair gives the following upper bound on the enlarged
completion count:
\begin{equation}
 \mathcal A_{m,h,k}(a-u)
 \le C_p^{\mathrm{comp}}\sum_{(\ell,t)\in\mathcal G_{a,u}}
 p^{E_{\rm fin}(m,h;a,u,\ell,t)}.
 \label{eq:completion-bound-before-max}
\end{equation}
The set $\mathcal G_{a,u}$ has at most $(m+1)^2$ elements.  Therefore
Eq.~\eqref{eq:completion-bound-before-max} implies
\begin{align}
 \mathcal A_{m,h,k}(a-u)
 &\le C_p^{\mathrm{comp}}(m+1)^2
 p^{\max_{(\ell,t)\in\mathcal G_{a,u}}
 E_{\rm fin}(m,h;a,u,\ell,t)}.
 \label{eq:completion-max-term-bound}
\end{align}
This is Eq.~\eqref{eq:uniform-completion-bound}.  The constant and the bound
$|\mathcal G_{a,u}|\le(m+1)^2$ are independent of $(a,u)$, proving the stated
uniformity.
\end{proof}

Normalization is by the full visible-coordinate alphabet.  Its cardinality is
\begin{equation}
 \lvert Y\rvert=p^{D_{\rm vis,fin}},
 \quad
 D_{\rm vis,fin}:=mh+\frac{m(m+1)}2.
 \label{eq:D-vis-fin}
\end{equation}
Substituting $h=n-m+1$ gives
\begin{align}
 mh+\frac{m(m+1)}2
 =m(n-m+1)+\frac{m(m+1)}2
 =(n+1)m-\frac{m(m-1)}2,
 \label{eq:visible-dimension-finite-algebra}
\end{align}
which agrees with the exact exponent of $|Y|$ in
Eq.~\eqref{eq:Y-size}.  For any fixed feasible row-space parameters
$(a,u)$, normalization by the visible alphabet subtracts
$D_{\rm vis,fin}$ from the completion exponent
$\max_{(\ell,t)\in\mathcal G_{a,u}}E_{\rm fin}$.

To pass from these finite exponents to the limiting optimization, let $m/n=\beta_n\to\beta$ and
$k/n=\alpha_n\to\alpha$.  Along any convergent subsequence of
admissible rank parameters, write
$a/n\to a_0$, $u/n\to u_0$, $\ell/n\to\lambda$, and
$t/n\to\tau$.  Dividing the finite rank budget
$a-u+2\ell+t\le k$ by $n$ gives
$a_0-u_0+2\lambda+\tau\le\alpha$.
In the normalized expressions $E$ and $D$ of
Eqs.~\eqref{eq:direct-finite-E} and \eqref{eq:direct-finite-D},
respectively, the finite ratios $\alpha_n,\beta_n$ are replaced by
their limits $\alpha,\beta$.  The quadratic terms of
$E_{\rm fin}/n^2$ and $D_{\rm vis,fin}/n^2$ then converge to
these expressions, while the linear and logarithmic corrections vanish.
Thus the limiting normalized completion contribution is $E-D$;
no finite-size ratio is identified with its limit.

The branch-specific matrix data have now been reduced to a finite domain and
an exponent indexed by the probability pair $(a,u)$ and the completion pair
$(\ell,t)$.  The proof below assembles the completion bound, the row-space
probability estimate, and the uniform exponent estimate into the stated
structural theorem.

\subsection{Completion of the proof of the odd-prime finite structural estimate}
\label{subsec:odd-finite-proof-assembly}

\begin{proof}[Proof of Theorem~\ref{thm:converse-odd-prime-finite-structural-estimate}]
The proof now combines three independently established finite statements: the
rank-constrained completion bound, the lower row-space probability estimate,
and the uniform completion-exponent bound.  The first bounds the number of occurring full moments, the
second supplies the lower-event probability weight, and the third places the completion
count on the exponent scale required by the common asymptotic theorem.

\smallskip
\noindent\textbf{Step 1: Rank-constrained completion and clipping.}
Fix a lower-coordinate value $\zeta$ and put
\begin{equation}
 B_{\zeta}:=
 \begin{pmatrix}Q_{22}&S_2\\S_2^T&k\end{pmatrix}.
 \label{eq:theorem-B-B-zeta}
\end{equation}
The canonical block decomposition in Proposition~\ref{prop:converse-weighted-moment-fiber-reduction}
shows that every occurring visible coordinate triple above $\zeta$ produces a
symmetric completion of rank at most $k$:
\begin{equation}
 \begin{pmatrix}Q_{11}&C\\C^T&B_{\zeta}\end{pmatrix}.
 \label{eq:theorem-B-completion-matrix}
\end{equation}
Lemma~\ref{lem:visible-to-ambient-completion} enlarges
the realizable set only in the upper-bound direction.  Congruence separates
the nondegenerate part of $B_{\zeta}$ from its radical, and
Lemma~\ref{lem:radical-rank-identity} shows that a radical coupling of rank
$\ell$ costs exactly $2\ell$ units of rank.  The remaining symmetric block has
rank $t$.  Proposition~\ref{prop:exact-ambient-completion-count} therefore
gives the exact enlarged count in Eq.~\eqref{eq:exact-completion-count}, which
depends on the fixed lower matrix only through
$r_B=\rank B_{\zeta}$.

For a lower block with row-space parameters $(a,u)$,
Eq.~\eqref{eq:lower-block-rank-row-space} gives
\begin{equation}
 r_B=a-u.
 \label{eq:theorem-B-rB-au}
\end{equation}
The completion indices consequently obey
\begin{equation}
 0\le\ell\le\min\{m,h-a+u\},\quad
 0\le t\le m-\ell,\quad
 a-u+2\ell+t\le k.
 \label{eq:theorem-B-proof-finite-completion-ranges}
\end{equation}
Since the actual number
of occurring visible coordinates is at most both this enlarged count and
$\lvert Y\rvert$, the clipped factor used in the exponent comparison is
\begin{equation}
 \frac{d_{\zeta}}{\lvert Y\rvert}
 \le
 \min\left\{1,
 \frac{\mathcal A_{m,h,k}(a-u)}{\lvert Y\rvert}\right\}.
 \label{eq:theorem-B-proof-clipped-completion-bound}
\end{equation}
The minimum with one is not an asymptotic convenience; it is the exact finite
cardinality cap.

\smallskip
\noindent\textbf{Step 2: Probability of the row-space event.}
Let $Z_2=(X_2;\1_k^T)$ and let $R=\operatorname{row}Z_2$.  Write
$a=\dim R$ and $u=\dim(R\cap R^\perp)$.  The kernel argument in
Eq.~\eqref{eq:lower-block-rank-row-space} gives
$\rank(Z_2Z_2^T)=a-u$, consistently with the completion rank used above.  Counting row spaces of
dimension $a$ and radical dimension $u$, and then counting matrices with that
row space, gives Proposition~\ref{prop:lower-block-class-probability}:
\begin{equation}
 \mathsf p_{n,m,k}^{\mathrm{odd}}(a,u)
 \le
 (K_p^{\mathrm{class}})^k
 p^{-(k-a)(n-m-a)-u(u+1)/2}.
 \label{eq:theorem-B-proof-lower-block-probability}
\end{equation}
The constant depends only on the fixed prime and the estimate is uniform over
all positive-mass pairs.  The totally isotropic-subspace estimates used here
are upper bounds, so no unproved classification or equality assumption enters
the probability estimate.

\smallskip
\noindent\textbf{Step 3: Combining the finite factors.}
Group the weighted bound of
Proposition~\ref{prop:converse-weighted-moment-fiber-reduction} by the pair $(a,u)$.
Multiplying the probability of each $(a,u)$ class by its clipped
completion factor gives
\begin{equation}
 P_{\mathrm{id}}^{(k),*}(n,m;p)
 \le\sum_{(a,u)\in\mathfrak P_n}
 \mathsf p_{n,m,k}^{\mathrm{odd}}(a,u)
 \min\left\{1,
 \frac{\mathcal A_{m,h,k}(a-u)}{\lvert Y\rvert}\right\}.
 \label{eq:finite-probability-comparison-3}
\end{equation}
By Eq.~\eqref{eq:theorem-B-D-vis-fin},
$\lvert Y\rvert=p^{D_{\rm vis,fin}}$.  Substitute the finite bounds from
Proposition~\ref{prop:uniform-completion-bound} and
Eq.~\eqref{eq:lower-block-class-probability} into
Eq.~\eqref{eq:finite-probability-comparison-3}.  In particular, the
completion factor is bounded by
\[
 C_p^{\mathrm{comp}}(m+1)^2
 p^{\max_{(\ell,t)\in\mathcal G_{a,u}}E_{\rm fin}-D_{\rm vis,fin}},
\]
and the probability factor retains $(K_p^{\mathrm{class}})^k$.
These are the explicit finite factors of
Eq.~\eqref{eq:finite-probability-comparison-4}.

For the asymptotic presentation, the same rank-count estimates give
the following uniform exponent form.
Applying the uniform rectangular- and symmetric-rank estimates term by term to
the exact count yields Proposition~\ref{prop:uniform-completion-bound}:
\begin{equation}
 \mathcal A_{m,h,k}(a-u)
 \le
 p^{\max_{(\ell,t)\in\mathcal G_{a,u}}
 E_{\rm fin}(m,h;a,u,\ell,t)+O_p(\log n)}.
 \label{eq:theorem-B-proof-uniform-completion-bound}
\end{equation}
The $O_p(\log n)$ error in the exponent of
Eq.~\eqref{eq:theorem-B-proof-uniform-completion-bound} is uniform over
the admissible rank parameters when $m,k=O(n)$.

Taking the base-$p$ logarithm uses the exact identity
\begin{equation}
 \log_p\min\{1,p^x\}=\min\{0,x\}.
 \label{eq:finite-clipping-log-identity}
\end{equation}
Moreover, the map $x\mapsto\min\{0,x\}$ is one-Lipschitz.  Consequently an
additive remainder bounded uniformly before clipping remains bounded by the
same remainder after clipping, including when the maximizing exponent crosses
zero.  No separation from the clipping boundary is required.

Finally, clipping and the probability prefactor are converted to base-$p$
exponents by
\begin{equation}
 \min\{1,p^x\}=p^{\min\{0,x\}},
 \quad
 (K_p^{\mathrm{class}})^k=p^{O_p(k)}.
 \label{eq:theorem-B-clipping-identities}
\end{equation}
When $m,h,k=O(n)$ these identities convert
Eq.~\eqref{eq:finite-probability-comparison-4} into
Eq.~\eqref{eq:finite-probability-comparison-5}, with uniform remainders.  This is
the claimed odd-prime finite structural estimate.
\end{proof}

The odd-prime branch is now complete.  The individual lower moments and all
completion matrices have been absorbed into the two structural parameter
pairs, their finite feasible domain, the exponent $E_{\rm fin}$, and the
clipping rule.  Only this interface is carried into the common asymptotic
reduction.

\section{Proof of the characteristic-two finite structural estimate}
\label{app:binary-finite-structural-proof}
This appendix proves Theorem~\ref{thm:converse-binary-finite-structural-estimate}, the characteristic-two input to the all-prime converse.  The binary case cannot be obtained by formally setting $p=2$ in the odd-prime quadratic argument: the enhanced moment contains an associated binary bilinear matrix and additional mod-$4$ diagonal data.  The proof is therefore separated into a bilinear completion count, the comparison with enhanced moments and its $2^m$ lift cost, the common row-space probability estimate, and the finite exponent assembly.

\subsection{Bilinear completion count}
\label{sec:integrated-qubit-completion}
The characteristic-two argument first counts completions of the bilinear form
associated with the enhanced quadratic data.  At this level the
radical-coupling and residual-rank parameterization parallels the odd-prime
calculation.  The distinction between the bilinear form and its enhanced
diagonal data is restored in the following subsection rather than being
hidden inside this count.

\subsubsection{Row-space parameters and bilinear completion formula}
\label{subsubsec:binary-complete-finite-proof}
We now derive the exact bilinear completion count.  In the present finite calculation
the enhanced mod-$4$ diagonal data are deliberately absent; the next
subsection maps the actual enhanced moments to the bilinear objects counted
here and restores their diagonal lifts.  Let $Y$ be uniformly distributed on
$\F_2^{(n-m)\times k}$.  Its row space and the two invariants used by the
bilinear completion problem are
\begin{equation}
 W_2(Y):=\operatorname{row}Y\subseteq\F_2^k,
 \quad
 A_2(Y):=\dim W_2(Y),
 \quad
 U_2(Y):=\dim\bigl(W_2(Y)\cap W_2(Y)^\perp\bigr).
 \label{eq:integrated-qubit-au-lower-block-rank}
\end{equation}
The joint event probability and the induced lower bilinear rank are
\begin{equation}
 \mathsf p_{n,m,k}^{(2)}(a,u)
 :=\Pr\{A_2(Y)=a,\ U_2(Y)=u\},
 \quad r_{\rm bin}:=a-u.
 \label{eq:integrated-binary-unreported-mass-function}
\end{equation}
Thus $a$ and $u$ are the row-space dimension and radical dimension,
respectively, while $r_{\rm bin}=a-u$ is the lower bilinear rank used in the
completion count.  For a fixed lower bilinear rank $r_{\rm bin}$, the following completion count of
rank-constrained symmetric block completions is
\begin{align}
 \mathcal N_{m,r,k}^{(2)}(r_{\rm bin})
 :={}&2^{m r_{\rm bin}}
 \sum_{\ell=0}^{\min\{m,r-r_{\rm bin},\lfloor(k-r_{\rm bin})/2\rfloor\}}
 N_{\rm rect}(m,r-r_{\rm bin},\ell;2)
 2^{\ell m-\ell(\ell-1)/2}\notag\\
 &\quad\times
 \sum_{t=0}^{\min\{m-\ell,k-r_{\rm bin}-2\ell\}}
 N_{\rm sym}(m-\ell,t;2).
 \label{eq:integrated-qubit-completion-count}
\end{align}
We first prove that this is the exact number of bilinear block completions.
The next three lemmas eliminate the concrete completion matrices in favor of
the coupling rank $\ell$ and residual rank $t$.

\subsubsection{Binary Schur reduction and radical coupling}
\label{subsubsec:binary-schur-radical-reduction}

\begin{lemma}[Binary Schur reduction]
\label{lem:binary-Schur-reduction}
Let $B\in\Sym_r(\F_2)$ have rank $r_{\rm bin}$ and nullity
$s:=r-r_{\rm bin}$.  Choose a congruence separating its invertible and radical
parts:
\begin{equation}
 P^TBP=\begin{pmatrix}D&0\\0&0_s\end{pmatrix},
 \quad D\in\Sym_{r_{\rm bin}}(\F_2)\text{ invertible}.
 \label{eq:binary-lower-block-rank-decomposition}
\end{equation}
Split the coupling matrix in the same basis and form the Schur translate
\begin{equation}
 A:=Q-C_1D^{-1}C_1^T.
 \label{eq:binary-Schur-translation}
\end{equation}
The completed rank is the sum of the fixed lower-block rank and the residual
radical-coupling rank:
\begin{equation}
 \rank\begin{pmatrix}Q&C\\C^T&B\end{pmatrix}
 =r_{\rm bin}+
 \rank\begin{pmatrix}A&C_0\\C_0^T&0_s\end{pmatrix},
 \label{eq:binary-first-rank-reduction}
\end{equation}
and, for fixed $C_1$, $Q\mapsto A$ is a bijection of
$\Sym_m(\F_2)$.
\end{lemma}
\begin{proof}
The field-independent Schur-complement reduction of
Lemma~\ref{lem:first-congruence} applies over $\F_2$ with $r_B=r_{\rm bin}$.
It gives Eq.~\eqref{eq:binary-first-rank-reduction} directly.  In
characteristic two, the translated block may equivalently be written
$A=Q+C_1D^{-1}C_1^T$.  Since $D^{-1}$ is symmetric, this translation preserves
$\Sym_m(\F_2)$ and is bijective for fixed $C_1$.
\end{proof}

\begin{lemma}[Binary radical-coupling rank identity]
\label{lem:binary-radical-rank-identity}
Let $A\in\Sym_m(\F_2)$ and let $C_0\in\F_2^{m\times s}$ have rank
$\ell$.  Choose bases that put the coupling into standard rank form and split
the transformed symmetric matrix accordingly:
\begin{equation}
 RC_0S=\begin{pmatrix}I_\ell&0\\0&0\end{pmatrix},
 \quad
 RAR^T=\begin{pmatrix}A_{11}&A_{12}\\A_{12}^T&A_{22}\end{pmatrix}.
 \label{eq:binary-radical-normal-forms}
\end{equation}
The coupling contributes exactly $2\ell$ units of rank, with the remaining
contribution carried by $A_{22}$:
\begin{equation}
 \rank\begin{pmatrix}A&C_0\\C_0^T&0_s\end{pmatrix}
 =2\ell+\rank A_{22}.
 \label{eq:binary-radical-rank-identity}
\end{equation}
\end{lemma}
\begin{proof}
With $U:=\operatorname{diag}(R^T,S)$, substitution of
Eq.~\eqref{eq:binary-radical-normal-forms} into $U^T(\cdot)U$ gives
\begin{equation}
 U^T\begin{pmatrix}A&C_0\\C_0^T&0_s\end{pmatrix}U
 =\begin{pmatrix}
 A_{11}&A_{12}&I_\ell&0\\
 A_{12}^T&A_{22}&0&0\\
 I_\ell&0&0&0\\
 0&0&0&0_{s-\ell}
 \end{pmatrix}.
 \label{eq:binary-radical-full-normal-form}
\end{equation}
The last coordinate block is a zero direct summand.  Delete it and call the
remaining matrix $K$.  For a column vector $(x,z,y)$ with block sizes
$\ell,m-\ell,\ell$, the equation $K(x,z,y)^T=0$ is equivalent to
\begin{align}
 A_{11}x+A_{12}z+y=0,~
 A_{12}^Tx+A_{22}z=0,~
 x=0.
 \label{eq:binary-radical-kernel-system}
\end{align}
The third equation forces $x=0$; the second then becomes $A_{22}z=0$; and the
first uniquely determines $y=A_{12}z$.  Thus the kernel of
$K$ has the explicit parametrization and dimension
\begin{equation}
 \ker K=\{(0,z,A_{12}z):z\in\ker A_{22}\},
 \quad \dim\ker K=(m-\ell)-\rank A_{22}.
 \label{eq:binary-radical-kernel-parametrization}
\end{equation}
Since $K$ has size $m+\ell$, rank--nullity gives
\begin{equation}
 \rank K=(m+\ell)-((m-\ell)-\rank A_{22})
 =2\ell+\rank A_{22}.
 \label{eq:binary-radical-rank-from-kernel}
\end{equation}
This establishes Eq.~\eqref{eq:binary-radical-rank-identity}.  Notice that the
argument never divides by two and does not require eliminating $A_{11}$ by a
congruence.
\end{proof}

\subsubsection{Binary parameterization and completion count}
\label{subsubsec:binary-completion-parameterization}

Equation~\eqref{eq:binary-radical-rank-identity} separates the
rank contribution $2\ell$ of the radical coupling from the rank
of the remaining block $A_{22}$.
To count each completion once, we fix a pair of row and column basis
changes for each coupling matrix $C_0$.  We then use the resulting block
coordinates as independent parameters.

\begin{lemma}[Explicit binary completion parameterization]
\label{lem:binary-completion-parameterization}
Fix $B$, the matrix $P$ in Eq.~\eqref{eq:binary-lower-block-rank-decomposition}, and
ranks $(\ell,t)$.  For each rank-$\ell$ matrix $C_0$, fix one pair
$R(C_0),S(C_0)$ satisfying Eq.~\eqref{eq:binary-radical-normal-forms}.
For fixed ranks $\rank C_0=\ell$ and $\rank A_{22}=t$, the completion pairs
$(Q,C)$ are parametrized bijectively by
\begin{equation}
 (C_1,C_0,A_{11},A_{12},A_{22}),
 \label{eq:binary-completion-parameter-tuple}
\end{equation}
where $C_1$ is arbitrary, $C_0$ has rank $\ell$,
$A_{11}\in\Sym_\ell(\F_2)$ and
$A_{12}\in\F_2^{\ell\times(m-\ell)}$ are arbitrary, and
$A_{22}\in\Sym_{m-\ell}(\F_2)$ has rank $t$.
\end{lemma}
\begin{proof}
By the binary Schur-complement reduction, for each fixed $C_1$ the variable
$Q$ is in affine bijection with $A\in\Sym_m(\F_2)$.  For each $C_0$, the
changes of basis $R(C_0),S(C_0)$ fixed in the statement give the unique
block decomposition
\[
 R(C_0)AR(C_0)^T
 =\begin{pmatrix}A_{11}&A_{12}\\A_{12}^T&A_{22}\end{pmatrix}.
\]
Conversely, these blocks reconstruct $A$, then $Q$, while
$C=(C_1\ C_0)P^{-1}$.  Thus the correspondence is bijective.  Fixing one
pair of basis changes for each $C_0$ prevents overcounting; any other fixed
choice merely
relabels the block coordinates bijectively.
\end{proof}

\begin{lemma}[Binary completion count]
\label{lem:integrated-binary-completion-count}
Let $B\in\Sym_r(\F_2)$ have rank $r_{\rm bin}$.  The number of pairs
$(Q,C)\in\Sym_m(\F_2)\times\F_2^{m\times r}$ satisfying
\begin{equation}
 \rank\begin{pmatrix}Q&C\\C^T&B\end{pmatrix}\le k
 \label{eq:integrated-binary-completion-condition}
\end{equation}
is the quantity $\mathcal N_{m,r,k}^{(2)}(r_{\rm bin})$ in
Eq.~\eqref{eq:integrated-qubit-completion-count}; in particular, it depends on
$B$ only through $r_{\rm bin}$.
\end{lemma}
\begin{proof}
Combining Lemma~\ref{lem:binary-Schur-reduction} with
Lemma~\ref{lem:binary-radical-rank-identity} and writing
$t:=\rank A_{22}$ gives the exact completion rank
\begin{equation}
 \rank\begin{pmatrix}Q&C\\C^T&B\end{pmatrix}
 =r_{\rm bin}+2\ell+t.
 \label{eq:binary-full-completion-rank}
\end{equation}
For fixed $(\ell,t)$ satisfying $r_{\rm bin}+2\ell+t\le k$,
Lemma~\ref{lem:binary-completion-parameterization} gives a bijective sequence
of choices.  Their cardinalities are
\begin{align}
 |\{C_1\}|&=2^{mr_{\rm bin}},
 \label{eq:binary-count-C1}\\
 |\{C_0:\rank C_0=\ell\}|&=N_{\rm rect}(m,s,\ell;2),
 \label{eq:binary-count-C0}\\
 |\{(A_{11},A_{12})\}|
 &=2^{\ell(\ell+1)/2+\ell(m-\ell)}
 =2^{\ell m-\ell(\ell-1)/2},
 \label{eq:binary-count-free-A}\\
 |\{A_{22}:\rank A_{22}=t\}|
 &=N_{\rm sym}(m-\ell,t;2).
 \label{eq:binary-count-A22}
\end{align}
The dimensions and Eq.~\eqref{eq:binary-full-completion-rank} give exactly the
summation ranges in Eq.~\eqref{eq:integrated-qubit-completion-count}.
Multiplying Eqs.~\eqref{eq:binary-count-C1}--
\eqref{eq:binary-count-A22} and summing over these ranges proves that formula.
The maps $C\mapsto CP$ and $Q\mapsto A$ are bijections, and the remaining
counts use only $r_{\rm bin}$ and $s=r-r_{\rm bin}$; hence the result is
independent of the entries and bilinear type of $B$.
\end{proof}

The bilinear counting target is now fixed.  We next return to the actual
enhanced moments, forget only their extra diagonal information, and control its
restoration by a separate lift factor.

\subsection{Quadratic data, associated bilinear forms, and diagonal lifts}
We compare binary enhanced completions with completions of their associated
bilinear forms.  Each enhanced completion determines a bilinear completion;
the additional diagonal information contributes at most $2^m$ possible
lifts for each bilinear completion.

The lower enhanced moment records mod-$4$ diagonal data and mod-$2$
lower--lower edge data, while its bilinear form retains the binary Gram
matrix.  The binary lower block has dimension $n-m$.  In contrast, the
odd-prime augmented lower block has dimension $n-m+1$ because it includes
a row of ones.  The binary argument uses the separate lift bound instead.

\subsubsection{Enhanced lower data and bilinear completion sets}

For a lower matrix $Y\in\F_2^{(n-m)\times k}$, retain both its mod-$4$
diagonal data and mod-$2$ edge data in the enhanced lower moment
\begin{equation}
 \zeta_2(Y)
 :=\left(
   \left(\sum_{j=1}^kY_{aj}\pmod4\right)_{a=1}^{n-m},
   \left(\sum_{j=1}^kY_{aj}Y_{bj}\pmod2\right)_{1\le a<b\le n-m}
   \right)
 \in\mathcal L_{n,m,2}.
 \label{eq:binary-lower-enhanced-moment-of-Y}
\end{equation}
Recall the matrix associated with a fixed lower enhanced value
$\zeta=(r^{\rm low},T^{\rm low})$ in
Eq.~\eqref{eq:lower-moment-matrix-binary}.  Its diagonal entries are the
reductions modulo two of $r^{\rm low}$, and its off-diagonal entries are
the edge coordinates $T^{\rm low}$:
\begin{equation}
 B_\zeta\in\Sym_{n-m}(\F_2),
 \quad
 (B_\zeta)_{aa}=r_a^{\rm low}\pmod2,
 \quad
 (B_\zeta)_{ab}=T_{ab}^{\rm low}\quad(a<b).
 \label{eq:binary-bilinear-shadow}
\end{equation}
Whenever $Y$ realizes the enhanced value $\zeta$, its ordinary binary Gram
matrix is exactly this associated form:
\begin{equation}
 B_\zeta=YY^T.
 \label{eq:binary-shadow-is-Gram}
\end{equation}

For a fixed lower enhanced value
$\zeta=(r^{\rm low},T^{\rm low})$, collect the occurring full enhanced
moments in
\begin{equation}
 \mathcal E_\zeta
 :=\left\{z=(r,T)\in\mathcal C_\zeta^{(2)}:
                 N_z^{(2)}>0\right\}.
 \label{eq:binary-enhanced-occurring-set}
\end{equation}
Thus $|\mathcal E_\zeta|=d_\zeta^{(2)}$.  For $z\in\mathcal E_\zeta$, reduce the visible enhanced diagonal modulo two
and combine it with the visible edge data to form
\begin{equation}
 Q_z\in\Sym_m(\F_2),
 \quad
 (Q_z)_{ii}:=r_i\pmod2,
 \quad
 (Q_z)_{ij}:=T_{ij}\quad(i<j).
 \label{eq:binary-visible-bilinear-block}
\end{equation}
The visible--lower edge coordinates form the cross block
\begin{equation}
 C_z\in\F_2^{m\times(n-m)},
 \quad
 (C_z)_{ia}:=T_{i,m+a}.
 \label{eq:binary-cross-bilinear-block}
\end{equation}
The corresponding rank-constrained bilinear completion set is
\begin{equation}
 \mathcal R[B_\zeta]
 :=\left\{(Q,C)\in\Sym_m(\F_2)\times
                   \F_2^{m\times(n-m)}:
 \rank\begin{pmatrix}Q&C\\C^T&B_\zeta\end{pmatrix}\le k
 \right\}.
 \label{eq:binary-shadow-completion-set}
\end{equation}
By Lemma~\ref{lem:integrated-binary-completion-count},
\begin{equation}
 |\mathcal R[B_\zeta]|
 =\mathcal N_{m,n-m,k}^{(2)}(\rank B_\zeta).
 \label{eq:binary-shadow-completion-cardinality}
\end{equation}

\subsubsection{Enhanced moments and bilinear completions}
\label{subsubsec:enhanced-bilinear-comparison}

The binary moment coordinates retain mod-$4$ phase information that is
absent from an ordinary symmetric bilinear completion.  We compare them
by forgetting the visible phase lifts.  The image lies in the
rank-constrained bilinear completion set, while the possible choices of
visible lifts give a uniform bound on each fiber.  This bounds the
number of enhanced moments without treating the forgetful map as
injective.

\begin{proposition}[Comparison between enhanced moments and bilinear completions]
\label{prop:binary-enhanced-bilinear-comparison}
For every lower enhanced value $\zeta$, forgetting the visible mod-$4$ lifts
defines the map
\begin{equation}
 f_\zeta:\mathcal E_\zeta\longrightarrow
          \Sym_m(\F_2)\times\F_2^{m\times(n-m)},
 \quad
 f_\zeta(z):=(Q_z,C_z).
 \label{eq:binary-enhanced-forgetful-map}
\end{equation}
This map satisfies
\begin{equation}
 f_\zeta(\mathcal E_\zeta)\subseteq\mathcal R[B_\zeta].
 \label{eq:binary-forgetful-image-in-completion-set}
\end{equation}
Each visible diagonal parity has at most two mod-$4$ lifts, so every fiber of
$f_\zeta$ has cardinality at most $2^m$.  The enhanced support
size therefore satisfies
\begin{equation}
 d_\zeta^{(2)}
 \le 2^m\mathcal N_{m,n-m,k}^{(2)}(\rank B_\zeta).
 \label{eq:binary-enhanced-support-completion-bound}
\end{equation}
\end{proposition}
\begin{proof}
Take $z\in\mathcal E_\zeta$.  Choose an ordered tuple realizing it and split
the associated coordinate matrix into visible and lower rows:
\begin{equation}
 X=\begin{pmatrix}X_1\\Y\end{pmatrix},
 \quad
 X_1\in\F_2^{m\times k},
 \quad Y\in\F_2^{(n-m)\times k},
 \label{eq:binary-realizing-tuple-blocks}
\end{equation}
with enhanced moment $z$.  Equations
\eqref{eq:binary-visible-bilinear-block} and
\eqref{eq:binary-cross-bilinear-block}, together with $x^2=x$ over
$\F_2$, give
\begin{equation}
 Q_z=X_1X_1^T,
 \quad C_z=X_1Y^T,
 \quad B_\zeta=YY^T.
 \label{eq:binary-realized-bilinear-blocks}
\end{equation}
The three block identities in Eq.~\eqref{eq:binary-realized-bilinear-blocks}
give the full bilinear Gram factorization
\begin{equation}
 \begin{pmatrix}Q_z&C_z\\C_z^T&B_\zeta\end{pmatrix}
 =XX^T.
 \label{eq:binary-shadow-full-Gram-factorization}
\end{equation}
The factorization has rank at most $k$.  This proves
Eq.~\eqref{eq:binary-forgetful-image-in-completion-set}.

Now fix $(Q,C)$ in the image.  All off-diagonal visible--visible and
visible--lower enhanced coordinates are already determined by $(Q,C)$.
For each visible index $i\le m$, the forgotten coordinate is
$r_i\in\mathbb Z_4$, and its parity is fixed by the diagonal entry of $Q$:
\begin{equation}
 r_i\pmod2
 =\sum_{j=1}^k(X_1)_{ij}
 =\sum_{j=1}^k(X_1)_{ij}^2
 =(Q)_{ii}.
 \label{eq:binary-enhanced-diagonal-parity}
\end{equation}
A prescribed parity class contains exactly two elements of $\mathbb Z_4$.
Therefore each of the $m$ forgotten visible diagonal coordinates has at most
two lifts, while all remaining enhanced coordinates are already fixed by
$(Q,C)$ and the lower value $\zeta$.  Therefore
$|f_\zeta^{-1}(Q,C)|\le2^m$.  Summing the fiber bound over
$f_\zeta(\mathcal E_\zeta)$, using the inclusion in
Eq.~\eqref{eq:binary-forgetful-image-in-completion-set}, and then applying
Eq.~\eqref{eq:binary-shadow-completion-cardinality} gives
Eq.~\eqref{eq:binary-enhanced-support-completion-bound}.
\end{proof}

For the binary lower matrix $Y$,
Lemma~\ref{lem:common-Gram-radical-rank} expresses the rank of its
Gram matrix in terms of the row-space invariants $a,u$:
\begin{equation}
 \rank(YY^T)=a-u.
 \label{eq:binary-shadow-rank-au}
\end{equation}
The identity for $\rank(YY^T)$ depends only on the bilinear form and
the radical dimension; the binary quadratic refinement is not needed
for this rank calculation.


\subsubsection{Finite weighted completion reduction}
\label{subsubsec:binary-finite-weighted-reduction}

By Eq.~\eqref{eq:binary-shadow-rank-au}, the lower bilinear matrix has rank
$a-u$.  Combining the completion bound for this rank with the lower
row-space probability gives the following finite bound on exact identification.

\begin{proposition}[Binary finite weighted completion reduction]
\label{prop:binary-finite-weighted-completion-reduction}
For all integers $n\ge1$, $1\le m\le n$, and $k\ge0$, the binary
exact-identification optimum satisfies the following finite weighted
completion bound:
\begin{equation}
 P_{\mathrm{id}}^{(k),*}(n,m;2)
 \le\sum_{a,u}\mathsf p_{n,m,k}^{(2)}(a,u)
 \min\left\{1,
 \frac{2^m\mathcal N_{m,n-m,k}^{(2)}(a-u)}{\lvert Y_{n,m,2}\rvert}
 \right\}.
 \label{eq:integrated-qubit-finite-converse}
\end{equation}
\end{proposition}
\begin{proof}
For every $\zeta$, the trivial inclusion
$\mathcal E_\zeta\subseteq\mathcal C_\zeta^{(2)}$ gives
$d_\zeta^{(2)}\le\lvert Y_{n,m,2}\rvert$.  Combining the trivial coset-size cap with the enhanced-to-bilinear comparison
gives the finite normalized moment-count bound capped at one
\begin{equation}
 \frac{d_\zeta^{(2)}}{\lvert Y_{n,m,2}\rvert}
 \le\min\left\{1,
 \frac{2^m\mathcal N_{m,n-m,k}^{(2)}(\rank B_\zeta)}
 {\lvert Y_{n,m,2}\rvert}\right\}.
 \label{eq:binary-clipped-completion-bound}
\end{equation}
Insert Eq.~\eqref{eq:binary-clipped-completion-bound} into
Eq.~\eqref{eq:binary-weighted-coset-bound}.  Under the probability law
$\pi_\zeta^{(2)}$, choose $Y$ uniformly and set $\zeta=\zeta_2(Y)$.
By Eq.~\eqref{eq:binary-shadow-is-Gram} and
Lemma~\ref{lem:common-Gram-radical-rank},
\begin{equation}
 \rank B_{\zeta_2(Y)}
 =\rank(YY^T)=A_2(Y)-U_2(Y).
 \label{eq:binary-zeta-rank-regrouping}
\end{equation}
Regrouping any function of the associated bilinear rank by the row-space
invariants gives
\begin{equation}
 \sum_\zeta\pi_\zeta^{(2)}F(\rank B_\zeta)
 =\sum_{a,u}\mathsf p_{n,m,k}^{(2)}(a,u)F(a-u).
 \label{eq:binary-zeta-to-au-regrouping}
\end{equation}
Apply this identity to the clipped function in
Eq.~\eqref{eq:binary-clipped-completion-bound}.  The result is
Eq.~\eqref{eq:integrated-qubit-finite-converse}.
\end{proof}

The factor $2^m$ in Eq.~\eqref{eq:integrated-qubit-finite-converse}
bounds the multiplicity of enhanced diagonal lifts; it is not part of the
exact bilinear completion count.  The proof has converted the enhanced
occurring-moment count above a fixed lower value into a bound depending on the completion rank $a-u$ and then used
Eq.~\eqref{eq:binary-zeta-to-au-regrouping} to weight that bound by the
row-space law.  To complete the converse, the next subsection bounds that
probability law.  The Gram--radical identity in
Eq.~\eqref{eq:binary-shadow-rank-au} identifies the rank required for this
last estimate.

\subsection{Binary row-space probability bound}
The preceding completion estimate controls the number of possible
visible coordinates at each lower Gram rank.  To complete the weighted
bound, we must also control the probability of the corresponding lower
row-space event.  The next lemma applies the common row-space estimate
to the binary invariants $(a,u)$.

\begin{lemma}[Binary row-space probability bound]
\label{lem:integrated-binary-lower-block-probability}
For every positive-mass pair $(a,u)$, the common row-space estimate gives
\begin{equation}
 \mathsf p_{n,m,k}^{(2)}(a,u)
 \le (K_2^{\mathrm{class}})^k
 2^{-(k-a)(n-m-a)-u(u+1)/2}.
 \label{eq:binary-lower-block-probability-finite}
\end{equation}
Consequently, along any fixed admissible integer sequences with $m,k=O(n)$, the multiplicative factor $(K_2^{\mathrm{class}})^k$ contributes only $k\log_2 K_2^{\mathrm{class}}=O(n)=o(n^2)$ to the exponent, uniformly in the admissible pair $(a,u)$. Thus, along these sequences,
\begin{equation}
 \mathsf p_{n,m,k}^{(2)}(a,u)
 \le2^{-(k-a)(n-m-a)-u(u+1)/2+o(n^2)}.
 \label{eq:integrated-binary-lower-block-probability}
\end{equation}
The remainder can depend on the fixed integer sequences; it is uniform over admissible $(a,u)$ and is not asserted to be common to all sequences with unbounded implicit $O(n)$ constants.
\end{lemma}
\begin{proof}
Apply Proposition~\ref{prop:all-prime-row-space-event-bound} with $p=2$ and $r=n-m$. The factor in Eq.~\eqref{eq:binary-lower-block-probability-finite} is independent of $(a,u)$; taking its base-two logarithm proves the asserted sequencewise remainder bound.
\end{proof}

\subsection{Finite exponent comparison}
We now combine the lower row-space probability estimate with the
bilinear completion count and its enhanced-diagonal lift factor.
The probability parameters are $(a,u)$, while $(\ell,t)$ index the
completion rank counts.
The resulting finite exponent is compared with the common expression
from Eq.~\eqref{eq:main-binary-finite-exponent-identity}; the remaining
correction is linear in $n$ and vanishes after normalization by $n^2$.

\subsubsection{Uniform rank-count estimates}
The probability estimate above controls the row-space event.
The completion factor uses separate rectangular and symmetric matrix
counts.  The rectangular bound is
Eq.~\eqref{eq:rectangular-uniform-bound} with $p=2$; the following
lemma gives a dimension-independent symmetric bound.

\begin{lemma}[Uniform binary symmetric-rank bound]
\label{lem:integrated-binary-symmetric-rank-bound}
There is a dimension-independent constant $C_2>0$ for which every binary
symmetric-rank count satisfies
\begin{equation}
 N_{\rm sym}(d,t;2)
 \le C_2\,2^{t(2d-t+1)/2}
 \quad(0\le t\le d).
 \label{eq:integrated-binary-symmetric-rank-bound}
\end{equation}
\end{lemma}
\begin{proof}
A rank-$t$ symmetric form has a radical of dimension $d-t$.
The number of choices for this radical is at most
\begin{equation}
 \genfrac{[}{]}{0pt}{}{d}{d-t}_2
 \le(\kappa_2^{\mathrm{GL}})^{-1}2^{t(d-t)}.
 \label{eq:binary-symmetric-rank-radical-choice}
\end{equation}
For a fixed radical, forms whose radical contains it are precisely the
symmetric forms on the $t$-dimensional quotient, of which there are
$2^{t(t+1)/2}$.  Retaining only the nondegenerate quotient forms can only
decrease this number.  Hence the claim holds with
$C_2=(\kappa_2^{\mathrm{GL}})^{-1}$, because
$t(d-t)+t(t+1)/2=t(2d-t+1)/2$.
\end{proof}

The finite exponent correction satisfies
\begin{equation}
 0\le\frac{m-\ell}{n^2}\le\frac1n,
 \label{eq:binary-remark-normalized-correction}
\end{equation}
by $0\le\ell\le m\le n$.  Thus it is uniform over the finite rank domains.

\subsubsection{Binary finite domain and completion count}\label{subsubsec:binary-finite-domain-completion-count}
Let $Y_2$ be uniform on $\F_2^{(n-m)\times k}$.  The probability-side
row-space data are
\begin{equation}
 W_2(Y_2):=\operatorname{row}Y_2,
 \quad
 A_2(Y_2):=\dim W_2(Y_2),
 \quad
 U_2(Y_2):=\dim\bigl(W_2(Y_2)\cap W_2(Y_2)^\perp\bigr).
 \label{eq:main-binary-unreported-parameters}
\end{equation}
Their joint probability and the induced lower bilinear rank are
\begin{equation}
 \mathsf p_{n,m,k}^{(2)}(a,u)
 :=\Pr\{A_2(Y_2)=a,\ U_2(Y_2)=u\},
 \quad r_{\rm bin}:=a-u.
 \label{eq:main-binary-unreported-mass-function}
\end{equation}
Thus $r_{\rm bin}$ is the rank of the induced bilinear form on the row space
for the event appearing in this probability.
The probability-side dimensions necessarily lie in the finite domain
\begin{equation}
 0\le a\le\min\{k,n-m\},
 \quad 0\le u\le\min\{a,k-a\}.
 \label{eq:main-binary-unreported-ranges}
\end{equation}
The two completion ranks $(\ell,t)$ range over
$\mathcal G_{a,u}^{(2)}$ as defined in
Eq.~\eqref{eq:main-binary-completion-ranges}.
Lemma~\ref{lem:integrated-binary-completion-count}, specifically
Eq.~\eqref{eq:integrated-qubit-completion-count}, gives for fixed $(a,u)$
\begin{align}
 \mathcal N_{m,n-m,k}^{(2)}(a-u)
 ={}&2^{m(a-u)}
 \sum_{\ell}
 N_{\rm rect}(m,n-m-a+u,\ell;2)
 2^{\ell m-\ell(\ell-1)/2}\notag\\
 &\quad\times\sum_t N_{\rm sym}(m-\ell,t;2),
 \label{eq:main-binary-completion-count}
\end{align}
where the two sums have exactly the ranges in
Eq.~\eqref{eq:main-binary-completion-ranges}.  Substitute this completion count and the row-space-event distribution into
the finite weighted reduction of
Proposition~\ref{prop:binary-finite-weighted-completion-reduction}, namely
Eq.~\eqref{eq:integrated-qubit-finite-converse}.
The optimal exact-identification probability for the binary coordinate
model therefore satisfies the following finite bound:
\begin{equation}
 P_{\mathrm{id}}^{(k),*}(n,m;2)
 \le\sum_{a,u}\mathsf p_{n,m,k}^{(2)}(a,u)
 \min\left\{1,
 \frac{2^m\mathcal N_{m,n-m,k}^{(2)}(a-u)}
 {\lvert Y_{n,m,2}\rvert}\right\}.
 \label{eq:main-binary-finite-converse}
\end{equation}
The prefactor $2^m$ is the enhanced-diagonal lift bound proved in
Proposition~\ref{prop:binary-enhanced-bilinear-comparison}, specifically
Eq.~\eqref{eq:binary-enhanced-support-completion-bound}.  It is not part of the
exact bilinear completion count in
Eq.~\eqref{eq:integrated-qubit-completion-count}; it records the additional
visible mod-$4$ diagonal lifts above fixed binary bilinear completion data.  We retain it
explicitly, and it becomes the first term $m$ in
Eq.~\eqref{eq:detail-main-binary-E-fin}.  The probability input used with this count
is Lemma~\ref{lem:integrated-binary-lower-block-probability}, in particular
Eq.~\eqref{eq:integrated-binary-lower-block-probability}.

\subsubsection{Characteristic-two finite exponent decomposition}\label{subsubsec:binary-finite-exponent-decomposition}
For an admissible quadruple $(a,u,\ell,t)$, collect the logarithmic
contribution of the external lift and the four completion factors in
\begin{align}
 E_{\rm fin}^{(2)}(m,n;a,u,\ell,t)
 :={}&m+m(a-u)+\ell(n-a+u-\ell)
 +\ell m-\frac{\ell(\ell-1)}2\notag\\
 &+\frac{t[2(m-\ell)-t+1]}2.
 \label{eq:detail-main-binary-E-fin}
\end{align}
The five terms record, respectively, the external binary lift, coupling to the
nondegenerate lower quotient, rectangular coupling to the radical, free
symmetric coordinates, and residual symmetric rank.  The corresponding exponent contributions are
\begin{align}
 \text{external binary factor:}\quad &m,
 \label{eq:main-binary-exponent-external}\\
 \text{nondegenerate lower-block coupling:}\quad &m(a-u),
 \label{eq:main-binary-exponent-unreported-coupling}\\
 \text{rank-$\ell$ rectangular coupling:}\quad
 &\ell(n-a+u-\ell),
 \label{eq:main-binary-exponent-rectangular}\\
 \text{free symmetric coordinates:}\quad
 &\ell m-\frac{\ell(\ell-1)}2,
 \label{eq:main-binary-exponent-free-symmetric}\\
 \text{residual symmetric rank $t$:}\quad
 &\frac{t[2(m-\ell)-t+1]}2.
 \label{eq:main-binary-exponent-residual}
\end{align}
These expressions define the exponents associated with the corresponding finite factors.  The dimension-uniform rank-count bounds that justify them in
the probability estimate are introduced later, exactly where they are used.

\subsubsection{Clipped characteristic-two finite bound}\label{subsubsec:binary-clipped-finite-bound}
\label{subsec:binary-probability-completion-estimate}

We now return to the finite exact-identification probability.  The following
rank-count estimates are introduced here because this is the point at which
they enter the argument:
\begin{align}
 N_{\rm rect}(M,S,L;2)
 &\le C_{\rm rect}\,2^{L(M+S-L)},
 \label{eq:main-binary-rectangular-bound}\\
 N_{\rm sym}(D,T;2)
 &\le C_{\rm sym}\,2^{T(2D-T+1)/2},
 \label{eq:main-binary-symmetric-bound}
\end{align}
where $C_{\rm rect}$ and $C_{\rm sym}$ are absolute constants,
independent of dimensions and ranks.  We use the set
$\mathcal G_{a,u}^{(2)}$ of completion ranks defined in
Eq.~\eqref{eq:main-binary-completion-ranges}.
Applying both rank-count bounds to each term of
Eq.~\eqref{eq:main-binary-completion-count} gives
\begin{equation}
 2^m\mathcal N_{m,n-m,k}^{(2)}(a-u)
 \le C_{\rm rect}C_{\rm sym}
 \sum_{(\ell,t)\in\mathcal G_{a,u}^{(2)}}
 2^{E_{\rm fin}^{(2)}(m,n;a,u,\ell,t)}.
 \label{eq:binary-completion-before-maximum}
\end{equation}
The admissible ranks satisfy $0\le\ell,t\le n$, so
$|\mathcal G_{a,u}^{(2)}|\le(n+1)^2$.  Bounding this sum by its
largest term, we obtain
\begin{equation}
 2^m\mathcal N_{m,n-m,k}^{(2)}(a-u)
 \le
 2^{\max_{(\ell,t)}E_{\rm fin}^{(2)}(m,n;a,u,\ell,t)+O(\log n)},
 \label{eq:main-binary-completion-exponent-bound}
\end{equation}
where the maximum uses the ranges in
Eq.~\eqref{eq:main-binary-completion-ranges}.

The binary row-space-event estimate is
\begin{equation}
 \mathsf p_{n,m,k}^{(2)}(a,u)
 \le 2^{-(k-a)(n-m-a)-u(u+1)/2+o(n^2)}.
 \label{eq:main-binary-lower-block-probability}
\end{equation}
In the binary case, for any fixed admissible sequences $(m_n),(k_n)$ with
$m_n,k_n=O(n)$, the $o(n^2)$ term in
Eq.~\eqref{eq:main-binary-lower-block-probability} can be taken uniformly
over every admissible pair $(a,u)$.  Indeed,
Eq.~\eqref{eq:binary-lower-block-probability-finite} supplies the common
additive exponent $k_n\log_2 K_2^{\mathrm{class}}=O(n)$, independently of
$(a,u)$.  This statement does not require the same remainder for all
integer sequences with arbitrarily large implicit $O(n)$ constants.
Multiplying the probability of each $(a,u)$ class by its clipped completion
factor yields
\begin{align}
 P_{\mathrm{id}}^{(k),*}(n,m;2)
 \le\sum_{a,u}
 2^{-(k-a)(n-m-a)-u(u+1)/2+o(n^2)}
 \min\Bigl\{1,
 2^{\max_{(\ell,t)}E_{\rm fin}^{(2)}
       -D_{\rm vis,fin}+O(\log n)}\Bigr\}.
 \label{eq:detail-main-binary-finite-exponent-bound}
\end{align}
Here the visible-label exponent is the all-prime finite quantity
$D_{\rm vis,fin}$ defined in Eq.~\eqref{eq:D-vis-fin}; for $p=2$ it equals
$\log_2\lvert Y_{n,m,2}\rvert$.
After division by $n^2$, the deterministic scaling argument converts this
finite bound to the main-text normalized exponent
\[
 -I+\min\{0,E-D\},
\]
where $I$, $E$, and $D$ are defined in
Eqs.~\eqref{eq:direct-finite-I}--\eqref{eq:direct-finite-D}.  Thus the finite
hidden exponent in Eq.~\eqref{eq:theorem-C-binary-hidden-exponent}, the
maximized binary completion exponent, and $D_{\rm vis,fin}$ converge to the
three quantities already used in the common variational argument.  The
probability estimate is therefore complete before that argument begins.

The binary-specific data have now been reduced to the two common parameter
pairs, the clipped finite bound, and the correction $m-\ell$.  The proof below
assembles the bilinear completion bound, the enhanced-lift comparison, the
row-space probability estimate, and the finite exponent relation into the
stated structural theorem.

\subsection{Completion of the proof of the characteristic-two finite structural estimate}
\label{subsec:binary-finite-proof-assembly}

\begin{proof}[Proof of Theorem~\ref{thm:converse-binary-finite-structural-estimate}]
For characteristic two, the enhanced lower moment is first mapped to its
associated bilinear matrix, while the omitted diagonal data are counted by the
separate factor $2^m$.  The exact bilinear completion count, multiplied
by $2^m$, bounds the number of occurring enhanced moments.  The binary
row-space estimate supplies the probability weight for this bound.  These two finite ingredients are then combined before comparison with
the common formal exponent.  The probability bound is the all-prime row-space estimate specialized
to $p=2$; the completion and lift bounds are the binary statements.

\smallskip
\noindent\emph{Enhanced moment and associated bilinear form.}
Fix a lower enhanced moment $\zeta$.  Its associated bilinear form is the lower--lower
mod-$2$ block
\begin{equation}
 B_{\zeta}^{\rm bil}:=YY^T,
 \label{eq:theorem-C-bilinear-shadow}
\end{equation}
where $Y\in\F_2^{(n-m)\times k}$ realizes $\zeta$.  The parity constraints in
the enhanced moment make this matrix independent of the chosen realization.
Proposition~\ref{prop:binary-enhanced-bilinear-comparison} defines the forgetful
map from occurring enhanced visible moments to bilinear completion data.  Its
image is contained in the set of symmetric block completions of
$B_{\zeta}^{\rm bil}$ with total rank at most $k$.

For each bilinear completion, all off-diagonal mod-$2$ data are fixed, while
each of the $m$ visible mod-$4$ diagonal coordinates has at most two compatible
lifts.  The enhanced support size therefore obeys
\begin{equation}
 d_{\zeta}^{(2)}
 \le 2^m\mathcal N_{m,n-m,k}^{(2)}
       (\rank B_{\zeta}^{\rm bil}).
 \label{eq:theorem-C-enhanced-support-completion-bound}
\end{equation}
The factor $2^m$ is external to the exact bilinear completion count.  Keeping
it separate prevents it from being counted twice.

\smallskip
\noindent\emph{Bilinear completion count.}
Let $W=\operatorname{row}Y$, $a=\dim W$, and
$u=\dim(W\cap W^\perp)$.  Lemma~\ref{lem:common-Gram-radical-rank} gives
\begin{equation}
 \rank(YY^T)=a-u.
 \label{eq:theorem-C-shadow-rank}
\end{equation}
The binary Schur reduction and radical-coupling identity reduce every enlarged
completion to a rectangular radical coupling of rank $\ell$ and a residual
symmetric block of rank $t$.  Lemma~\ref{lem:integrated-binary-completion-count}
then gives the exact count
\begin{align}
 \mathcal N_{m,n-m,k}^{(2)}(a-u)
 ={}&2^{m(a-u)}
 \sum_{\ell}N_{\rm rect}(m,n-m-a+u,\ell;2)
 2^{\ell m-\ell(\ell-1)/2}\notag\\
 &\quad\times\sum_t N_{\rm sym}(m-\ell,t;2).
 \label{eq:theorem-C-bilinear-completion-count}
\end{align}
The summation ranges in this completion formula are the integer pairs $(\ell,t)$ satisfying
\begin{equation}
 0\le\ell\le\min\{m,n-m-a+u\},\quad
 0\le t\le m-\ell,\quad
 a-u+2\ell+t\le k.
 \label{eq:theorem-C-completion-ranges}
\end{equation}

\smallskip
\noindent\emph{Binary row-space probability.}
For uniform $Y\in\F_2^{(n-m)\times k}$ define
\begin{equation}
 \mathsf p_{n,m,k}^{(2)}(a,u)
 :=\Pr\{\dim\operatorname{row}Y=a,
 \ \dim(\operatorname{row}Y\cap\operatorname{row}Y^\perp)=u\}.
 \label{eq:theorem-C-unreported-mass-function}
\end{equation}
The finite form is the specialization of
Proposition~\ref{prop:all-prime-row-space-event-bound} used in
Lemma~\ref{lem:integrated-binary-lower-block-probability}.  We retain its explicit factor
$(K_2^{\mathrm{class}})^k$ when combining the bounds.

\smallskip
\noindent\emph{Combination of the finite bounds and exact exponent comparison.}
Apply Proposition~\ref{prop:converse-weighted-moment-fiber-reduction} in
characteristic two, use the lift-and-completion bound from the enhanced-to-bilinear and completion bounds, and group the
lower enhanced moments by $(a,u)$.  This gives
\begin{equation}
 P_{\mathrm{id}}^{(k),*}(n,m;2)
 \le\sum_{a,u}\mathsf p_{n,m,k}^{(2)}(a,u)
 \min\left\{1,
 \frac{2^m\mathcal N_{m,n-m,k}^{(2)}(a-u)}{\lvert Y_{n,m,2}\rvert}\right\}.
 \label{eq:theorem-C-finite-weighted-completion}
\end{equation}
For the fully finite conclusion, use the termwise rank-count bound
\eqref{eq:binary-completion-before-maximum} and at most $(m+1)^2$
admissible rank pairs.  This bounds the completion numerator by
$C_2^{\mathrm{comp}}(m+1)^2
2^{\max_{(\ell,t)}E_{\rm fin}^{(2)}}$.
Use the finite row-space estimate
\eqref{eq:binary-lower-block-probability-finite}, retaining
$(K_2^{\mathrm{class}})^k$.
Since $|Y_{n,m,2}|=2^{D_{\rm vis,fin}}$, these substitutions give
Eq.~\eqref{eq:main-binary-finite-exponent-bound} with its finite cap.

The final finite-size comparison identifies the binary exponent with the odd-prime finite exponent up to a uniform lower-order correction.  Put $h=n-m+1$ and substitute this value into
$E_{\rm fin}(m,h;a,u,\ell,t)$ from
Eq.~\eqref{eq:theorem-B-E-fin}.  The
rectangular term becomes
\begin{equation}
 \ell(m+h-a+u-\ell)
 =\ell(n+1-a+u-\ell)
 =\ell(n-a+u-\ell)+\ell.
 \label{eq:theorem-C-rectangular-comparison}
\end{equation}
All remaining terms agree with the binary expression, except for its external
lift contribution $m$.  Since the common rectangular term contains the extra
$\ell$ displayed above, the exact difference is
\begin{equation}
 E_{\rm fin}^{(2)}(m,n;a,u,\ell,t)
 -E_{\rm fin}(m,n-m+1;a,u,\ell,t)=m-\ell,
 \label{eq:theorem-C-exponent-difference}
\end{equation}
which is Eq.~\eqref{eq:main-binary-finite-exponent-identity}.  Finally,
$0\le\ell\le m\le n$ gives
$0\le(m-\ell)/n^2\le1/n$, uniformly over the complete finite domain.  This
establishes the exact finite comparison.

\smallskip
\noindent\emph{Passage to asymptotic exponents.}
Apply the uniform rectangular- and symmetric-rank estimates to each summand in
Eq.~\eqref{eq:theorem-C-bilinear-completion-count}.  As in the odd-prime
case, the admissible set contains at most $(m+1)^2$ pairs $(\ell,t)$.
Consequently the passage from the sum to its maximum contributes at most
$2\log_2(n+1)$ to the exponent.  The product constants in the rank estimates
contribute only an additional constant, so the total remainder is
$O(\log n)$ uniformly in $(a,u)$ and throughout all boundary cases.
The binary completion count therefore satisfies:
\begin{equation}
 2^m\mathcal N_{m,n-m,k}^{(2)}(a-u)
 \le
 2^{\max_{(\ell,t)}E_{\rm fin}^{(2)}(m,n;a,u,\ell,t)+O(\log n)}.
 \label{eq:theorem-C-completion-exponent-bound}
\end{equation}
The leading term $m$ in $E_{\rm fin}^{(2)}$ is exactly the separate lift factor
from the enhanced-to-bilinear comparison.
After division by $n^2$, this linear lift term vanishes, while it remains
essential in the finite inequality.  For sequences with $m/n=\beta_n\to\beta$ and
$k/n=\alpha_n\to\alpha$, take any convergent subsequence of
admissible normalized ranks $a/n,u/n,\ell/n,t/n$ with limits
$a_0,u_0,\lambda,\tau$.  Dividing the binary finite rank budget
$a-u+2\ell+t\le k$ by $n$ gives
$a_0-u_0+2\lambda+\tau\le\alpha$.
The remaining quadratic terms converge term by term to the common
visible-completion exponent, with $\alpha_n,\beta_n$ replaced by
their limits $\alpha,\beta$.
Therefore the two characteristics have the same limiting optimization, even
though their finite prefactors and moment interpretations differ.

The binary isotropic-subspace count and the count of matrices with a prescribed
row space give Lemma~\ref{lem:integrated-binary-lower-block-probability}:
\begin{equation}
 \mathsf p_{n,m,k}^{(2)}(a,u)
 \le 2^{-(k-a)(n-m-a)-u(u+1)/2+o(n^2)}.
 \label{eq:theorem-C-lower-block-probability}
\end{equation}
For the fixed integer sequences under consideration, taking logarithms
of the factor $(K_2^{\mathrm{class}})^k$ in
Eq.~\eqref{eq:binary-lower-block-probability-finite} gives this $o(n^2)$
term, with the rank uniformity stated in
Lemma~\ref{lem:integrated-binary-lower-block-probability}.
Together these estimates give the asymptotic presentation with remainders
uniform over the rank parameters along each fixed admissible sequence.
This proves the theorem.
\end{proof}

The combined estimates establish the characteristic-two finite structural bound stated in the main text.

\section{Uniform passage from finite to normalized exponents}
\label{app:common-clipping}
\label{subsubsec:binary-endpoint-clipping}

The finite exponent estimates must remain valid when a rank parameter
approaches the boundary of its allowed range.  Clipping a support
factor at one introduces a minimum with zero on the logarithmic
exponent scale; the following lemma shows that a uniform
$o(n^2)$ approximation is preserved under this operation.  We use it
to pass from the finite completion bounds to their limiting exponent
without excluding boundary values.

\begin{lemma}[Endpoint-uniform passage through clipping]
\label{lem:endpoint-uniform-clipping-passage}
For each $n$, let $\Theta_n$ be a nonempty finite parameter set,
let $F_n:\Theta_n\to\mathbb R$, and let
$F:\bigcup_{n\ge1}\Theta_n\to\mathbb R$.  Suppose that the uniform error
on the same set $\Theta_n$ satisfies
\begin{equation}
 \sup_{\theta\in\Theta_n}\left|\frac{F_n(\theta)}{n^2}-F(\theta)\right|
 \le \varepsilon_n,
 \quad \varepsilon_n\longrightarrow0.
 \label{eq:endpoint-uniform-preclipping-error}
\end{equation}
Since $x\mapsto\min\{0,x\}$ is $1$-Lipschitz, the bound
$\varepsilon_n$ from Eq.~\eqref{eq:endpoint-uniform-preclipping-error}
also controls the clipped finite and limiting functions uniformly over
$\Theta_n$:
\begin{equation}
 \sup_{\theta\in\Theta_n}\left|
 \min\left\{0,\frac{F_n(\theta)}{n^2}\right\}
 -\min\{0,F(\theta)\}\right|
 \le\varepsilon_n.
 \label{eq:endpoint-uniform-postclipping-error}
\end{equation}
The difference between the maxima of the two clipped functions over
the same nonempty set $\Theta_n$ is also at most $\varepsilon_n$.
\end{lemma}
\begin{proof}
Both $x\mapsto\min\{0,x\}$ and the maximum functional in the uniform norm are
one-Lipschitz.  Applying these two facts successively proves the claim.
\end{proof}

We now specify the parameter sets and functions to which the lemma is
applied.  Fix $1\le m\le n$ and $k\ge0$, and choose either characteristic
branch.  In odd characteristic let $\mathfrak P_n^{(p)}:=\mathfrak P_n$
and $\mathcal G_{a,u}^{(p)}:=\mathcal G_{a,u}$ from
Eqs.~\eqref{eq:theorem-B-positive-mass-domain}--\eqref{eq:theorem-B-G-au};
for $p=2$ use $\mathfrak P_n^{(2)}$ and
$\mathcal G_{a,u}^{(2)}$ from
Eqs.~\eqref{eq:theorem-C-positive-mass-domain}--\eqref{eq:main-binary-completion-ranges}.
For every positive-mass pair $(a,u)$ with a nonempty completion domain,
take the finite set
\begin{equation}
 \Theta_{n,p}(a,u):=\left\{
 (k/n,m/n,a/n,u/n,\ell/n,t/n):
 (\ell,t)\in\mathcal G_{a,u}^{(p)}\right\}.
 \label{eq:endpoint-clipping-actual-parameter-fiber}
\end{equation}
On this same set, define the unnormalized finite visible exponent
$F_{n,p}(\theta):=E_{\rm fin}^{(p)}-D_{\rm vis,fin}$, taking
$E_{\rm fin}^{(p)}=E_{\rm fin}(m,n-m+1;\cdot)$ for odd $p$ and
$E_{\rm fin}^{(2)}$ for $p=2$.  Let $F(\theta):=E(\theta)-D(\theta)$, where
$E,D$ are the polynomials in
Eqs.~\eqref{eq:direct-finite-E}--\eqref{eq:direct-finite-D}
evaluated at the normalized arguments in $\Theta_{n,p}(a,u)$.
Direct expansion of the finite formulas gives the exact errors
\begin{equation}
 \frac{F_{n,p}(\theta)}{n^2}-F(\theta)=
 \begin{cases}
 (3\ell+t-3m)/(2n^2),&p\text{ odd},\\
 (\ell+t-m)/(2n^2),&p=2.
 \end{cases}
 \label{eq:endpoint-clipping-exact-finite-errors}
\end{equation}
Because $0\le\ell\le m$, $0\le t\le m-\ell$, and $m\le n$,
Eq.~\eqref{eq:endpoint-clipping-exact-finite-errors} has absolute error
at most $3/(2n)$ for odd $p$ and $1/(2n)$ for $p=2$.  Both estimates
are uniform over the admissible $(a,u,\ell,t)$.
The exact hidden-probability exponent differs from $n^2I$ by $u/2$,
which is at most $n/2$.

Applying Lemma~\ref{lem:endpoint-uniform-clipping-passage} to each nonempty
fiber $\Theta_{n,p}(a,u)$ preserves the respective error bound after
clipping and maximization.  Empty fibers contribute zero and require no
maximum.

The odd-prime and binary branches retain different finite feasibility
conditions.  The clipping lemma neither identifies their domains nor
establishes finite lattice approximations to every point of a limiting
feasible set.  No endpoint rounding is needed for the one-sided uniform
finite-exponent bound above.

The polynomial comparison is also meaningful at algebraic boundary values,
but it does not extend the finite probability bounds, proved for
$1\le m\le n$, to $m=0$.  The principal fixed-constraint theorem permits
$0\le m_n\le n$ for each $n$; since $m_n/n\to\beta\in(0,1]$, we use
the finite auxiliary probability bounds only when $m_n\ge1$, which holds
for all sufficiently large $n$.  No additional main-theorem hypothesis
is introduced.

The normalization by the visible-label alphabet size is already present before
this passage: for each lower moment $\zeta$, the success contribution is the
probability of that moment multiplied by the normalized occurring-moment-sector
count $d_\zeta/|Y|$, with $d_\zeta/|Y|\le1$.  Thus clipping represents the
exact finite support-count cap, not a later relaxation, and the
probability-weighted sum remains normalized
when the finite factors are replaced by their exponent bounds.

\section{Support-dimension benchmark for exact identification}
\label{app:common-auxiliary-estimates}
\label{app:dimension-bound-insufficient}
This appendix gives an elementary support-dimension bound for exact identification and applies it to the partial-label ensemble.  It then compares the resulting copy-rate range with that of the weighted moment-fiber converse.  The comparison shows the range obtained by substituting the coarse estimate $D\le p^{nk}$ into the elementary bound; it is not used in the sharp weighted proof.  The characteristic-specific counting and rank estimates are established in their respective proof sections and auxiliary algebra appendices.

\subsection{The elementary support-dimension bound}
A dimension-only argument provides a useful baseline for the exact
converse.  It compares the number of distinguishable labels with the
dimension of the common support of their states, without using the
moment-fiber probabilities or the rank-constrained completion counts.
The next proposition states the resulting bound for an arbitrary
uniform finite ensemble.

\begin{proposition}[Elementary support-dimension bound]
\label{prop:elementary-support-dimension}
Let $\Theta$ be a finite label set with $E:=\lvert \Theta\rvert$, equipped with the
uniform prior.  Let $\{\rho_\theta:\theta\in\Theta\}$ be density operators
supported on a common subspace $\mathcal K$, and put $D:=\dim\mathcal K$.
The common support dimension alone then gives the exact-identification bound
\begin{equation}
 P_{\mathrm{id}}^*\le \min\left\{1,\frac{D}{E}\right\}.       \label{eq:elementary-support-dimension}
\end{equation}
\end{proposition}
\begin{proof}
Compress an arbitrary identifying POVM $\{M_\theta\}$ to $\mathcal K$.  Since
$0\le\rho_\theta\le I_{\mathcal K}$,
\[
 \Tr(M_\theta\rho_\theta)\le\Tr M_\theta.
\]
Summing this trace bound over the uniformly weighted labels gives
\[
 P_{\mathrm{id}}(\mathsf M)
 =\frac1E\sum_{\theta\in\Theta}\Tr(M_\theta\rho_\theta)
 \le\frac1E\sum_{\theta\in\Theta}\Tr M_\theta
 =\frac{D}{E}.
\]
The additional upper bound by one is probabilistic.
\end{proof}
The proposition uses no covariance and no information about how the states are
distributed inside their common support.  These features make it general but
also coarse.

\subsection{Application to partial-label identification}
The conditioned partial-label ensemble has the following number
of equiprobable labels:
\begin{equation}
 E=\lvert Y\rvert=p^{mh+m(m+1)/2},\quad h=n-m+1.                  \label{eq:appendix-label-count}
\end{equation}
All conditioned $k$-copy states are mixtures of vectors
$|\psi_{A,b}\rangle^{\otimes k}$ and are therefore supported on the symmetric
tensor power of the $p^n$-dimensional single-copy space.  This gives the common
support estimate
\begin{equation}
 D\le \dim\operatorname{Sym}^k(\mathbb C^{p^n})
 =\binom{p^n+k-1}{k}
 \le p^{nk}.                                             \label{eq:appendix-support-bound}
\end{equation}
The last inequality also follows directly because the symmetric tensor space
is a subspace of the full $k$-fold tensor product.

Under the linear scaling
\[
 m=\beta n+o(n),\quad k=\alpha n+o(n),
\]
Eqs.~\eqref{eq:appendix-label-count} and
\eqref{eq:appendix-support-bound} give
\begin{align}
 \log_p E
 &=\left(\beta-\frac{\beta^2}{2}\right)n^2+o(n^2),       \label{eq:appendix-label-rate}\\
 \log_p D
 &\le \alpha n^2+o(n^2).                                \label{eq:appendix-support-rate}
\end{align}
Comparing these two exponents through
Proposition~\ref{prop:elementary-support-dimension} yields
\begin{equation}
 P_{\mathrm{id}}^{(k),*}(n,m;p)
 \le
 p^{-\left(\beta-\beta^2/2-\alpha\right)n^2+o(n^2)}.     \label{eq:appendix-elementary-rate}
\end{equation}
This bound guarantees decay with a strictly negative quadratic exponent in the region
\begin{equation}
 \alpha<\beta-\frac{\beta^2}{2}.                         \label{eq:appendix-elementary-region}
\end{equation}
Outside this region, the resulting upper estimate has no strictly negative
quadratic exponent.  It therefore gives no quadratic-rate decay guarantee
there.

\subsection{Comparison with the weighted converse obtained from the moment formula}
\begin{center}
\small
\begin{tabularx}{\textwidth}{@{}Y >{\centering\arraybackslash}p{0.27\textwidth} Y@{}}
\toprule
Method & Region where exact success is proved to vanish & Information retained\\
\midrule
Elementary bound with $D\le p^{nk}$
 & $\alpha<\beta-\beta^2/2$ & Coarse upper bound on support dimension\\
Weighted converse obtained from the moment formula
 & $\alpha<1$ & Fiber weights and unreported-coset structure\\
\bottomrule
\end{tabularx}
\end{center}
The region obtained from the coarse estimate $D\le p^{nk}$ is limited because
$0<\beta\le1$ implies
\begin{equation}
 \beta-\frac{\beta^2}{2}\le\frac12.
 \label{eq:appendix-dimension-max-half}
\end{equation}
This coarse estimate gives no quadratic-rate decay guarantee on
$[1/2,1)$, and generally leaves a larger interval uncovered.  At the
complete-label endpoint $\beta=1$, it gives a negative quadratic exponent
only for $\alpha<1/2$, whereas
Theorem~\ref{thm:converse-common-asymptotic-reduction} proves vanishing exact success for every
$\alpha<1$.

The calculation above uses the common support only through the coarse
estimate $D\le p^{nk}$.  It leaves unused the numbers of tuples with each
moment $N_z$, their grouping inside $H^\perp$-cosets, the lower-block
probabilities $\pi_{\zeta}$, and the conditional occurring-coordinate
counts $d_{\zeta}$.  The weighted moment-fiber analysis retains these data.  Its finite
bound is
\[
 P_{\mathrm{id}}^{(k),*}
 \le\sum_{\zeta}\pi_{\zeta}\frac{d_{\zeta}}{\lvert Y\rvert}
\]
and ultimately yields the full below-one converse.  Thus the elementary bound is
useful as a diagnostic comparison, but it is not a component of the sharp
proof.  In particular, Appendix~\ref{app:common-auxiliary-estimates} is not used in the rank-truncated verification
comparison: that argument uses the quadratic exact converse, not the weaker
dimension bound.

\section{Odd-prime auxiliary algebra}
\label{app:odd-prime-algebra}
This appendix supplies the odd-prime overlap and rank-counting identities.
Character orthogonality and quadratic Gauss sums prove the coordinate overlap
formula in Proposition~\ref{prop:rank-dependent-overlap}; symmetric-rank
enumeration supplies the uniform completion-count bound used in the converse.
The main direct bound instead uses unrestricted stabilizer overlaps and
Lagrangian-intersection classes.  No separate learning task or optimization of verification scores is introduced here.  Throughout, $p$ is a fixed odd prime.

\subsection{Quadratic Gauss sums and overlaps}
\label{app:quadratic-Gauss-sums}

This subsection supplies the complete calculation deferred from
Proposition~\ref{prop:rank-dependent-overlap}.  The partial-verification score
is derived directly from its projector in
Eq.~\eqref{eq:odd-partial-score-character-sum}.

\subsubsection{Character orthogonality and additive-character conventions}
\label{appsubsec:character-orthogonality}

Throughout this appendix, $p$ is a fixed odd prime and the standard additive
character is generated by
\begin{equation}
  \omega=e^{2\pi i/p}.
  \label{eq:appB-omega}
\end{equation}
The factor $1/2$ appearing below denotes the inverse of $2$ in $\F_p$.
We first record the additive-character identity used repeatedly in the proof.
It supplies both the one-dimensional cancellation rule and its product form on
$\F_p^m$, which will be used to remove incompatible linear terms from the
quadratic sums below.

\begin{lemma}[Character orthogonality]
\label{lem:appB-character-orthogonality}
For $a\in\F_p$, the additive-character sum evaluates to
\begin{equation}
  \sum_{z\in\F_p}\omega^{az}
  =
  \begin{cases}
    p,&a=0,\\
    0,&a\ne0.
  \end{cases}.
  \label{eq:appB-one-dimensional-character-sum}
\end{equation}
More generally, the additive-character sum associated with $c\in\F_p^m$ is
\begin{equation}
  \sum_{z\in\F_p^m}\omega^{c^Tz}
  =p^m\mathbf 1\{c=0\}.
  \label{eq:appB-multidimensional-character-sum}
\end{equation}
\end{lemma}

\begin{proof}
If $a=0$, every term in
Eq.~\eqref{eq:appB-one-dimensional-character-sum} is one.  If $a\ne0$, the
map $z\mapsto az$ permutes $\F_p$, and the sum is the geometric series over
all $p$th roots of unity.  Equation
\eqref{eq:appB-multidimensional-character-sum} follows by factorizing the sum
coordinatewise.
\end{proof}

\subsubsection{One-dimensional quadratic Gauss sums}
\label{appsubsec:one-dimensional-Gauss-sums}

For $a\in\F_p$, denote the one-dimensional quadratic character sum by
\begin{equation}
  G_p(a):=\sum_{z\in\F_p}\omega^{\frac12az^2}.
  \label{eq:appB-Gauss-sum-definition}
\end{equation}
Only its absolute value is needed here.

\begin{lemma}[Absolute value of a nondegenerate Gauss sum]
\label{lem:appB-one-dimensional-Gauss-absolute-value}
For $a\ne0$, the absolute value of the quadratic Gauss sum is
\begin{equation}
  \lvert G_p(a)\rvert=\sqrt p.
  \label{eq:appB-one-dimensional-Gauss-absolute-value}
\end{equation}
For $a=0$, one has $G_p(0)=p$.
\end{lemma}

\begin{proof}
The case $a=0$ is immediate.  Suppose $a\ne0$.  Expanding the squared absolute value gives
\begin{align}
  \lvert G_p(a)\rvert^2
  &=\sum_{x,y\in\F_p}
    \omega^{\frac12a(x^2-y^2)}.                                  \label{eq:appB-Gauss-square-1}
\end{align}
Because $2$ is invertible in $\F_p$, the sum can be evaluated using the
bijective change of variables
\begin{equation}
  r:=x-y,
  \quad
  s:=x+y.
  \label{eq:appB-Gauss-change-variables}
\end{equation}
This change of variables is a bijection of $\F_p^2$.  Since $x^2-y^2=rs$,
\begin{align}
  \lvert G_p(a)\rvert^2
  =\sum_{r,s\in\F_p}\omega^{\frac12ars}
  =\sum_{r\in\F_p}
    \left(\sum_{s\in\F_p}\omega^{(ar/2)s}\right)
=p,                                                            \label{eq:appB-Gauss-square-2},
\end{align}
where the last equality follows from
Lemma~\ref{lem:appB-character-orthogonality}: the inner sum is nonzero only
for $r=0$.  Taking square roots gives
Eq.~\eqref{eq:appB-one-dimensional-Gauss-absolute-value}.
\end{proof}

\subsubsection{Diagonalization of symmetric forms}
\label{appsubsec:symmetric-form-diagonalization}

The multidimensional Gauss sum will be reduced to the one-dimensional result
by congruence diagonalization of its symmetric quadratic form.

\begin{lemma}[Congruence diagonalization]
\label{lem:appB-congruence-diagonalization}
Let $D\in\Sym_n(\F_p)$ have rank $t$.  There exists
$R\in\operatorname{GL}(n,\F_p)$ and nonzero elements
$\lambda_1,\ldots,\lambda_t\in\F_p$ such that
\begin{equation}
  R^TDR
  =\operatorname{diag}(\lambda_1,\ldots,\lambda_t,
                        0,\ldots,0).
  \label{eq:appB-congruence-diagonalization}
\end{equation}
\end{lemma}

\begin{proof}
We argue by induction on $n$.  The statement is trivial for $D=0$.  Suppose
$D\ne0$.

If some diagonal entry $D_{jj}$ is nonzero, permute coordinates so that
$D_{11}\ne0$.  For each $\ell>1$, replace the $\ell$th basis vector by
\begin{equation}
  e_\ell-D_{11}^{-1}D_{1\ell}e_1.
  \label{eq:appB-diagonal-elimination}
\end{equation}
The changes of basis in Eq.~\eqref{eq:appB-diagonal-elimination}
remove all off-diagonal entries in the first row and column.
The resulting congruent matrix is the direct sum of a nonzero
one-dimensional block and a symmetric $(n-1)\times(n-1)$ block.

If every diagonal entry is zero, then some off-diagonal entry is nonzero, say
$D_{12}\ne0$.  Because the characteristic is not two, the vector
$e_1+e_2$ has
\begin{equation}
  (e_1+e_2)^TD(e_1+e_2)=2D_{12}\ne0.
  \label{eq:appB-create-nonzero-diagonal}
\end{equation}
After changing the first basis vector to $e_1+e_2$, we reduce to the preceding
case.  Iterating the argument diagonalizes $D$ by congruence.  The number of
nonzero diagonal entries equals $\rank(D)=t$.
\end{proof}

\subsubsection{Degenerate quadratic Gauss sums}
\label{appsubsec:degenerate-quadratic-Gauss-sums}

For a symmetric quadratic part $D\in\Sym_n(\F_p)$ and a linear term
$d\in\F_p^n$, consider the possibly degenerate Gauss sum
\begin{equation}
  \mathcal G(D,d)
  :=\sum_{x\in\F_p^n}
    \omega^{\frac12x^TDx+d^Tx}.
  \label{eq:appB-general-Gauss-sum}
\end{equation}

\begin{lemma}[Image--kernel orthogonality]
\label{lem:appB-image-kernel-orthogonality}
For each symmetric matrix $D\in\Sym_n(\F_p)$, its image is the
orthogonal complement of its kernel with respect to the standard
bilinear pairing:
\begin{equation}
  \im(D)=(\ker D)^\perp.
  \label{eq:appB-image-kernel-orthogonality}
\end{equation}
\end{lemma}

\begin{proof}
Vectors $y=Dx$ and $z\in\ker D$ satisfy the orthogonality identity
\begin{equation}
  z^Ty=z^TDx=(Dz)^Tx=0,
  \label{eq:appB-image-contained-orthogonal}
\end{equation}
so $\im(D)\subseteq(\ker D)^\perp$.  Both spaces have dimension
$\rank(D)$, hence they are equal.
\end{proof}

\begin{proposition}[Degenerate quadratic Gauss-sum dichotomy]
\label{prop:appB-degenerate-Gauss-dichotomy}
Let $D\in\Sym_n(\F_p)$ have rank $t$, and let $d\in\F_p^n$.  If the linear
term is incompatible with the quadratic image, translation along the kernel
forces
\begin{equation}
  \mathcal G(D,d)=0
  \quad\text{if }d\notin\im(D).
  \label{eq:appB-incompatible-Gauss-vanishes}
\end{equation}
If instead $d\in\im(D)$, completing the square leaves a rank-$t$
nondegenerate contribution and gives
\begin{equation}
  \lvert\mathcal G(D,d)\rvert=p^{n-t/2}.
  \label{eq:appB-compatible-Gauss-absolute-value}
\end{equation}
\end{proposition}

\begin{proof}
Assume first that $d\notin\im(D)$.  By
Lemma~\ref{lem:appB-image-kernel-orthogonality}, there exists
$z\in\ker D$ such that $d^Tz\ne0$.  Translating $x\mapsto x+z$ in
Eq.~\eqref{eq:appB-general-Gauss-sum} gives
\begin{align}
  \mathcal G(D,d)
  =\sum_{x\in\F_p^n}
    \omega^{\frac12(x+z)^TD(x+z)+d^T(x+z)}
=\omega^{d^Tz}\mathcal G(D,d),
  \label{eq:appB-translation-cancellation}
\end{align}
because $Dz=0$.  Since $\omega^{d^Tz}\ne1$, the sum must vanish.

Now suppose $d\in\im(D)$.  Choose a vector that represents the linear term
through the quadratic form:
\begin{equation}
  Dy=d.
  \label{eq:appB-choose-completion-vector}
\end{equation}
Set $x=z-y$.  Using the symmetry of $D$ and Eq.~\eqref{eq:appB-choose-completion-vector},
\begin{align}
  \frac12x^TDx+d^Tx
  =\frac12(z-y)^TD(z-y)+d^T(z-y)
  =\frac12z^TDz-\frac12y^TDy.
  \label{eq:appB-complete-the-square}
\end{align}
The translated sum therefore differs from the homogeneous quadratic sum only
by a unit-modulus phase:
\begin{equation}
  \mathcal G(D,d)
  =\omega^{-\frac12y^TDy}\mathcal G(D,0),
  \label{eq:appB-linear-term-removed}
\end{equation}
so it remains to evaluate $\lvert\mathcal G(D,0)\rvert$.

Choose $R$ as in Lemma~\ref{lem:appB-congruence-diagonalization} and make the
bijective change of variables $x=Rz$.  The resulting Gauss sum
factors into one-dimensional sums:
\begin{align}
  \mathcal G(D,0)
  &=\sum_{z\in\F_p^n}
    \omega^{\frac12\sum_{j=1}^t\lambda_jz_j^2}                   \\
  &=p^{n-t}\prod_{j=1}^tG_p(\lambda_j).
  \label{eq:appB-factorized-Gauss-sum}
\end{align}
By Lemma~\ref{lem:appB-one-dimensional-Gauss-absolute-value},
\begin{equation}
  \lvert\mathcal G(D,0)\rvert
  =p^{n-t}p^{t/2}
  =p^{n-t/2}.
  \label{eq:appB-factorized-Gauss-absolute-value}
\end{equation}
The phase in Eq.~\eqref{eq:appB-linear-term-removed} has unit absolute value,
which proves Eq.~\eqref{eq:appB-compatible-Gauss-absolute-value}.
\end{proof}

If $y'$ is another solution of $Dy'=d$, then $y'-y\in\ker D$.  Symmetry gives
\begin{equation}
  y'^TDy'-y^TDy=0,
  \label{eq:appB-completion-vector-phase-independent}
\end{equation}
so the phase in Eq.~\eqref{eq:appB-linear-term-removed} is independent of the
chosen solution.  This fact is not needed for the absolute-value calculation,
but it confirms that completion of the square is well defined.

\subsubsection{Proof of the rank-dependent overlap formula}
\label{appsubsec:proof-rank-dependent-overlap}

We now apply the Gauss-sum dichotomy to pairs of quadratic stabilizer states.
Their coordinate representation is
\begin{equation}
  |\psi_{A,b}\rangle
  =p^{-n/2}\sum_{x\in\F_p^n}
   \omega^{\frac12x^TAx+b^Tx}|x\rangle.
  \label{eq:appB-state-recall}
\end{equation}
For two labels $(A,b)$ and $(A',b')$, substituting these expansions converts
the overlap into the general Gauss sum:
\begin{align}
  \langle\psi_{A',b'}|\psi_{A,b}\rangle
  =p^{-n}\sum_{x\in\F_p^n}
    \omega^{\frac12x^T(A-A')x+(b-b')^Tx}                          
  =p^{-n}\mathcal G(\Delta A,\Delta b).
  \label{eq:appB-overlap-as-Gauss-sum}
\end{align}
Here the matrix and vector differences appearing in the Gauss sum are
\begin{equation}
  \Delta A:=A-A',
  \quad
  \Delta b:=b-b'.
  \label{eq:appB-overlap-differences}
\end{equation}
The relevant quadratic rank is
\begin{equation}
  t:=\rank(\Delta A).
  \label{eq:appB-overlap-rank}
\end{equation}
If the character difference is incompatible with the image of the quadratic
difference, the dichotomy gives an orthogonal pair:
\begin{equation}
  \langle\psi_{A',b'}|\psi_{A,b}\rangle=0.
  \label{eq:appB-overlap-zero}
\end{equation}
If the character difference is compatible, the overlap magnitude depends only
on the rank $t$:
\begin{equation}
  \lvert\langle\psi_{A',b'}|\psi_{A,b}\rangle\rvert
  =p^{-n}p^{n-t/2}
  =p^{-t/2}.
  \label{eq:appB-overlap-absolute-value}
\end{equation}
Squaring and combining the compatible and incompatible cases yields the
rank-dependent overlap formula
\begin{equation}
  \lvert\langle\psi_{A',b'}|\psi_{A,b}\rangle\rvert^2
  =
  \begin{cases}
    p^{-t},&\Delta b\in\im(\Delta A),\\
    0,&\Delta b\notin\im(\Delta A),
  \end{cases},
  \label{eq:appB-rank-dependent-overlap-final}
\end{equation}
which proves Proposition~\ref{prop:rank-dependent-overlap}.

\label{appsubsec:overlap-special-cases}
Two special cases are useful for interpreting the complete $m=n$ endpoint.

If $A'=A$, then $t=0$ and $\im(A-A')=\{0\}$.  The
squared overlap consequently equals the Kronecker delta:
\begin{equation}
  \lvert\langle\psi_{A,b'}|\psi_{A,b}\rangle\rvert^2
  =\delta_{b,b'},
  \label{eq:appB-full-label-overlap}
\end{equation}
recovering the fixed-$A$ orthonormality of
Lemma~\ref{lem:fixed-A-orthonormality}.

If $A-A'$ is invertible, then $t=n$ and
$\im(A-A')=\F_p^n$.  Every character difference is compatible, and
\begin{equation}
  \lvert\langle\psi_{A',b'}|\psi_{A,b}\rangle\rvert^2=p^{-n},
  \label{eq:appB-full-rank-mutually-unbiased}
\end{equation}
for all $b,b'$.  Thus the two bases $\mathcal B_A$ and $\mathcal B_{A'}$ are
mutually unbiased whenever $A-A'$ is nonsingular.


\subsection{Fixed-rank symmetric matrices}
\label{app:symmetric-rank-counts}

This subsection records the exact number of symmetric matrices of a
prescribed rank and derives the uniform estimate used in the completion
exponent.  We state the result over a finite field
$\F_q$ of odd cardinality $q$ and later specialize to $q=p$.  The only external enumerative input is the classical full-rank count of
MacWilliams~\cite{MacWilliams1969}.  The reduction from prescribed rank $t$ to
that full-rank enumeration is proved below within the manuscript: one chooses
the radical and then a nondegenerate symmetric form on the quotient.

\subsubsection{Gaussian binomial coefficients}
\label{appsubsec:Gaussian-binomial}

For $0\le t\le n$, denote the number of $t$-dimensional subspaces of
$\F_q^n$ by the Gaussian binomial coefficient
\begin{equation}
  \genfrac{[}{]}{0pt}{}{n}{t}_q
  :=\prod_{j=0}^{t-1}\frac{q^{n-j}-1}{q^{t-j}-1}.
  \label{eq:Gaussian-binomial-definition}
\end{equation}
By symmetry, it is also the number of codimension-$t$ subspaces.  Factoring out
the powers of $q$ gives
\begin{equation}
  \genfrac{[}{]}{0pt}{}{n}{t}_q
  =q^{t(n-t)}
   \frac{\prod_{j=0}^{t-1}(1-q^{-(n-j)})}
        {\prod_{j=1}^{t}(1-q^{-j})}.
  \label{eq:Gaussian-binomial-normalized}
\end{equation}

Define the positive constant
\begin{equation}
  \kappa_q
  :=\prod_{j=1}^{\infty}(1-q^{-j}).
  \label{eq:kappa-q-definition}
\end{equation}
The infinite product converges to a strictly positive number.  From
Eq.~\eqref{eq:Gaussian-binomial-normalized},
\begin{equation}
  \kappa_q q^{t(n-t)}
  \le \genfrac{[}{]}{0pt}{}{n}{t}_q
  \le \kappa_q^{-1}q^{t(n-t)}
  \quad(0\le t\le n).
  \label{eq:Gaussian-binomial-uniform-bounds}
\end{equation}
Indeed, both finite products in
Eq.~\eqref{eq:Gaussian-binomial-normalized} lie between $\kappa_q$ and $1$.
The bounds include $t=0$ and $t=n$, where the empty-product
convention applies.  Thus the constants do not deteriorate at either endpoint.
Taking logarithms, the entire product correction is bounded in absolute value
by $|\log_q\kappa_q|$, independently of $n$ and $t$; after division by $n^2$
it is therefore uniformly $O(n^{-2})$.

\subsubsection{Nonsingular symmetric matrices}
\label{appsubsec:nonsingular-symmetric-count}

Let $S_q(t)$ denote the number of nonsingular symmetric $t\times t$
matrices over $\F_q$:
\begin{equation}
  S_q(t)
  :=\lvert\{B\in\Sym_t(\F_q):\det B\ne0\}\rvert.
  \label{eq:nonsingular-symmetric-count-definition}
\end{equation}

\begin{theorem}[MacWilliams full-rank count]
\label{thm:MacWilliams-full-rank-count}
For odd $q$ and every nonnegative integer $t$, the full-rank symmetric-matrix count is
\begin{equation}
  S_q(t)
  =q^{t(t+1)/2}
   \prod_{\substack{1\le j\le t\\ j\ \mathrm{odd}}}
   (1-q^{-j}).
  \label{eq:MacWilliams-full-rank-count}
\end{equation}
The empty product for $t=0$ is one.
\end{theorem}

The cited enumeration is used here only for fields of odd cardinality;
no characteristic-two specialization is invoked.

Theorem~\ref{thm:MacWilliams-full-rank-count} is the full-rank
specialization of MacWilliams' rank enumeration in
\cite[Theorem~2, pp.~154--155]{MacWilliams1969}; simplifying that
specialization gives Eq.~\eqref{eq:MacWilliams-full-rank-count}.  MacWilliams'
Theorem~2 is not restricted to characteristic two, while the present theorem
uses it only for fields of odd cardinality.  The prescribed-rank formula below
is derived separately here by selecting the radical and pulling back a
nondegenerate form from the quotient.  For orientation, the first cases
are
\begin{align}
  S_q(0)&=1,                                                     \label{eq:S-q-zero}\\
  S_q(1)&=q-1,                                                   \label{eq:S-q-one}\\
  S_q(2)&=q^2(q-1),                                             \label{eq:S-q-two}\\
  S_q(3)&=q^2(q-1)(q^3-1).                                     \label{eq:S-q-three}
\end{align}

Let
\begin{equation}
  \kappa_q^{\mathrm{odd}}
  :=\prod_{m=0}^{\infty}(1-q^{-(2m+1)}).
  \label{eq:kappa-q-odd-definition}
\end{equation}
Then $\kappa_q\le\kappa_q^{\mathrm{odd}}<1$, and
Eq.~\eqref{eq:MacWilliams-full-rank-count} immediately yields
\begin{equation}
  \kappa_q^{\mathrm{odd}}q^{t(t+1)/2}
  \le S_q(t)
  \le q^{t(t+1)/2}.
  \label{eq:nonsingular-symmetric-uniform-bounds}
\end{equation}

\subsubsection{Exact rank-\texorpdfstring{$t$}{t} enumeration}
\label{appsubsec:exact-rank-t-enumeration}

For the prescribed-rank enumeration, write
\begin{equation}
  N_{\mathrm{sym}}(n,t;q)
  :=\lvert\{D\in\Sym_n(\F_q):\rank(D)=t\}\rvert.
  \label{eq:rank-t-symmetric-count-q}
\end{equation}

\begin{theorem}[Symmetric matrices of rank \texorpdfstring{$t$}{t}]
\label{thm:exact-symmetric-rank-count}
For odd $q$ and $0\le t\le n$, the number of symmetric matrices of rank $t$ is
\begin{align}
  N_{\mathrm{sym}}(n,t;q)
  &=\genfrac{[}{]}{0pt}{}{n}{t}_q S_q(t)                                       \label{eq:rank-t-count-quotient-form}\\
  &=\genfrac{[}{]}{0pt}{}{n}{t}_q q^{t(t+1)/2}
    \prod_{\substack{1\le j\le t\\j\ \mathrm{odd}}}
    (1-q^{-j}).                                                  \label{eq:rank-t-count-exact}
\end{align}
\end{theorem}

\begin{proof}
A symmetric matrix $D$ defines the following symmetric bilinear form
on $V:=\F_q^n$:
\begin{equation}
  \beta_D(x,y):=x^TDy.
  \label{eq:bilinear-form-from-matrix}
\end{equation}
Its radical is
\begin{equation}
  \operatorname{rad}(\beta_D)
  =\{x\in V:\beta_D(x,y)=0\ \text{for all }y\in V\}
  =\ker D.
  \label{eq:radical-equals-kernel}
\end{equation}
Thus $D$ has rank $t$ exactly when its radical $K$ has codimension $t$.  The
Gaussian binomial coefficient counts the possible radicals:
there are $\genfrac{[}{]}{0pt}{}{n}{t}_q$ choices for $K$.

Fix one such $K$.  A symmetric form on $V$ whose radical contains $K$ descends
uniquely to a symmetric form on the quotient $V/K$ through
\begin{equation}
  \overline\beta(x+K,y+K):=\beta(x,y).
  \label{eq:form-descends-to-quotient}
\end{equation}
Its radical is exactly $K$ if and only if $\overline\beta$ is nondegenerate.
Conversely, every nondegenerate symmetric form on $V/K$ pulls back uniquely to
a symmetric form on $V$ with radical $K$.  After choosing any basis of the
$t$-dimensional quotient, the number of such forms is $S_q(t)$; this number is
basis independent.  The choices of the radical and the nondegenerate quotient form are independent.
Multiplying their counts proves
Eq.~\eqref{eq:rank-t-count-quotient-form}.  Substitution of
Eq.~\eqref{eq:MacWilliams-full-rank-count} gives
Eq.~\eqref{eq:rank-t-count-exact}.
\end{proof}

For the smallest dimension $n=1$, the formula gives
\begin{equation}
  N_{\mathrm{sym}}(1,0;q)=1,
  \quad
  N_{\mathrm{sym}}(1,1;q)=q-1,
  \label{eq:rank-count-n-one-check}
\end{equation}
which sums to $q$, the number of $1\times1$ symmetric matrices.  More
generally, the rank counts sum to the total number of symmetric
$n\times n$ matrices:
\begin{equation}
  \sum_{t=0}^nN_{\mathrm{sym}}(n,t;q)
  =q^{n(n+1)/2}.
  \label{eq:rank-count-sums-to-all-symmetric}
\end{equation}

\subsubsection{Uniform exponent bounds}
\label{appsubsec:uniform-rank-exponent}

The exact rank formula separates into a subspace count and a nonsingular-form
count.  Applying their uniform bounds gives constants independent of both
$n$ and $t$.

\begin{corollary}[Uniform rank-count bounds]
\label{cor:uniform-symmetric-rank-count}
For fixed odd $q$, the rank count obeys the following two-sided bound for every integer $0\le t\le n$:
\begin{equation}
  \kappa_q\kappa_q^{\mathrm{odd}}
  q^{t(2n-t+1)/2}
  \le
  N_{\mathrm{sym}}(n,t;q)
  \le
  \kappa_q^{-1}
  q^{t(2n-t+1)/2}.
  \label{eq:uniform-rank-count-multiplicative}
\end{equation}
Taking base-$q$ logarithms gives the following additive exponent
estimate, uniformly over $n$ and $t$:
\begin{equation}
  \log_qN_{\mathrm{sym}}(n,t;q)
  =\frac{t(2n-t+1)}2+O_q(1).
  \label{eq:uniform-rank-count-logarithmic}
\end{equation}
\end{corollary}

\begin{proof}
The exponent is the sum
\begin{equation}
  t(n-t)+\frac{t(t+1)}2
  =\frac{t(2n-t+1)}2.
  \label{eq:rank-count-exponent-identity}
\end{equation}
The multiplicative factors are bounded by
Eqs.~\eqref{eq:Gaussian-binomial-uniform-bounds} and
\eqref{eq:nonsingular-symmetric-uniform-bounds}, proving
Eq.~\eqref{eq:uniform-rank-count-multiplicative}.  Taking logarithms gives
Eq.~\eqref{eq:uniform-rank-count-logarithmic}.
\end{proof}

For the paper's odd-prime notation, set $q=p$.  Corollary
\ref{cor:uniform-symmetric-rank-count} gives the following uniform logarithmic form of the rank-count estimate:
\begin{equation}
  \log_pN_{\mathrm{sym}}(n,t;p)
  =\frac{t(2n-t+1)}2+O_p(1).
  \label{eq:paper-uniform-rank-count}
\end{equation}
The error term in Eq.~\eqref{eq:paper-uniform-rank-count} is bounded
uniformly over the admissible ranks $t$, with the prime $p$ fixed.
Summing over polynomially many rank indices adds at most
$O_p(\log n)$ to the logarithm of the finite converse count;
it does not affect the exponent on the $n^2$ scale.  The
direct argument uses the separate Lagrangian-intersection count in
Lemma~\ref{lem:lagrangian-intersection-shells}.

\section{Characteristic-two auxiliary algebra}
\label{app:characteristic-two-algebra}
This appendix evaluates characteristic-two overlaps and Fourier energy
as an independent consistency check.  The $\mathbb Z_4$-valued quadratic
phases give the binary counterpart of the odd-prime Gauss-sum calculation:
we evaluate the overlap phase and magnitude, then sum the powers of the
magnitude over the binary rank classes.

The main direct proof uses the unrestricted transitive pure-stabilizer
PGM in Theorem~\ref{thm:unrestricted-stabilizer-direct} and the
Lagrangian-intersection count in
Lemma~\ref{lem:lagrangian-intersection-shells}.  The finite binary
converse is proved in Appendix~\ref{app:binary-finite-structural-proof},
and the refined verification comparison uses the binary score spectrum in
Lemma~\ref{lem:integrated-qubit-score-spectrum}.

\subsection{Radical compatibility and overlap magnitude}
\label{subsubsec:z4-radical-overlap}
For $g\in G_{n,2}$, let $q_g:\F_2^n\to\mathbb Z_4$ be the enhanced quadratic
form in Eq.~\eqref{eq:integrated-qubit-phase}, and let $B_g$ be its symmetric
binary polarization.  The radical, restricted enhancement, and bilinear rank
are
\begin{equation}
 R_g:=\operatorname{rad}B_g,
 \quad \eta_g:=q_g/2|_{R_g}:R_g\to\F_2,
 \quad t_g:=\rank B_g.
 \label{eq:integrated-qubit-gauss-invariants}
\end{equation}
If $\eta_g\ne0$, translation by a radical vector introduces a nontrivial
phase and forces the quadratic Gauss sum to vanish.  If $\eta_g=0$, the
enhancement descends to a nondegenerate form $\bar q_g$ on the quotient
$\F_2^n/R_g$, giving the complex overlap
\begin{equation}
 \langle\psi_{0,0}^{(2)}|\psi_g^{(2)}\rangle
 =2^{-t_g/2}
 \exp\!\left(\frac{\pi i}{4}\operatorname{Br}(\bar q_g)\right),
 \label{eq:integrated-qubit-overlap-classification}
\end{equation}
where $\operatorname{Br}(\bar q_g)\in\mathbb Z_8$ is the Brown invariant,
normalized so that the Gauss-sum phase is
$\exp(\pi i\operatorname{Br}(\bar q_g)/4)$
\cite{Brown1972Kervaire,Taylor2022GaussSums}.  This invariant records the
normalized Gauss-sum phase; it does not by itself classify the quadratic
enhancement up to isometry.
Thus the Brown invariant controls the characteristic-two phase, while the
nonzero magnitude depends only on $t_g$.  Fourier energy uses the $2k$th power
of the modulus, so the phase cancels exactly:
\begin{equation}
 \left|2^{-t_g/2}
 \exp\!\left(\frac{\pi i}{4}\operatorname{Br}(\bar q_g)\right)\right|^{2k}
 =2^{-kt_g};
 \label{eq:integrated-qubit-Brown-phase-cancels}
\end{equation}
the Brown phase disappears at this exact step, before the fixed-rank classes are
summed.  It specifies the phase of the complex overlap, but it has no effect on
the energy exponent.

\subsection{Rank enumeration and Fourier-energy decay}
\label{subsubsec:z4-rank-energy-decay}

Fix a symmetric binary form $B$ of rank $t$ on $V=\F_2^n$, and put
$R=\operatorname{rad}B$, so that $\dim R=n-t$.  Choose one enhancement
$q_\star$ of $B$.  Eq.~\eqref{eq:integrated-qubit-phase} parametrizes all
$2^n$ enhancements of $B$ uniquely as
\[
 q_\ell=q_\star+2\ell,\qquad \ell\in V^*.
\]
On $R$, the polarization identity in
Eq.~\eqref{eq:integrated-qubit-polarization} makes $q_\star$ additive
with even values.  Thus $\eta_\star:=(q_\star/2)|_R$ is an
$\F_2$-linear functional, and the restricted enhancement
$\eta_\ell:=(q_\ell/2)|_R$ satisfies
\[
 \eta_\ell=\eta_\star+\ell|_R.
\]
Every linear functional on $R$ extends to $V$, so the restriction map
$\operatorname{res}_R:V^*\to R^*$, defined by
$\operatorname{res}_R(\ell)=\ell|_R$, is surjective.
The condition $\eta_\ell=0$ selects the fiber over $\eta_\star$.
Since $\dim R=n-t$, this fiber has cardinality
\[
 |\operatorname{res}_R^{-1}(\eta_\star)|
 =2^{\dim V-\dim R}=2^t.
\]
Hence exactly $2^t$ enhancements have nonzero overlap.  Each has overlap
magnitude $2^{-t/2}$, so their $2k$th powers contribute
$2^t2^{-kt}=2^{-(k-1)t}$ for each fixed $B$.

At rank zero, the unique enhancement with trivial radical restriction is
$q=0$, which corresponds to the identity label $g=0$ and is omitted.
All other enhancements of the zero form have zero overlap.
Grouping the rank classes therefore gives, for every integer $k\ge1$,
the exact off-identity Fourier energy
\begin{equation}
 \mathcal E_{n,k}^{(2)}
 :=\sum_{g\in G_{n,2}\setminus\{0\}}
 |\langle\psi_{0,0}^{(2)}|\psi_g^{(2)}\rangle|^{2k}
 =\sum_{t=1}^nN_{\rm sym}(n,t;2)2^{-(k-1)t}.
 \label{eq:integrated-qubit-exact-fourier-energy}
\end{equation}
To bound this rank sum, specialize
Lemma~\ref{lem:integrated-binary-symmetric-rank-bound} to $d=n$:
\begin{equation}
 N_{\rm sym}(n,t;2)\le C_2\,2^{t(2n-t+1)/2}\le C_2\,2^{nt}.
 \label{eq:integrated-qubit-rank-bound}
\end{equation}
Combining this count with the factor $2^{-(k-1)t}$ makes the rank sum
geometric.  Hence, for every fixed $0<\delta<\alpha-1$ and
$k=\alpha n+o(n)$,
\begin{equation}
 \mathcal E_{n,k}^{(2)}=O(2^{-\delta n}).
 \label{eq:integrated-qubit-energy-decay}
\end{equation}

\section*{Statements and Declarations}
\noindent\textbf{Funding.}
M.H. was supported in part by the General R\&D Projects of the 1+1+1
CUHK--CUHK(SZ)--GDST Joint Collaboration Fund (Grant No.~GRDP2025-022) and
the Guangdong Provincial Quantum Science Strategic Initiative
(Grant No.~GDZX2505003).
Y.L. was supported in part by the National Natural Science Foundation of China under Grant No.~12475023, Dushi Program, 
and a startup funding from Yau Mathematical Sciences Center.

\noindent\textbf{Competing interests.} 
The authors have no relevant financial or non-financial interests to disclose.

\noindent\textbf{Code availability.}
The Lean~4 source code, machine-verification records, and
reproducibility instructions are available in version~1.0.0
of the accompanying repository~\cite{stabilizerIdentificationLeanV1}.

\end{document}